\documentclass[10pt]{article}
\usepackage{amsmath}
\usepackage{amsthm}
\usepackage{amssymb}
\usepackage{fullpage}
\usepackage{cite}
\usepackage[linesnumbered,nofillcomment,noend]{algorithm2e}
\usepackage{hyperref}
\usepackage{tikz,ifthen,calculator,xcolor,comment}
\usepackage{rotating}
\usepackage{subcaption}
\usepackage{enumitem}
\setlist{itemsep=0mm}
\setlist{topsep=1mm}
\setlist{parsep=0mm}
\usepackage{apxproof}
\usepackage{pifont}
\input{commands-tam.sty}

\newtheorem{theorem}{Theorem}
\newtheorem{definition}{Definition}

\newtheorem{lemma}{Lemma}
\newtheorem{corollary}{Corollary}
\newtheorem{observation}{Observation}

\newcommand{\mydirection}{->}
\newcommand{\mystartpos}{0}
\newcommand{\myendpos}{0}
\newcommand{\myaboveorbelow}{above}
\newcommand{\mylevel}{0}
\newcommand{\myanglestart}{0}
\newcommand{\myanglemidin}{0}
\newcommand{\myanglemidout}{0}
\newcommand{\myangleend}{0}
\newcommand{\piw}{\pi_{\textmd{win}}}
\newcommand{\pil}{\pi_{\textmd{lose}}}

\newcommand{\arrow}[4]{
  \SUBTRACT{#1}{1}{\mystartpos}
  \SUBTRACT{#2}{1}{\myendpos}
  \ifthenelse{#3 < 0}{      
    \renewcommand{\myaboveorbelow}{below}
    \SUBTRACT{0}{#3}{\mylevel}
    \renewcommand{\myanglestart}{270}
    \ifthenelse{#1 < #2}{
      \renewcommand{\mydirection}{->}
      \renewcommand{\myanglemidin}{180}
      \renewcommand{\myanglemidout}{0}
    } {
      \renewcommand{\mydirection}{<-}
      \renewcommand{\myanglemidin}{0}
      \renewcommand{\myanglemidout}{180}
    }
    \renewcommand{\myangleend}{270}
  } {
    \renewcommand{\myaboveorbelow}{above}
    \renewcommand{\mylevel}{#3}
    \renewcommand{\myanglestart}{90}
    \ifthenelse{#1 < #2}{
      \renewcommand{\mydirection}{->}
      \renewcommand{\myanglemidin}{180}
      \renewcommand{\myanglemidout}{0}
    } {
      \renewcommand{\mydirection}{<-}
      \renewcommand{\myanglemidin}{0}
      \renewcommand{\myanglemidout}{180}
    }
    \renewcommand{\myangleend}{90}
  }
  
  \MULTIPLY{\mystartpos}{10}{\tmp}
  \ADD{6}{\tmp}{\mystartx}
  \ADD{6}{0}{\mystarty}
  
  \MULTIPLY{\myendpos}{10}{\tmp}
  \ADD{6}{\tmp}{\myendx}
  \ADD{6}{0}{\myendy}
  
  \ADD{\mystartx}{\myendx}{\tmp}
  \DIVIDE{\tmp}{2}{\mymidx}
  \MULTIPLY{\mylevel}{5}{\tmp}
  \ADD{\tmp}{6}{\mymidy}

  \ifthenelse{\equal{\myaboveorbelow}{below}}{      
    \SUBTRACT{0}{\mystarty}{\mystarty}
    \SUBTRACT{0}{\myendy}{\myendy}
    \SUBTRACT{0}{\mymidy}{\mymidy}
  } {}

  \draw[->] (\mystartx,\mystarty) to
       [out=\myanglestart,in=\myanglemidin] (\mymidx,\mymidy)
       to [out=\myanglemidout,in=\myangleend] (\myendx,\myendy);           
}

\newcommand{\goingforward}[1]{
  \vspace*{2mm}\noindent\fbox{
  \begin{minipage}{0.97\linewidth}
    Going forward, #1
  \end{minipage}
  }\vspace*{2mm}
}

\begin{document}

\title{\textbf{Sequential non-determinism in tile self-assembly: a general framework and an application to efficient temperature-1 self-assembly of squares}}

\author{%
David Furcy\thanks{Computer Science Department, University of Wisconsin Oshkosh, Oshkosh, WI 54901, USA,\protect\url{furcyd@uwosh.edu}.}
\and
Scott M. Summers\thanks{Computer Science Department, University of Wisconsin Oshkosh, Oshkosh, WI 54901, USA,\protect\url{summerss@uwosh.edu}. This author's research was supported in part by University of Wisconsin Oshkosh Research Sabbatical (S581) during Fall 2023. }
}


\date{}
\maketitle

\begin{abstract}
In this paper, we work in a 2D version of the probabilistic variant of Winfree's abstract Tile Assembly Model defined by Chandran, Gopalkrishnan and Reif (SICOMP 2012) in which attaching tiles are sampled uniformly with replacement.
First, we develop a framework called ``sequential non-determinism'' for analyzing the probabilistic correctness of a non-deterministic, temperature-1 tile assembly system (TAS) in which most (but not all) tile attachments are deterministic and the non-deterministic attachments always occur in a specific order.
Our main sequential non-determinism result equates the probabilistic correctness of such a TAS to a finite product of probabilities, each of which 1) corresponds to the probability of the correct type of tile attaching at a point where it is possible for two different types to attach, and 2) ignores all other tile attachments that do not affect the non-deterministic attachment.
We then show that sequential non-determinism allows for efficient and geometrically expressive self-assembly.
To that end, we constructively prove that for any positive integer $N$ and any real $\delta \in (0,1)$, there exists a TAS that self-assembles into an $N \times N$ square with probability at least $1 - \delta$ using only $O\left( \log N + \log \frac{1}{\delta} \right)$ types of tiles.
Our bound improves upon the previous state-of-the-art bound for this problem by Cook, Fu and Schweller (SODA 2011). 
\end{abstract}

\section{Introduction}
\label{sec:introduction}
The notion of ``self-assembly'' is generally said to involve seemingly simple, fundamental components, evolving through instances of local interaction, to a terminal assembly whose complexity is greater than the sum of its parts. Self-assembly is ubiquitous in nature, e.g., atoms self-assemble into molecules, molecules self-assemble into macromolecules, etc. 

While self-assembly is ubiquitous throughout the natural world, researchers have recently begun investigating the extent to which the power of self-assembly can be harnessed for the systematic creation of atomically-precise computational, biomedical and mechanical devices at the nano-scale. For example, in the early 1980s, Nadrian Seeman introduced a revolutionary experimental technique for controlling nano-scale self-assembly known as ``DNA tile self-assembly'' \cite{Seem82} thus initiating the scientific study of algorithmic self-assembly. Erik Winfree's abstract Tile Assembly Model (aTAM) \cite{Winf98} is a discrete mathematical model of DNA tile self-assembly, which we describe briefly in the next paragraph. 

In the aTAM, a DNA tile is represented as an un-rotatable unit square {\em tile type}, each side of which has a corresponding {\em glue} comprising a string {\em label} and a positive integer {\em strength}. A {\em tile} is an instance of a tile type that is {\em placed} at some integer point in the 2D Cartesian coordinate space. The aTAM restricts the number of tile types in a tile set to be finite but infinitely many tiles of the same type may be placed. If two tiles are placed next to each other and their opposing glues match, then the tiles permanently {\em bind} with the strength of the glue, thus creating a tile {\em assembly}, i.e., a mapping of integer points to tile types, of size 2. Before the process of self-assembly begins, a positive integer {\em temperature} value is chosen, which is usually 1 or 2 and a fixed {\em seed} tile is placed at a designated integer point defining the initial seed-containing assembly. The tile set, along with the associated temperature and placement of the seed tile make up a {\em tile assembly system} (TAS).
Self-assembly in the aTAM is modeled by an {\em assembly sequence} of the TAS, which is a series of {\em tile attachment steps} such that in each step, the assembly sequence {\em attaches} a tile if the tile can be placed at an unoccupied integer point and bind to a tile of the seed-containing assembly with total strength at least the temperature.
If more than one tile may attach in a single step, then one tile is chosen non-deterministically to attach. 
Each tile attachment step in an assembly sequence {\em results} in a new assembly one tile larger than the assembly in the previous step, and we say the resulting assembly is {\em produced} by the TAS. The steps of an assembly sequence continue until a {\em terminal assembly} is reached to which no further tiles can bind. Figure~\ref{fig:intro-example1} depicts an example of a {\em non-deterministic}, temperature-1 TAS $\mathcal{T}$ in which two different terminal assemblies are produced. 
\begin{figure}[!ht]
  \centering
   \includegraphics[width=.65\linewidth]{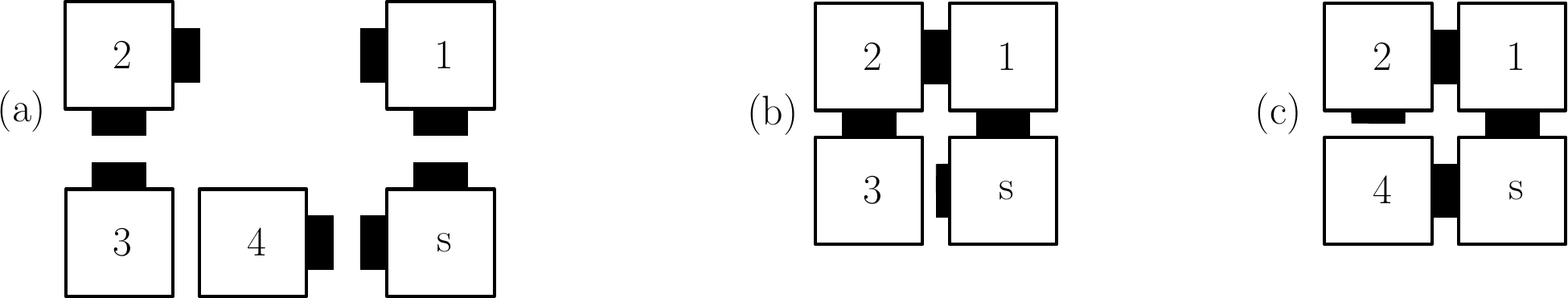}
  \caption{\label{fig:intro-example1} Example of a TAS $\mathcal{T}$ with the tile set shown in (a), in which the temperature is 1 and s is the seed tile. Sub-figures (b) and (c) show assembly sequences of $\mathcal{T}$ that result in two different terminal assemblies.}
\end{figure}

Observe that temperature-1 self-assembly, which involves a TAS whose temperature is set to 1, does not possess a synchronization mechanism that can be used to enforce cooperative binding, i.e., a tile attaching at a point only after and to at least two of its neighboring tiles attach. 
However, temperature-2 self-assembly provides a mechanism to enforce cooperative binding, which can be leveraged to exhibit 1) optimal self-assembly of squares \cite{RotWin00,AdlemanCGH01} and algorithmically-specified finite shapes \cite{SolWin07}, 2) efficient self-assembly of rectangles \cite{AGKS05g}, and 3) even an intrinsically universal tile set, i.e., a tile set that can, for an arbitrary input TAS, be configured to carry out a natural simulation of the input TAS but at a higher scale factor \cite{USA}.
With the lack of a cooperative binding mechanism in temperature-1 self-assembly, it would seem natural to speculate that the computational and geometric expressiveness of temperature-1 self-assembly is fundamentally limited.
Such speculation was formalized back in 2008 by Doty, Patitz and Summers \cite{LSAT1csp} through their pumpability conjecture in which they claimed that the only kinds of patterns that result from temperature-1 self-assembly in which a unique terminal assembly is produced effectively correspond to regular languages.
Even before the pumpability conjecture, temperature-1 self-assembly had been shown by Rothemund and Winfree \cite{RotWin00} to hinder the efficient self-assembly of shapes when tile assemblies are required to be fully connected, i.e., if two tiles are neighbors, then they bind to each other. 
Manuch, Stacho and Stoll \cite{ManuchSS10} generalized this result to temperature-1 self-assembly that always results in tile assemblies that merely contain no glue mismatches. 
Subsequent results showing that temperature-1 self-assembly is neither intrinsically universal  \cite{WindowMovieLemma,MeunierW17} nor capable of bounded Turing computation \cite{MeunierW17} continued to elucidate the supposed weakness of not being able to prevent non-cooperative binding.
A major breakthrough by Meunier, Regnault and Woods \cite{MeunierRW20} established a general pumping lemma for temperature-1 self-assembly.
Although it did not directly affirm the pumpability conjecture, Meunier and Regnault \cite{MeunierR21} used it in their proof of the pumpability conjecture.
Yet the lack of enforceable cooperative binding is not solely to blame for the weakness of temperature-1 self-assembly.
In fact, generalizations of the aTAM such as \cite{DDFIRSS07,SFTSAFT,FeketeHPRS15,FuPSS12,GilbertHPR16} enjoy a striking resemblance of computational and geometric expressiveness to that of temperature-2 self-assembly. 
This is also the case when the aTAM is only mildly generalized, e.g., by allowing: 1) the use of a unique negative glue \cite{SingleNegative}, 2) duple tiles \cite{HendricksPRS18}, 3) the attachment of cubic tiles in at most one additional plane \cite{CookFuSch11,jFurcyMickaSummers,FurcySummersWendlandtTCS,FurcyS18,FurcySW23,FurcySV23}, or 4) non-determinism \cite{CookFuSch11} in which multiple terminal assemblies may be produced.
The results in this paper are motivated by the problem of computing the probability that a particular terminal assembly is produced by non-deterministic, temperature-1 self-assembly. To study this problem for a given non-deterministic, temperature-1 TAS $\mathcal{T}$, we use the Markov chain $\mathcal{M}_{\mathcal{T}}$, in which the states are the assembly sequences of $\mathcal{T}$ and there is a transition in $\mathcal{M}_{\mathcal{T}}$ from one assembly sequence to another with a positive probability equal to one over the number of distinct assemblies that can result from attaching a single tile to the result of the former. 
This is the {\em tile concentration programming} aTAM subject to the restriction that all tile types have the same concentration (see \cite{KaoS08,Dot10} for a more general and thorough development of the tile concentration programming model). Note that this is also the same as the probabilistic aTAM (PTAM), which was defined by Chandran, Gopalkrishnan, and Reif \cite{ChandranGR12}, but extended to 2D, which we will refer to as the 2DPTAM going forward. 
It is easy to see that $\mathcal{M}_{\mathcal{T}}$ is a rooted tree. Then, the {\em probability} that a given TAS $\mathcal{T}$ produces a terminal assembly is the probability of $\mathcal{M}_{\mathcal{T}}$ reaching any state that corresponds to an assembly sequence of $\mathcal{T}$ that results in the terminal assembly. For example, the probability that the TAS shown in Figure~\ref{fig:intro-example1} produces the terminal assembly of part (c) can be easily computed as $\frac{7}{8}$ and the corresponding $\mathcal{M}_{\mathcal{T}}$ is shown in Figure~\ref{fig:intro-example1-noAB}. In that and all subsequent figures, all edge label values equal to 1 are omitted (but implied). 

\begin{figure}[ht!]
	\centering
	\includegraphics[width=.8in]{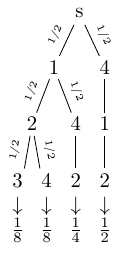}
	\caption{\label{fig:intro-example1-noAB} $\mathcal{M}_{\mathcal{T}}$ corresponding to the TAS in Figure~\ref{fig:intro-example1}. Each path from the root to a leaf node corresponds to an assembly sequence of $\mathcal{T}$. For the sake of simplification, the label of each node is merely the type of the last tile to attach in the assembly sequence to which that node corresponds, rather than the full assembly sequence, which can easily be reconstructed by following (and then reversing) the path back up to the root. Each leaf node corresponds to an assembly sequence that results in  a terminal assembly. The number under each leaf node is the probability of the corresponding sequence being produced by $\mathcal{T}$.}
\end{figure}

While $\mathcal{M}_{\mathcal{T}}$ can be equivalently formulated in terms of a directed acyclic graph (e.g., \cite{CookFuSch11}), we use a tree to aid in the visualization of some of our results.

Computing the probability that $\mathcal{T}$ from Figure~\ref{fig:intro-example1} produces one of its terminal assemblies is easy. However, consider the TAS shown in Figure~\ref{fig:intro-example2}, which we will denote as $\mathcal{T}$, and is the same as the $\mathcal{T}$ from Figure~\ref{fig:intro-example1} but with the addition of the A and B tile types, along with corresponding east-facing glues to the 1 and s tiles.

\begin{figure}[!ht]
  \centering
  \includegraphics[width=.65\linewidth]{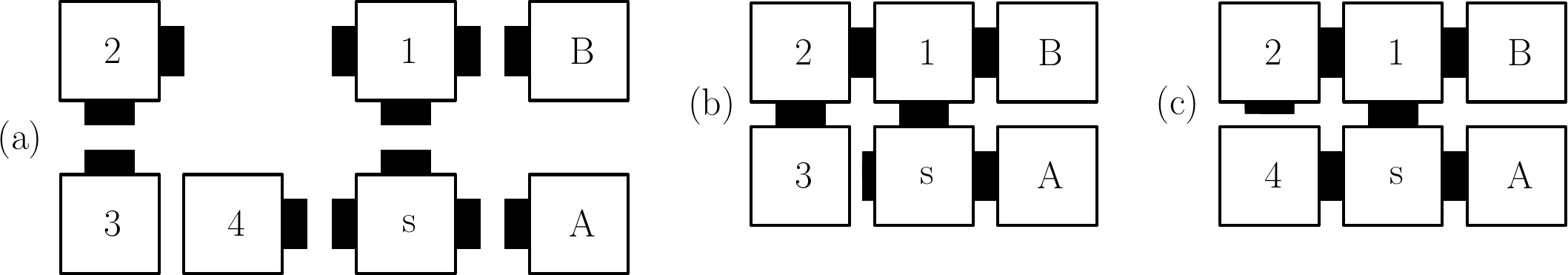}
  \caption{\label{fig:intro-example2} Example of another TAS $\mathcal{T}$ with the tile set shown in (a), in which the temperature is 1 and s is the seed tile. Sub-figures (b) and (c) show assembly sequences of $\mathcal{T}$ that result in two different terminal assemblies.}
\end{figure}

The presence of the A and B tiles in $\mathcal{T}$ does not affect, but rather complicates the computation of, the probability that $\mathcal{T}$ produces one of its terminal assemblies. To see this, Figure~\ref{fig:lnd-tree0} shows $\mathcal{M}_{\mathcal{T}}$, from which the probability that $\mathcal{T}$ produces the terminal assembly that contains the 4 tile can be derived.

\begin{figure}[ht!]
	\centering
	\includegraphics[width=\linewidth]{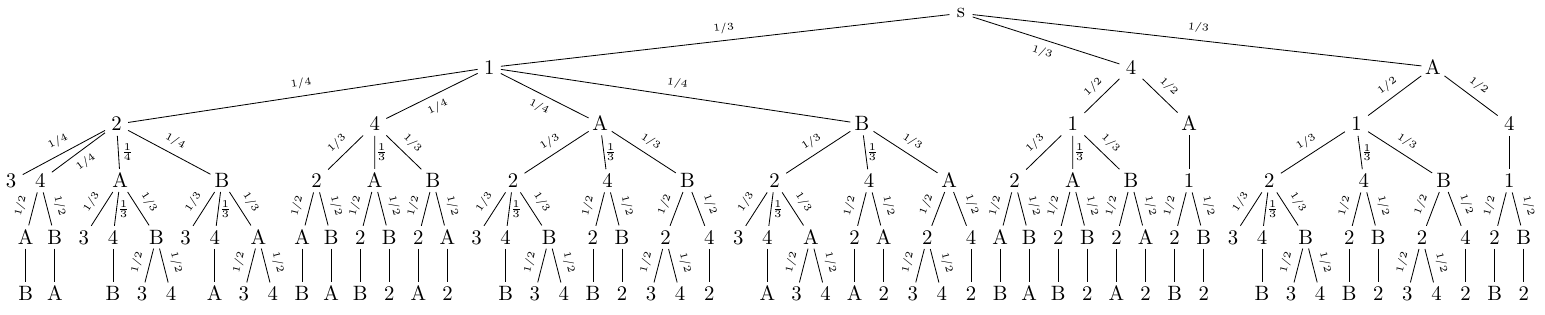}
  	\caption{\label{fig:lnd-tree0} Markov chain $\mathcal{M}_{\mathcal{T}}$, where $\mathcal{T}$ is the TAS from Figure~\ref{fig:intro-example2}. Note that this tree is not quite complete because it only contains the assembly sequences that result in the terminal assembly that contains the 4 tile. All of the sub-trees whose root is labeled with a $3$ and does not appear on the lowest level of the tree have been pruned below their root node (e.g., the leftmost node in the figure).}
\end{figure}

Based on $\mathcal{M}_{\mathcal{T}}$ shown in Figure~\ref{fig:lnd-tree0}, it is easy (although tedious) to see that the probability that $\mathcal{T}$ produces the terminal assembly that contains the 4 tile is:
\begin{multline*}
\begin{gathered}
\frac{1}{3} \frac{1}{4} \frac{1}{4} \frac{1}{2} + 				
\frac{1}{3} \frac{1}{4} \frac{1}{4} \frac{1}{2} + 				
\frac{1}{3} \frac{1}{4} \frac{1}{4} \frac{1}{3} + 				
\frac{1}{3} \frac{1}{4} \frac{1}{4} \frac{1}{3} \frac{1}{2} + 	
\frac{1}{3} \frac{1}{4} \frac{1}{4} \frac{1}{3} + 				
\frac{1}{3} \frac{1}{4} \frac{1}{4} \frac{1}{3} \frac{1}{2} +	
\frac{1}{3} \frac{1}{4} \frac{1}{3} \frac{1}{2} + 				
\frac{1}{3} \frac{1}{4} \frac{1}{3} \frac{1}{2} + 				
\frac{1}{3} \frac{1}{4} \frac{1}{3} \frac{1}{2} + 				
\frac{1}{3} \frac{1}{4} \frac{1}{3} \frac{1}{2} + 	\\			
\frac{1}{3} \frac{1}{4} \frac{1}{3} \frac{1}{2} + 				
\frac{1}{3} \frac{1}{4} \frac{1}{3} \frac{1}{2} + 				
\frac{1}{3} \frac{1}{4} \frac{1}{3} \frac{1}{3} + 				
\frac{1}{3} \frac{1}{4} \frac{1}{3} \frac{1}{3} \frac{1}{2} +	
\frac{1}{3} \frac{1}{4} \frac{1}{3} \frac{1}{2} + 				
\frac{1}{3} \frac{1}{4} \frac{1}{3} \frac{1}{2} +				
\frac{1}{3} \frac{1}{4} \frac{1}{3} \frac{1}{2} \frac{1}{2} +	
\frac{1}{3} \frac{1}{4} \frac{1}{3} \frac{1}{2} +				
\frac{1}{3} \frac{1}{4} \frac{1}{3} \frac{1}{3} +				
\frac{1}{3} \frac{1}{4} \frac{1}{3} \frac{1}{3} \frac{1}{2} +\\	
\frac{1}{3} \frac{1}{4} \frac{1}{3} \frac{1}{2} +				
\frac{1}{3} \frac{1}{4} \frac{1}{3} \frac{1}{2} +				
\frac{1}{3} \frac{1}{4} \frac{1}{3} \frac{1}{2} \frac{1}{2} +	
\frac{1}{3} \frac{1}{4} \frac{1}{3} \frac{1}{2} +				
\frac{1}{3} \frac{1}{2} \frac{1}{3} \frac{1}{2}	+				
\frac{1}{3} \frac{1}{2} \frac{1}{3} \frac{1}{2}	+				
\frac{1}{3} \frac{1}{2} \frac{1}{3} \frac{1}{2}	+				
\frac{1}{3} \frac{1}{2} \frac{1}{3} \frac{1}{2}	+				
\frac{1}{3} \frac{1}{2} \frac{1}{3} \frac{1}{2}	+ 				
\frac{1}{3} \frac{1}{2} \frac{1}{3} \frac{1}{2}	+ \\			
\frac{1}{3} \frac{1}{2} \frac{1}{2}	+							
\frac{1}{3} \frac{1}{2} \frac{1}{2}	+							
\frac{1}{3} \frac{1}{2} \frac{1}{3} \frac{1}{3}	+				
\frac{1}{3} \frac{1}{2} \frac{1}{3} \frac{1}{3} \frac{1}{2} +	
\frac{1}{3} \frac{1}{2} \frac{1}{3} \frac{1}{2}	+ 				
\frac{1}{3} \frac{1}{2} \frac{1}{3} \frac{1}{2}	+ 				
\frac{1}{3} \frac{1}{2} \frac{1}{3} \frac{1}{2} \frac{1}{2}	+	
\frac{1}{3} \frac{1}{2} \frac{1}{3} \frac{1}{2}	+				
\frac{1}{3} \frac{1}{2} \frac{1}{2}	+							
\frac{1}{3} \frac{1}{2} \frac{1}{2} = \frac{7}{8}.				
\end{gathered}
\end{multline*} 
Perhaps unsurprisingly, the above probability is equal to the probability of $\mathcal{T}$ producing its terminal assembly that contains the 4 tile. 
After all, the A and B tiles should not affect the probability of $\mathcal{T}$ producing its terminal assembly that contains the 4 tile. So $\mathcal{T}$ should be equivalent to $\mathcal{T}$ from Figure~\ref{fig:intro-example1} in a probabilistic sense. 
Nevertheless, the point of this example is to show how analyzing the probabilistic correctness of even a seemingly simple non-deterministic, temperature-1 TAS in the 2DPTAM can become unwieldy very quickly. 
This is because the probability that the TAS produces some terminal assembly depends on all the tile attachment steps of all the assembly sequences of the TAS that result in the assembly, including tile attachment steps in which tiles attach in such a way that should not affect the probability.
In this paper, we formalize the intuition that tiles that attach in such a way that should not affect the probability of a certain kind of non-deterministic, temperature-1 TAS in fact do not affect the probability.
To that end, we develop a first-of-its-kind general framework called ``sequential non-determinism'' for analyzing the probabilistic correctness of such a TAS, e.g., ${\mathcal{T}}$ shown in Figure~\ref{fig:intro-example2}.
In sequential non-determinism, we assume that for a given non-deterministic, temperature-1 TAS and in every step of every assembly sequence thereof, a tile initially attaches via exactly one of its glues and is the only type of tile that can attach at that location and via that glue.
We also assume there exist (potentially infinitely many) {\em points of competition (POC)}, at which one of only two types of tiles may (and may only) be placed via two corresponding competing paths of tiles that originate from a common starting point and are disjoint until potentially meeting at the POC, e.g., the location at which either the 3 or the 4 tiles in Figure~\ref{fig:intro-example2} may attach.
Without going into all the technical details, we require that the POCs be ordered, sufficiently spaced out, and that, in every assembly sequence of the TAS, tiles attach at the POCs in the specified order.
We assume that each POC is either {\em essential} or {\em inessential}, that the number of essential ones is finite, and that each essential POC is associated with a type of tile that is the correct type of tile that should be placed at that POC. Under these assumptions, the TAS produces a unique {\em correct} terminal assembly.
Our main sequential non-determinism theorem says that the probability in the 2DPTAM that a given sequentially non-deterministic, temperature-1 TAS produces the correct (potentially infinite) terminal assembly is a finite product of probabilities, each of which corresponds to an essential POC having the correct type of tile placed at it.
Moreover, our framework allows for the computation of this probability to be carried out while effectively ignoring the tile attachments that should not affect whether the correct type of tile attaches at an essential POC. 
For example, our main theorem allows us to analyze the TAS in Figure~\ref{fig:intro-example2} as if it were the TAS in Figure~\ref{fig:intro-example1}.
We then show that the requirements of sequential non-determinism are not overly restrictive, since sequential non-determinism does not hinder sufficiently geometrically and computationally expressive self-assembly.
To that end, we constructively prove that for any real $\delta \in (0,1)$ and any $N \in \mathbb{Z}^+$, there exists a TAS that produces the correct terminal assembly with probability at least $1 - \delta$, whose domain is the  $N \times N$ square $\left\{ 0, 1, \ldots, N - 1 \right\}^2$ and whose tile set has a size that is $O\left( \log N + \log \frac{1}{\delta} \right)$.
Note that, for a fixed $\delta$, our bound is simply $O\left( \log N \right)$, which is an improvement over the previous state-of-the-art bound of $O\left( \log^2 N + \log N \log \frac{1}{\delta} \right)$ for this problem by Cook, Fu and Schweller \cite{CookFuSch11}.
In addition to the previous bound by Cook, Fu and Schweller, there are some other notable previous results related to probabilistic self-assembly in the aTAM.
Working in the general tile concentration programming aTAM, where the concentration of each tile type can be set accordingly to affect the behavior of the TAS, Kao and Schweller \cite{KaoS08} showed that for some fixed TAS and any positive integer $N$, it is possible to set the tile concentrations of the tile set of the TAS so that it produces, with high probability, a terminal assembly whose domain corresponds to an ``approximation'' of an $N \times N$ square.
Technically, they showed that there exists a TAS such that for any $N \in \mathbb{Z}^+$ and for any real $\delta, \varepsilon > 0$, the tile concentrations can be set so that the TAS produces with probability $1 - \delta$ a terminal assembly whose domain is $\left\{ 0, 1, \ldots, N' - 1 \right\}^2$ such that $(1 - \varepsilon)N \leq N' \leq (1+\varepsilon)N$.
Among several other results, Doty \cite{Dot10} proved  an ``exact'' analog of the aforementioned ``approximate'' result by demonstrating the existence of a TAS with constant size tile set such that, for any $N \in \mathbb{Z}^+$ and for any real $\delta > 0$, the concentrations of the tile types of the TAS can be set so that it produces with probability $1 - \delta$ a terminal assembly whose domain is $\left\{ 0, 1, \ldots, N - 1 \right\}^2$.
Chandran, Gopalkrishnan, and Reif \cite{ChandranGR12} investigated the size of the smallest tile set such that there is a corresponding singly-seeded, temperature-1, non-deterministic TAS that produces a terminal assembly whose domain is expected to be a {\em linear assembly}, i.e., an assembly whose domain, for some $N \in \mathbb{Z}^+$, is $\left\{ 0, 1, \ldots, N-1\right\} \times \{0\}$.
While the previously-mentioned results by Doty, and Kao and Schweller utilize arbitrarily-defined tile concentrations, in several of their results, Chandran, Gopalkrishnan and Reif worked in the PTAM (in 1D) and assumed the same uniform tile concentration programming model that we do in this paper.
Among numerous other results, they proved that for any $N \in \mathbb{Z}^+$, there exists a TAS, whose tile set has size $O\left( \log N \right)$, that produces a linear terminal assembly with expected but perhaps widely varying length $N$.
Chandran, Gopalkrishnan and Reif followed this result up by showing that there is a TAS whose tile set has size $O\left( \log^3 N \right)$ that does the same thing but can more precisely control the variance of the length of the produced terminal assembly.

\section{Preliminaries}
\label{sec:preliminaries}

In this section, we give definitions for the abstract Tile Assembly Model (aTAM) as well as probabilistic self-assembly in the aTAM, yielding a model that we call the 2DPTAM.

\subsection{Formal description of the abstract Tile Assembly Model}
\label{sec:tam-formal}

Fix an alphabet $\Sigma$.
$\Sigma^*$ is the set of finite strings over $\Sigma$. Let $\Z$, $\Z^+$, and $\N$ denote the set of integers, positive integers, and nonnegative integers, respectively. Given $V \subseteq \Z^2$, the \emph{full grid graph} of $V$ is the undirected graph $\fullgridgraph_V=(V,E)$,
and for all $\vec{x}, \vec{y}\in V$, $\left\{\vec{x},\vec{y}\right\} \in E \iff \| \vec{x} - \vec{y}\| = 1$.

A \emph{tile type} is a tuple $t \in (\Sigma^* \times \N)^{4}$, that is, a unit square with four sides listed in some standardized order, and each side having a \emph{glue} $g \in \Sigma^* \times \N$ consisting of a finite string \emph{label} and nonnegative integer \emph{strength}. We assume a finite set of tile types, usually denoted $T$, but an infinite number of copies of each tile type. Intuitively, a \emph{tile} is a copy of a tile type that is placed at a point but technically it is a pair $\left(\vec{x},t\right) \in \mathbb{Z}^2 \times T$. 

A {\em configuration} is a possibly empty arrangement of tile types on the integer lattice $\Z^2$, i.e., a partial function $\alpha:\Z^2 \dashrightarrow T$, or, equivalently, a subset of $\mathbb{Z}^2 \times T$. We say that the tile $\left(\vec{x},t\right)$ is \emph{placed} by $\alpha$ if $\alpha\left( \vec{x} \right) = t$.
Two adjacent tiles in a configuration \emph{bind} if the glues on their opposing sides are equal (in both label and strength) and have positive strength.
Each configuration $\alpha$ induces a \emph{binding graph} $\bindinggraph_\alpha$, a grid graph whose vertices are positions occupied by tile types, according to $\alpha$, with an edge between two vertices if the tile types at those vertices bind.

An \emph{assembly} is a connected, non-empty configuration, i.e., a partial function $\alpha:\Z^2 \dashrightarrow T$ such that $\fullgridgraph_{\dom \alpha}$ is connected and $\dom \alpha \neq \emptyset$. The \emph{shape} $S_\alpha \subseteq \Z^2$ of $\alpha$ is $\dom \alpha$ and its size is $| \dom{\alpha}|$.

Given $\tau\in\Z^+$, $\alpha$ is \emph{$\tau$-stable} if every cut of~$\bindinggraph_\alpha$ has weight at least $\tau$, where the weight of an edge is the strength of the glue it represents. When $\tau$ is clear from the context, we say $\alpha$ is \emph{stable}.
Given two assemblies $\alpha$ and $\beta$, we say that $\alpha$ and $\beta$ \emph{agree} if, for all $\vec{x} \in \dom{\alpha} \cap \dom{\beta}$, $\alpha\left(\vec{x}\right) = \beta\left(\vec{x}\right)$. Furthermore, we say that $\alpha$ is a \emph{subassembly} of $\beta$, and we write $\alpha \sqsubseteq \beta$, if $S_\alpha \subseteq S_\beta$ and  $\alpha$ and $\beta$ agree.

A \emph{tile assembly system} (TAS) is a triple $\mathcal{T} = (T,\sigma,\tau)$, where $T$ is a tile set, $\sigma:\Z^2 \dashrightarrow T$ is the finite, $\tau$-stable, \emph{seed assembly}, and $\tau\in\Z^+$ is the \emph{temperature}.
We say that $\mathcal{T}$ is {\em singly-seeded} if $\left| \dom{\sigma} \right| = 1$. From this point on, unless explicitly stated otherwise, we assume that every TAS is singly-seeded, with $\dom{\sigma}=\{\vec{s\,}\}$, and has temperature $1$.

Given two $\tau$-stable assemblies $\alpha,\beta$, we write $\alpha \to_1^{\mathcal{T}} \beta$ if $\alpha \sqsubseteq \beta$ and $|S_\beta \setminus S_\alpha| = 1$. In this case we say $\alpha$ \emph{$\mathcal{T}$-produces $\beta$ in one step}. If $\alpha \to_1^{\mathcal{T}} \beta$, $ S_\beta \setminus S_\alpha=\{\vec{p}\}$, and $t=\beta(\vec{p})$, we write $\beta = \alpha + (\vec{p}, t)$ to denote the \emph{tile attachment step} in which a tile of type $t$ \emph{attaches} to $\alpha$ at $\vec{p}$ to produce $\beta$. We use the notation $\beta \backslash \alpha$ to denote $\left(\vec{p}, t \right)$, i.e., the tile that attaches to $\alpha$ to produce $\beta$ in one step. If $\beta \backslash \alpha$ attaches to $\alpha$ with total strength 1 and the tile in $\alpha$ to which $\beta \backslash \alpha$ binds is located at point $\vec{p}\;^\prime$, we use $\vec{u}_{\beta \backslash \alpha}$ to denote the unit vector $\vec{p}\;^\prime - \vec{p}$. 

The \emph{$\mathcal{T}$-frontier} of $\alpha$ is the set $\partial^\mathcal{T} \alpha = \bigcup_{\alpha \to_1^\mathcal{T} \beta} S_\beta \setminus S_\alpha$, the set of empty locations at which a tile could stably attach to $\alpha$. When $\mathcal{T}$ is clear from the context, then we omit the superscript $\mathcal{T}$ on the $\mathcal{T}$-frontier of $\alpha$. The \emph{$t$-frontier} $\partial_t \alpha \subseteq \partial \alpha$ of $\alpha$ is the set $\setr{\vec{p}\in\partial \alpha}{\alpha \to_1^\mathcal{T} \beta \text{ and } \beta(\vec{p})=t}.$

Let $\mathcal{A}^T$ denote the set of all assemblies of tiles from $T$, and let $\mathcal{A}^T_{< \infty}$ denote the set of finite assemblies of tiles from $T$.
A sequence of $k\in\Z^+ \cup \{\infty\}$ assemblies $\alpha_1,\alpha_2,\ldots$ over $\mathcal{A}^T$ denoted as $\vec{\alpha} = \left( \alpha_i \mid 0 \leq i-1 < k \right)$ is a \emph{$\mathcal{T}$-assembly sequence} if, for all $1 \leq i-1 < k$, $\alpha_{i-1} \to_1^\mathcal{T} \alpha_{i}$. We use this notation for defining the range of values for $i$ because it allows us to handle both cases for $k$ (i.e., finite or infinite) simultaneously. For integers $1 \leq i \leq j \leq k$, $\vec{\alpha}[i\ldots j]$ denotes the subsequence $\left( \alpha_i, \ldots, \alpha_j \right)$. When $\vec{\alpha}$ is finite, we say that $\vec{\alpha}$ \emph{terminates} at $\vec{p} \in \mathbb{Z}^2$, where $\left\{ \vec{p} \right\} = \dom{\alpha_1}$ if $k=1$, and $\left\{ \vec{p} \right\} = \dom{\alpha_{k}} \backslash \dom{\alpha_{k-1}}$ otherwise. The {\em result} of an assembly sequence $\vec{\alpha}$, written as $\res{\vec{\alpha}}$ is the unique assembly such that $\dom{\res{\vec{\alpha}} = \bigcup_{i=1}^{k}{\dom{\alpha_i}}}$ and for all $0 \leq i -1 < k$, $\alpha_i \sqsubseteq \res{\vec{\alpha}}$. In other words, $\res{\vec{\alpha}}$ is the unique limiting assembly, which for a finite sequence, is the final assembly in the sequence. 
If $\res{\vec{\alpha}} = \alpha$ and $\vec{x} \in \dom{\alpha}$, then $\textmd{index}_{\vec{\alpha}}\left( \vec{x} \right) = \min\left\{ \; i \; \left| \; \vec{x} \in \dom{\alpha_i} \right. \right\}$.

If $\alpha = \res{\vec{\alpha}}$ and, for some $t \in T$, $\vec{x} \in \partial_t \alpha$, then the notation $\vec{\alpha} + \left( \vec{x}, t \right)$ denotes \emph{appending the tile} $\left( \vec{x}, t \right)$ \emph{to} $\vec{\alpha}$, and is the assembly sequence $\left( \alpha_1, \ldots, \alpha_k, \alpha_{k+1} \right)$, where $\alpha_{k+1} = \alpha_k + \left( \vec{x}, t \right)$.

For $k\in \mathbb{Z}^+$, $m \in \mathbb{Z}^+ \cup \left\{ \infty \right\}$, and $\mathcal{T}$-assembly sequences $\vec{\alpha} = \left( \alpha_i \mid 1 \leq i \leq  k \right)$ and $\vec{\beta} = \left( \beta_i \mid 0 \leq i-1 < m \right)$, if $\vec{\varepsilon} = \left( \varepsilon_i \mid 0 \leq i-1 < k+m \right)$ is a $\mathcal{T}$-assembly sequence, where for all integers $0 \leq i-1 < k$, $\varepsilon_i = \alpha_i$, and for all $k \leq i-1 < k+m$, $\varepsilon_i = \beta_{i-k}$, then we say that $\vec{\varepsilon}$ is an \emph{extension} of $\vec{\alpha}$ \emph{by} $\vec{\beta}$. In this case, we say that $\vec{\alpha}$ is a \emph{prefix} of $\vec{\varepsilon}$ and $\vec{\beta}$ is a \emph{suffix} of $\vec{\varepsilon}$. 
For $k\in \mathbb{Z}^+$, $m \in \mathbb{Z}^+ \cup \left\{ \infty \right\}$, and $\mathcal{T}$-assembly sequences $\vec{\alpha} = \left( \alpha_i \mid 1 \leq i \leq k \right)$ and $\vec{\beta} = \left( \beta_i \mid 0 \leq i-1 < m \right)$, we say that $\vec{\alpha}$ is \emph{embedded} in $\vec{\beta}$ if $k \leq m$ and there exists an \emph{embedding function of} $\vec{\alpha}$ \emph{in} $\vec{\beta}$ satisfying the following conditions: (1) $f: \left\{ 1,2, \ldots, k \right\} \rightarrow \left\{i \in \Z^+ \mid 0 \leq i-1 <m\right\}$, (2) if $f(1)=1$, then $\alpha_1 = \beta_1$, otherwise $\alpha_1 = \{ \beta_{f(1)} \backslash \beta_{f(1)-1} \}$, (3) for all $2 \leq i \leq k$, $\beta_{f(i)} \backslash \beta_{f(i)-1} = \alpha_i \backslash \alpha_{i-1}$, and (4) for $1 \leq i < j \leq k$, $f(i) < f(j)$. Note that the last condition implies that the embedding function is one-to-one.

We write $\alpha \to^\mathcal{T} \beta$, and we say $\alpha$ \emph{$\mathcal{T}$-produces} $\beta$ (in 0 or more steps) if there is a $\mathcal{T}$-assembly sequence $\left( \alpha_1,\alpha_2,\ldots \right)$ of length $k = |S_\beta \setminus S_\alpha| + 1$ such that
1) $\alpha = \alpha_1$,
2) $S_\beta = \bigcup_{0 \leq i -1 < k} S_{\alpha_i}$, and
3) for all $0 \leq i-1 < k$, $\alpha_{i} \sqsubseteq \beta$.

If $k$ is finite then it is routine to verify that $\beta = \alpha_k$. 
We say $\alpha$ is \emph{$\mathcal{T}$-producible} if $\sigma \to^\mathcal{T} \alpha$, and we write $\prodasm{\mathcal{T}}$ to denote the set of $\mathcal{T}$-producible assemblies. 
We say that $\vec{\alpha}$ is $\mathcal{T}$-producing if $\vec{\alpha} = (\sigma, \ldots)$ and $\res{\vec{\alpha}} \in \mathcal{A}[\mathcal{T}]$.
An assembly $\alpha$ is \emph{$\mathcal{T}$-terminal} if $\alpha$ is $\tau$-stable and $\partial^\mathcal{T} \alpha=\emptyset$.
We write $\termasm{\mathcal{T}} \subseteq \prodasm{\mathcal{T}}$ to denote the set of $\mathcal{T}$-producible, $\mathcal{T}$-terminal assemblies.

We say that a TAS $\mathcal{T}$ \emph{strictly self-assembles} a shape $X \subseteq \Z^2$ if, for all $\alpha \in \termasm{\mathcal{T}}$, $S_{\alpha} = X$; i.e., if every terminal assembly produced by $\mathcal{T}$ assigns tile types to all and only the points in the set $X$.

\subsection{The probabilistic aTAM or 2DPTAM}

We model probabilistic self-assembly in the aTAM with non-empty, rooted, directed, weighted tree structures in which all edges point away from the root and, for each node $v$, $S_v$ denotes the sum of all the numerical weights over all the outgoing edges on $v$. If $v$ is a leaf node, then $S_v = 0$. We say that a non-empty, rooted, weighted tree $\mathcal{Q}$ is a \emph{sub-probability tree} (SPT) if:
\begin{enumerate}
	\item all edges of $\mathcal{Q}$ are labeled with some \emph{probability} $p \in (0, 1]$, and
	\item for each node $v$ in $\mathcal{Q}$, $S_v \leq 1$.
\end{enumerate}

If $\mathcal{Q}$ is an SPT and $v$ is a node of  $\mathcal{Q}$, then we say $v$ is \emph{normalized} in $\mathcal{Q}$ if $S_v = 1$. If all internal nodes of an SPT $\mathcal{Q}$ are normalized, then $\mathcal{Q}$ is a discrete time Markov chain, or simply Markov chain.

For an SPT $\mathcal{Q}$, if $v$ is a node of $\mathcal{Q}$, then we define the \emph{probability} of $v$, denoted as $\textmd{Pr}_{\mathcal{Q}}[v]$, as follows:
\begin{enumerate}
	\item If $v$ is the root node of $\mathcal{Q}$, then $\textmd{Pr}_{\mathcal{Q}}[v] = 1$.
	\item Otherwise, if $u$ is the parent of $v$, and $p$ is the probability of the edge from $u$ to $v$, then $\textmd{Pr}_{\mathcal{Q}}[v] = \textmd{Pr}_{\mathcal{Q}}\left[u\right] \cdot p$.
\end{enumerate}
Let $m \in \mathbb{Z}^+ \cup \left\{ \infty \right\}$,  $\pi = \left(q_i \mid 0 \leq i-1 < m\right)$ be a sequence of nodes in $\mathcal{Q}$ where $q_1$ is the root of $\mathcal{Q}$, for all $1 \leq i-1 < m$, there is an edge from $q_{i-1}$ to $q_i$ in $\mathcal{Q}$, and $p_i$ is the probability on this edge. We say that $\pi$ is a \emph{maximal path in} $\mathcal{Q}$ if  $m = \infty$, or $m \in \mathbb{Z}^+$ and $q_m$ is a leaf node of $\mathcal{Q}$. If $m = 1$, then $\displaystyle\textmd{Pr}_{\mathcal{Q}}[\pi] = 1$. Otherwise, if $m > 1$, then we define $\displaystyle\textmd{Pr}_{\mathcal{Q}}[\pi] = \prod_{\left\{i \in \Z^+\mid 1\leq i-1<m\right\}}{p_i}$. Note that if $m = \infty$, then $\pi$ is a ray in $\mathcal{Q}$, starting from the root. 

\begin{observation}
\label{obs:leaf-node-maximal-path}
If $\mathcal{Q}$ is an SPT and every maximal path in  $\mathcal{Q}$  is finite, then there is a one-to-one correspondence between the set of leaf nodes of $\mathcal{Q}$ and the set of maximal paths in $\mathcal{Q}$, and thus $\displaystyle\sum_{u \textmd{ leaf node of } \mathcal{Q}}{\textmd{Pr}_{\mathcal{Q}}[u]} = \sum_{\pi \textmd{ maximal path in } \mathcal{Q}}{\textmd{Pr}_{\mathcal{Q}}[\pi]}$.
\end{observation}

We define the \emph{probability of SPT} $\mathcal{Q}$ to be $\displaystyle\textmd{Pr}[\mathcal{Q}] = \sum_{\pi \textmd{ maximal path in } \mathcal{Q}}{\textmd{Pr}_{\mathcal{Q}}[\pi]}$.

If $\mathcal{Q}$ is a rooted tree and $u$ is a node in  $\mathcal{Q}$, then $\mathcal{Q}^u$ is the subtree of $\mathcal{Q}$ that contains $u$, along with all descendants of $u$ in $\mathcal{Q}$. Note that, if $\mathcal{Q}$ is an SPT, then $\mathcal{Q}^u$ is also an SPT. We say that any subtree $\mathcal{Q}'$ of $\mathcal{Q}$ is \emph{full relative to} $\mathcal{Q}$ if, for every internal node $v$ of $\mathcal{Q}'$, all children of $v$ in $\mathcal{Q}$ are children of $v$ in $\mathcal{Q}'$. We say $\mathcal{Q}$ is \emph{finitely branching} (or \emph{locally finite)} if, for every internal node $v$ of $\mathcal{Q}$, there exists $c \in \mathbb{Z}^+$ such that $v$ has $c$ children in $\mathcal{Q}$.

Assume $\pi = \left( x_i \mid 0 \leq i-1 < m \right)$ for $m \in \mathbb{Z}^+ \cup \left\{ \infty \right\}$ denotes a simple path whose elements may be, for example, either points in $\mathbb{Z}^2$, or nodes in an SPT. Then $\dom{\pi}$ denotes the set of elements $\left\{ x_i \mid i \geq 1 \right\}$, $\left| \pi \right| = m$ denotes the length of $\pi$, $\pi[i]$ denotes $x_i$ for $i \geq 1$, and, for integers $0 \leq i-1 \leq j-1 < m$, $\pi[i\ldots j]$ denotes the subsequence $\left( x_i, \ldots, x_j \right)$. Any subsequence  $\pi[1\ldots j]$ for integers $j \geq 1$ is called a \emph{prefix} of $\pi$, whereas, assuming $\pi$ is a finite path, any subsequence  $\pi[i\ldots m]$ for integers $1\leq i \leq m$ is called a \emph{suffix} of $\pi$.

We study probabilistic self-assembly in the aTAM using a simplified version of the tile concentration programming model, which we call the 2DPTAM, where tiles are assumed to have uniform concentration (see \cite{Dot10} for a formal development of the general model). 

We model probabilistic self-assembly of a TAS $\mathcal{T} = (T,\sigma,1)$ with the SPT $\mathcal{M}_{\mathcal{T}}$, where:
\begin{enumerate}
	\item the root node of $\mathcal{M}_{\mathcal{T}}$ is the $\mathcal{T}$-assembly sequence $(\sigma)$, 
	\item there is a one-to-one correspondence between nodes in $\mathcal{M}_{\mathcal{T}}$ and finite $\mathcal{T}$-producing assembly sequences,
	\item there is an edge in $\mathcal{M}_{\mathcal{T}}$ from $\vec{\alpha}$ to $\vec{\beta}$, assuming $\res{\vec{\alpha}} = \alpha$ and $\res{\vec{\beta}} = \beta$, if $\alpha \rightarrow^{\mathcal{T}}_1 \beta$, and
	\item the probability of the edge in $\mathcal{M}_{\mathcal{T}}$ from $\vec{\alpha}$ to $\vec{\beta}$, assuming $\res{\vec{\alpha}} = \alpha$ and $\res{\vec{\beta}} = \beta$, is $\frac{1}{M_{\alpha}}$, where $M_{\alpha} = \left| \left\{ \beta \left| \alpha \rightarrow^{\mathcal{T}}_1 \beta \right. \right\} \right| > 0$.
\end{enumerate}
Note that $\mathcal{M}_{\mathcal{T}}$ is a finitely branching SPT. To see that it is a tree, note that its nodes are $\mathcal{T}$-producing assembly sequences and tiles cannot detach in the aTAM. It is finitely branching because each one of its nodes is a finite $\mathcal{T}$-producing assembly sequence that results in a finite assembly, the frontier of which is finite. Moreover, $\mathcal{M}_{\mathcal{T}}$ is a Markov chain because, for each node $\vec{\alpha}$ with $\res{\vec{\alpha}} = \alpha$, the probability on every outgoing edge of $\vec{\alpha}$ is $\frac{1}{M_{\alpha}}$, and there are $M_{\alpha}$ such edges.

Since there is a one-to-one correspondence between nodes of $\mathcal{M}_{\mathcal{T}}$ and finite $\mathcal{T}$-producing assembly sequences, we define the probability of a $\mathcal{T}$-producible assembly in terms of the $\mathcal{T}$-producing assembly sequences from which it may result. For $\alpha \in \mathcal{A}[\mathcal{T}]$, we define the \emph{probability} of $\alpha$ as:

\centerline{$\textmd{Pr}_{\mathcal{T}}[\alpha]=\displaystyle \sum_{\substack{\vec{\alpha} \textmd{ is a } \mathcal{T}\textmd{-producing assembly sequence}\\ \res{\vec{\alpha}}=\alpha}}{\textmd{Pr}_{\mathcal{M}_{\mathcal{T}}}[\vec{\alpha}]}$.}

We say that $\mathcal{T}$ {\em strictly self-assembles a shape} $X \subseteq \mathbb{Z}^2$ with probability at least $p \in [0,1]$ if $\displaystyle \sum_{\substack{\alpha \in \mathcal{A}_{\Box}[\mathcal{T}]\\ \dom{\alpha}=X}}{\textmd{Pr}_{\mathcal{T}}[\alpha]}~\geq~p$. We say that $X$ {\em strictly self-assembles} with probability at least $p$ if there exists some TAS that strictly self-assembles it with probability at least $p$. 

\section{Defining sequential non-determinism}
\label{sec:snd}
In this section, we build up the definitions that we will use to define the notion of sequential non-determinism.
In Figure~\ref{fig:def-example-tas}, we define a sample TAS that we
will use as a running example to illustrate the key concepts that we
develop in this section.
\begin{figure}[!h]
      \begin{minipage}{3in}
        \includegraphics[width=3in]{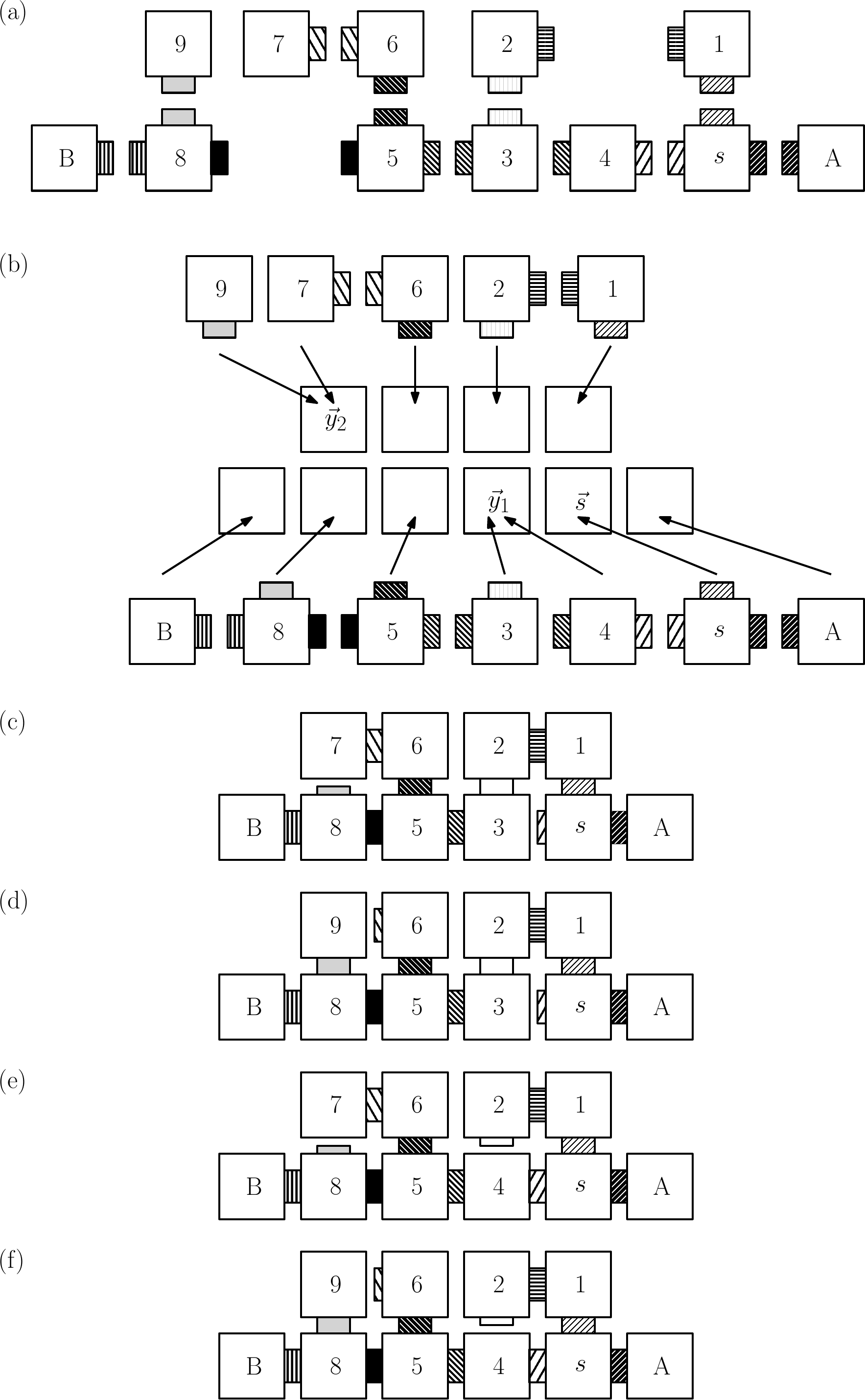}
            \end{minipage}\hfill
    \begin{minipage}{3in}
      \caption{\label{fig:def-example-tas}The TAS
        $\mathcal{T}=(T,\sigma,1)$ we are using as a running example
        throughout this and the next section, with $|\termasm{T}|$ = 4\\
        (a) Depiction of the tile set $T=\{s,1,2,3,4,5,6,$ $7,8,9,A,B\}$,
        in which $s$ is the seed tile and the other tiles are labeled with
        either an uppercase letter or a positive integer; all glue strengths
        are 1; the different glue labels are depicted with shades of gray or
        tiling patterns\\
        (b) Depiction (in the middle) of the points in $\mathbb{Z}^2$ at
        which tiles in $T$ may attach; $\vec{s}$ is the point where the
        seed tile $s$ is always placed in this example, i.e.,
        $\sigma=\{(\vec{s},s)\}$; 
        the tiles in $T$ are shown either above or below the points, with
        an arrow indicating the unique point at which the tile may
        attach; $\vec{y}_1$ and $\vec{y}_2$ are the only points where
        more than one tile may attach\\
        (c) $\alpha_{3,7}$ is the  assembly in
        $\termasm{T}$ containing the two tiles $(\vec{y}_1,3)$ and
        $(\vec{y}_2,7)$\\
        (d) $\alpha_{3,9}$ is the  assembly in
        $\termasm{T}$ containing the two tiles $(\vec{y}_1,3)$ and
        $(\vec{y}_2,9)$\\
        (e) $\alpha_{4,7}$ is the assembly in
        $\termasm{T}$ containing the two tiles $(\vec{y}_1,4)$ and
        $(\vec{y}_2,7)$\\
        (f) $\alpha_{4,9}$, is the assembly in
        $\termasm{T}$ containing the two tiles $(\vec{y}_1,4)$ and
        $(\vec{y}_2,9)$
      }
    \end{minipage}
  \end{figure}

\goingforward{assume $\mathcal{T} = (T,\sigma,1)$ is a singly-seeded TAS with $\dom{\sigma}=\{\vec{s}\,\}$.}

In this section, as well as in Section~\ref{sec:proof}, we use a box like the one above to state one or more assumptions that hold by default, from this point on and until the end of Section~\ref{sec:proof},  unless otherwise stated. This practice allows us, for example, to significantly shorten the statement of several lemmas that share an increasingly long sequence of assumptions.
\begin{definition} 
\label{def:competing-paths}
Let $\vec{x},\vec{y} \in \mathbb{Z}^2$ with $\vec{x} \neq \vec{y}$, and $\pi$ and $\pi'$ be finite simple paths in $G^{\textmd{f}}_{\mathbb{Z}^2}$ from $\vec{x}$ to $\vec{y}$. We say that $\pi$ and $\pi'$ are \emph{competing for} $\vec{y}$, \emph{from} $\vec{x}$ \emph{in} $\mathcal{T}$, if all of the following conditions hold.
\begin{enumerate}
	\item \label{def:cp-1} $\pi \ne \pi'$
	\item \label{def:cp-2} $\dom{\pi} \cap \dom{\pi'} = \left\{ \vec{x}, \vec{y} \right\}$
	\item \label{def:cp-3} There exists a path $p \in \left\{ \pi, \pi' \right\}$ and an assembly $\alpha \in \mathcal{A}[\mathcal{T}]$ such that  $p$ is a simple path in $G^{\textmd{b}}_{\alpha}$.
 	\item \label{def:cp-4} For all $\alpha \in \mathcal{A}[\mathcal{T}]$, each path $p \in \left\{ \pi, \pi' \right\}$ and all integers $1 \leq l < \left| p \right|$, if $p'$ is any finite simple path from $\vec{s}$ to $p[l]$ in $G^{\textmd{b}}_{\alpha}$, then:
	\begin{enumerate}
		\item \label{def:cp-4a} if $p'$ goes through $p[1]$, then $p[1\ldots l]$ is a suffix of $p'$,
		\item \label{def:cp-4b} otherwise $p'$ is a prefix of every finite simple path from $\vec{s}$ to $p[1]$ in $G^{\textmd{b}}_{\alpha}$. 
	\end{enumerate}
	\item \label{def:cp-5} For all $\alpha \in \mathcal{A}[\mathcal{T}]$, if $p'$ is any simple path from $\vec{s}$ to $\vec{y}$ in $G^{\textmd{b}}_{\alpha}$, then there exists a path $p \in \left\{ \pi, \pi' \right\}$ such that $p$ is a suffix of $p'$.
    \item \label{def:cp-6} For all $\alpha \in \mathcal{A}[\mathcal{T}]$ such that $\vec{x} \in \dom{\alpha}$, the following conditions hold.
    \begin{enumerate}
    	\item \label{def:cp-6a} For each path $p \in \left\{ \pi, \pi' \right\}$, there exists a $\mathcal{T}$-assembly sequence $\vec{\beta} = \left( \beta_i \mid 1 \leq i \leq k\right)$ such that $\beta_1 = \left\{ \left(\vec{x}, \alpha\left(\vec{x}\right) \right) \right\}$, and $\dom{\beta_k} = \dom{p}$.
		\item \label{def:cp-6b} For all $\mathcal{T}$-assembly sequences $\vec{\beta} = \left( \beta_i \mid 1 \leq i \leq k\right)$, if $\beta_1 = \left\{ \left(\vec{x}, \alpha\left(\vec{x}\right) \right) \right\}$ and $p$ is a simple path from $\vec{x}$ to $\vec{y}$ in $G^{\textmd{b}}_{\beta_k}$, then $p \in \left\{ \pi, \pi' \right\}$.
	\end{enumerate}   
\end{enumerate}
We call $p$ a \emph{competing path} if there exist $\vec{x}$ and $\vec{y}$ such that $\pi$ and $\pi'$ are competing for $\vec{y}$, from $\vec{x}$ in $\mathcal{T}$ and $p \in \left \{ \pi, \pi' \right\}$. 
\end{definition}
Figure~\ref{fig:def-competing-paths} both depicts visually and explains informally the implications of this definition.
\begin{figure}[!h]
  \begin{minipage}[c]{0.65\linewidth} \caption{\label{fig:def-competing-paths} 
The two paths $\pi$ and $\pi'$ are competing for $\vec{y}$ from $\vec{x}$ in $\mathcal{T}=(T,\sigma,1)$. 
In this figure, each arrow represents a sequence of tile attachment steps that \emph{follow} a simple path in the sense that tiles attach at points along a simple path and in order. 
According to Definition~\ref{def:competing-paths}, the solid paths
must exist, the dashed path may exist, and the dotted ones cannot exist.
We interpret the meaning of the conditions in this definition as follows.\\ 
(a) Conditions~\ref{def:cp-1} and~\ref{def:cp-2} together mean that the two competing paths only share their starting and ending points and that their domains are not identical. Since these paths only intersect at their extremities, tiles that attach along one of them will never block the other path.\\
(b) Condition~\ref{def:cp-3} means that there exists a $\mathcal{T}$-producing assembly sequence that results in an assembly that places tiles on all the points in $\pi$ or $\pi'$, which also implies that there exists a $\mathcal{T}$-producing assembly sequence that results in an assembly that places tiles on all the points in any prefix of $\pi$ or $\pi'$. We do not require the existence of a $\mathcal{T}$-producing assembly sequence for both $\pi$ and $\pi'$ because it might be the case that one of these paths is blocked in every $\mathcal{T}$-producing assembly sequence. For example, $\pi$ could be blocked by tiles placed along the dashed arrow.
  } \end{minipage}\hfill
  \begin{minipage}{2in} \includegraphics[width=2in]{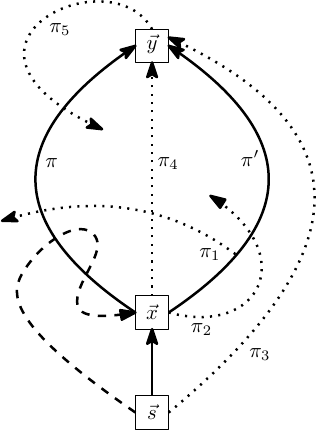} \end{minipage} \vspace*{-3mm}

  (c) Condition~\ref{def:cp-4a} means that, if the binding graph $G$
  of a $\mathcal{T}$-producible assembly contains a simple path $p$
  from $\vec{s}$ to any non-ending point on one of the competing paths
  such that $p$ goes through $\vec{x}$, then $p$ must first go through
  $\vec{x}$ and follow $\pi$ or $\pi'$. So, for example, no
  $\mathcal{T}$-producing assembly sequence may first follow the path
  along the solid arrow from $\vec{s}$ to $\vec{x}$ and then follow
  $\pi_2$. Similarly, no $\mathcal{T}$-producing assembly sequence may
  first follow the path along the solid arrow from $\vec{s}$ to
  $\vec{x}$, then follow the prefix of $\pi'$, and finally follow
  $\pi_1$. As yet another example, no $\mathcal{T}$-producing assembly
  sequence that terminates at $\vec{y}$ may be extended by an assembly
  sequence that follows $\pi_5$.

  (d) Condition~\ref{def:cp-4b} means that, if the binding graph $G$
  of a $\mathcal{T}$-producible assembly contains a simple path $p$
  from $\vec{s}$ to any non-ending point on one of the competing paths
  such that $p$ does not go through $\vec{x}$, then $p$ must be a
  prefix of all paths from $\vec{s}$ to $\vec{x}$ in $G$.  For
  example, if the path represented by the dashed arrow is contained in
  the binding graph of some assembly, then the latter may not contain
  the path represented by the solid arrow from $\vec{s}$ to $\vec{x}$.
  
(e) Condition~\ref{def:cp-5} means that all $\mathcal{T}$-assembly sequences that follow a path from $\vec{s}$ to $\vec{y}$ must go through $\vec{x}$ and then follow either $\pi$ or $\pi'$. In other words, no $\mathcal{T}$-producing assembly sequence may follow $\pi_3$.

  (f) Condition~\ref{def:cp-6a} means that it is possible for either competing path to self-assemble in isolation in a $\mathcal{T}$-assembly sequence starting from any tile that is placed at $\vec{x}$ by some $\mathcal{T}$-producible assembly.
  
  (g) Condition~\ref{def:cp-6b} means that, if there is a
  $\mathcal{T}$-assembly sequence that starts from any tile that is
  placed at $\vec{x}$ by some $\mathcal{T}$-producible assembly and
  follows a simple path to $\vec{y}$, then this simple path
  must be one of the two competing paths $\pi$ or $\pi'$. In other words, no
  $\mathcal{T}$-assembly sequence starting from such a tile may follow $\pi_4$.
\end{figure}

We now describe a TAS $\mathcal{T}$, shown in Figure~\ref{fig:def1-tas}, that we use only here to illustrate the conditions that make up
Definition~\ref{def:competing-paths}. $\mathcal{T}$ is also useful in illustrating the kind of
gadgets (i.e., subsets of tiles that perform a specific task like
writing or reading the value of a bit) that we use in our main
construction in Section~\ref{sec:square}. $\mathcal{T}$ contains a bit-writing gadget
and a bit-reading gadget. Both in our main construction and in this
example, the reading gadgets are non-deterministic because a 0-valued
bit may be erroneously read as a 1-valued bit, as can be seen in
Figure~\ref{fig:def1-tas}(e). Whereas the bit-writing gadgets are
deterministic in our main construction (i.e., a specific bit value is
encoded in the design of each gadget and written with no errors by the
gadget), we made the bit-writing gadget in this example
non-deterministic, since either one of the W$^0_a$ and W$^1_a$ tiles
may attach to the seed tile S, as shown in
Figures~\ref{fig:def1-tas}(d) and (c), respectively, which allows us
to more fully illustrate the conditions in Definition~\ref{def:competing-paths} while keeping
$\mathcal{T}$ small.

Table~\ref{tab:def1-tas} summarizes how the two pairs of paths
depicted in Figure~\ref{fig:def1-tas}(b) fare against all of the
conditions imposed by Definition~\ref{def:competing-paths}. We now explain how each entry in
the table was determined.

{\bf Conditions~\ref{def:cp-1} and~\ref{def:cp-2}.}
Figure~\ref{fig:def1-tas}(b) shows that both pairs satisfy
conditions~\ref{def:cp-1} and ~\ref{def:cp-2} since the paths in each
pair are distinct and only share their starting and ending points,
namely the points in the figure containing either a black disk or a
pair of joining arrowheads, respectively.

{\bf Condition~~\ref{def:cp-3}.}
Figure~\ref{fig:def1-tas}(c),~\ref{fig:def1-tas}(d),
and~\ref{fig:def1-tas}(e), respectively, show that $\pi$ and $\pi'_1$ both self-assemble in $\alpha_{1,1}$, that $\pi'$ and $\pi_1$ both self-assemble in $\alpha_{0,0}$, and that $\pi$ and $\pi_1$ both self-assemble in $\alpha_{0,1}$. In conclusion, both pairs of paths satisfy condition~\ref{def:cp-3} in Definition~\ref{def:competing-paths}, since at
least one path in each pair self-assembles in at least one assembly of
$\mathcal{T}$. Note that, if the bit-writing gadget in $\mathcal{T}$
were deterministic and always wrote a 1 like in
Figure~\ref{fig:def1-tas}(c), then $\pi'$ would never be followed in
full by a $\mathcal{T}$-producing assembly sequence, because any such
sequence would get blocked by the bit bump. This is the intended
behavior in our main construction (see Section~\ref{sec:square}) and is the reason
why condition~\ref{def:cp-3} in Definition~\ref{def:competing-paths} does not require that both competing
paths be followed by $\mathcal{T}$-producing assembly sequences.

{\bf Condition~\ref{def:cp-4}.} Figures~\ref{fig:def1-tas}(c)-(e)
together show that $\pi$ satisfies condition~\ref{def:cp-4a} in
Definition~\ref{def:competing-paths} in all cases since, in any
subassembly of $\alpha_{1,1}$, $\alpha_{0,0}$, or $\alpha_{0,1}$ in
which a tile is placed at $\pi[l]$ for any value $l \in \{1,2,3,4,5\}$
by following $p'$, then this path goes through $\pi[1]$ (i.e., the
point at which the R tile is placed) and $\pi[1\ldots l]$ is a suffix
of $p'$. Furthermore, Figures~\ref{fig:def1-tas}(c)-(e) together show
that $\pi'$ satisfies either condition~\ref{def:cp-4a} or
condition~\ref{def:cp-4b} in Definition~\ref{def:competing-paths} in
all cases.  For example, in any subassembly of $\alpha_{0,0}$ or
$\alpha_{0,1}$ in which a tile is placed at $\pi'[l]$ for any value $l
\in \{1,2,3\}$ by following $p'$, then this path goes through
$\pi'[1]$ (i.e., the point at which the R tile is placed) and
$\pi'[1\ldots l]$ is a suffix of $p'$, thereby satisfying
condition~\ref{def:cp-4a} in all of these cases. In addition, in any
subassembly of $\alpha_{1,1}$ in which a tile is placed at $\pi'[l]$
for any value $l \in \{2,3\}$ by following $p'$, then this path does
not go through $\pi'[1]$ but it is a prefix of the only path in
$\alpha_{1,1}$ from the seed to $\pi'[1]$, namely the path
S-W$^1_a$-W$^1_b$-W$^1_c$-W$^1_d$-W$^1_e$-R, thereby satisfying
condition~\ref{def:cp-4b} in all of these cases. In conclusion, the
pair $\{\pi,\pi'\}$ satisfies condition~\ref{def:cp-4} in
Definition~\ref{def:competing-paths}.

In contrast, whereas Figures~\ref{fig:def1-tas}(c)-(e) together show
that $\pi_1$ satisfies condition~\ref{def:cp-4} in all cases, $\pi'_1$ does not
because, for example, for $\alpha = \alpha_{0,0}$ and $l=3$,
Figure~\ref{fig:def1-tas}(d) shows that the path
$p'=$ S-W$^0_a$-W$^0_b$-W$^0_c$-R-R$^0_a$ does go through $\pi'_1[1]$
but $\pi'_1[1\ldots 3]$ is not a suffix of $p'$. Therefore, the pair
$\{\pi_1,\pi'_1\}$ does not satisfy condition~\ref{def:cp-4} in Definition~\ref{def:competing-paths}, which
means that $\pi_1$ and $\pi'_1$ are not competing paths in
$\mathcal{T}$.

{\bf Condition~\ref{def:cp-5}.} Figures~\ref{fig:def1-tas}(c)-(e) together show
that both pairs of paths in this example satisfy condition~\ref{def:cp-5} in
Definition~\ref{def:competing-paths} since, for each pair and any assembly sequence placing
tiles along a simple path $p'$ from the seed to the point where the
paths in the pair join, exactly one path in the pair is a suffix of
$p'$.

{\bf Condition~\ref{def:cp-6a}.} The pair $\{\pi,\pi'\}$ satisfies
condition~\ref{def:cp-6a} since, starting with the R tile at the
starting point of both paths, the $\mathcal{T}$-assembly sequence
placing the tiles R$^1_a$, R$^1_b$, R$^1_c$, R$^1_d$, and 1 in this
order along $\pi$ can fully assemble in the absence of any other
tiles. The same holds when placing the tiles R$^0_a$, R$^0_b$, and 0
in this order along $\pi'$.  In contrast, the pair $\{\pi_1,\pi'_1\}$
does not satisfy condition~\ref{def:cp-6a} because, once the tile
W$^1_a$ is placed at the starting point of $\pi_1$, there is no
$\mathcal{T}$-assembly sequence that will follow $\pi$ since the
W$^1_a$ tile has no glue on its east side.

{\bf Condition~\ref{def:cp-6b}.} The pair $\{\pi,\pi'\}$ satisfies condition~\ref{def:cp-6a}
because Figures~\ref{fig:def1-tas}(c)-(e) show that there is no third
path sharing the same starting and ending points in the binding graphs
of any of the three $\mathcal{T}$-terminal assemblies. The same holds for
the pair $\{\pi_1,\pi'_1\}$.

\begin{figure}[!h]
  \begin{minipage}{\linewidth}
    \centering \includegraphics[width=\linewidth]{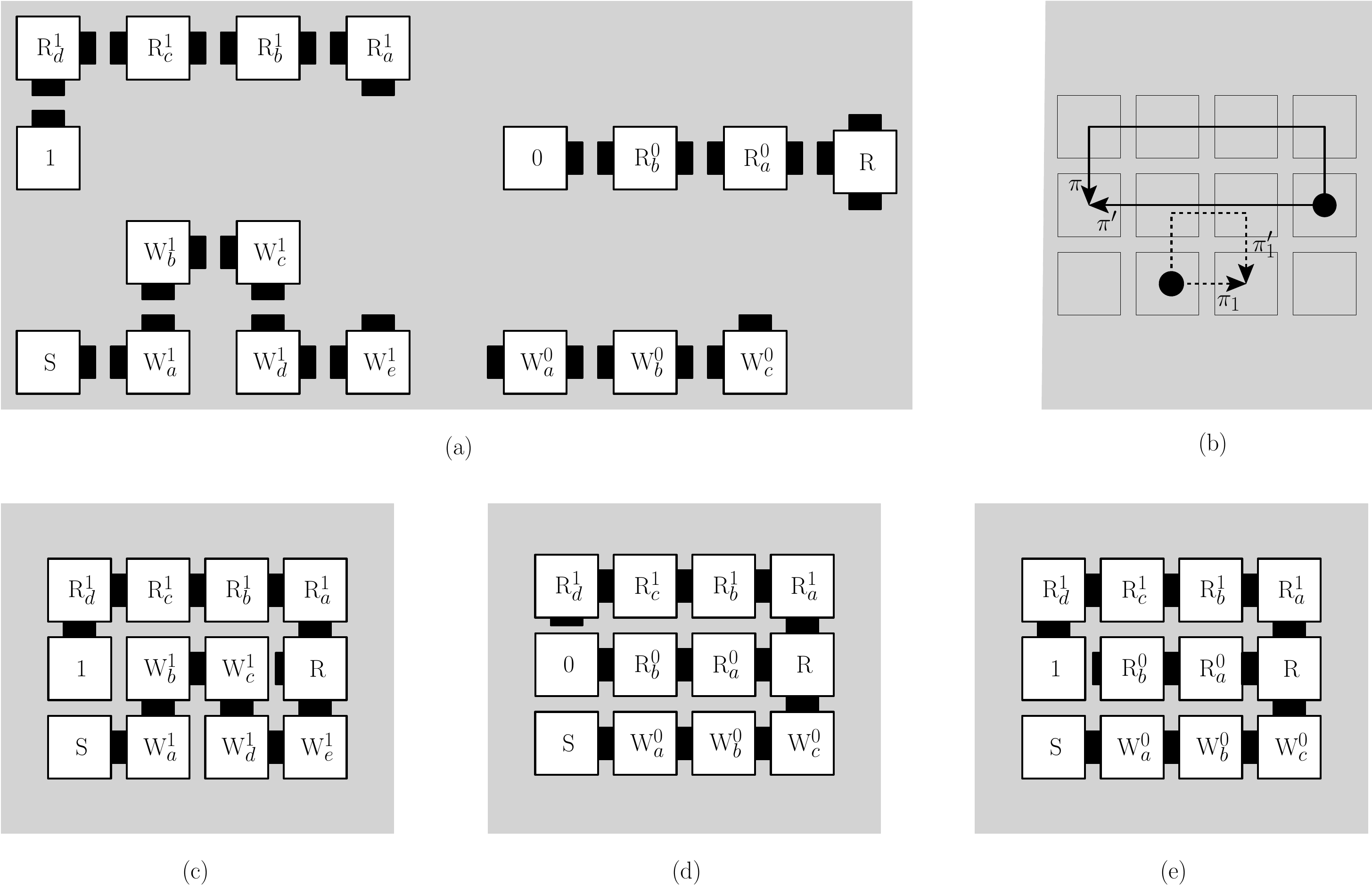}
      \end{minipage}      
      \caption{\label{fig:def1-tas}The TAS $\mathcal{T}=(T,\sigma,1)$
        with $|\termasm{T}|$ = 3 we are using to illustrate Definition
        1. \\(a) Depiction of the tile set $T$=\{S, R, 0, 1, W$^0_a$,
        W$^0_b$, W$^0_c$, W$^1_a$, W$^1_b$, W$^1_c$, W$^1_d$, W$^1_e$,
        R$^0_a$, R$^0_b$, R$^1_a$, R$^1_b$, R$^1_c$, R$^1_d$\} in
        which S is the seed tile, the subset of W$^i$ (resp., R$^i$)
        tiles represent the writing (resp., reading) of an $i$-valued
        bit, for $i \in \{0,1\}$, R is the tile that transitions from
        writing mode to reading mode in all assembly sequences that
        place it, and 0 (resp., 1) is the tile representing the fact
        that the bit value 0 (resp., 1) was read. The geometric
        representation of a bit used here, namely the ``bump'' formed
        by the tiles W$^1_b$ and W$^1_c$ to represent a 1 and the
        absence of such a bump to represent a 0, is the same as the
        one we use in our main construction in
        Section~\ref{sec:square}. \\ (b) The twelve points (i.e.,
        light gray squares) that make up the domain of each assembly
        in $\termasm{T}$, together with two pairs of candidate
        competing paths in $\mathbb{Z}^2$, namely $\pi$ and $\pi'$
        depicted as solid arrows, and $\pi_1$ and $\pi'_1$ depicted as
        dashed arrows. The starting and ending points of each path are
        represented by a black disk and an arrowhead,
        respectively. The seed assembly of $\mathcal{T}$ is $\sigma =
        \{(\vec{p}, $S$)\}$, where $\vec{p}$ is the point depicted in
        the bottom-left corner of the figure.\\ (c) Terminal assembly
        $\alpha_{1,1}$ of $\mathcal{T}$ in which a bit value of 1 is
        both written and then read back. \\ (d) Terminal assembly
        $\alpha_{0,0}$ of $\mathcal{T}$ in which a bit value of 0 is
        both written and then read back. \\ (e) Terminal assembly
        $\alpha_{0,1}$ of $\mathcal{T}$ in which a bit value of 0 is
        written and then mistakenly read back as a 1.  }
\end{figure}


\begin{table}[!h]
\begin{minipage}{\linewidth}
      \centering

  \begin{tabular}{|c|c|c|c|c|c|c|c|c|}\hline
    & \multicolumn{7}{|c|}{Conditions in Definition~\ref{def:competing-paths}} & \\\cline{2-9}
    Pair of paths  & \ref{def:cp-1} & \ref{def:cp-2} & \ref{def:cp-3} & \ref{def:cp-4} & \ref{def:cp-5} & \ref{def:cp-6a} & \ref{def:cp-6b} & Competing paths? \\\hline
    $\{\pi,\pi'\}$ & \checkmark & \checkmark & \checkmark & \checkmark & \checkmark & \checkmark & \checkmark & yes \\ \hline
    $\{\pi_1,\pi'_1\}$ & \checkmark & \checkmark & \checkmark & \ding{55} & \checkmark & \ding{55} & \checkmark & no \\ \hline

  \end{tabular}
\end{minipage}
\caption{\label{tab:def1-tas}Conditions of Definition~\ref{def:competing-paths} satisfied by the pairs of paths depicted in Figure~\ref{fig:def1-tas}b}
\end{table}

\begin{definition} 
\label{def:point-of-competition}
Let $\vec{y} \in \mathbb{Z}^2$. We call the point $\vec{y}$ a \emph{point of competition}, or \emph{POC} for short, \emph{in} $\mathcal{T}$, if there exist $\vec{x} \in \mathbb{Z}^2$ and simple paths $\pi$ and $\pi'$ in $G^{\textmd{f}}_{\mathbb{Z}^2}$ that are competing for $\vec{y}$, from $\vec{x}$ in $\mathcal{T}$. We say that $\pi$ and $\pi'$ are the \emph{competing paths that correspond to} $\vec{y}$, and $\vec{x}$ is the corresponding starting point. 
\end{definition}

Figure~\ref{fig:def-example-tas-POC} depicts the two POCs and the
two corresponding pairs of competing paths in our running example.
\begin{figure}[!h]
  \begin{minipage}{\linewidth}
    \centering
        \includegraphics[width=2in]{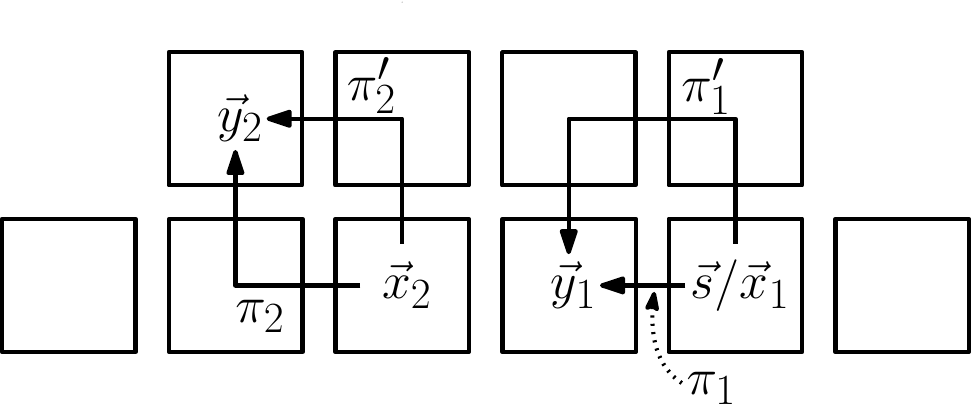}
      \end{minipage}      
      \caption{\label{fig:def-example-tas-POC}In our running example
        $\mathcal{T}$, the points $\vec{y}_1$ and $\vec{y}_2$ are two
        distinct points of competition, with corresponding starting
        points $\vec{x}_1 = \vec{s}$ and $\vec{x}_2$,
        respectively. The figure shows two pairs of competing paths,
        namely the path $\pi_1$ that is competing for $\vec{y}_1$ from
        $\vec{x}_1$ in $\mathcal{T}$, the path $\pi'_1$ that is also
        competing for $\vec{y}_1$ from $\vec{x}_1$ in $\mathcal{T}$,
        the path $\pi_2$ that is competing for $\vec{y}_2$ from
        $\vec{x}_2$ in $\mathcal{T}$, and the path $\pi'_2$ that is
        also competing for $\vec{y}_2$ from $\vec{x}_2$ in
        $\mathcal{T}$.  }
\end{figure}

\begin{lemma}  
\label{lem:only-two-competing-paths}
If $\vec{y} \in \mathbb{Z}^2$ is any POC in $\mathcal{T}$ with corresponding starting point $\vec{x} \in \mathbb{Z}^2$ and corresponding competing paths $\pi$ and $\pi'$, then $\pi$ and $\pi'$ are the only two paths competing for $\vec{y}$, from $\vec{x}$ in $\mathcal{T}$.
\end{lemma}

\begin{proof}
Let $\mathcal{T} = (T, \sigma, 1)$ be a TAS and $\vec{y} \in\mathbb{Z}^2$ be any POC in $\mathcal{T}$ with corresponding starting point $\vec{x} \in \mathbb{Z}^2$ and corresponding competing paths $\pi$ and $\pi'$.  
By condition~\ref{def:cp-1} in Definition~\ref{def:competing-paths} applied to $\pi$ and $\pi'$, these two paths are distinct.
Now, for the sake of obtaining a contradiction, let $\pi''$ be a simple path in $\fullgridgraph_{\mathbb{Z}^2}$ from $\vec{x}$ to $\vec{y}$ such that $\pi''$ is a third path competing for $\vec{y}$, from $\vec{x}$ in $\mathcal{T}$. 
Thus: 
\begin{enumerate} 
	\item $\pi''$ and $\pi$ are competing paths for $\vec{y}$, from $\vec{x}$ in
$\mathcal{T}$, and 
	\item $\pi''$ and $\pi'$ are competing paths for $\vec{y}$, from $\vec{x}$ in $\mathcal{T}$.  
\end{enumerate} 
By condition~\ref{def:cp-1} of Definition~\ref{def:competing-paths}, $\pi''$ is distinct from both $\pi$ and $\pi'$.
Let $\alpha \in \mathcal{A}[\mathcal{T}]$ be such that $\vec{x} \in \dom{\alpha}$.
By condition~\ref{def:cp-3} of Definition~\ref{def:competing-paths}, such an $\alpha$ exists.
By condition~\ref{def:cp-6a} of Definition~\ref{def:competing-paths} applied to $\pi''$, there exists a $\mathcal{T}$-assembly sequence $\vec{\alpha}'' = \left(\alpha''_i \mid 1 \leq i \leq k \right)$ such that $\alpha''_1 = \left\{ \left( \vec{x}, \alpha\left( \vec{x} \right) \right) \right\}$ and $\dom{\alpha''_k} = \dom{\pi''}$.
But condition~\ref{def:cp-6b} of Definition~\ref{def:competing-paths} applied to $\pi$ and $\pi'$ says that for all $\mathcal{T}$-assembly sequences $\vec{\alpha} = \left( \alpha_i \mid 1 \leq i \leq k\right)$, if $\alpha_1 = \left\{ \left(\vec{x}, \alpha\left(\vec{x}\right) \right) \right\}$ and $p$ is a simple path from $\vec{x}$ to $\vec{y}$ in $G^{\textmd{b}}_{\alpha_k}$, then $p \in \left\{ \pi, \pi' \right\}$.
Since $\pi''$ is distinct from both $\pi$ and $\pi''$, such a $\mathcal{T}$-assembly sequence $\vec{\alpha}''$ cannot exist. 

\end{proof}

\goingforward{let $r \in \mathbb{Z}^+$ and assume $\mathcal{T}$ has a non-empty, possibly infinite set $P$ of POCs with $|P| \geq r$.}

\begin{definition} 
\label{def:winner-function}
Let $Y = \left\{\vec{y}_1, \ldots, \vec{y}_r\right\}$ be an $r$-element subset $P$. A \emph{winner function} (for $\mathcal{T}$) is a function of the form $w: Y \rightarrow T$.
Having defined $w$, we call the set $Y$ the set of {\em essential} POCs in $\mathcal{T}$.
\end{definition}
\goingforward{assume $Y = \left\{\vec{y}_1, \ldots, \vec{y}_r\right\}$ is the $r$-element set of essential POCs in $\mathcal{T}$ and $S_Y = \left\{ \vec{x}_1, \ldots, \vec{x}_r \right\}$ is the corresponding set of starting points.}

For example, a TAS could be designed to strictly self-assemble an $N \times N$ square such that its set of essential POCs is a particular subset of $\{0,\ldots,N-1\}^2$. However, it may induce infinite, erroneous $\mathcal{T}$-assembly sequences that attach tiles at infinitely many POCs. 
The motivation behind Definition~\ref{def:winner-function} is that $\mathcal{T}$ may produce multiple (possibly infinite) terminal assemblies and we use $w$ to determine the correctness of such an assembly in the following sense.
\begin{definition} 
\label{def:correct-assembly}
Let $\alpha \in \mathcal{A}[\mathcal{T}]$ and $w:Y \rightarrow T$ be a winner function for $\mathcal{T}$. If:
\begin{enumerate}
\item \label{def:correct-assembly-1}for all $\vec{p} \in \dom{\alpha} \cap Y$, $\alpha\left(\vec{p}\,\right) = w\left(\vec{p}\,\right)$ and
\item  \label{def:correct-assembly-2}$\dom{\alpha} \cap (P\backslash Y)= \emptyset$,
\end{enumerate}
  then we say that $\alpha$ is a $w$-\emph{correct} $\mathcal{T}$-\emph{assembly}.
If $\vec{\alpha}$ is any $\mathcal{T}$-assembly sequence such that $\res{\vec{\alpha}} = \alpha$, then we say that $\vec{\alpha}$ is a $w$-\emph{correct} $\mathcal{T}$-\emph{assembly sequence}.
\end{definition}

Condition~\ref{def:correct-assembly-1} of this definition means that
$\alpha$ may only place correct tiles at the essential POCs included in its domain. Condition~\ref{def:correct-assembly-2} means that a tile cannot attach to an inessential POC unless it is attaching to a $\mathcal{T}$-producible assembly that is not $w$-correct.

In our running example, $w=\{(\vec{y}_1,3),(\vec{y}_2,7)\}$ and
    $w'=\{(\vec{y}_1,4),(\vec{y}_2,7))\}$ are two possible winner functions.
    Then the assembly $\alpha_{3,7} \in \mathcal{A}_{\Box}[\mathcal{T}]$, shown in
    Figure~\ref{fig:def-example-tas}c, is $w$-correct, whereas
    the assembly $\alpha_{4,7} \in \mathcal{A}_{\Box}[\mathcal{T}]$, shown in
    Figure~\ref{fig:def-example-tas}e, is $w'$-correct.

\begin{definition} 
\label{def:dd-ad-pt}
Let $\vec{p} \in \Z^2$, $\alpha, \alpha',\beta,\beta' \in \mathcal{A}[\mathcal{T}]$ and $t_\alpha, t_\beta \in T$ such that $\alpha' = \alpha + \left( \vec{p}, t_\alpha \right)$ and $\beta' = \beta + \left( \vec{p}, t_\beta \right)$ are valid tile attachment steps.
We say that $\vec{p}$ is a \emph{directionally deterministic} point in $\mathcal{T}$ if either $\vec{p} = \vec{s}$, or the following two conditions hold:
\begin{enumerate}
	\item \label{def:dd-ad-pt-2} If $\vec{p} \in P$, then the fact that $t_\alpha \ne t_\beta$ implies that:
	\begin{enumerate}
		\item \label{def:dd-ad-pt-2a} $\vec{u}_{\alpha' \backslash \alpha} \ne \vec{u}_{\beta' \backslash \beta}$, or
		\item \label{def:dd-ad-pt-2b} $\alpha\left( \vec{p} + \vec{u}_{ \alpha' \backslash \alpha }\right) \ne \beta\left( \vec{p} + \vec{u}_{ \beta' \backslash \beta} \right)$.
	\end{enumerate}
	\item \label{def:dd-ad-pt-1} If $\vec{p} \in Z^2 \backslash P$, then the fact that $t_\alpha \ne t_\beta$ implies that $\alpha$ and $\beta$ do not agree.
	
\end{enumerate}
\end{definition}

Condition~\ref{def:dd-ad-pt-1} of this definition means that, when two $\mathcal{T}$-producible assemblies agree, no more than one tile may attach at any non-POC that belongs to the intersection of their frontiers.
Condition~\ref{def:dd-ad-pt-2}
says that, if any two distinct tiles ever attach to any
$\mathcal{T}$-producible assemblies at any POC $\vec{p}$, then these tiles
must attach via glues on different sides or bind to different tile
types.  Note that the definition of a directionally deterministic
point $\vec{p}$ allows for different $\mathcal{T}$-producible
assemblies to place tiles of different types at $\vec{p}$ (see
Figure~\ref{fig:intro-example1}) or to place tiles of the same type at
$\vec{p}$ via different input glues (see
Figure~\ref{fig:dd-TAS-example4}). We now discuss these two examples.

The TAS shown in Figure~\ref{fig:intro-example1} produces two terminal
assemblies, respectively placing the 3 tile at the point $\vec{p}$ in
the bottom-left corner, as depicted in
Figure~\ref{fig:intro-example1}(b), or placing the 4 tile at the same
point, as depicted in Figure~\ref{fig:intro-example1}(c). Note that,
in this example, the 3 tile binds only via its north glue, whereas the
4 tile binds only via it east glue. Therefore, the point $\vec{p}$ is
directionally deterministic in this TAS.
  
As another example, the TAS $\mathcal{T}$ depicted in
Figure~\ref{fig:dd-TAS-example4} illustrates the importance of
condition~\ref{def:dd-ad-pt-2b} in Definition~\ref{def:dd-ad-pt}. As shown in
part (b) and (c) of the figure, the point
$\vec{p}+\vec{u}_{\alpha'\backslash\alpha}$ is a POC in
$\mathcal{T}$, with the seed tile's point as its corresponding
starting point. The path $\pi$ followed by the tiles s and 3' and the
path $\pi'$ followed by the tiles s, 1, 2 and 3 are the two competing
paths for this POC. Since the distinct tiles 4 and 4' can both be
placed at $\vec{p}$, $\vec{u}_{\alpha'\backslash\alpha}
=\vec{u}_{\beta'\backslash\beta}$, and
$\alpha\left(\vec{p}+\vec{u}_{\alpha'\backslash\alpha}\right)
= \mathrm{tile\ 3} \neq \mathrm{tile\ 3'}
= \beta\left(\vec{p}+\vec{u}_{\beta'\backslash\beta}\right)$, point
$\vec{p}$ satisfies condition~\ref{def:dd-ad-pt-2b}. Suppose the set of
tiles $\{s,1,2,3,3'\}$ together form, for example, a bit reading
``gadget'' similar to those described in
Section~\ref{sec:square} (see Step~2c in the proof of Lemma~\ref{lem:high-probability-counter}). Then the tile placed at the POC via the winning path $\pi$ is able to propagate
different information from the one propagated by the tile placed at
the POC via $\pi'$ (e.g., a different value of the carry bit) through
tile 4 or 4' to another gadget on its right (not shown in the figure)
in order to achieve a bigger goal, say, incrementing the value of a counter.

\begin{figure}[!ht]
  \centering \includegraphics[width=.7\linewidth]{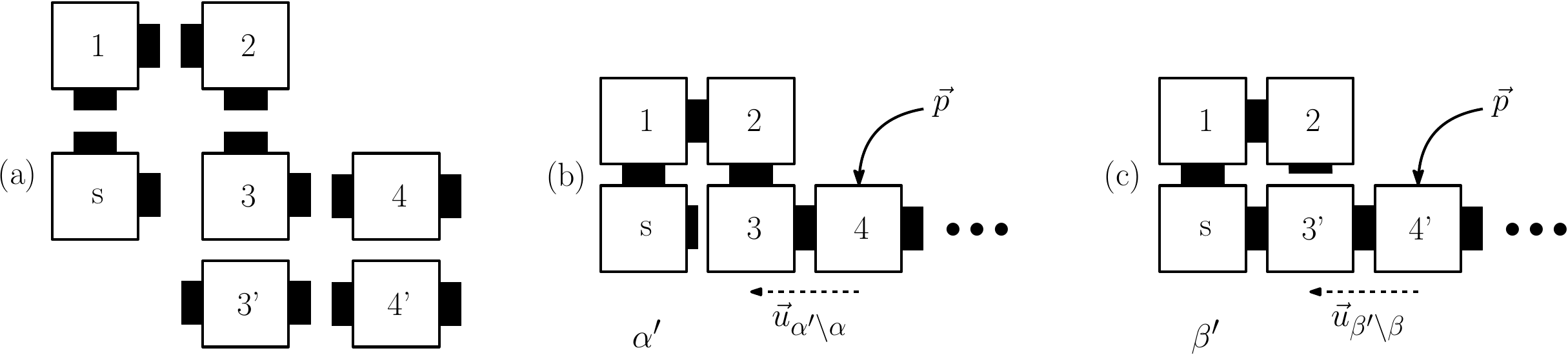} \caption{\label{fig:dd-TAS-example4}
  Example of a TAS $\mathcal{T}$ with the tile set shown in (a), in
  which the temperature is 1 and s is the seed tile. Sub-figures (b)
  and (c) show the only two terminal assemblies of
  $\mathcal{T}$. Each one of the points at which $\mathcal{T}$
  may place tiles is directionally deterministic in
  $\mathcal{T}$. 
}
\end{figure}

\begin{definition}
\label{def:dd-ad-tas}
We say that $\mathcal{T}$ is a \emph{directionally deterministic} TAS if all points in the domains of all $\mathcal{T}$-producible assemblies are directionally deterministic.
\end{definition}

The next lemma says that if $\mathcal{T}$ is directionally deterministic and $\alpha$ is a $w$-correct $\mathcal{T}$-terminal assembly, then any $w$-correct $\mathcal{T}$-producing assembly sequence must result in a subassembly of $\alpha$. 

\begin{lemma} 
\label{lem:w-correct-subassembly-w-correct}
Let $w:Y \rightarrow T$ be a winner function for $\mathcal{T}$ and $\alpha \in \mathcal{A}_{\Box}[\mathcal{T}]$ be $w$-correct. If $\mathcal{T}$ is 
directionally deterministic 
and $\beta$ is any $w$-correct, $\mathcal{T}$-producible assembly, then $\beta \sqsubseteq \alpha$.
\end{lemma}

In fact, $\mathcal{T}$ need only satisfy condition~\ref{def:dd-ad-pt-1} of Definition~\ref{def:dd-ad-pt} for all points in the domains of all of its producible assemblies in order for the following proof to go through.

\begin{proof}
Let $m \in \mathbb{Z}^+ \cup \{ \infty \}$.
Assume $\beta = \res{\vec{\beta}}$ with $\vec{\beta} = \left( \beta_i \mid 0 \leq i - 1 < m\right)$.
We will show that for all $k \in \mathbb{Z}^+$ such that $0 \leq k - 1 < m$, $\beta_k \sqsubseteq \alpha$.
We proceed by induction on $k$.
For $k = 1$, we have $\vec{\beta} = \left(\beta_1\right) = \left( \sigma \right)$ and thus $\beta_1 \sqsubseteq \alpha$. 
Now, assume $k$ is some integer with $1 \leq k < m$ and $\beta_k \sqsubseteq \alpha$.
We will show that $\beta_{k+1} \sqsubseteq \alpha$.
From $k < m$, it follows that $\partial^{\mathcal{T}} \beta_k \ne \emptyset$.
This means there exist $\vec{p} \in \mathbb{Z}^2$ and $t_\beta \in T$ such that $\beta_{k+1} = \beta_k + \left( \vec{p}, t_\beta \right)$.
Note that  $\vec{p} \in \dom{\alpha}$, since  $\beta_k \sqsubseteq \alpha$,  $p \in \partial^{\mathcal{T}} \beta_k$, and $\alpha$ is terminal.
Let $t_{\alpha} = \alpha(\vec{p}\,)$.
Since either $\vec{p} \in Y$ or $\vec{p} \notin Y$, we consider each case in turn.
\begin{enumerate}
\item Suppose $\vec{p} \in Y$.          
  Since, in the hypothesis of the lemma, $\beta$ is assumed to be $w$-correct, we have $t_\beta = \beta_{k+1}\left( \vec{p}\, \right) = w\left( \vec{p}\, \right)$.
  Since, in the hypothesis of the lemma, $\alpha$ is assumed to be $w$-correct, we have $t_{\alpha} = w\left( \vec{p}\, \right)$.
  Therefore, $t_\beta = \beta_{k+1}\left( \vec{p}\, \right) = w\left( \vec{p}\, \right) = t_{\alpha}$, which, combined with the inductive hypothesis, yields $\beta_{k+1} \sqsubseteq \alpha$. 
\item Suppose $\vec{p} \not \in Y$. Note that  $\beta_{k+1} \sqsubseteq \alpha$ directly follows from our inductive hypothesis (i.e., $\beta_k \sqsubseteq \alpha$), as long as $t_\alpha = t_\beta$. We now prove this fact by establishing the conditions of, and then applying, 
condition~\ref{def:dd-ad-pt-1} of Definition~\ref{def:dd-ad-pt}.

  First, since $\vec{p} \not \in Y$, $\vec{p} \in \dom{\alpha}$, and $\alpha$ is $w$-correct, condition~\ref{def:correct-assembly-2} of Definition~\ref{def:correct-assembly} implies that $\vec{p} \not \in P$.

  Second, let $\beta' = \beta_k$. Thus, $\beta' + \left(\vec{p}, t_\beta \right)$ is a valid tile attachment step. Furthermore, since $\vec{p} \in \dom{\alpha}$, there exist a $\mathcal{T}$-producible assembly $\alpha' \sqsubseteq \alpha$ such that $\alpha' + \left( \vec{p}, t_\alpha \right)$ is a valid tile attachment step.
  
  Third, since $\alpha' \sqsubseteq \alpha$ and, by our inductive hypothesis, $\beta' = \beta_k \sqsubseteq \alpha$, it follows that $\alpha'$ and $\beta'$ agree.
  
  Finally, since $\mathcal{T}$ is 
  directionally deterministic, 
  condition~\ref{def:dd-ad-pt-1} of Definition~\ref{def:dd-ad-pt}
   implies that $t_\alpha = t_\beta$, which concludes the proof of this case.
\end{enumerate}	
\end{proof}

%

\begin{corollary} 
\label{cor:unique-w-correct-assembly}
 If $\mathcal{T}$ is directionally deterministic
 and there exists a $w$-correct assembly $\alpha \in \mathcal{A}_{\Box}[\mathcal{T}]$, then $\alpha$ is unique.
\end{corollary}

\begin{proof}
Assume that $\alpha \in \mathcal{A}_{\Box}[\mathcal{T}]$ is $w$-correct. Therefore, $\left|\left\{ \alpha \in \mathcal{A}_{\Box}[\mathcal{T}] \mid \alpha \textmd{ is } w\textmd{-correct} \right\}\right| \geq 1$. Suppose, for the sake of obtaining a contradiction, that $\left|\left\{ \alpha \in \mathcal{A}_{\Box}[\mathcal{T}] \mid \alpha \textmd{ is } w\textmd{-correct} \right\}\right| > 1$. Then there exists a $w$-correct assembly $\beta \in \mathcal{A}_{\Box}[\mathcal{T}]$ such that $\alpha \ne \beta$.
By Lemma~\ref{lem:w-correct-subassembly-w-correct}, $\beta \sqsubseteq \alpha$. Together with $\beta \ne \alpha$, this implies  $\partial^{\mathcal{T}} \beta \ne \emptyset$, contradicting the fact that $\beta \in \mathcal{A}_{\Box}[\mathcal{T}]$.
\end{proof}
Note that if $\mathcal{T}$
does not satisfy condition~\ref{def:dd-ad-pt-1} of Definition~\ref{def:dd-ad-pt} for all points in the domains of all of its producible assemblies, then Corollary~\ref{cor:unique-w-correct-assembly} need not hold. 
An example of such a TAS is depicted in Figure~\ref{fig:dd_but_disagree}.
\begin{figure}[h!]
\captionsetup[subfigure]{labelformat=empty}
    \centering
    \begin{subfigure}[t]{\textwidth}
        \centering
        \includegraphics[width=1.45in]{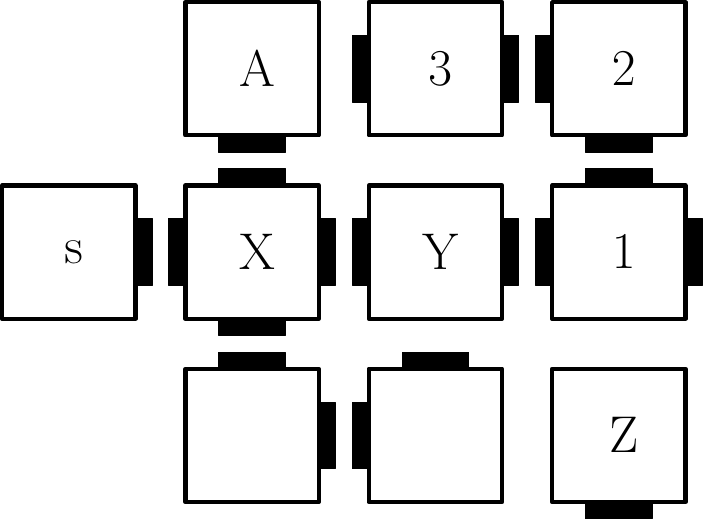}
        \caption{\label{fig:dd_but_disagree_tiles}(a) The tile set,
          with s as the seed tile. This TAS has only one POC, namely
          the point at which the Y tile is placed in the three
          sub-figures below. Tile Z is the other tile that may be
          placed at the POC. We assume $w$ is such that an assembly
          that places the Y tile is $w$-correct.}
    \end{subfigure}
    \\
    \vspace{10pt}
    \hfill
    \begin{subfigure}[t]{0.3\textwidth}
        \centering
        \includegraphics[width=1.8in]{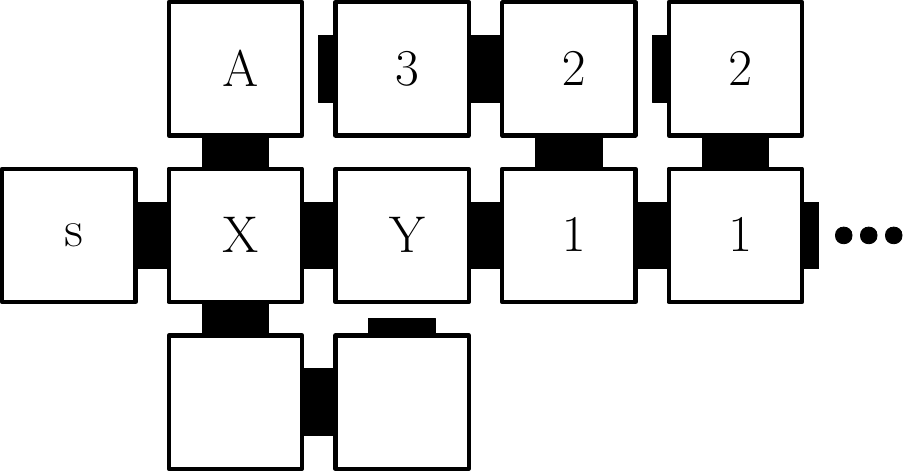}
        \caption{\label{fig:dd_but_disagree_1}(b) A $w$-correct
          assembly. Note that the point at which the A tile is
          placed is not a POC, even though multiple types of tiles can
          be placed at this location, similarly to what happens at a
          POC. }
    \end{subfigure}
    ~
    \begin{subfigure}[t]{0.3\textwidth}
        \centering
        \includegraphics[width=1.8in]{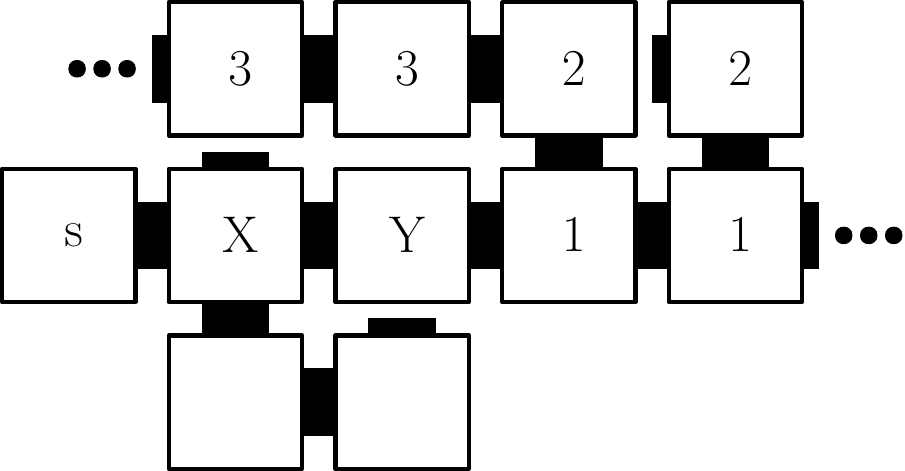}
        \caption{\label{fig:dd_but_disagree_2}(c) Another $w$-correct
          assembly in which a tile of type 3 is placed at the point
          of the A tile from
          part~(\subref{fig:dd_but_disagree_1}). Note that this
          location is not a POC because
          conditions~\ref{def:cp-4},~\ref{def:cp-5},
          and~\ref{def:cp-6b} of Definition~\ref{def:competing-paths}
          do not hold. }
    \end{subfigure}
  \hfill  
  \begin{subfigure}[t]{0.3\textwidth}
        \centering
        \includegraphics[width=1.8in]{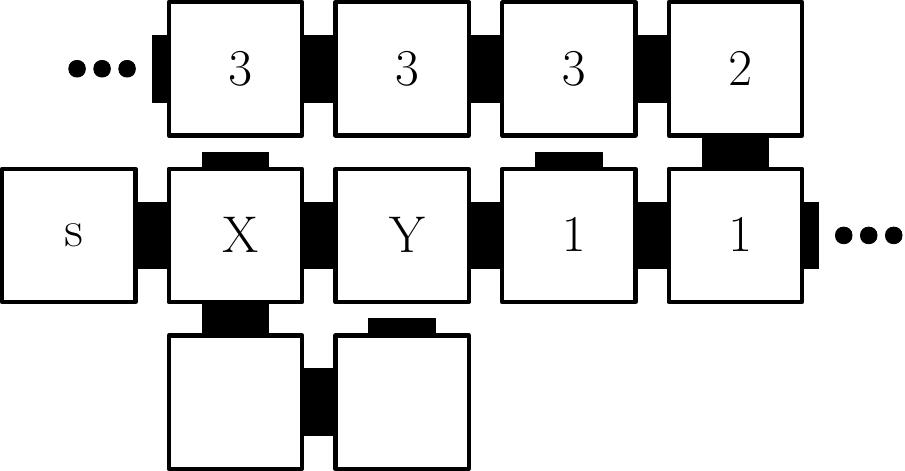}
        \caption{\label{fig:dd_but_disagree_3}(d) Another $w$-correct
          assembly in which a tile of type 3 is placed at the point
          of the A tile along a path that blocks the path from X to 3
          shown in part~(\subref{fig:dd_but_disagree_2}). }
    \end{subfigure}
    ~
      \caption{\label{fig:dd_but_disagree} 
     An example of a TAS that 
does not satisfy condition~\ref{def:dd-ad-pt-1} of Definition~\ref{def:dd-ad-pt} for all points in the domains of all of its producible assemblies,
and produces more than one
        $w$-correct terminal assembly. Its tile set is shown in
        part~(\subref{fig:dd_but_disagree_tiles});
        parts~(\subref{fig:dd_but_disagree_1})
        through~(\subref{fig:dd_but_disagree_3}) depict distinct
        $w$-correct assemblies. Although the depicted assemblies are
        not terminal, it is easy to see how each one can be extended
        to an infinite terminal assembly. In fact, this TAS induces an
        infinite number of (infinite)  $w$-correct terminal assemblies.}
        \end{figure}
\goingforward{assume $w: Y \rightarrow T$ is a winner function for $\mathcal{T}$. }

The following lemma says that any $w$-correct, finite $\mathcal{T}$-producing assembly sequence of a directionally deterministic TAS
that results in a non-terminal assembly can always be extended to a sequence whose result is the $w$-correct $\mathcal{T}$-terminal assembly $\alpha$, which, by Corollary~\ref{cor:unique-w-correct-assembly}, is unique.
\begin{lemma} 
\label{lem:w-correct-extended-to-w-correct-terminal}
Assume $\mathcal{T}$ is directionally deterministic
and $\alpha$ is the unique $w$-correct $\mathcal{T}$-terminal assembly.
For every finite, $w$-correct, $\mathcal{T}$-producing assembly sequence $\vec{\beta}$ such that $\res{\vec{\beta}} \not \in \mathcal{A}_{\Box}\left[ \mathcal{T} \right]$, there exists an extension $\vec{\alpha}$ of $\vec{\beta}$ by some $\mathcal{T}$-assembly sequence such that $\vec{\alpha}$ results in $\alpha$.
\end{lemma}
Similarly to Lemma~\ref{lem:w-correct-subassembly-w-correct}, $\mathcal{T}$ need only satisfy condition~\ref{def:dd-ad-pt-1} of Definition~\ref{def:dd-ad-pt} for all points in the domains of all of its producible assemblies in order for the following proof, which references Lemma~\ref{lem:w-correct-subassembly-w-correct} and its corresponding corollary, to go through. 
\begin{proof} 
Assume that $\vec{\beta} = \left( \beta_i \mid 1 \leq i \leq k \right)$ for some $k \in \mathbb{Z}^+$ is a $w$-correct, $\mathcal{T}$-producing assembly sequence such that $\res{\vec{\beta}} \not \in \mathcal{A}_{\Box}[\mathcal{T}]$.
Let $\vec{e}=(\vec{p}_1,\vec{p}_2,\ldots)$ be any enumeration of $\mathbb{Z}^2$.
We will use $\alpha$ to construct a $w$-correct extension of $\vec{\beta}$ via the following algorithm.
  \begin{minipage}[t]{\linewidth}
    \begin{algorithm}[H]
      \DontPrintSemicolon
  
       $\vec{\varepsilon} \leftarrow \left( \beta_1, \ldots, \beta_k \right)$
       \tcp*[f]{ initialize $\vec{\varepsilon}$ to $\vec{\beta}$ }
        
        \While{$\res{\vec{\varepsilon\,}} \not \in \mathcal{A}_{\Box}[\mathcal{T}]$}{
        	
			Let $\vec{p}$ be the first point in $\vec{e}$ such that $\vec{p} \in \partial^{\mathcal{T}} \res{\vec{\varepsilon}\,}$
			\tcp*[f]{ choose the next point and \ldots}

			$t \leftarrow \alpha\left( \vec{p}\, \right)$
      	      		\tcp*[f]{\ldots\ tile type to append to $\vec{\varepsilon}$ }
			
			$\vec{\varepsilon} \leftarrow \vec{\varepsilon} + \left( \vec{p}, t \right)$
			\tcp*[f]{ append $\left( \vec{p}, t \right)$ to $\vec{\varepsilon}$ }

        }    
\end{algorithm}
\end{minipage}
Note that:
\begin{itemize}
\item the while loop executes at least one iteration because $\res{\vec{\beta}} = \beta_k$ is not terminal and
\item $\vec{\varepsilon}$ is always an extension of $\vec{\beta}$ because it is initialized to $\vec{\beta}$ in step 1 and at least one tile attachment step is added to it in step 5.
\end{itemize}
Now, consider the following invariant. Prior to each iteration of the while loop:
\begin{enumerate}
	\item $\vec{\varepsilon}$ is a $\mathcal{T}$-producing assembly sequence and
	\item $\vec{\varepsilon}$ is $w$-correct.
\end{enumerate}
We will show that this invariant holds prior to each iteration of the while loop.
First, note that the invariant holds prior to the first iteration of the while loop because $\vec{\varepsilon}$ is initialized to $\vec{\beta}$, and $\vec{\beta}$ is assumed to be a $w$-correct, $\mathcal{T}$-producing assembly sequence.
Now, assume that the invariant holds prior to some iteration of the while loop, and let $\varepsilon = \res{\vec{\varepsilon}\,}$ just prior to that iteration.
In this case, $\varepsilon \not \in \mathcal{A}_{\Box}[\mathcal{T}]$, which means there exists $\vec{p} \in \partial^{\mathcal{T}} \varepsilon$.
In particular, let $\vec{p}$ be the first point in $\vec{e}$ such that $\vec{p} \in \partial^{\mathcal{T}} \varepsilon$.
By Lemma~\ref{lem:w-correct-subassembly-w-correct}, $\varepsilon \sqsubseteq \alpha$.
Since, after step~4 of the algorithm, $t=\alpha\left(\vec{p}\,\right)$ and the temperature of $\mathcal{T}$ is 1, $\vec{p} \in \partial^{\mathcal{T}}_t \varepsilon$.
This means $\vec{\varepsilon} + \left(\vec{p},t\right)$ is a valid tile attachment step. Thus, after step 5 of the algorithm has executed:
\begin{enumerate}
	\item $\vec{\varepsilon}$ is a $\mathcal{T}$-producing assembly sequence         and 
	\item $\vec{\varepsilon}$ is $w$-correct. 
\end{enumerate}
The second part of the invariant still holds because either $\vec{p} \notin Y$ or else $t = \alpha\left(\vec{p}\,\right) = w\left(\vec{p}\,\right)$,  because $\alpha$ is $w$-correct.
Thus, the invariant holds prior to the next iteration of the while loop. 
If $\alpha$ is finite, then the while loop eventually terminates with $\vec{\varepsilon}$ being a $w$-correct, $\mathcal{T}$-producing assembly sequence such that $\res{\vec{\varepsilon}\,} \in \mathcal{A}_{\Box}[\mathcal{T}]$.
Let $\varepsilon$ denote $\res{\vec{\varepsilon}\,}$ after the while loop terminates.
Then, by Lemma~\ref{lem:w-correct-subassembly-w-correct}, $\varepsilon \sqsubseteq \alpha$.
But since $\varepsilon \in \mathcal{A}_{\Box}[\mathcal{T}]$, and in combination with Corollary~\ref{cor:unique-w-correct-assembly}, we must have $\varepsilon = \alpha$.
Now, consider the case where $\alpha$ is infinite.
Then the while loop never terminates. We let $\vec{\varepsilon}$ denote $\left( \varepsilon_i \mid 1 \leq i \right)$ such that, for all $1 \leq i \leq k$, $\varepsilon_i = \beta_i$ and, for all $i > k$, $\varepsilon_i=\res{\vec{\varepsilon}_i}$, where $\res{\vec{\varepsilon}_i}$ is the result of the execution of step~5 during the $(i-k)^{th}$ iteration of the while loop.
The loop invariant ensures that $\vec{\varepsilon}_i$ satisfies the hypotheses of Lemma~\ref{lem:w-correct-subassembly-w-correct}, which means that, for all $i \geq 1$, $\varepsilon_i \sqsubseteq \alpha$.
Let $\varepsilon = \res{\vec{\varepsilon}\,}$.
We now show that $\partial^{\mathcal{T}}\varepsilon = \emptyset$.
Assume for the sake of obtaining a contradiction that 
$\partial^{\mathcal{T}}\varepsilon \ne \emptyset$. Let $\vec{x} \in \partial^{\mathcal{T}}\varepsilon$.
Since $\vec{e}$ is an enumeration of $\mathbb{Z}^2$, there exists $i \geq 1$ such that $\vec{x}$ is the first point in $\vec{e}$ with $\vec{x} \in \partial^{\mathcal{T}} \varepsilon_i$.
In this case, some tile must attach  at $\vec{x}$ in step 5 of the corresponding iteration of the loop, contradicting our assumption that $\partial^{\mathcal{T}}\varepsilon \ne \emptyset$.
Thus, $\partial^{\mathcal{T}}\varepsilon = \emptyset$, which means that $\varepsilon \in \mathcal{A}_{\Box}[\mathcal{T}]$.
\end{proof}
The following definition is the cornerstone of the framework we introduce in this paper. For the sake of completeness, we explicitly state all of the assumptions that were instantiated in previous ``Going forward'' boxes. 
\begin{definition} 
\label{def:seq-non-deterministic}
Assume $\mathcal{T} = (T,\sigma,1)$ is a singly-seeded TAS with $\dom{\sigma} = \left\{\vec{s} \, \right\}$, 
$r \in \mathbb{Z}^+$, $\mathcal{T}$ has a non-empty, possibly infinite set $P$ of POCs with $|P| \geq r$, 
$S_P$ denotes the set $\left\{ \vec{x} \mid \vec{x} \textmd{ is the starting point for some } \vec{y} \in P \right\}$, 
$Y = \left\{ \vec{y}_1, \ldots \vec{y}_r \right\}$ is the $r$-element set of essential POCs in $\mathcal{T}$ with
corresponding set of starting points $S_Y = \left\{\vec{x}_1, \ldots, \vec{x}_r\right\}$, and 
$w : Y \rightarrow T$ is a winner function for $\mathcal{T}$ such that $\alpha$ is the unique $w$-correct $\mathcal{T}$-terminal assembly with $Y \subseteq \dom{\alpha}$.
We say that $\mathcal{T}$ is $w$-\emph{sequentially non-deterministic} if:
\begin{enumerate}
	\item \label{def:snd-1} $\mathcal{T}$ is directionally deterministic, 
    \item \label{def:snd-3} for all $\mathcal{T}$-producing assembly sequences $\vec{\alpha}$,
     if 
     $Y \subseteq \dom{\res{\vec{\alpha}}}$,
     then for all integers $1 \leq i < r$, $\textmd{index}_{\vec{\alpha}}\left( \vec{y}_i \right) < \textmd{index}_{\vec{\alpha}}\left( \vec{y}_{i+1} \right)$, and
	\item \label{def:snd-4} $S_P \cap P = \emptyset$.
\end{enumerate}
\end{definition}

The idea is that if $\mathcal{T}$ is $w$-sequentially non-deterministic, then it produces $\alpha$ in a sequential, mostly-deterministic fashion.

\goingforward{assume $\mathcal{T}$ is $w$-sequentially non-deterministic.}
The next two lemmas establish that, in every $\mathcal{T}$-producing assembly sequence, a tile can only be placed at a POC after a tile has been placed at the corresponding starting point (Lemma~\ref{lem:x-i-y-i-ordering-r-1}) but always before a tile is placed at the starting point corresponding to the next POC (Lemma~\ref{lem:x-i-y-i-ordering}).
\begin{lemma} 
\label{lem:x-i-y-i-ordering-r-1}
Let $i \in \mathbb{Z}^+$. For all $\mathcal{T}$-producing assembly sequences $\vec{\beta}$, if $Y \subseteq \dom{\res{\vec{\beta}}}$ and $1 \leq i \leq r$, then $\vec{x}_i \in \dom{\res{\vec{\beta}}}$ and $\textmd{index}_{\vec{\beta}}\left( \vec{x}_i \right) < \textmd{index}_{\vec{\beta}}\left( \vec{y}_i \right)$.
\end{lemma}

\begin{proof}
  Let $\vec{\beta}$ be any $\mathcal{T}$-producing assembly sequence with $\beta = \res{\vec{\beta}}$. 
  Assume that $\left\{ \vec{y}_1, \ldots, \vec{y}_r \right\} \subseteq \dom{\beta}$. 
  Let $i$ be any integer such that $1 \leq i \leq r$.
  Since $\vec{y}_i \in \dom{\beta}$, there exists a simple path from
  $\vec{s}$ to $\vec{y}_i$ in
  $G^{\textmd{b}}_{\beta}$.  Condition~\ref{def:cp-5} in
  Definition~\ref{def:competing-paths} implies that $\vec{x}_i \in
  \dom{\beta}$ and $\textmd{index}_{\vec{\beta}}\left( \vec{x}_i
  \right) \leq \textmd{index}_{\vec{\beta}}\left( \vec{y}_i
  \right)$. Finally, condition~\ref{def:snd-4} in
  Definition~\ref{def:seq-non-deterministic} implies that
  $\textmd{index}_{\vec{\beta}}\left( \vec{x}_i \right) <
  \textmd{index}_{\vec{\beta}}\left( \vec{y}_i \right)$.
\end{proof}

\begin{lemma} 
\label{lem:x-i-y-i-ordering}
Let $i \in \mathbb{Z}^+$ and $r > 1$. For all $\mathcal{T}$-producing assembly sequences $\vec{\beta}$, if $Y \subseteq \dom{\res{\vec{\beta}}}$ and $1 \leq i < r$, then $\vec{x}_i, \vec{x}_{i+1} \in \dom{\res{\vec{\beta}}}$ and $\textmd{index}_{\vec{\beta}}\left( \vec{x}_i \right) < \textmd{index}_{\vec{\beta}}\left( \vec{y}_i \right) < \textmd{index}_{\vec{\beta}}\left( \vec{x}_{i+1} \right)$.
\end{lemma}  

\begin{proof}

	Let $\vec{\beta}$ be any $\mathcal{T}$-producing assembly sequence with
  $\beta = \res{\vec{\beta}}$. Assume that $Y \subseteq \dom{\beta}$. Let
  $i$ be any integer such that $1 \leq i < r$. 
  By Lemma~\ref{lem:x-i-y-i-ordering-r-1}, $\vec{x}_i, \vec{x}_{i+1} \in \dom{\beta}$ and $\textmd{index}_{\vec{\beta}}\left( \vec{x}_i \right) <
  \textmd{index}_{\vec{\beta}}\left( \vec{y}_i \right)$.

  For the sake of obtaining a contradiction, assume that
  $\textmd{index}_{\vec{\beta}}\left( \vec{y}_i \right) \geq
  \textmd{index}_{\vec{\beta}}\left( \vec{x}_{i+1}
  \right)$. Condition~\ref{def:snd-4} of
  Definition~\ref{def:seq-non-deterministic} implies that
  $\vec{y}_i \neq \vec{x}_{i+1}$ and thus
  $\textmd{index}_{\vec{\beta}}\left( \vec{y}_i \right) >
  \textmd{index}_{\vec{\beta}}\left( \vec{x}_{i+1} \right)$.  
  In other words, $\vec{\beta}$ has the form $(\beta_1=\sigma,
  \ldots,\beta_{j-1}, \beta_j, \beta_{j+1},\ldots, \beta_k,\ldots)$
  where $j = \textmd{index}_{\vec{\beta}}\left( \vec{x}_{i+1}
  \right)$, $1 \leq j < |\vec{\beta}|$, $k =
  \textmd{index}_{\vec{\beta}}\left( \vec{y}_{i} \right)$, $1 < k
  \leq |\vec{\beta}|$, and $j < k$.
   Let $\vec{\beta}_p$ denote the (possibly empty) prefix $(\beta_1, \ldots,
   \beta_{j-1})$ of $\vec{\beta}$ and let $\pi$ and $\pi'$ be the two
   paths competing for $\vec{y}_{i+1}$.
   Condition~\ref{def:cp-6a} of Definition~\ref{def:competing-paths}
   implies that there exist two $\mathcal{T}$-assembly sequences
   $\vec{\gamma}$ and $\vec{\gamma}'$ that first place the tile type
   $\beta_j\left(\vec{x}_{i+1}\right)$ at $\vec{x}_{i+1}$ and then
   follow $\pi$ and $\pi'$, respectively.
   We say that {\it $\vec{\beta}_p$ blocks $\pi$ (resp., $\pi'$)} if
   $\vec{\beta}_p$ places at least one tile at a point in $\dom{\pi}$
   (resp., $\dom{\pi'}$) that makes it impossible for $\vec{\beta}_p$ to be 
   extended with $\vec{\gamma}$ (resp., $\vec{\gamma}'$) to yield a
   $\mathcal{T}$-producing sequence that places a tile at
   $\vec{y}_{i+1}$.
   We now consider two cases.

   {\bf Case 1: $\vec{\beta}_p$ blocks both $\pi$ and $\pi'$.}  In
   this case, tile attachment steps in $\vec{\beta_p}$ prevent both
   $\vec{\gamma}$ and $\vec{\gamma}'$ from extending $\vec{\beta}_p$ to
   place a tile at $\vec{y}_{i+1}$.  Since $j =
   \textmd{index}_{\vec{\beta}}\left( \vec{x}_{i+1} \right)$,
   Lemma~\ref{lem:x-i-y-i-ordering-r-1} implies that $\vec{y}_{i+1}
   \notin \dom{\beta_j}$. Thus $\vec{\beta}_p$ is a prefix of
   $\vec{\beta}$ that does not place a tile at $\vec{y}_{i+1}$. If
   there is an extension of $\vec{\beta}_p$ that results in an
   assembly whose domain contains $\vec{y}_{i+1}$ then, by
   condition~\ref{def:cp-5} in Definition~\ref{def:competing-paths},
   the binding graph of such an assembly must contain either $\pi$ or
   $\pi'$, which is ruled out by the fact that both $\pi$ and $\pi'$
   are blocked by $\vec{\beta}_p$ in this case. Therefore, no
   extension of $\vec{\beta}_p$ can result in an assembly that places
   a tile at $\vec{y}_{i+1}$. However, since $\vec{\beta}$ is by
   definition such an extension, we get a contradiction with the
   lemma's assumption that $Y \subseteq \dom{\res{\vec{\beta}}}$.

   {\bf Case 2: $\vec{\beta}_p$ does not block both $\pi$ and
     $\pi'$.}  In this case, WLOG assume that $\vec{\beta}_p$ does
   not block $\pi$, whose length is equal to $l$. Therefore, since
   $\beta_j$ and $\gamma=\res{\vec{\gamma}}$ both place the same tile
   type at $\vec{x}_{i+1}$, it is possible to embed $\vec{\gamma}$ in
   $\vec{\beta}$ by inserting it right after $\beta_j$ to yield the
   $\mathcal{T}$-producing sequence $\vec{\beta}' = (\beta_1=\sigma,
   \ldots, \beta_{j-1},\varepsilon_1,\ldots, \varepsilon_l,\varepsilon_{l+1},\ldots,
   \varepsilon_{l+k-j},\ldots )$ where:\vspace{-2mm}
  \begin{itemize}
  \item $\varepsilon_1 = \beta_j$,
  \item for all integer values of $i$ between 2 and $l$, $\varepsilon_i =
    \varepsilon_{i-1} + (\gamma_{i}\backslash\gamma_{i-1})$, and
  \item for all applicable integer values of $i$ greater than $l$, $\varepsilon_i =
    \varepsilon_{i-1} + (\beta_{i+j-l}\backslash\beta_{i+j-l-1})$.
  \end{itemize}\vspace{-2mm}
In this case, $\textmd{index}_{\vec{\beta}'}\left( \vec{y}_{i+1}
\right) =j-1+l< k-1+l = \textmd{index}_{\vec{\beta}'}\left(
\vec{y}_{i} \right)$, which contradicts condition~\ref{def:snd-3} in
Definition~\ref{def:seq-non-deterministic}.

  In conclusion, since we reached a contradiction in both cases, 
  $\textmd{index}_{\vec{\beta}}\left(
  \vec{y}_{i} \right) < \textmd{index}_{\vec{\beta}}\left(
  \vec{x}_{i+1} \right)$ must hold.

\end{proof}

Figure~\ref{fig:def-seq-nd} depicts visually and explains informally the implications of Definition~\ref{def:seq-non-deterministic} and Lemma~\ref{lem:x-i-y-i-ordering}.
\begin{figure}[!h]
      \centerline{\includegraphics[width=2.5in]{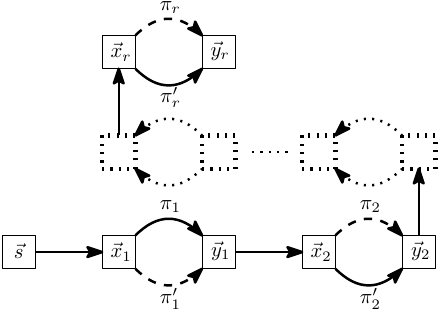}}
    \caption{\label{fig:def-seq-nd} In this example,
      $\mathcal{T}=(T,\sigma,1)$ is a $w$-sequentially non-deterministic
      TAS with $\dom \sigma = \{\vec{s}\, \}$, 
      $r > 2$, and $\alpha$ is the unique $w$-correct, $\mathcal{T}$-terminal assembly whose domain contains all of the essential POCs. Note that each arrow
      represents a sequence of tile attachment steps that follow a simple path.
      Definition~\ref{def:seq-non-deterministic} and Lemmas~\ref{lem:x-i-y-i-ordering-r-1} and~\ref{lem:x-i-y-i-ordering}
      together require that:\\      
      a) [Condition~\ref{def:snd-1}] every point at which a tile attaches in the assembly sequence is directionally deterministic, \\
      b) 
      [Assumption in Definition~\ref{def:seq-non-deterministic}] the $\vec{y}_i$ points are all the essential POCs in $\mathcal{T}$,\\
      c)
      [Condition~\ref{def:snd-3}] in all $\mathcal{T}$-assembly sequences whose results place tiles at every $\vec{y}_i$,
         the tiles always attach at the $\vec{y}_i$ points in order of 
         increasing $i$ values,\\
      d)
      [Lemmas~\ref{lem:x-i-y-i-ordering-r-1} and~\ref{lem:x-i-y-i-ordering}] in all $\mathcal{T}$-producing assembly sequences whose results place tiles at every $\vec{y}_i$,
      tiles also attach at all of the $\vec{x}_i$ points, and tile attachment
      steps always occur in the order $\vec{x}_1$, $\vec{y}_1$, $\vec{x}_2$, $\vec{y}_2$, \ldots, $\vec{x}_r$, $\vec{y}_r$, and\\
      e)
      [Condition~\ref{def:snd-4}] none of the $\vec{x}_i$ points is a point of
      competition.\\
      The solid arrows in the figure depict a $\mathcal{T}$-producible assembly that is the result of a $\mathcal{T}$-producing assembly
      sequence that attaches tiles at points along a simple path in $\mathbb{Z}^2$ and in the following,
      required order: from
      $\vec{s}$ (the seed's point) to $\vec{x}_1$, then to $\vec{y}_1$ via
      $\pi_1$, then to $\vec{x}_2$, then to $\vec{y}_2$ via $\pi'_2$, \ldots, to
      $\vec{x}_r$, and finally to $\vec{y}_r$ via $\pi'_r$.
    }
    \end{figure}

\begin{observation}
\label{obs:distinct-starting-points}
Lemma~\ref{lem:x-i-y-i-ordering} implies that the starting points corresponding to the elements of $Y$ are all distinct. In other words, if $|Y|=r$ for some $r \in \mathbb{Z}^+$, then $|S_Y|=r$.
\end{observation}

In the following, we define the TAS that corresponds to a pair of competing paths competing for a POC.

\begin{definition} 
\label{def:competition}
For all integers $1 \leq i \leq r$, let $\mathcal{C}_i = \left( T, \sigma_i, 1 \right)$ be the TAS with $\sigma_i=\left\{\left(\vec{x}_i,\alpha\left( \vec{x}_i \right)\right)\right\}$.
We say that $\mathcal{C}_i$ is \emph{competition} $i$ in $\mathcal{T}$ associated with POC $\vec{y}_i$. 
\end{definition}

Note that Definition~\ref{def:competition} defines each competition
$\mathcal{C}$ as the TAS $\mathcal{T}$ with a different seed
tile. This allows us to reason about the behavior of $\mathcal{C}$ in
the absence of any other tiles, such as, e.g., a simple path of tiles
from the seed of $\mathcal{T}$ to the starting point of $\mathcal{C}$
that may or may not block one of the competing paths of
$\mathcal{C}$. Importantly, Definition~\ref{def:competition} does not
restrict the set of $\mathcal{C}$-producible assemblies to only those
that place tiles on, and only, in the domains of the associated
competing paths. In fact, it is always possible for a
$\mathcal{C}$-producible assembly to be the result of a
$\mathcal{C}$-producing assembly sequence that self-assembles
backwards from its seed tile to the seed tile of
$\mathcal{T}$. However, in the following definition, we will characterize a
subset of ``winning'' $\mathcal{C}$-producible assemblies that place
tiles on every point in the winning path and possibly a prefix of the
losing path. We will later define
(see~Definition~\ref{def:competition-probability}) the probability of
  $\mathcal{C}$ only considering $\mathcal{C}$-producing assembly
  sequences that not only result in a ``winning''
  $\mathcal{C}$-producible assembly but are also comprised of tile
  attachment steps that faithfully follow the competing
  paths. Moreover, defining $\mathcal{C}$ to merely be $\mathcal{T}$
  with a different seed tile, instead of, for example, defining
  $\mathcal{C}$ with a new tile type not in the tile set of $\mathcal{C}$,
  will allow for much greater notational convenience when (in
  Section~\ref{sec:proof}) we  1) embed $\mathcal{C}$-producing
  assembly sequences within $\mathcal{T}$-producing assembly sequences
  and 2) identify within $\mathcal{T}$-producing assembly sequences
  corresponding $\mathcal{C}$-producing assembly sequences.

\begin{definition} 
\label{def:winning-path-assembly-competition}
Let $\vec{y} \in Y$ with corresponding competing paths $\pi$ and $\pi'$, $\mathcal{C}$ be the competition in $\mathcal{T}$ associated with $\vec{y}$, and $\gamma \in \mathcal{A}[\mathcal{C}]$ such that:
\begin{enumerate}
	\item \label{def:wp-1} $\pi$ is a path in $G^{\textmd{b}}_{\gamma}$,
	\item \label{def:wp-2} $\gamma \sqsubseteq \alpha$, and

        \item \label{def:wp-3} every simple path in  $G^{\textmd{b}}_{\gamma}$ that starts at $\vec{x}$ is a prefix of either $\pi$ or $\pi'$.
\end{enumerate}
We say that $\pi$ is the \emph{winning path for} $\mathcal{C}$,  $\pi'$ is the \emph{losing path for} $\mathcal{C}$, and that $\gamma$ is a \emph{winning assembly} for $\mathcal{C}$.
Let $W_{\mathcal{C}}$ denote the set of all winning assemblies for $\mathcal{C}$. Any $\mathcal{C}$-producing assembly sequence $\vec{\gamma}$ that terminates at $\vec{y}$ and whose result is in $ W_{\mathcal{C}}$ is a \emph{winning assembly sequence} for $\mathcal{C}$.
\end{definition}

Consider again the assembly $\alpha_{4,9}$, shown in
    Figure~\ref{fig:def-example-tas}(f). Let $\alpha'_{4,9}$
    be the subassembly of $\alpha_{4,9}$ with only the $s$ and 4 tile
    types attached.  Let $\alpha''_{4,9}$ be the subassembly of $\alpha_{4,9}$
    with only the $s$, 1, and 4 tile types attached.  If
    $w'=\{(\vec{y}_1,4),(\vec{y}_2,9)\}$, then both $\alpha'_{4,9}$ and
    $\alpha''_{4,9}$ are winning assemblies for $\mathcal{C}_1$ in
    $\mathcal{T}$ relative to $w'$. In this example, $\pi_1$, shown in
    Figure~\ref{fig:def-example-tas-POC}, is the
    winning path for $\mathcal{C}_1$.

    \begin{figure}[!h]

      \centerline{\includegraphics[width=5in]{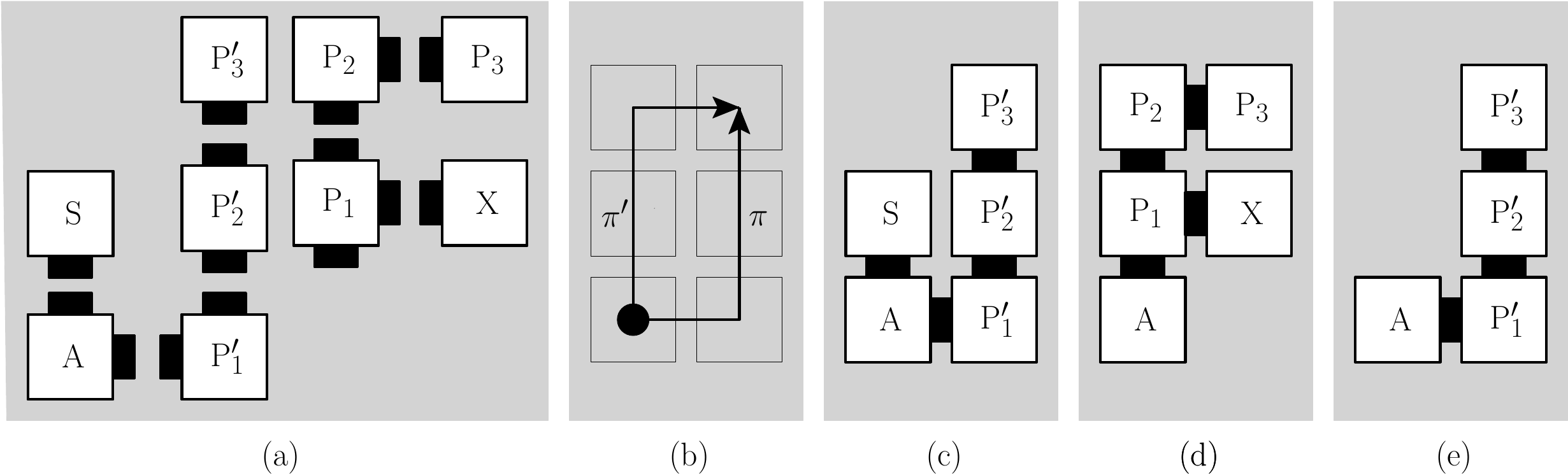}}

    \caption{\label{fig:weird-competition} A sample TAS
      $\mathcal{T}=(T,\sigma,1)$ that is sequentially
      non-deterministic but whose only competition $\mathcal{C}$ is
      such that there exists a $\mathcal{C}$-producing assembly
      sequence in which $\pi'$ blocks $\pi$ via the placement of the X
      tile. Note that such a placement is only possible in a
      $\mathcal{C}$-producing assembly sequence and not in any
      $\mathcal{T}$-producing assembly sequence. If the placement of
      the X tile were possible in some $\mathcal{T}$-producing
      assembly sequence, then condition~\ref{def:cp-4a} of Definition~\ref{def:competing-paths} would be
      violated. (a) The tile set $T$ with S
      being the seed tile (b) The two paths competing for the
      top-right corner point in $\mathcal{T}$, from the bottom-left
      corner point (c) The unique $\mathcal{T}$-terminal assembly (d)
      One $\mathcal{C}$-producible assembly that follows $\pi'$ but
      also contains the X tile that blocks $\pi$. Since $\pi'$ is
      always blocked in the $\mathcal{T}$-terminal assembly, this
      $\mathcal{C}$-producible assembly can never be a subassembly of
      any $\mathcal{T}$-producible assembly (e) One
      $\mathcal{C}$-producible assembly that follows $\pi$ and that
      is never blocked in any $\mathcal{T}$-producible assembly
      since the X tile is always absent}
\end{figure}

    \begin{figure}[!h]

      \centerline{\includegraphics[width=6in]{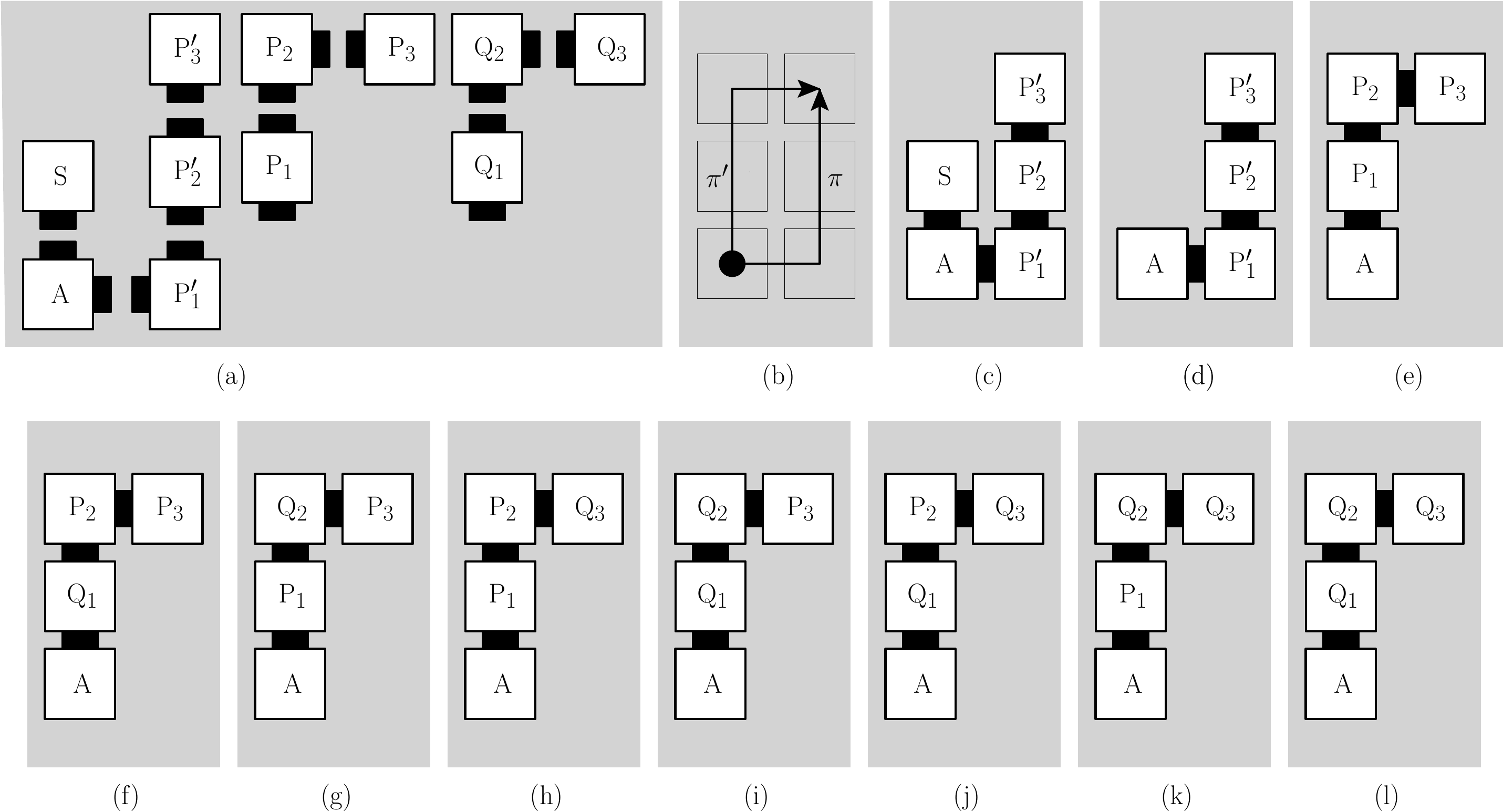}}
      
    \caption{\label{fig:weird-competition2} A sample TAS
      $\mathcal{T}=(T,\sigma,1)$ that is sequentially
      non-deterministic even though its only competition $\mathcal{C}$
      exhibits a form of non-determinism that is not allowed in our
      framework since there are exactly two distinct tile types,
      namely the P$_i$ and Q$_i$ tiles, that may attach at point
      $\pi'[i+1]$, for $i \in \{1,2,3\}$. Note that this
      non-determinism is not exhibited in $\mathcal{T}$ because $\pi'$
      is blocked in every $\mathcal{T}$-producible assembly sequence. 
      If this were not the case, then $\mathcal{T}$ would violate condition~\ref{def:snd-1} of Definition~\ref{def:seq-non-deterministic}.
      (a) The tile set $T$
      with S being the seed tile (b) The two paths competing for the
      top-right corner point in $\mathcal{T}$, from the bottom-left
      corner point (c) The unique $\mathcal{T}$-terminal assembly (d)
      The unique $\mathcal{C}$-producible assembly that follows $\pi$
      (e) through (l) The eight $\mathcal{C}$-producible assemblies
      that follow $\pi'$. Since $\pi'$ is always blocked in the
      $\mathcal{T}$-terminal assembly, none of these
      $\mathcal{C}$-producible assemblies can ever be a subassembly of
      any $\mathcal{T}$-producible assembly }
\end{figure}

Figure~\ref{fig:weird-competition} describes a TAS that is
sequentially non-deterministic even though the losing competing path
$\pi'$, associated with its unique competition blocks the winning
path $\pi$. Our framework handles such cases correctly because
the losing path is blocked in any $w$-correct terminal assembly,
thereby always preventing the X tile from attaching.
Condition~\ref{def:wp-3} in
Definition~\ref{def:winning-path-assembly-competition} ensures that
any $\mathcal{C}$-producing assembly sequence that places the X tile
does not qualify as a winning assembly sequence for $\mathcal{C}$.

Figure~\ref{fig:weird-competition2} describes a TAS that is
sequentially non-deterministic even though its only competition
exhibits, along its losing path, a form of non-determinism that
violates both the directional and agreeable determinism
conditions. Our framework handles such cases correctly because the
losing path is blocked in any $w$-correct terminal assembly, thereby
always preventing the unwanted non-determinism.
Condition~\ref{def:wp-2} in
Definition~\ref{def:winning-path-assembly-competition} ensures that
any $\mathcal{C}$-producing assembly sequence that places a tile type
at any point in $\pi'$ that is not the same tile type that $\alpha$ places at that point does not qualify as a winning assembly sequence for $\mathcal{C}$.

  \goingforward{assume $\vec{y} \in Y$ with corresponding starting point $\vec{x} \in S_Y$, $\vec{\beta}_{\vec{x}}$ is a  $\mathcal{T}$-producing sequence that terminates at $\vec{x}$ and whose result $\beta_{\vec{x}}$ is a subassembly of $\alpha$, and $\mathcal{C}$ is the competition in $\mathcal{T}$ associated with $\vec{y}$ with corresponding winning and losing paths $\piw$ and $\pil$, respectively.}

  The following definition describes the notion that a competition need not be ``fair'' relative to a starting assembly.
\begin{definition}  
\label{def:competition-rigged}
We say that $\mathcal{C}$ is rigged by $\vec{\beta}_{\vec{x}}$ if there exists an integer $1 < l < \left| \pil\right|$
such that:
\begin{enumerate}
	\item \label{def:rigged-1} $\pil[l] \in \dom{\beta_{\vec{x}}}$ and
	\item \label{def:rigged-2} every simple path from $\vec{s}$ to $\vec{x}$ in $G^{\textmd{b}}_{\beta_{\vec{x}}}$ goes through  $\pil[l]$.
\end{enumerate}
\end{definition}

In a competition rigged by $\vec{\beta}_{\vec{x}}$,  $\pil$ is ``blocked'' in $\vec{\beta}_{\vec{x}}$, in the sense implied by the following lemma.

\begin{lemma} 
\label{lem:rigged}
If $\mathcal{C}$ is rigged by $\vec{\beta}_{\vec{x}}$ then, for every $\mathcal{T}$-producing sequence $\vec{\beta}$ that extends $\vec{\beta}_{\vec{x}}$, $\pil$  cannot be a simple path in $\bindinggraph_{\res{\vec{\beta}}}$.
\end{lemma}

\begin{proof}
  Assume, for the sake of obtaining a contradiction, that
  $\vec{\beta}_{\vec{x}}$ can be extended to a $\mathcal{T}$-producing
  sequence $\vec{\beta}$ such that $\res{\vec{\beta}}= \beta$ and
  $\pil$ is a simple path in
  $\bindinggraph_{\beta}$. Since $\mathcal{C}$ is rigged by
  $\vec{\beta}_{\vec{x}}$, there exists a simple blocking path $p_b$
  from $\vec{s}$ to $\vec{x}$ in $\bindinggraph_{\beta}$ that goes
  through at least one point in $\pil$ that is neither
  $\vec{x}$ nor $\vec{y}$. Let $i = \max\left\{ \,j \left| \;
  \pil[\,j] \in \dom{p_b} \right. \right\}$ and
  $\vec{p} = \pil[i]$. In words, $\vec{p}$ is the point
  in $p_b$ that is the closest point to $\vec{y}$ in
  $\pil$. Let the path $p_1$ be the prefix of $p_b$
  that (starts at $\vec{s}$ and) ends at $\vec{p}$. Note that $\vec{x}
  \notin \dom p_1$ and $\vec{y} \notin \dom p_1$. Now, by construction
  of $\vec{\beta}$, there also exists a simple path $p_2$ from
  $\vec{p}$ to $\vec{y}$ in $\bindinggraph_{\beta}$ that is a proper
  suffix of $\pil$. From the way $\vec{p}$ was
  selected, it follows that $\dom{p_1} \cap \dom{p_2} =
  \{\vec{p}\,\}$.  Since $p_1$ is a simple path in
  $\bindinggraph_{\beta}$ that ends at $\vec{p}$, $p_2$ is a simple
  path in $\bindinggraph_{\beta}$ that starts at $\vec{p}$, and
  $\vec{p}$ is the only point that belongs to the intersection of
  $\dom{p_1}$ and $\dom{p_2}$, if follows that the path $p$ obtained
  by concatenating $p_1$ and $p_2$ is a simple path from $\vec{s}$ to
  $\vec{y}$ in $\bindinggraph_{\beta}$. Since $p_1$ is a prefix of
  $p_b$ that ends at $\vec{p} \neq \vec{x}$, $\vec{x} \notin
  \dom{p_1}$. Since $p_2$ is a suffix of $\pil$ that
  starts at $\vec{p} \neq \vec{x}$, $\vec{x} \notin
  \dom{p_2}$. Thus, $\vec{x} \notin \dom{p}$. Therefore, neither
  $\piw$ nor $\pil$ is a suffix of $p$
  which, together with the fact that $p$ is a simple path from
  $\vec{s}$ to $\vec{y}$ in the binding graph of the
  $\mathcal{T}$-producible assembly $\beta$, contradicts
  condition~\ref{def:cp-5} in Definition~\ref{def:competing-paths}.
\end{proof}

By the definitions of $\alpha$ and $\vec{\beta}_{\vec{x}}$, the latter may not place a tile at any point on the winning path except for $\vec{x}$. So if $\mathcal{C}$ is not rigged by $\vec{\beta}_{\vec{x}}$, then the following result holds.
\begin{lemma} 
\label{lem:not-rigged}
If $\mathcal{C}$ is not rigged by $\vec{\beta}_{\vec{x}}$, then $\dom{\beta_{\vec{x}}} \cap \left( \dom{\piw} \cup \dom{\pil} \right) = \left\{ \vec{x} \right\}$. 
\end{lemma}
\begin{proof}
Since $\mathcal{C}$ is not rigged by $\vec{\beta}_{\vec{x}}$, then either condition~\ref{def:rigged-1} or condition~\ref{def:rigged-2} of Definition~\ref{def:competition-rigged} does not hold.
If the former does not hold for any $l$, then the conclusion of this lemma follows immediately.
Note that condition~\ref{def:rigged-2} of Definition~\ref{def:competition-rigged} cannot hold if condition~\ref{def:rigged-1} does not. 
Therefore, assume for the sake of obtaining a contradiction that condition~\ref{def:rigged-1} of Definition~\ref{def:competition-rigged} holds but condition~\ref{def:rigged-2} does not.
%
%
Thus there exists a simple path $p$ from $\vec{s}$ to $\vec{x}$ in $G^{\textmd{b}}_{\beta_{\vec{x}}}$ that does not go through $\pil[l]$. 
Let $p'$ be any simple path from $\vec{s}$ to $\pil[l]$ in $G^{\textmd{b}}_{\beta_{\vec{x}}}$.
Such a path exists because condition~\ref{def:rigged-1} of Definition~\ref{def:competition-rigged} holds. 
Then, either $p'$ goes through $\vec{x}$ or it does not.
\begin{enumerate}
\item If $p'$ does not go through $\vec{x}$, then by condition~\ref{def:cp-4b} of Definition~\ref{def:competing-paths},  $p'$ must be a prefix of $p$, which contradicts the facts that $\pil[l] \in \dom{p'}$ and $\pil[l] \not \in \dom{p}$.

	\item If $p'$ goes through $\vec{x}$, then by condition~\ref{def:cp-4a} of Definition~\ref{def:competing-paths}, $p[1\ldots  l]$ must be a suffix of $p'$ in $G^{\textmd{b}}_{\beta_{\vec{x}}}$, which contradicts the fact that $\vec{\beta}_{\vec{x}}$ terminates at $\vec{x}$. 
\end{enumerate}
Since a contradiction arises in both cases, it follows that condition~\ref{def:rigged-1} of Definition~\ref{def:competition-rigged} cannot hold. 
\end{proof}

\begin{definition}  
\label{def:rigged}
We say that $\mathcal{C}$ is \emph{rigged in} $\mathcal{T}$ if it is
rigged by every $\mathcal{T}$-producing sequence that terminates at
$\vec{x}$ and whose result is a subassembly of $\alpha$.
\end{definition}
%
%
%
%
%
%
%
%

The next lemma counts the number of winning assemblies for a given competition that is not rigged, which we will eventually use (Lemma~\ref{lem:domain-competition}) to count the number of winning assembly sequences for the competition and derive the probability of a winning assembly sequence for the competition. 

\begin{lemma} 
  \label{lem:number-of-winning-assemblies} Let $\vec{y} \in Y$ with corresponding competing paths $\piw$ and $\pil$ and starting point $\vec{x}$, and $\mathcal{C}$ be the competition in $\mathcal{T}$ associated with $\vec{y}$. 
  If  $\mathcal{C}$ is not rigged in $\mathcal{T}$ and $\piw$ and $\pil$ have length $l \geq 2$ and $l' \geq 2$, respectively, then $\left|W_{\mathcal{C}}\right| = l' - 1$.
\end{lemma}

For example, Figure~\ref{fig:def-Wi} depicts the set $W_{\mathcal{C}_1}$ in our running running $\mathcal{T}$, with $w'=\{(\vec{y}_1,4),(\vec{y}_2,9)\}$.

\begin{figure}[!h]

      \centerline{\includegraphics[width=2in]{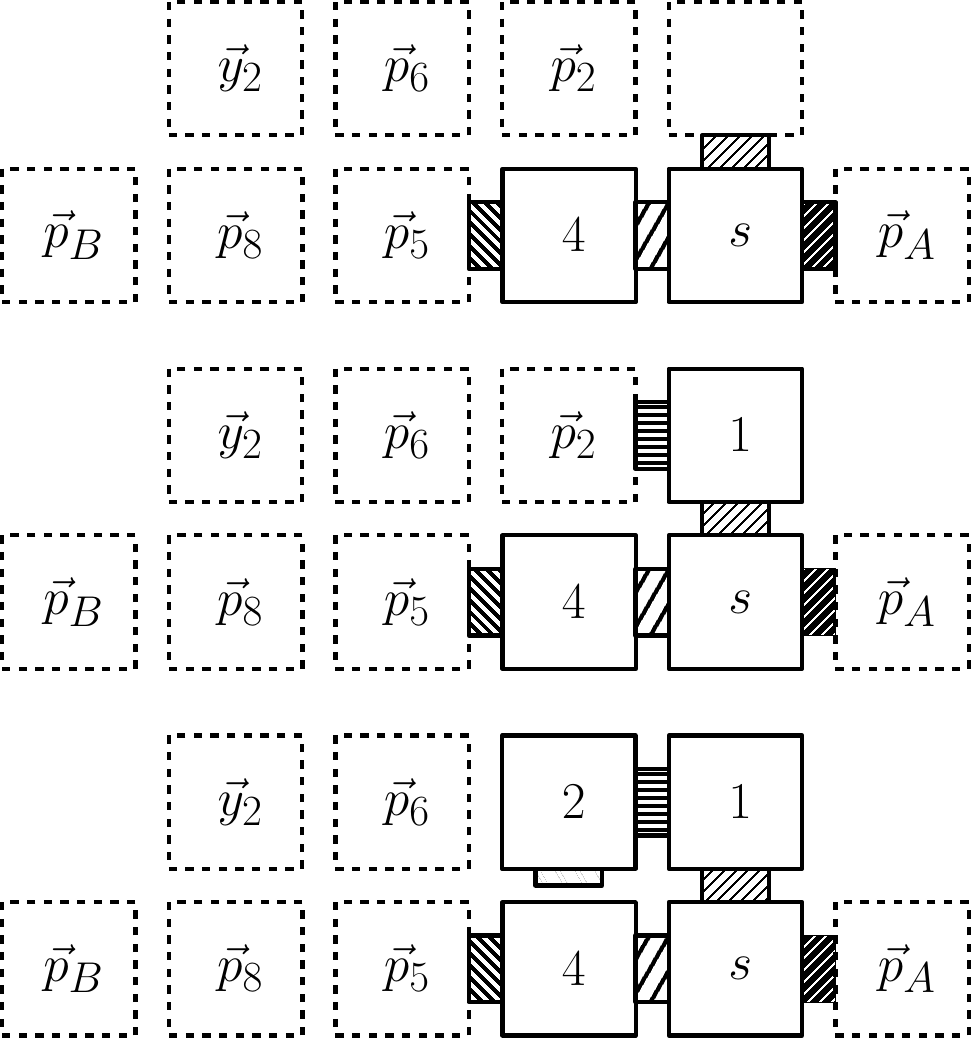}}

      \caption{\label{fig:def-Wi} The three winning assemblies that
        make up the set $W_{\mathcal{C}_1}$ in our
        running running $\mathcal{T}$, with $w'=\{(\vec{y}_1,4),(\vec{y}_2,9)\}$.
        In this example, $\pi_1$ and
        $\pi'_1$, shown in Figure~\ref{fig:def-example-tas-POC},
        have length $l=2$ and
        $l'=4$, respectively. In agreement with Lemma~\ref{lem:number-of-winning-assemblies},
        $|W_{\mathcal{C}_1}| = 3 = l'-1$.  }
    \end{figure}

\begin{proof}
  Since
$\vec{y} \in Y$ and $Y \subseteq \dom{\alpha}$, let $p'$ be any path
from $\vec{s}$ to $\vec{y}$ in $\bindinggraph_{\alpha}$. By
condition~\ref{def:cp-5} in Definition~\ref{def:competing-paths},
either $\piw$ or $\pil$ is a suffix of $p'$. Since $\alpha$ is
$w$-correct, $\piw$ must be a suffix of $p'$. Let $\gamma$ be the
subassembly of $\alpha$ whose domain is the domain of $\piw$.
By construction:\vspace*{-0mm}
\begin{enumerate}[label=\alph*)]
\item \label{item:sizeWc-1} $\gamma \sqsubseteq \alpha$,
\item \label{item:sizeWc-2} $\dom{\gamma} = \dom{\piw}$,
\item \label{item:sizeWc-3} $\piw$ is a path in  $\bindinggraph_{\gamma}$
(because it is a suffix of $p'$),
\item \label{item:sizeWc-4}  $\gamma$ is a
    stable assembly in $\mathcal{A}^T$ (because it is a subassembly of
     the stable, $\mathcal{T}$-producible assembly $\alpha$), and
\item \label{item:sizeWc-5}  $\gamma(\vec{x}) = \alpha(\vec{x})$ (because
of item~\ref{item:sizeWc-1} above).

\end{enumerate}
Therefore, $\gamma$ meets all of the requirements of
  Definition~\ref{def:winning-path-assembly-competition}, as we now
  show.

\begin{itemize}
  \item $\gamma \in \mathcal{A}[\mathcal{C}]$ because of
      items~\ref{item:sizeWc-4} and~\ref{item:sizeWc-5} above,
      combined with the fact that $\alpha\left(\vec{x}\right)$ is, by
      Definition~\ref{def:competition}, the tile type in $T$ placed by
      $\alpha$ at the seed location of $\mathcal{C}$.
      
  \item Condition~\ref{def:wp-1} in
      Definition~\ref{def:winning-path-assembly-competition} is the
      same as item~\ref{item:sizeWc-3} above.

\item Condition~\ref{def:wp-2} in
      Definition~\ref{def:winning-path-assembly-competition} is the
      same as item~\ref{item:sizeWc-1} above.

\item Condition~\ref{def:wp-3} in
      Definition~\ref{def:winning-path-assembly-competition} holds
      because, by item~\ref{item:sizeWc-2} above, $\dom{\gamma}
      = \dom{\piw} \subseteq \dom{\piw} \cup \dom{\pil}$.
\end{itemize}         

Not only is $\gamma$ in $W_{\mathcal{C}}$, but the reasoning above
further implies that $\gamma$ is the unique, smallest assembly that
satisfies the definition of a winning assembly for $\mathcal{C}$
because, if any tile were removed from $\gamma$, then its domain would
become a proper subset of $\dom{\piw}$, which would contradict
condition~\ref{def:wp-1} in
Definition~\ref{def:winning-path-assembly-competition}.
However, assuming $l' > 2$, $\gamma$ is not the only assembly that satisfies the definition of a winning assembly for $\mathcal{C}$, which we now show.
Since $\mathcal{C}$ is not rigged in $\mathcal{T}$, there exists a $w$-correct $\mathcal{T}$-producing assembly sequence $\vec{\beta}_{\vec{x}}=(\beta_1, \beta_2, ..., \beta_m = \beta_{\vec{x}})$ of length $m \in \mathbb{Z}^+$, terminating at $\vec{x}$, and with result $\beta_{\vec{x}}$ such that $\mathcal{C}$ is not rigged in $\vec{\beta}_{\vec{x}}$.
Lemma~\ref{lem:not-rigged} applied to $\mathcal{C}$ and $\vec{\beta}_{\vec{x}}$ implies that $\dom{\beta_{\vec{x}}}~\cap~\dom{\pil} = \left\{ \vec{x} \right\}$, which means that $\vec{\beta}_{\vec{x}}$ can be extended to follow $\pil\left[ 1 \ldots l' - 1 \right]$.
To that end, note that condition~\ref{def:cp-6a} of Definition~\ref{def:competing-paths} implies the existence of a $\mathcal{T}$-assembly sequence $\vec{\beta} = \left( \beta_i \mid 1 \leq i < l' \right)$ with result $\beta$ such that 1) 
$\beta_1 = \left\{ \vec{x}, \alpha\left(\vec{x}\right) \right\}$, 2) for all $1 < i < l'$, $\beta_i = \beta_{i-1} + \left( \pil\left[ i \right], t_i \right)$ with $t_i \in T$, and 3) $\dom{\beta} = \dom{\pil\left[ 1\ldots l' -1 \right]}$.
Importantly, $\vec{\beta}$ attaches tiles along $\pil$ up to, but not including, $\vec{y}$ in the following sense: for $1 < i < l'$, if $\beta_i = \beta_{i-1} + \left( \pil[i], t_i \right)$, then $\vec{u}_{\beta_i \backslash \beta_{i-1}} = \pil[i-1] - \pil[i]$.
Thus $\vec{\beta}$ can be used to construct a corresponding extension $\vec{\beta}'_{\vec{x}} = (\beta'_1=\beta_1, \beta'_2=\beta_2, ..., \beta'_m=\beta_m = \beta_{\vec{x}}, \beta'_{m+1}, ..., \beta'_{m+l'-2})$  of $\vec{\beta}_{\vec{x}}$ such that the result of $\vec{\beta}'_{\vec{x}}$ has a domain equal to $\dom{\beta_{\vec{x}}}~\cup~\dom{\pil[2\ldots l' -1]}$.
Since $\mathcal{T}$ is directionally deterministic, condition~\ref{def:dd-ad-pt-1} of Definition~\ref{def:dd-ad-pt} implies that, for all $m+1 \leq i \leq m+l'-2$, there exists a unique $t_i \in T$ such that $\beta'_i = \beta'_{i-1} + \left( \pil\left[ i \right], t_i \right) = \beta'_{i-1} + \left( \pil\left[ i \right], \alpha\left( \pil\left[ i \right]\right) \right)$.
And this property of $\vec{\beta}'$ implies the following property of $\vec{\beta}$: for all $1 < i < l'$, there exists a unique $t_i \in T$ such that $\beta_i = \beta_{i-1} + \left( \pil\left[ i \right], t_i \right) = \beta_{i-1} + \left( \pil\left[ i \right], \alpha\left( \pil\left[ i \right] \right) \right)$.
In other words, since $\mathcal{T}$ is directionally deterministic and $\mathcal{C}$ is not rigged by $\vec{\beta}_{\vec{x}}$, $\vec{\beta}$ is the unique assembly sequence whose existence is implied by condition~\ref{def:cp-6a} of Definition~\ref{def:competing-paths}.
Thus $\vec{\beta}$ can be used to extend $\gamma$ by attaching tiles at points along any prefix of $\pil\left[ 2\ldots l'-1\right]$ such that we choose the type of tile to attach subject to the following constraint: for $2 \leq i < l'$,  if the tile attaches at $\pil[i]$, then it must attach to the tile placed at $\pil[i-1]$.
In such an extension of $\gamma$, each tile added is unique because each tile that attaches in $\vec{\beta}$ is unique. 
Moreover, such an extended assembly would still satisfy all of the conditions in Definition~\ref{def:winning-path-assembly-competition}:
\begin{enumerate}
	\item As noted above, condition~\ref{def:wp-1} is satisfied by $\gamma$.
	\item Condition~\ref{def:wp-2} is satisfied because $\gamma$ was extended along a prefix of $\pil\left[2\ldots l'-1\right]$ using $\vec{\beta}$, which places the same tile types that $\alpha$ places along the points in $\dom{\pil\left[2\ldots l'-1\right]}$, i.e., for all $2 \leq j \leq l'-1$, if $\gamma$ is extended along the prefix $\pil[2\ldots j]$, then for all $2 \leq i \leq j$, $\beta\left( \pil[i] \right) = \alpha\left( \pil[i] \right)$.
	\item Condition~\ref{def:wp-3} is satisfied because:
	\begin{enumerate}
		\item $\gamma$ was extended using $\vec{\beta}$, which follows $\pil$ up to but not including $\vec{y}$, i.e., for $1 < i < l'$, if $\beta_i = \beta_{i-1} + \left( \pil[i], t_i \right)$, then $\vec{u}_{\beta_i \backslash \beta_{i-1}} = \pil[i-1] - \pil[i]$, and
		\item assuming $\gamma'$ is the result of extending $\gamma$ by $\vec{\beta}$, then $\gamma'\left(\pil[l' - 1]\right)$ does not bind to $\gamma'\left(\vec{y} \,\right)$ because such a situation can be ruled out as follows. 
		Suppose for the sake of obtaining a contradiction that there is an edge between the nodes $\vec{y}$ and $\pil[l' - 1]$ in $G^{\textmd{b}}_{\gamma'}$.
		Since $\alpha\left( \vec{y} \,\right) = \gamma'\left( \vec{y} \, \right)$ and $\alpha\left( \pil[l' - 1]\right) = \gamma'\left( \pil[l' - 1] \right)$, it follows that there is an edge between the nodes $\vec{y}$ and $\pil[l' - 1]$ in $G^{\textmd{b}}_{\alpha}$.
		Then, there exists a path $p''$ in $G^{\textmd{b}}_{\alpha}$ from $\vec{s}$ to $\pil[l' - 1]$ such that $s = (\vec{x} = \piw[1], \ldots, \vec{y} = \piw[l], \pil[l' - 1])$ is a suffix of $p''$.
		Note that $p''$ is in fact a simple path because 1) $\piw$ is a simple path, 2) $l' \geq 3$ implies that $\pil[l' -1] \not \in \left\{\vec{x},\vec{y}\right\}$, and 3) conditions~\ref{def:cp-1} and~\ref{def:cp-2} of Definition~\ref{def:competing-paths} together imply that $\pil[l'-1] \not \in \dom{\piw}$.
		Since $p''$ is a simple path in $G^{\textmd{b}}_{\alpha}$ from $\vec{s}$ to $\pil[l' -1]$ that goes through $\vec{x}$, it follows by condition~\ref{def:cp-4a} of Definition~\ref{def:competing-paths} that $\pil[1 \ldots l'-1]$ is a suffix of $p''$. 
		Since $l' \geq 3$, we have $\vec{y} \in \dom{\pil[1 \ldots l'-1]}$.
		But this is impossible because $\pil$ is a simple path from $\vec{x}$ to $\vec{y}$. 
		Thus, $\gamma'\left(\pil[l' - 1]\right)$ does not bind to $\gamma'\left(\vec{y} \,\right)$.
	\end{enumerate}
\end{enumerate}
Since the number of distinct path prefixes of $\pil[2\ldots l'-1]$ is equal to $l'- 2$, this is also the number of distinct winning assemblies for $\mathcal{C}$ that are larger than $\gamma$. 
In conclusion, $\left|W_{\mathcal{C}}\right| = (l'-2)+1 = l'-1$. Note that the winning assembly is unique when the length of (the losing path) $\pil$ is equal to 2.
\end{proof}

We now define the probability of a competition as the sum of the probabilities of all of its winning assembly sequences.
The probability of a winning assembly sequence of the competition is defined relative to a modified version of its corresponding SPT where certain nodes are removed to simplify the computation of the probability of the winning assembly sequence.
\begin{definition}  
\label{def:competition-probability}
Let $\vec{y} \in Y$ with corresponding starting point $\vec{x}$ and competing paths $\piw$ and $\pil$. Let $\mathcal{C}$ be the competition in $\mathcal{T}$ associated with $\vec{y}$ and $\mathcal{M}'_{\mathcal{C}}$ be the SPT obtained from $\mathcal{M}_{\mathcal{C}}$ by removing all nodes $\vec{\beta}$ of $\mathcal{M}_{\mathcal{C}}$ such that the binding graph of $\res{\vec{\beta}}$ contains a simple path that starts at $\vec{x}$ and is not a prefix of either $\piw$ or $\pil$. 
If $\vec{\gamma}$ is any winning assembly sequence for  $\mathcal{C}$, we define the \emph{probability of} $\vec{\gamma}$ as  $\textmd{Pr}_{\mathcal{C}}\left[ \vec{\gamma} \right] = \textmd{Pr}_{\mathcal{M}'_{\mathcal{C}}}[ \vec{\gamma} ]$, and the \emph{probability of competition} $\mathcal{C}$ as 
$$\displaystyle
\textmd{Pr}\left[ \mathcal{C} \right]=\sum_{\vec{\gamma} \textmd{ winning assembly sequence for } \mathcal{C}}
{\textmd{Pr}_{\mathcal{C}}[ \vec{\gamma} ]}.
$$
\end{definition}

Intuitively, $\textmd{Pr}[\mathcal{C}]$ is defined so that it can be computed easily, i.e., without having to account for potentially numerous points at which tiles could attach outside of the competing paths but in doing so should not affect $\textmd{Pr}\left[\mathcal{C}\right]$. Furthermore, $\mathcal{C}$-producible assembly sequences that place a tile in $\dom{\piw} \cup \dom{\pil}$ without following a prefix of $\piw$ or $\pil$, like the X tile in Figure~\ref{fig:weird-competition}, are excluded from the summation in the probability computation since these sequences are not winning assembly sequences for $\mathcal{C}$.
The next lemma first counts the number of winning assembly sequences for a competition that is not rigged that all result in a certain winning assembly and then derives the probability of each such winning assembly sequence. 

\begin{lemma} 
  \label{lem:domain-competition} Let $\vec{y} \in Y$ with
  corresponding competing paths $\piw$ and $\pil$ of length $l \geq 2$ and $l' \geq 2$, respectively. Let $\mathcal{C}$ be the competition associated with $\vec{y}$ in $\mathcal{T}$. 
For any integer $0 \leq i \leq l' - 2$, let $\pi_i$ denote the (possibly empty) simple path $\pil\left[2\ldots(i+1)\right]$ and $\gamma$ denote the winning assembly in $W_{\mathcal{C}}$ whose domain equals $\dom{\piw} \cup \dom{\pi_i}$.   If  $\mathcal{C}$ is not rigged in $\mathcal{T}$, then:
  \begin{enumerate}
  	\item \label{lem:dc-1} The number of winning assembly sequences for $\mathcal{C}$ with result $\gamma$ is ${l+i-2 \choose l-2}$.
  	\item \label{lem:dc-2} If $\vec{\gamma}$ is a winning assembly sequence for $\mathcal{C}$, with result $\gamma$, then $\textmd{Pr}_{\mathcal{C}}\left[\vec{\gamma}\right] = \frac{1}{2^{\left|\vec{\gamma}\right|-1}} = \frac{1}{2^{l+i-1}}$.
  \end{enumerate}
\end{lemma}

\begin{proof}
Let $\vec{x}$ be the starting point corresponding to $\vec{y}$.
Let $\vec{\gamma}$ be any winning assembly sequence for $\mathcal{C}$ with
     $\res{\vec{\gamma}} = \gamma$.  Since, by conditions~\ref{def:cp-1} and~\ref{def:cp-2} in
     Definition~\ref{def:competing-paths}, $\piw$ and $\pil$ are two
     distinct simple paths such that $ \dom \piw\ \cap\ \dom \pil
     = \{\vec{x},\vec{y}\}$, $\dom \gamma
     = \{\vec{x}=\piw[1]=\pil[1], \piw[2], \ldots, \vec{y}=\piw[l]=\pil[l'], \pil[2], \ldots, \pil[i+1]\}$.
     Thus $|\vec{\gamma}|=\left|\dom \gamma \right|=l+i$.
     Since 1) each assembly step in $\vec{\gamma}$ places a tile at a
     point that extends a prefix of either one of the simple paths
     $\piw$ and $\pi_i$ such that every simple path in
     $G^{\textmd{b}}_{\gamma}$ that starts at $\vec{x}$ is a prefix of
     either $\piw$ or $\pil$ and 2) $\vec{\gamma}$ terminates at
     $\vec{y}$,
     $\vec{\gamma} =
     (\gamma_1,\gamma_2, \ldots, \gamma)$ must satisfy the
     following conditions:
\begin{enumerate}[label=(\alph*)]
   \item $\gamma_1=\left\{(\vec{x},
   \alpha\left(\vec{x}\right))\right\}$,
    \item $\gamma = \gamma_{l+i-1} +  \left( \vec{y}, \alpha\left(\vec{y}\right) \right)$, and
   \item for every $2 \leq k < l+i$, $\gamma_k = \gamma_{k-1}
      + (\vec{p} , t)$ for some $\vec{p} \in \dom \gamma$ and $t\in T$, and either:

      Case 1: there exists $2 \leq u < l$ such that
      $\vec{p}=\piw[u]$ and, for all $2 \leq v < u$,
      $\textmd{index}_{\vec{\gamma}}\left( \piw[v] \right) <
      \textmd{index}_{\vec{\gamma}}\left( \piw[u] \right)$,

      or
      
      Case 2: there exists $2 \leq u < l'$ such that $\vec{p}=\pil[u]$
      and, for all $2 \leq v < u$,
      $\textmd{index}_{\vec{\gamma}}\left( \pil[v] \right)~<~\textmd{index}_{\vec{\gamma}}\left( \pil[u] \right)$.

    \end{enumerate} 
Or, informally, tiles along each one of the paths in $\{\piw,\pi_i\}$
must attach to $\gamma$ in a fixed order.  Therefore, since exactly
one tile type may attach at each point in $\dom{\piw}\cup\dom{\pil}$
(namely the tile type that $\alpha$ places there), the set of all
winning assembly sequences $\vec{\gamma}$ for $\mathcal{C}$ with
$\res{\vec{\gamma}} = \gamma$ corresponds to the set of all possible
ways of interleaving the two ordered sets of points that define $\piw$
and $\pi_i$. Since picking the indices of the tile attachment steps in
which the tiles assemble along $\piw$ in $\vec{\gamma}$ fully
determines the indices of the tile attachment steps in which the tiles
assemble along $\pi_i$ in $\vec{\gamma}$, we will count the number of
possible ways for tiles to assemble along $\piw$ in $\vec{\gamma}$.
Note that $|\piw|=l$, $|\vec{\gamma}|=l+i$, and $\vec{\gamma}$
terminates at $\vec{y}$.
      Since $\gamma(\piw[0])$ must be placed in $\gamma_1$ and all
       of the other tiles to be placed by $\vec{\gamma}$ (in both
       $\piw$ and $\pi_i$) must assemble before $\gamma(\piw[l])$
       attaches, there remain exactly $l-2$ other tiles to be
       placed along $\piw$ by $\vec{\gamma}$. Since the relative
       ordering of these $l-2$ tile attachment steps is fixed and
       these steps must take place between steps $2$ and $l+i-1$
       (inclusive) of $\vec{\gamma}$, the number of assembly sequences
       that place the tile $w(\vec{y})$ at $\piw[l]$ in $\gamma$ is
       equal to ${l+i-2 \choose l-2}$.
	Recall that $\textmd{Pr}_{\mathcal{C}}[ \vec{\gamma} ]$, where $\vec{\gamma}$ is a winning assembly sequence for $\mathcal{C}$ and $\res{\vec{\gamma}}=\gamma$, is the probability of the Markov
chain $\mathcal{M}'_{\mathcal{C}}$ being in state $\vec{\gamma}$.
This means $\textmd{Pr}_{\mathcal{C}}[ \vec{\gamma} ]$ is the probability of a tile being placed at $\piw[l]$ in $\gamma$, subject to the restriction that
every tile that attaches extends a prefix of either $\piw$ or $\pil$ in the sense of condition~\ref{def:wp-3} of Definition~\ref{def:winning-path-assembly-competition}, i.e., for all $p \in \{ \piw, \pil \}$ and all $1 < k < |\gamma|$, if $j$ is such that $1 < j \leq |p|$ and $\gamma_k = \gamma_{k-1} + (p[j], \alpha(p[j]))$, then $\vec{u}_{\gamma_k \backslash \gamma_{k-1}} = p[j-1] - p[j]$.
Thus $\textmd{Pr}_{\mathcal{C}}[ \vec{\gamma} ]$ is computed as follows.
    Due to the requirement of each tile placement step in
    $\vec{\gamma}$ having to extend one of the competing paths, for
    all $2 \leq j \leq l+i$, one of two tile types may attach at step
    $j$ of $\vec{\gamma}$, namely the unique tile type that could
    attach to the first empty location along $\piw$ in $\gamma_{j-1}$
    and the unique tile type that could attach to the first empty
    location along $\pil$ in $\gamma_{j-1}$. Each one of these two
    tile types exists and is unique, following the fact that
    $\mathcal{C}$ is not rigged in $\mathcal{T}$ and the same reasoning
    used in the second half of the proof of
    Lemma~\ref{lem:domain-competition}. This means that, in our model's notation, $M_{\gamma_{j-1}} = 2$. Thus,
      the probability that each one of these two tiles will attach
      during this step is equal to $\frac{1}{M_{\gamma_{j-1}}}
      = \frac{1}{2}$. Since only two tiles are available for
      attachment at each step and the selection of one of these two
      tiles to be attached at this step is independent of all previous
      tile attachment steps, the probability that the winning tile attaches at
      the point $\piw[l]$ at step $l+i$ of $\vec{\gamma}$ is equal to
      $\displaystyle\prod_{k=2}^{l+i}\frac{1}{2}
      = \left(\frac{1}{2}\right)^{l+i-1}$. The product starts at index
      $k=2$ because only one tile type may attach during $\gamma_1$,
      namely $\alpha\left(\vec{x}\right)$.
\end{proof}

\begin{corollary} 
  \label{cor:competition-probability-sum}
  Under the same assumtions as Lemma~\ref{lem:domain-competition}:
  \[
  \textmd{Pr}\left[ \mathcal{C} \right] =
  \sum_{i=0}^{l'-2}{l+i-2 \choose l-2}\cdot \left(\frac{1}{2}\right)^{l+i-1}
  \]  
\end{corollary}

\begin{proof}
According to Definition~\ref{def:competition-probability},
    $\displaystyle
    \textmd{Pr}\left[ \mathcal{C} \right]=\sum_{\vec{\gamma} \textmd{ winning assembly sequence for } \mathcal{C}} {\textmd{Pr}_{\mathcal{C}}[ \vec{\gamma} ]}.
    $
   According to
    Lemma~\ref{lem:number-of-winning-assemblies}, $W_{\mathcal{C}}$
    contains $l'-1$ winning assemblies for $\mathcal{C}$, the domain
    for each one of them containing $\dom{\piw}$ together with a fixed
    number of points in any (possibly empty) prefix of $\pil[2\ldots
    l'-1]$.
  Therefore, 
$$\textmd{Pr}\left[ \mathcal{C} \right]
= \sum_{i = 0}^{l'-2}{\left(\sum_{\substack{\vec{\gamma} \textmd{ is a winning assembly sequence for } \mathcal{C}\\ |\vec{\gamma}|=l+i \textmd{\ and\ } \res{\vec{\gamma}}=\gamma}}{\textmd{Pr}_{\mathcal{C}}[ \vec{\gamma} ]}\right)}.
$$
By Lemma~\ref{lem:domain-competition}:
\begin{enumerate}
	\item there are ${l+i-2 \choose l-2}$ winning assembly sequences for $\mathcal{C}$, with result $\gamma$, and,
	\item for each such assembly sequence $\vec{\gamma}$, we have $\textmd{Pr}_{\mathcal{C}}\left[\vec{\gamma}\right] = \frac{1}{2^{l+i-1}}$. 
\end{enumerate}
It follows that
$$\textmd{Pr}\left[ \mathcal{C} \right]
= \sum_{i = 0}^{l'-2}{\left(\sum_{\substack{\vec{\gamma} \textmd{ is a winning assembly sequence for } \mathcal{C}\\  |\vec{\gamma}|=l+i \textmd{\ and\ } \res{\vec{\gamma}}=\gamma}}{\textmd{Pr}_{\mathcal{C}}[ \vec{\gamma} ]}\right)} = \sum_{i=0}^{l'-2}{{l+i-2 \choose l-2} \cdot \left(\frac{1}{2}\right)^{l+i-1}}.
$$
\end{proof}

The following is the first main result of this paper, which basically says that the probability that the $w$-sequentially non-deterministic TAS $\mathcal{T}$ strictly self-assembles the shape of the unique $\alpha$ is at least the product of the probabilities of all of its competitions, which themselves can be easily computed using our framework.

\begin{theorem}
\label{thm:local-non-determinism-theorem}
If $\mathcal{C}_1, \ldots, \mathcal{C}_r$ are the competitions in $\mathcal{T}$, then $\mathcal{T}$ strictly self-assembles $\dom{\alpha}$ with probability at least $\displaystyle\prod_{i=1}^{r}{\textmd{Pr}\left[ \mathcal{C}_i \right]}$.
\end{theorem}

To prove Theorem~\ref{thm:local-non-determinism-theorem}, we must show that
$$\sum_{\substack{\alpha' \in \mathcal{A}_{\Box}[\mathcal{T}], \\ \dom{\alpha'}= \dom{\alpha}}}{\left(\sum_{\substack{\vec{\alpha} \textmd{ is a } \mathcal{T}\textmd{-producing assembly sequence,}\\ \res{\vec{\alpha}}=\alpha'}}{\textmd{Pr}_{\mathcal{M}_{\mathcal{T}}}[\vec{\alpha}]}\right)} \geq \prod_{i=1}^{r}{\textmd{Pr}\left[ \mathcal{C}_i \right]}
.$$

We will actually prove the following, stronger result, which implies
the inequality above:

\begin{equation}
  \label{eqn:main-inequality}
\sum_{\substack{\vec{\alpha} \textmd{ is a } \mathcal{T}\textmd{-producing assembly sequence,}\\ \res{\vec{\alpha}}=\alpha}}{\textmd{Pr}_{\mathcal{M}_{\mathcal{T}}}[\vec{\alpha}]} \geq \prod_{i=1}^{r}{\textmd{Pr}\left[ \mathcal{C}_i \right]}.
\end{equation}

Note that this inequality only considers $\mathcal{T}$-producing
assembly sequences that result in $\alpha$ and ignores all of the other ones
that result in an assembly whose domain happens to be identical to
$\dom{\alpha}$, as the next example illustrates. 
Before proving Theorem~\ref{thm:local-non-determinism-theorem}, we give an example in which we apply it to the example TAS $\mathcal{T}$ depicted in Figure~\ref{fig:def-example-tas-2}, which is a modified version of our running example TAS $\mathcal{T}$ from Figure~\ref{fig:def-example-tas} that produces multiple terminal assemblies having different respective shapes.
\begin{figure}[!h]
      \begin{minipage}{3in}
        \includegraphics[width=3in]{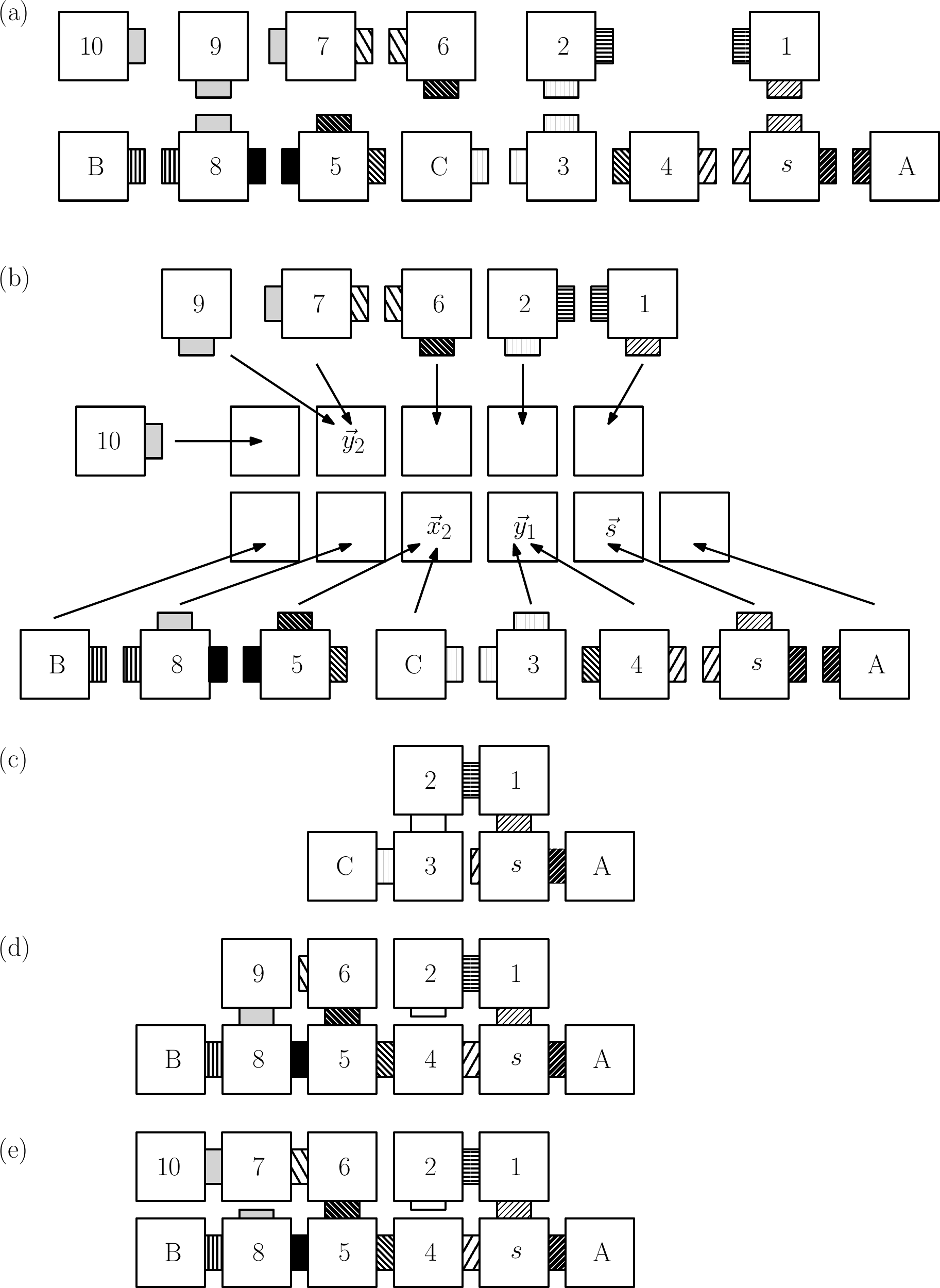}
            \end{minipage}\hfill
    \begin{minipage}{3in}
      \caption{\label{fig:def-example-tas-2}The TAS
        $\mathcal{T}=(T,\sigma,1)$ depicted here is a modified version of our  running example from Figure~\ref{fig:def-example-tas} with $|\termasm{T}|$ = 3\\
        (a) Depiction of the tile set $T=\{s,1,2,3,4,5,6,$ $7,8,9,10,A,B,C\}$,
        in which $s$ is the seed tile and the other tiles are labeled with
        either an uppercase letter or a positive integer; all glue strengths
        are 1; the different glue labels are depicted with shades of gray or
        tiling patterns\\
        (b) Depiction (in the middle) of the points in $\mathbb{Z}^2$ at
        which tiles in $T$ may attach; $\vec{s}$ is the point where the
        seed tile $s$ is always placed in this example, i.e.,
        $\sigma=\{(\vec{s},s)\}$; 
        the tiles in $T'$ are shown either above or below the points, with
        an arrow indicating the unique point at which the tile may
        attach; $\vec{y}_1$, $\vec{y}_2$ and the point to the west of $\vec{x}_2$ are the only points where
        more than one type of tile may attach\\
        (c) $\alpha_{3}$ is the  assembly in
        $\termasm{T}$ containing the tile $(\vec{y}_1,3)$\\
        (d) $\alpha_{4,9}$ is the  assembly in
        $\termasm{T}$ containing the two tiles $(\vec{y}_1,4)$ and
        $(\vec{y}_2,9)$\\
        (e) $\alpha_{4,7}$ is the assembly in
        $\termasm{T}$ containing the two tiles $(\vec{y}_1,4)$ and
        $(\vec{y}_2,7)$      }
    \end{minipage}
  \end{figure}
We first define the set $P$ of POCs of $\mathcal{T}$, and the set $Y \subseteq P$ of essential POCs of $\mathcal{T}$.
With these sets, we can then define a corresponding winner function $w: Y \rightarrow T$. 
Let $P$ be the POC set for $\mathcal{T}$ with essential POC set $Y = \left\{ \vec{y}_1, \vec{y}_2 \right\}$, where $P = Y$, and $\vec{y}_1$ and $\vec{y}_2$ are indicated in Figure~\ref{fig:def-example-tas-2}(b). 
We order $Y$ according to the sequence in which tiles are placed at the POCs, which is a fixed order by the way we created $\mathcal{T}$.
We can then define the winner function $w: Y \rightarrow T$ as $w = \left\{ \left(\vec{y}_1, 4\right), \left(\vec{y}_2, 7\right)  \right\}$.
In this case, the assembly $\alpha_{4,7}$ shown in Figure~\ref{fig:def-example-tas-2}(e) is $w$-correct.
Our goal is to use Theorem~\ref{thm:local-non-determinism-theorem} to derive a lower bound on the probability that $\mathcal{T}$ strictly self-assembles any terminal assembly with the shape of $\alpha_{4,7}$.
We now show that $\mathcal{T}$ is $w$-sequentially non-deterministic such that $\alpha_{4,7}$ is the unique $w$-correct assembly.
Note that $\mathcal{T}$ has $r = 2$ competitions, $\mathcal{C}_1$ and $\mathcal{C}_2$, with respective pairs of competing paths $\pi_1$ and $\pi'_1$ and $\pi_2$ and $\pi'_2$, where $\pi_1$ and $\pi_2$ are the winning paths, and $\pi'_1$ and $\pi'_2$ are the losing paths.
The definitions of $\pi_1$, $\pi'_1$, $\pi_2$ and $\pi'_2$ can be inferred from Figures~\ref{fig:def-example-tas-2}(c)-(e), and it is routine to verify that both pairs of competing paths satisfy the conditions of Definition~\ref{def:competing-paths}. 
Moreover, neither competition for $\mathcal{T}$ is rigged in $\mathcal{T}$ and $\alpha_{4,7}$ is the unique $\mathcal{T}$-terminal assembly with shape $\dom{\alpha_{4,7}}$.
This means Theorem~\ref{thm:local-non-determinism-theorem} gives the exact probability that $\mathcal{T}$ strictly self-assembles the shape $\dom{\alpha_{4,7}}$.
Note that deleting the 10 tile type from $T$ and applying Theorem~\ref{thm:local-non-determinism-theorem} would still give the same probability, but it would no longer be exact because $\mathcal{T}$ (without the 10 tile type) would produce distinct terminal assemblies that both happen to have the shape $\dom{\alpha_{4,7}}$.
Let $D \subset \mathbb{Z}^2$ denote the finite set of points at which a tile is placed by at least one assembly in $\mathcal{A}_{\Box}[\mathcal{T}]$.
We now show that $\mathcal{T}$ satisfies all of the conditions of Definition~\ref{def:seq-non-deterministic} as follows:
\begin{enumerate}
	\item
	Condition~\ref{def:snd-1} of Definition~\ref{def:seq-non-deterministic}: Here, we show that $\mathcal{T}$ is directionally deterministic, i.e., satisfies Definition~\ref{def:dd-ad-pt}:
	Let $\vec{p} \in D$ be arbitrary.
	Assume $\vec{p}$ is any non-POC point in $D$, $\alpha$ and $\beta$ are any $\mathcal{T}$-producible assemblies and $t_\alpha$ and $t_\beta$ are any tile types such that $\alpha + \left( \vec{p}, t_\alpha \right)$ and $\beta + \left(\vec{p}, t_\beta \right)$ are valid tile attachment steps.
	We assume that $\alpha$ and $\beta$ agree and establish, using a proof by cases, that $t_\alpha = t_\beta$.
	\begin{itemize}
		\item $\vec{p} = \vec{x}_2$:
		In this case, two types of tiles may attach at $\vec{x}_2$.
		However, since $\alpha$ and $\beta$ agree, either $\alpha\left(\vec{y}_1\right) = 3 = \beta\left(\vec{y}_1\right)$ or $\alpha\left(\vec{y}_1\right) = 4 = \beta\left(\vec{y}_1\right)$.
		If $\alpha\left(\vec{y}_1\right) = 3 = \beta\left(\vec{y}_1\right)$, then we have $\alpha\left(\vec{p}\right) = \textmd{C} = \beta\left(\vec{p}\right)$.
		Otherwise, we have $\alpha\left(\vec{p}\right) = 5 = \beta\left(\vec{p}\right)$.
		In both cases, we have $t_\alpha = t_\beta$. 
		\item $\vec{p} \ne \vec{x}_2$:  
		Recall that $\vec{p}$ was originally defined to be any non-POC point in $D$.
		In this case, we also have $\vec{p} \ne \vec{x}_2$.
		Thus, at most one tile type may be placed at $\vec{p}$ and it follows that $t_\alpha = t_\beta$.
	\end{itemize}
	Thus, the contrapositive of condition~\ref{def:dd-ad-pt-1} of Definition~\ref{def:dd-ad-pt} is satisfied.
	Now, suppose $\vec{p}$ is any POC point in $D$.
	Since $\vec{y}_1$ and $\vec{y}_2$ are the only POCs, it follows that $\vec{p}$ is either one of these two points.
		We will now show in turn that condition~\ref{def:dd-ad-pt-2} of Definition~\ref{def:dd-ad-pt} holds for the points $\vec{y}_1$ and $\vec{y}_2$:
		\begin{itemize}
		\item $\vec{y}_1$: Exactly two types of tiles may attach at $\vec{y}_1$, namely the 3 and 4 tile types such that the former attaches from the north whereas the latter attaches from the east.
		Thus, condition~\ref{def:dd-ad-pt-2a} of Definition~\ref{def:dd-ad-pt}
		holds for $\vec{y}_1$.
		\item $\vec{y}_2$: Exactly two types of tiles may attach at $\vec{y}_2$, namely the 7 and 9 tile types such that the former attaches from the east whereas the latter attaches from the south.
	Thus, condition~\ref{def:dd-ad-pt-2a} of Definition~\ref{def:dd-ad-pt}
holds for $\vec{y}_2$.
\end{itemize}

	Thus, condition~\ref{def:dd-ad-pt-2} of Definition~\ref{def:dd-ad-pt} is satisfied and it follows that $\mathcal{T}$ is directionally deterministic.

	\item 
	Condition~\ref{def:snd-3} of Definition~\ref{def:seq-non-deterministic}: 

	Note that, if $\vec{\alpha}$ is a $\mathcal{T}$-producing assembly sequence such that $Y$ is a subset of the domain of its resulting assembly, then the 4 tile must attach in some step prior to the step in which the 7 tile attaches, i.e., $\textmd{index}_{\vec{\alpha}}\left(\vec{y}_1\right) < \textmd{index}_{\vec{\alpha}}\left( \vec{y}_2 \right)$.
	\item Condition~\ref{def:snd-4} of Definition~\ref{def:seq-non-deterministic}: Finally, note that $\mathcal{T}$ has two competitions, namely $\mathcal{C}_1$ and $\mathcal{C}_2$ associated with $\vec{y}_1$ and $\vec{y}_2$, respectively with corresponding starting points $\vec{x}_1 = \vec{s}$ and $\vec{x}_2$.
	This means $S_P \cap P = \emptyset$. 
\end{enumerate}
Thus, $\mathcal{T}$ is $w$-sequentially non-deterministic.
We now compute $\textmd{Pr}\left[ \mathcal{C}_1 \right]$ and $\textmd{Pr}\left[ \mathcal{C}_2 \right]$.
First, applying Corollary~\ref{cor:competition-probability-sum} to $\mathcal{C}_1$ with $l = 2$ and $l' = 4$ gives: $\textmd{Pr}\left[ \mathcal{C}_1 \right] =  \left( \frac{1}{2} \right)^{1} +  \left( \frac{1}{2} \right)^{2} + \left( \frac{1}{2} \right)^{3} = \frac{7}{8}$.
Then, applying Corollary~\ref{cor:competition-probability-sum} to $\mathcal{C}_2$ with $l = 3$ and $l' = 3$ gives: $\textmd{Pr}\left[ \mathcal{C}_2 \right] =  { 1 \choose 1 } \cdot \left( \frac{1}{2} \right)^{2} +  { 2 \choose 1 } \cdot \left( \frac{1}{2} \right)^{3} = \frac{1}{2}$.
Finally, Theorem~\ref{thm:local-non-determinism-theorem} says that $\mathcal{T}$ strictly self-assembles the shape $\dom{\alpha_{4,7}}$ with probability at least $\textmd{Pr}\left[ \mathcal{C}_1 \right] \cdot \textmd{Pr}\left[ \mathcal{C}_2 \right] = \frac{7}{8}\cdot \frac{1}{2} = \frac{7}{16} = .4375$.
In the next section, we develop the machinery that we
will use to prove (\ref{eqn:main-inequality}).

\section{Proof of Theorem 1}
\label{sec:proof}

In this section, we prove (\ref{eqn:main-inequality}) by interpreting the summation on its left-hand side as a corresponding SPT, to which we apply a series of transformations that never decrease its probability below the product on its right-hand side.
We first define the main SPT structure that we will use to analyze the correctness of $\mathcal{T}$.
\begin{definition} 
\label{def:w-pruned}
We denote as $\mathcal{P}_w$ the SPT that contains all (and only) maximal paths in $\mathcal{M}_{\mathcal{T}}$ that correspond to $\mathcal{T}$-producing assembly sequences resulting in $\alpha$.
\end{definition}
Note that, by 
the assumptions of Definition~\ref{def:seq-non-deterministic}, $\mathcal{P}_w$ is not empty. Intuitively, Definition~\ref{def:w-pruned} defines the SPT $\mathcal{P}_{w}$ to be the result of ``pruning'' out all the ``bad'' maximal paths of $\mathcal{M}_{\mathcal{T}}$. %
One could define $\mathcal{P}_w$ as the result of removing from $\mathcal{M}_{\mathcal{T}}$ all of the subtrees rooted at nodes that terminate at some POC $\vec{y} \in Y$ and whose resulting assembly places at $\vec{y}$ a tile whose type is different from  $w\left( \vec{y} \right)$.
Such a definition would not eliminate from $\mathcal{P}$ infinite, ``unfair'' $\mathcal{T}$-assembly sequences in which no tile attachment step ever attaches a tile to any POC.
However, by Observation 2.1 of \cite{Dot10}, such assembly sequences do not contribute to $\textmd{Pr}[\mathcal{P}_w]$.

\goingforward{we use $\mathcal{P}$ to denote $\mathcal{P}_w$, where $w$ is a winner function such that $\alpha$ is the unique $w$-correct assembly satisfying $\alpha  \in \mathcal{A}_{\Box}[\mathcal{T}]$ and $Y \subseteq \dom{\alpha}$.}

Figure~\ref{fig:def-example-tas-Mt} depicts the SPT $\mathcal{M}_\mathcal{T}$ that corresponds to our running example $\mathcal{T}$, whereas Figure~\ref{fig:def-example-tas-P_t_w} depicts $\mathcal{P}_{w}$ where $\mathcal{T}$ is our running example TAS and $w'~=~\{(\vec{y}_1,3),(\vec{y}_2,9)\}$.

\begin{figure}[!h]
      \centerline{\includegraphics[width=\linewidth]{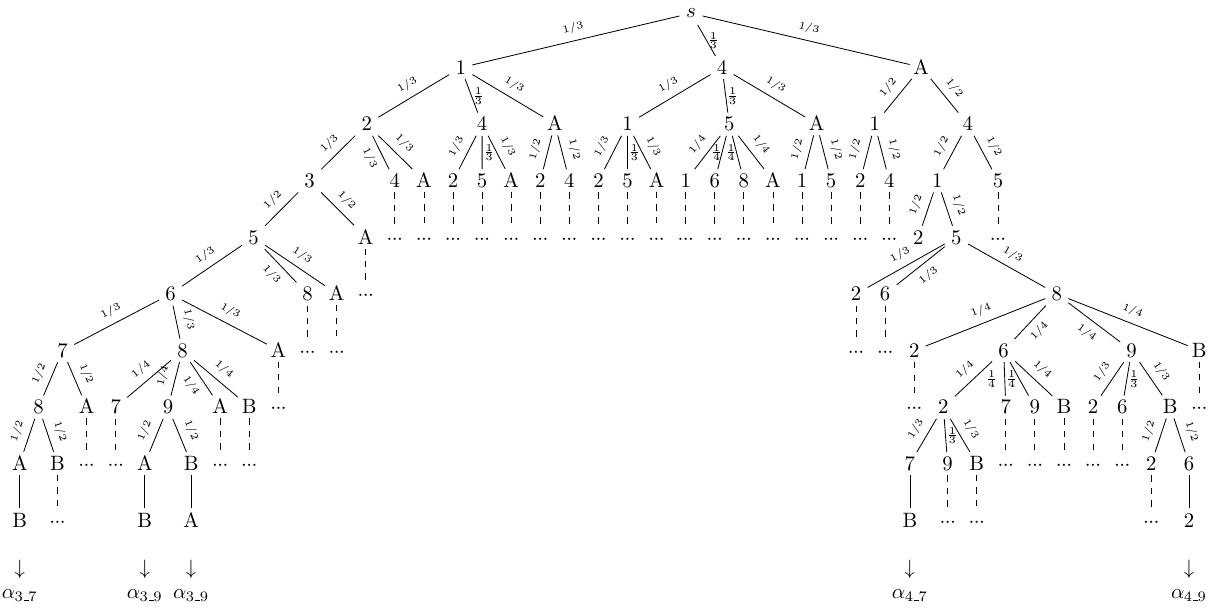}}
      \caption{\label{fig:def-example-tas-Mt} The (partial) Markov
        chain $\mathcal{M}_{\mathcal{T}}$ for our running example TAS
        $\mathcal{T}$. Note that the full tree would have all of its
        leaf nodes at the same depth as the five included leaf nodes and
        would not fit in one page.  In this and the following tree
        figures, each node is labeled not with the assembly sequence
        it corresponds to, but rather with the tile type that attaches
        in the last tile attachment step of the sequence. The sequence
        itself can be recovered by following (and reversing) the path
        from the node back up to the root of the tree. For
        illustrative purposes, the assembly resulting from each leaf
        node is shown under each one of the five included leaf nodes.
        These terminal assemblies are the ones depicted in
        Figure~\ref{fig:def-example-tas}c through
        Figure~\ref{fig:def-example-tas}f. Note that two different
        sequences (i.e., leaf nodes) are shown for $\alpha_{3,9}$.}
    \end{figure}

The following observation follows directly from Definition~\ref{def:w-pruned}.
\begin{observation} 
\label{obs:alpha-w-correct-iff-node-in-p-corr-alpha}
Every node $\vec{\beta}$ of $\mathcal{P}$ is a finite, $w$-correct $\mathcal{T}$-producing assembly sequence such that $\res{\vec{\beta}}~\sqsubseteq~\alpha$.
\end{observation}
\begin{figure}[!h]
      \centerline{\includegraphics[width=\linewidth]{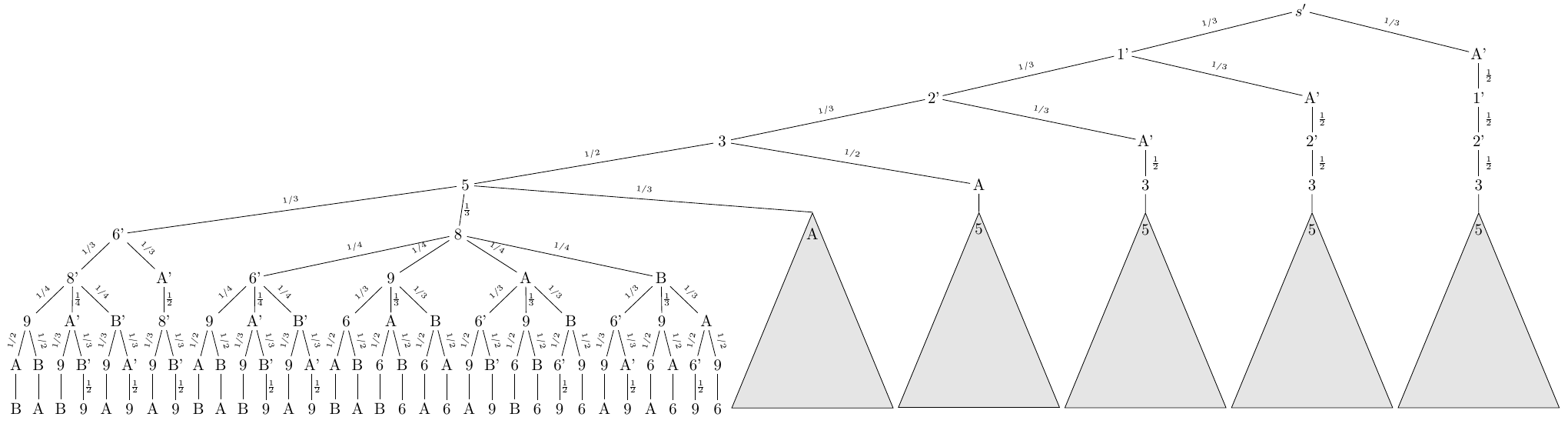}}

      \begin{minipage}{0.79\linewidth}
        \caption{\label{fig:def-example-tas-P_t_w} $\mathcal{P}_{w'}$
          where $\mathcal{T}$ is our running example TAS and
          $w'~=~\{(\vec{y}_1,3),(\vec{y}_2,9)\}$. Therefore, every leaf node is a $\mathcal{T}$-producing assembly sequence whose
          result is $\alpha_{3,9}$, shown in Figure~\ref{fig:def-example-tas}d.
          In order to reduce the width of this large tree, every occurrence of
          the repeated subtree shown on the right is depicted as a shaded
          triangle in the tree above. Note that every node whose label ends
          with prime had a child node pruned because it corresponds to placing
          a tile on a POC that is not the tile allowed by $w'$.
          }
      \end{minipage}
      \begin{minipage}{0.2\linewidth}
        \centering
        \centerline{\includegraphics[width=1in]{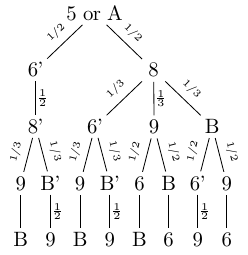}}
      \end{minipage}
    \end{figure}

The next two definitions give us a way of identifying ``levels'' of nodes within rooted trees  such as $\mathcal{P}$. 

\begin{definition}  
\label{def:bottleneck-set-of-nodes}
Let $\mathcal{Q}$ be a rooted tree and $B$ be a set containing nodes of Q (among possibly other elements). We say that $B$ is a \emph{bottleneck of} $\mathcal{Q}$ if, for every maximal path $\pi$ in $\mathcal{Q}$, there exists a unique $b \in B$ such that $b \in \dom{\pi}$.
\end{definition}

Note that the parenthesized expression in this definition applies to scenarios where $B$ contains extra nodes that do not belong to $\mathcal{Q}$. For example, in Lemma~\ref{lem:p-restricted-u} below, $B$ may contain nodes that belong to a larger tree that contains $\mathcal{Q}$ as a subtree.

\begin{definition} 
\label{def:spt-tas-restricted-set-of-nodes}
If $\mathcal{Q}$ is a tree rooted at node $u$ and $B$ is a bottleneck of $\mathcal{Q}$, we define $\mathcal{Q} \upharpoonright B$ to be the restriction of $\mathcal{Q}$ to $B$, that is, the unique subtree of  $\mathcal{Q}$ containing all of the paths in $\mathcal{Q}$ from $u$ to any node in $B$.
\end{definition}

The next lemma is a straightforward characterization of the probability of an SPT in terms of a bottleneck set of its nodes.

\begin{lemma} 
\label{lem:alternative-characterization-of-pr}
If $\mathcal{Q}$ is an SPT and $B$ is a bottleneck of $\mathcal{Q}$, then $\textmd{Pr}[\mathcal{Q}] = \displaystyle\sum_{v \in B}{\left( \textmd{Pr}_{\mathcal{Q}}\left[v\right]\cdot \textmd{Pr}\left[ \mathcal{Q}^{v} \right] \right)}$.
\end{lemma}
\begin{proof}
Assume $u$ is the root of $\mathcal{Q}$.
Let $v$ be a node of $\mathcal{Q}$ and $\pi'$ be a maximal path in $\mathcal{Q}^v$.
Since every maximal path in $\mathcal{Q}^v$ starts at $v$ and there exists a unique simple path from $u$ to $v$ in $\mathcal{Q}$, it follows that there exists a unique maximal path $\pi = \left( u, \ldots, v, \ldots \right)$ in $\mathcal{Q}$ such that $\textmd{Pr}_{\mathcal{Q}}\left[ v \right]\cdot \textmd{Pr}_{\mathcal{Q}^v}\left[ \pi' \right] = \textmd{Pr}_{\mathcal{Q}}\left[ \pi \right]$.
Let $\pi$ be a maximal path in $\mathcal{Q}$.
Since $B$ is a bottleneck of $\mathcal{Q}$, there exists a unique node $v \in B$ such that $v \in \dom{\pi}$.
This means we can write $\pi = \left( u, \ldots, v, \ldots \right)$ such that:
\begin{enumerate}
	\item $\left(u, \ldots, v\right)$ is the unique path in $\mathcal{Q}$ from the root to $v$, and 
	\item $\pi' = \left( v, \ldots \right)$ is the unique maximal path in $\mathcal{Q}^v$ such that $\textmd{Pr}_{\mathcal{Q}}\left[ v \right]\cdot \textmd{Pr}_{\mathcal{Q}^v}\left[ \pi' \right] = \textmd{Pr}_{\mathcal{Q}}\left[ \pi \right]$.
\end{enumerate}
The previous two paragraphs imply the following equality:
\begin{equation}
\label{eqn:prob-max-path}
\displaystyle\sum_{v \in B}{\left( \sum_{\pi' \textmd{ maximal path in } \mathcal{Q}^{v}}{\left( \textmd{Pr}_{\mathcal{Q}}\left[v\right] \cdot \textmd{Pr}_{\mathcal{Q}^{v}}\left[ \pi'\right] \right)} \right)} = \displaystyle\sum_{\pi \textmd{ maximal path in } \mathcal{Q}}{\textmd{Pr}_{\mathcal{Q}}[\pi]}
\end{equation}
Then, we have:
\[
\begin{array}{llll}
\displaystyle\sum_{v \in B}{\left( \textmd{Pr}_{\mathcal{Q}}\left[v\right]\cdot \textmd{Pr}\left[ \mathcal{Q}^{v} \right] \right)} & = & \displaystyle\sum_{v \in B}{\left( \textmd{Pr}_{\mathcal{Q}}\left[v\right] \cdot \left( \sum_{\pi' \textmd{ maximal path in } \mathcal{Q}^{v}}{\textmd{Pr}_{\mathcal{Q}^{v}}\left[\pi'\right]} \right) \right)} & \textmd{ Definition of } \textmd{Pr}\left[\mathcal{Q}^{v}\right] \\

	& = & \displaystyle\sum_{v \in B}{\left( \sum_{\pi' \textmd{ maximal path in } \mathcal{Q}^{v}}{\left( \textmd{Pr}_{\mathcal{Q}}\left[v\right] \cdot \textmd{Pr}_{\mathcal{Q}^{v}}\left[ \pi'\right] \right)} \right)} & \\
	
	& = & \displaystyle\sum_{\pi \textmd{ maximal path in } \mathcal{Q}}{\textmd{Pr}_{\mathcal{Q}}[\pi]} & \textmd{ Equation~(\ref{eqn:prob-max-path}) } \\
	
	& = & \textmd{Pr}[\mathcal{Q}] & \textmd{ Definition of } \textmd{Pr}[\mathcal{Q}].
\end{array}
\]
\end{proof}

The following corollary is analogous to Lemma 1 of \cite{ChandranGR12} but is stated and proved in the framework and notation of this paper.  

\begin{corollary} 
\label{cor:normalized-probability-1}
If $\mathcal{Q}$ is an SPT and all of its internal nodes are normalized, then $\textmd{Pr}[\mathcal{Q}] = 1$.
\end{corollary}

\begin{proof}
We prove this lemma by structural induction on SPTs. For the basis
step, if $\mathcal{Q}$ is made up of a single node, then the
corollary holds vacuously. For the inductive step, define $B$ to be the
(non-empty) set of all the child nodes in $\mathcal{Q}$ of the root of
$\mathcal{Q}$. Trivially, $B$ is a bottleneck of $\mathcal{Q}$ and
$\mathcal{Q}^b$ is an SPT for each node $b \in B$.  Now assume
(inductive hypothesis) that, for each $b \in B$, if all of the
internal nodes of $\mathcal{Q}^b$ are normalized, then
$\textmd{Pr}[\mathcal{Q}^b] = 1$. We now prove that the corollary
holds for $\mathcal{Q}$. Assume that its antecedent holds for $\mathcal{Q}$, i.e., all
of the internal nodes of $\mathcal{Q}$ are normalized. Therefore:
\begin{enumerate}
\item[a)] the root of $\mathcal{Q}$, as one of its internal nodes, is normalized
and
\item[b)] for each $b \in B$, all of the internal nodes of $\mathcal{Q}^b$ (if any) are normalized.
\end{enumerate}

Item  b) above, combined  with our inductive hypothesis, yields: 
\begin{equation}
\label{eqn:prQb=1}
\textmd{for each } b \in B \textmd{, Pr}[\mathcal{Q}^b] = 1
\end{equation}

We now prove that the conclusion of the corollary holds for  $\mathcal{Q}$:

\[
\begin{array}{llll}
\textmd{Pr}[\mathcal{Q}] & = & \displaystyle\sum_{b \in B}{\left( \textmd{Pr}_{\mathcal{Q}}\left[b\right]\cdot \textmd{Pr}\left[ \mathcal{Q}^{b} \right] \right)} & \textmd{ Lemma~\ref{lem:alternative-characterization-of-pr}} \\

		 & = & \displaystyle\sum_{b \in B}{\left( \textmd{Pr}_{\mathcal{Q}}\left[b\right]\cdot 1 \right)} & \textmd{ Equation~\ref{eqn:prQb=1} }\\

		& = & 1 & \textmd{ Item a) above combined with the definition of } B\\

\end{array}
\]
\end{proof}
    
The next lemma says that a finite, $w$-correct $\mathcal{T}$-producing assembly sequence cannot be extended by attaching a tile at an inessential POC.
\begin{lemma} 
\label{lem:w-correct-no-inessential-poc}
If 
%
$\vec{y} \in P \backslash Y$, and
    $\vec{\beta} = \left( \beta_i \mid 1 \leq i \leq k \right)$ is a finite, $w$-correct $\mathcal{T}$-producing assembly sequence with $\beta = \res{\vec{\beta}} \not \in \mathcal{A}_{\Box}[\mathcal{T}]$, 
    then $\vec{y} \not \in \partial^{\mathcal{T}}\beta$.
\end{lemma}
%
Lemma~\ref{lem:w-correct-no-inessential-poc} and Observation~\ref{obs:alpha-w-correct-iff-node-in-p-corr-alpha} together imply that no node of $\mathcal{P}$ can be extended by attaching a tile at an inessential POC.
%
\begin{proof}
Assume, for the sake of obtaining a contradiction, that $\vec{y} \in \partial^{\mathcal{T}} \beta$.
Assume there exists a $w$-correct $\mathcal{T}$-terminal assembly $\alpha$ which, by Corollary~\ref{cor:unique-w-correct-assembly}, is unique.
Note that, by Lemma~\ref{lem:w-correct-subassembly-w-correct}, we have $\beta \sqsubseteq \alpha$.
We also know that $\alpha \in \mathcal{A}_{\Box}[\mathcal{T}]$ and $\vec{y} \in \partial^{\mathcal{T}} \beta$.
Thus, $\vec{y} \in \dom{\alpha}$.
Since $\vec{\beta}$ is a finite, $w$-correct, $\mathcal{T}$-producing assembly sequence and $\beta_k \not \in \mathcal{A}_{\Box}[\mathcal{T}]$, by Lemma~\ref{lem:w-correct-extended-to-w-correct-terminal}, there exist $m \in \Z^+\cup\{\infty\}$ and an extension $\vec{\beta}' = \left( \beta'_i \mid 0 \leq i-1 < m\right)$ of $\vec{\beta}$ that results in $\alpha$.
Since $\vec{y} \in \dom{\alpha}$, there exists an integer $i>1$ such that $\beta'_i = \beta'_{i-1} + \left( \vec{y}, \alpha\left(\vec{y}\, \right) \right)$. 
Therefore, $\vec{y} \in \dom{\beta'_i}$. Since $\vec{y} \in P \backslash Y$, condition~\ref{def:correct-assembly-2} of Definition~\ref{def:correct-assembly} implies that $\beta'_i$ is not $w$-correct.
But this contradicts the fact that $\beta'_i$ must be $w$-correct, since  $\vec{\beta}'$ is an extension of the $w$-correct $\vec{\beta}$ that results in the $w$-correct $\alpha$. This contradiction means that  $\vec{y} \not \in \partial^{\mathcal{T}} \beta$. 
\end{proof}
%
%
%

The next lemma says that every $w$-correct node in $\mathcal{M}_{\mathcal{T}}$ that is a child of some node in $\mathcal{P}$ is a node in $\mathcal{P}$.

\begin{lemma} 
\label{lem:extend-to-alpha}
Let $\vec{\beta}$ be a finite, $w$-correct, $\mathcal{T}$-producing assembly sequence such that $\beta = \res{\vec{\beta}}$ and $\beta \not \in \mathcal{A}_{\Box}[\mathcal{T}]$. If $\vec{\beta}'$ is a finite, $w$-correct, $\mathcal{T}$-producing assembly sequence, and is a child of $\vec{\beta}$ in $\mathcal{M}_{\mathcal{T}}$, then $\vec{\beta}'$ is a child of $\vec{\beta}$ in $\mathcal{P}$. 
\end{lemma}

To show that a node of (a subtree that is full relative to) $\mathcal{P}$ is normalized, it suffices to show that the set of its child nodes in $\mathcal{M}_{\mathcal{T}}$ is the set of its child nodes in (the subtree of) $\mathcal{P}$. 
We will use Lemma~\ref{lem:extend-to-alpha} as an intermediate step in showing that such a node is normalized. 

\begin{proof}
Let $\beta^{\prime } = \res{\vec{\beta}^{ \prime}}$.
We consider two cases:
\begin{enumerate}
	\item Assume $\beta^{\prime } \in \mathcal{A}_{\Box}[\mathcal{T}]$.
	By Observation~\ref{obs:alpha-w-correct-iff-node-in-p-corr-alpha}, every leaf node of $\mathcal{P}$ corresponds to a $w$-correct, $\mathcal{T}$-producing assembly sequence that, by Corollary~\ref{cor:unique-w-correct-assembly}, results in $\alpha$.
	This means $\vec{\beta}^{\prime }$ is a leaf node in $\mathcal{P}$ and therefore a child of $\vec{\beta}$ in $\mathcal{P}$.
	\item Assume $\beta^{\prime } \not \in \mathcal{A}_{\Box}[\mathcal{T}]$. Then, by Lemma~\ref{lem:w-correct-extended-to-w-correct-terminal}, there is an extension of $\vec{\beta}^{\prime }$ by some $\mathcal{T}$-assembly sequence that results in $\alpha$.
	Such an extension testifies to $\vec{\beta}^{\prime }$ being a child of $\vec{\beta}$ in $\mathcal{P}$ because it means $\vec{\beta}^{\prime }$ is on a maximal path of $\mathcal{P}$ that corresponds to a $w$-correct, $\mathcal{T}$-producing assembly sequence that results in $\alpha$. 
\end{enumerate}
In each case, $\vec{\beta}^{\prime }$ is a child of $\vec{\beta}$ in $\mathcal{P}$.
\end{proof}

The following lemma gives sufficient conditions for when the probability of a certain kind of subtree of $\mathcal{P}$ is equal to 1.

\begin{lemma} 
\label{lem:deterministic-spt-probability-1-generalized}
Let $\mathcal{Q}$ be a subtree of $\mathcal{P}$ such that $\mathcal{Q}$ is full relative to $\mathcal{P}$. If no nodes in $\mathcal{Q}$, with the possible exception of its root, terminate at an essential POC in $\mathcal{T}$, then every internal node of $\mathcal{Q}$ is normalized.
\end{lemma}
\begin{proofsketch}
In the non-vacuous case, we let $\vec{\beta}$ be an arbitrary internal
node of $\mathcal{Q}$ and proceed as follows:
\begin{enumerate}
\item Prove that no child of  $\vec{\beta}$ in  $\mathcal{M}_{\mathcal{T}}$ may terminate at a POC (essential or otherwise) in $\mathcal{T}$.
\item Prove that any child of  $\vec{\beta}$ in  $\mathcal{M}_{\mathcal{T}}$ must be a child of  $\vec{\beta}$ in  $\mathcal{Q}$.
\item Prove that any internal node of $\mathcal{Q}$ is normalized.
\end{enumerate}\vspace*{-5mm}
\end{proofsketch}

\begin{proof}
If $\mathcal{Q}$ is made up of a single node, then the lemma holds vacuously.
So, going forward, assume $\mathcal{Q}$ contains more than one node
and that every child node in $\mathcal{Q}$ terminates at a point that is not in $Y$.
For some $k \in \mathbb{Z}^+$, let $\vec{\beta} = \left( \beta_i \mid 1 \leq i \leq k \right)$ be an arbitrary internal node of $\mathcal{Q}$. 
By Observation~\ref{obs:alpha-w-correct-iff-node-in-p-corr-alpha},  $\vec{\beta}$ is a finite, $w$-correct, $\mathcal{T}$-producing assembly sequence.
Since $\vec{\beta}$ is an internal node of $\mathcal{Q}$, and $Q$ is a subtree of $\mathcal{P}$, $\vec{\beta}$ has at least one child in $\mathcal{M}_{\mathcal{T}}$. 
Let $\vec{\beta}^{\prime}$ be an arbitrary child of $\vec{\beta}$ in $\mathcal{M}_{\mathcal{T}}$.
{\em Step 1.}
Suppose, for the sake of obtaining a contradiction, that $\vec{\beta}^{\prime}$ terminates at some $\vec{y} \in P$. 
Since, by Lemma~\ref{lem:w-correct-no-inessential-poc}, no child of $\vec{\beta}$ in $\mathcal{M}_{\mathcal{T}}$, including $\vec{\beta}'$, may terminate at any point in $P\backslash Y$, $\vec{y} \in Y$.
Let $t = w\left(\vec{y}\right) \in T$.
Then $\vec{y} \in \partial^{\mathcal{T}}_t \beta_k$ because $\vec{\beta}^{\prime}$ terminates at $\vec{y}$ and $\vec{\beta}$ is a $w$-correct assembly sequence that results in a $\mathcal{T}$-assembly that is not terminal.
Thus $\vec{\beta}^{\prime } = \left( \beta_1, \ldots, \beta_k, \beta_k + \left(\vec{y},t\right) \right)$ is finite, $w$-correct, and $\mathcal{T}$-producing, because $\vec{\beta}$ satisfies all of these conditions. 
By Lemma~\ref{lem:extend-to-alpha}, $\vec{\beta}^{\prime}$ is a child of $\vec{\beta}$ in $\mathcal{P}$.
Moreover, by definition, $\vec{\beta}^{\prime}$ is a node that terminates at an essential POC in $\mathcal{T}$.
Therefore, $\vec{\beta}^{\prime}$ is a child of $\vec{\beta}$ in $\mathcal{Q}$, since  $\mathcal{Q}$ is assumed to be full relative to  $\mathcal{P}$.
However, such a $\vec{\beta}^{\prime }$ being a child of $\vec{\beta}$ in $\mathcal{Q}$ contradicts the hypothesis of this lemma.
Therefore, no child of $\vec{\beta}$ in $\mathcal{M}_{\mathcal{T}}$ may terminate at $\vec{y} \in P$. 
{\em Step 2.}
Assume $\vec{\beta}^{\prime} = \left( \beta_1, \ldots, \beta_k, \beta_{k+1} \right)$ and let $\vec{p}\in\Z^2$ be the single element in $\dom{\beta_{k+1}} \backslash \dom{\beta_k}$.
Then $\vec{p}\not \in P$, because we ruled out in Step 1 the possibility of any child of $\vec{\beta}$ terminating at a POC.
This means $\vec{\beta}^{\prime}$ is finite, $w$-correct, and $\mathcal{T}$-producing, because $\vec{\beta}$ satisfies all of these conditions.
By Lemma~\ref{lem:extend-to-alpha}, it follows that $\vec{\beta}^{\prime}$ is a child of $\vec{\beta}$ in $\mathcal{Q}$, since  $\mathcal{Q}$ is assumed to be full relative to  $\mathcal{P}$.
{\em Step 3.} Since we assume that $\mathcal{Q}$ is full relative to $\mathcal{P}$, if $\vec{\beta}^{\prime}$ is a child of $\vec{\beta}$ in $\mathcal{M}_{\mathcal{T}}$ (and thus a child of $\vec{\beta}$ in $\mathcal{P}$ as well, since $\vec{\beta}^{\prime}$ is $w$-correct), then $\vec{\beta}^{\prime}$ is also a child of $\vec{\beta}$ in $\mathcal{Q}$.
Since all children of  $\vec{\beta}$ in $\mathcal{M}_{\mathcal{T}}$ are also children of $\vec{\beta}$ in $\mathcal{Q}$, the assignment of probabilities in $\mathcal{M}_{\mathcal{T}}$ implies that every internal node of $\mathcal{Q}$ is normalized. 
\end{proof}

\begin{corollary} 
\label{cor:deterministic-spt-probability-1-generalized}
$\textmd{Pr}\left[ \mathcal{Q} \right] = 1$
\end{corollary}

\begin{proof}
Corollary~\ref{cor:normalized-probability-1} implies that  $\textmd{Pr}[\mathcal{Q}] = 1$.
\end{proof}

If $\mathcal{Q}$ is finitely branching (which $\mathcal{M}_{\mathcal{T}}$ and any of its subtrees are), then any restriction of $\mathcal{Q}$ to one of its  bottlenecks is a finite tree.

\begin{lemma} 
\label{lem:restricted-tree-finite}
If $\mathcal{Q}$ is a rooted, finitely branching tree, and $B$ is a bottleneck of $\mathcal{Q}$, then $\mathcal{Q} \upharpoonright B$ is finite and its set of leaf nodes is $B$.
\end{lemma}

\begin{proof}
Since $B$ is a bottleneck of $\mathcal{Q}$, by Definition~\ref{def:bottleneck-set-of-nodes}, every maximal path in $\mathcal{Q}$ contains a unique element of $B$.
By Definition~\ref{def:spt-tas-restricted-set-of-nodes}, $\mathcal{Q} \upharpoonright B$ is the unique subtree of $\mathcal{Q}$ comprising the set of all paths that start at the root of  $\mathcal{Q}$ and end at some element of $B$.
Thus, $\mathcal{Q} \upharpoonright B$ is well-defined and its set of leaf nodes is $B$.
Assume, for the sake of obtaining a contradiction that, $\mathcal{Q} \upharpoonright B$ is infinite. 
First, note that $\mathcal{Q}$ is finitely branching. Thus $\mathcal{Q} \upharpoonright B$ cannot contain an internal node with infinitely many children.
Then, K\H{o}nig's Lemma says that $\mathcal{Q} \upharpoonright B$ contains an infinite simple path $\pi$ starting at its root, which means that $\pi$ is a maximal path in $\mathcal{Q} \upharpoonright B$.
However, $B$ is a bottleneck of $\mathcal{Q}$. Thus $B \cap \pi \ne \emptyset$ and $\pi$ cannot be an infinite path in $\mathcal{Q} \upharpoonright B$.
\end{proof}

The next definition gives us a convenient notation for restricting $\mathcal{P}$ based on the points at which $\mathcal{T}$-assembly sequences terminate.

\begin{definition} 
\label{def:spt-tas-restricted-general}
Let $\vec{p} \in \mathbb{Z}^2$, $\mathcal{Q}$ be a subtree of $\mathcal{M}_{\mathcal{T}}$ such that every maximal path of $\mathcal{Q}$ contains a node that terminates at $\vec{p}$, and $B_{\vec{p}}$ be the bottleneck of $\mathcal{Q}$ whose elements terminate at $\vec{p}$. Then, we define $\mathcal{Q} \upharpoonright \vec{p} = \mathcal{Q}~\upharpoonright~B_{\vec{p}}$.
\end{definition}

Figure~\ref{fig:def-example-tas-P_t_w-restricted}
    depicts $\mathcal{P}_{w} \upharpoonright \vec{y}_2$
    where $\mathcal{T}$ is our
    running example TAS and $w~=~\{(\vec{y}_1,3),(\vec{y}_2,9)\}$.
    Note that $\mathcal{P}_{w} \upharpoonright \vec{x}_1
    = \mathcal{P}_{w} \upharpoonright \vec{s}$, which is the SPT
    containing only the root node of $\mathcal{P}_{w}$. 
    
    \begin{figure}[!h]
      \centerline{\includegraphics[width=\linewidth]{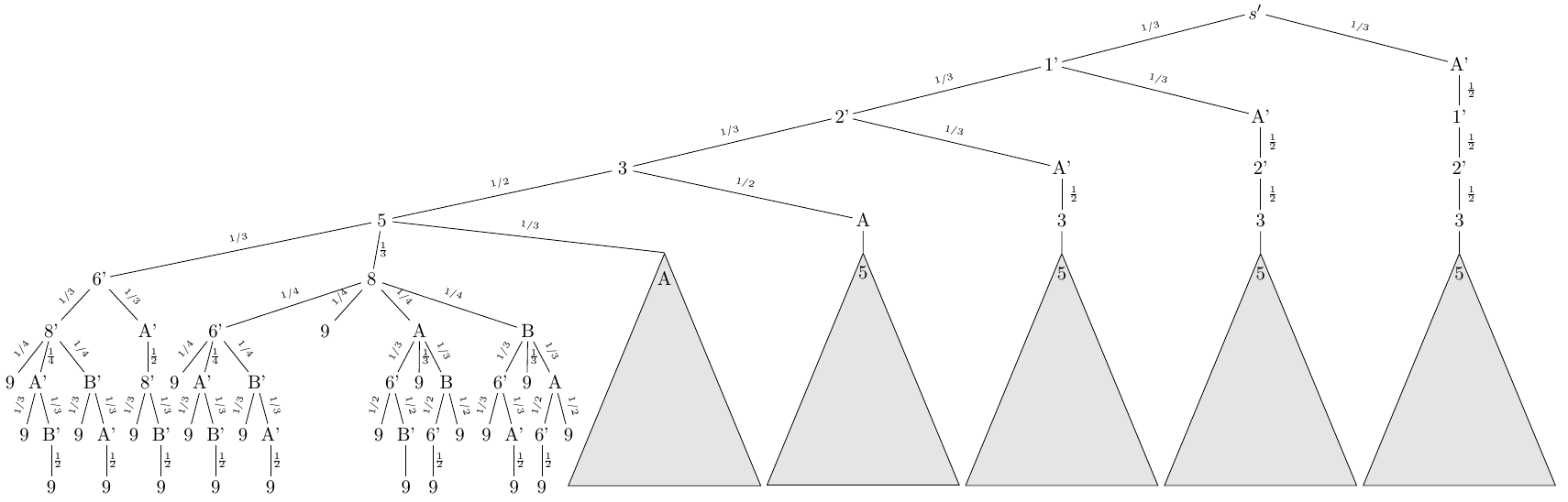}}

      \begin{minipage}{0.79\linewidth}
        \caption{\label{fig:def-example-tas-P_t_w-restricted}
          $\mathcal{P}_{w} \upharpoonright \vec{y}_2$,
          where $\mathcal{T}$ is our running example TAS and
          $w~=~\{(\vec{y}_1,3),(\vec{y}_2,9)\}$, and $\vec{y}_2$ is
          the second point of competition in $\mathcal{T}$ at which
          tile type 9 must attach according to $w$. Therefore, this
          tree is a restriction of the tree in
          Figure~\ref{fig:def-example-tas-P_t_w} in which all of the
          subtrees under the nodes labeled with a 9 have been pruned.
          Again, every occurrence of the repeated subtree shown on
          the right is depicted as a shaded triangle in the tree
          above.          
        }
      \end{minipage}
      \begin{minipage}{0.2\linewidth}
        \centering
        \centerline{\includegraphics[width=1in]{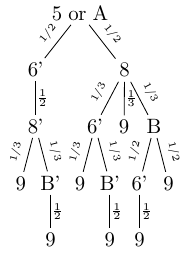}}
      \end{minipage}
        \end{figure}

    Figure~\ref{fig:p-levels} depicts $\mathcal{P}$ in terms of ``levels'' defined by essential POCs and corresponding starting points. 
    The top level of  $\mathcal{P}$ is $\mathcal{P} \upharpoonright \vec{x}_1$. 
    The following lemma allows us to ``collapse'' this top level when computing the probability of $\mathcal{P}$. 

\begin{lemma} 
\label{lem:probability-x-1}
$\mathcal{P} \upharpoonright \vec{x}_1$ is well-defined and $\textmd{Pr}\left[ \mathcal{P} \upharpoonright \vec{x}_1 \right] = 1$.
\end{lemma}

The basic proof idea here is to show that no node in $\mathcal{P} \upharpoonright \vec{x}_1$ terminates at an essential POC in $\mathcal{T}$ and then apply Corollary~\ref{cor:deterministic-spt-probability-1-generalized}.

\begin{figure}[!h]
	\centerline{\includegraphics[width=\linewidth]{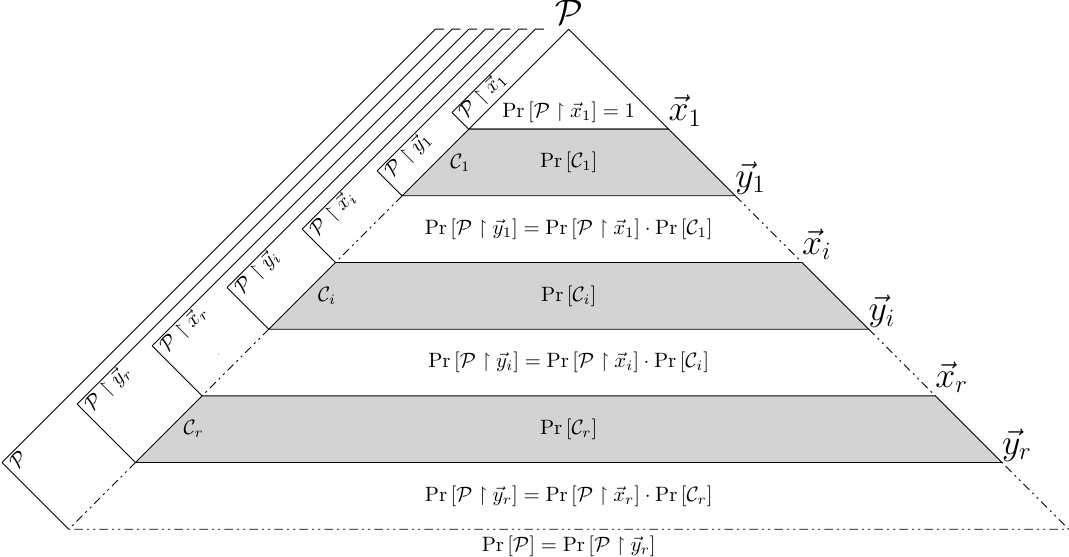}}
	\caption{\label{fig:p-levels} Depiction of $\mathcal{P}$ in terms of ``levels'' defined by essential POCs and corresponding starting points. The competitions are shown in grey. Intuitively, the white portions of $\mathcal{P}$ correspond to tile attachment steps that, as we will prove formally in Theorem~\ref{thm:local-non-determinism-theorem}, do not affect $\textmd{Pr}[\mathcal{P}]$, and therefore can be ``collapsed,'' leaving just $\textmd{Pr}\left[ \mathcal{P} \right] = \textmd{Pr}\left[ \mathcal{C}_1 \right] \cdot\ \cdots\ \cdot \textmd{Pr}\left[ \mathcal{C}_r \right]$.}
\end{figure}

\begin{proof}
It suffices to show that no node in $\mathcal{P} \upharpoonright \vec{x}_1$ terminates at an essential POC in $\mathcal{T}$.
Let $\pi$ be an arbitrary maximal path in $\mathcal{P}$.
By Definition~\ref{def:w-pruned}, $\pi$ corresponds to a $w$-correct, $\mathcal{T}$-producing assembly sequence $\vec{\alpha}$ that results in $\alpha$. 
Recall that $\alpha$ is the unique, $w$-correct, terminal assembly of $\mathcal{T}$ such that $Y \subseteq \dom{\alpha}$. Therefore, by Lemma~\ref{lem:x-i-y-i-ordering-r-1}, $\left\{ \vec{x}_1, \ldots, \vec{x}_r \right\} \subseteq \dom{\alpha}$. Thus there exists $\vec{\beta}_{\vec{x}_1} \in \dom{\pi}$ that terminates at $\vec{x}_1$, which means $\mathcal{P} \upharpoonright \vec{x}_1$ is well-defined. 

Now, if $r=1$, by Lemma~\ref{lem:x-i-y-i-ordering-r-1},  $\textmd{index}_{\vec{\alpha}}\left( \vec{x}_1 \right) < \textmd{index}_{\vec{\alpha}}\left( \vec{y}_1 \right)$. If $r>1$, repeated applications of Lemma~\ref{lem:x-i-y-i-ordering} yield
$\textmd{index}_{\vec{\alpha}}\left( \vec{x}_1 \right) < \textmd{index}_{\vec{\alpha}}\left( \vec{y}_1 \right) < \textmd{index}_{\vec{\alpha}}\left( \vec{x}_2 \right) < \ldots < \textmd{index}_{\vec{\alpha}}\left( \vec{x}_{r-1} \right) < \textmd{index}_{\vec{\alpha}}\left( \vec{y}_{r-1} \right) < \textmd{index}_{\vec{\alpha}}\left( \vec{x}_r \right)$. Then one more application of  Lemma~\ref{lem:x-i-y-i-ordering-r-1} yields $\textmd{index}_{\vec{\alpha}}\left( \vec{x}_r \right) < \textmd{index}_{\vec{\alpha}}\left( \vec{y}_r \right)$. Consequently, for all values of $r\in\Z^+$ and all integers $1 \leq i \leq r$, $\textmd{index}_{\vec{\alpha}}\left( \vec{x}_1 \right) < \textmd{index}_{\vec{\alpha}}\left( \vec{y}_i \right)$. Therefore, if $\vec{\beta} \in \dom{\pi}$ and $\vec{\beta}$ terminates at an essential POC in $\mathcal{T}$, then $\vec{\beta}$ is a descendant of $\vec{\beta}_{\vec{x}_1}$  in $\mathcal{P}$.
In other words, no ancestor of $\vec{\beta}_{\vec{x}_1}$ in $\mathcal{P}$ terminates at an essential POC in $\mathcal{T}$.
Therefore, no node in $\mathcal{P} \upharpoonright \vec{x}_1$ terminates at an essential POC in $\mathcal{T}$.
Since $\mathcal{P} \upharpoonright \vec{x}_1$ is full relative to $\mathcal{P}$, Corollary~\ref{cor:deterministic-spt-probability-1-generalized} implies that $\textmd{Pr}\left[\mathcal{P} \upharpoonright \vec{x}_1 \right] = 1$.
\end{proof}

The following lemma says that any subtree of $\mathcal{P}$ whose root terminates at the last POC can be ``collapsed'' when computing the probability of $\mathcal{P}$.

\begin{lemma} 
\label{lem:terminate-at-y-r-pr-1}
If $\vec{\beta}_{\vec{y}_r}$ is a node of $\mathcal{P}$ that terminates at $\vec{y}_r$, then $\textmd{Pr}\left[ \mathcal{P}^{\vec{\beta}_{\vec{y}_r}} \right] = 1$.
\end{lemma}

\begin{proof}

It suffices to show that no node in  $\mathcal{P}^{\vec{\beta}_{\vec{y}_r}}$ (except for its root) terminates at an essential POC in $\mathcal{T}$.
Let $\pi$ be any path obtained by concatenating the unique path
in $\mathcal{P}$ from its root to  $\vec{\beta}_{\vec{y}_r}$ with an arbitrary
maximal path in $\mathcal{P}^{\vec{\beta}_{\vec{y}_r}}$.
By Definition~\ref{def:w-pruned}, $\pi$ corresponds to a $w$-correct, $\mathcal{T}$-producing assembly sequence  $\vec{\alpha}$ that results in $\alpha$, the unique $w$-correct terminal assembly of $\mathcal{T}$ such that $Y \subseteq \dom{\alpha}$.
By condition~\ref{def:snd-3} of Definition~\ref{def:seq-non-deterministic}, for all integers $1 \leq i < r$, $\textmd{index}_{\vec{\alpha}}\left( \vec{y}_r \right) > \textmd{index}_{\vec{\alpha}}\left(\vec{y}_i \right)$.
This implies that, if $\vec{\beta} \in \dom{\pi}$ and $\vec{\beta}$ terminates at an essential POC in $\mathcal{T}$, then either $\vec{\beta} = \vec{\beta}_{\vec{y}_r}$ or $\vec{\beta}$ is an ancestor  of $\vec{\beta}_{\vec{y}_r}$ in $\mathcal{P}$. 
Thus, no descendant of $\vec{\beta}_{\vec{y}_r}$ in $\mathcal{P}$ terminates at an essential POC in $\mathcal{T}$.
In other words, no node in $\mathcal{P}^{\vec{\beta}_{\vec{y}_r}}$ (except for its root) terminates at an essential POC in $\mathcal{T}$.
Since $\mathcal{P}^{\vec{\beta}_{\vec{y}_r}}$ is full relative to $\mathcal{P}$, Corollary~\ref{cor:deterministic-spt-probability-1-generalized} implies that $\textmd{Pr}\left[ \mathcal{P}^{\vec{\beta}_{\vec{y}_r}} \right] = 1$.
\end{proof}

\begin{observation}
\label{obs:prob-subtree-same-root}
If $\mathcal{Q}$ and $\mathcal{Q}'$ are SPTs with the same root, $\mathcal{Q}'$ is a subtree of $\mathcal{Q}$, and $v$ is a node of $\mathcal{Q}'$, then $\textmd{Pr}_{\mathcal{Q}'}[v] = \textmd{Pr}_{\mathcal{Q}}[v]$. 
\end{observation}

Observation~\ref{obs:prob-subtree-same-root} follows from the fact that the unique simple path in  $\mathcal{Q}'$ from its root to $v$ is identical to the corresponding path in $\mathcal{Q}$.

The next lemma ``collapses'' the portion of $\mathcal{P}$ below the last POC. 

\begin{lemma} 
\label{lem:probability-y-r}
$\mathcal{P} \upharpoonright \vec{y}_r$ is well-defined and $\textmd{Pr}\left[ \mathcal{P} \right] = \textmd{Pr}\left[ \mathcal{P} \upharpoonright \vec{y}_r \right]$.
\end{lemma}
\begin{proof}
Let $\pi$ be an arbitrary maximal path in $\mathcal{P}$.
By Definition~\ref{def:w-pruned}, $\pi$ corresponds to a $w$-correct, $\mathcal{T}$-producing assembly sequence $\vec{\alpha}$ that results in $\alpha$, the unique $w$-correct terminal assembly of $\mathcal{T}$ such that $Y \subseteq \dom{\alpha}$.
This means there is a node in $\dom{\pi}$ that terminates at $\vec{y}_r$.
Thus, $\mathcal{P} \upharpoonright \vec{y}_r$ is a well-defined subtree of $\mathcal{P}$.
If we define $B$ as the set of nodes of $\mathcal{P}$ that terminate at $\vec{y}_r$, then $B$ is a bottleneck of $\mathcal{P}$.
Moreover, Lemma~\ref{lem:restricted-tree-finite} implies that $\mathcal{P} \upharpoonright \vec{y}_r$ is finite and the set of its leaf nodes is $B$.
Then, we have:
\[
\begin{array}{llll}
\textmd{Pr}[\mathcal{P}] & = & \displaystyle \sum_{\vec{\beta}_{\vec{y}_r} \in B}{\left( \textmd{Pr}_{\mathcal{P}}\left[\vec{\beta}_{\vec{y}_r}\right] \cdot \textmd{Pr}\left[ \mathcal{P}^{\vec{\beta}_{\vec{y}_r}} \right] \right)} & \textmd{ Lemma~\ref{lem:alternative-characterization-of-pr} with } \mathcal{Q} = \mathcal{P} \\

	& = & \displaystyle \sum_{\vec{\beta}_{\vec{y}_r} \in B}{\left( \textmd{Pr}_{\mathcal{P}}\left[\vec{\beta}_{\vec{y}_r}\right] \cdot 1 \right)} & \textmd{ Lemma~\ref{lem:terminate-at-y-r-pr-1} } \\
	
	& = & \displaystyle \sum_{\vec{\beta}_{\vec{y}_r} \textmd{ leaf node of } \mathcal{P} \upharpoonright \vec{y}_r }{ \textmd{Pr}_{\mathcal{P}}\left[\vec{\beta}_{\vec{y}_r}\right] } & \textmd{ Definition of $B$ } \\
	
	& = & \displaystyle \sum_{\vec{\beta}_{\vec{y}_r} \textmd{ leaf node of } \mathcal{P} \upharpoonright \vec{y}_r }{ \textmd{Pr}_{\mathcal{P} \upharpoonright \vec{y}_r }\left[\vec{\beta}_{\vec{y}_r}\right] } & \textmd{ Observation~\ref{obs:prob-subtree-same-root} with } \mathcal{Q} = \mathcal{P} \textmd{ and } \mathcal{Q}' = \mathcal{P} \upharpoonright \vec{y}_r \\
	
	& = & \displaystyle \sum_{\pi \textmd{ maximal path in } \mathcal{P} \upharpoonright \vec{y}_r }{ \textmd{Pr}_{\mathcal{P} \upharpoonright \vec{y}_r }\left[\pi\right] } & \textmd{ Observation~\ref{obs:leaf-node-maximal-path} with } Q = \mathcal{P} \upharpoonright \vec{y}_r \\
	
	& = & \textmd{Pr}\left[ \mathcal{P} \upharpoonright \vec{y}_r \right] & \textmd{ Definition of } \textmd{Pr}\left[ \mathcal{P} \upharpoonright \vec{y}_r \right].
\end{array}
\]
\end{proof}

The following lemma says that any subtree of $\mathcal{P}$ whose root terminates at a POC and whose leaf nodes terminate at the starting point of the next POC can be ``collapsed'' when computing $\mathcal{P}$. 

\begin{lemma} 
\label{lem:terminate-at-y-i-x-i-plus-1-pr-1}
If $r > 1$, $i \in \mathbb{Z}^+ \cap [1, r-1]$, and $\vec{\beta}_{\vec{y}_i}$ is a node of $\mathcal{P}$ that terminates at $\vec{y}_i$, then $\mathcal{P}^{\vec{\beta}_{\vec{y}_i}} \upharpoonright \vec{x}_{i+1}$ is well-defined and $\textmd{Pr}\left[ \mathcal{P}^{\vec{\beta}_{\vec{y}_i}} \upharpoonright \vec{x}_{i+1} \right] = 1$.
\end{lemma}

\begin{proof}
It suffices to show that no node in $\mathcal{P}^{\vec{\beta}_{\vec{y}_i}} \upharpoonright \vec{x}_{i+1}$ (except for its root) terminates at an essential POC in $\mathcal{T}$.
Let $\pi$ be any path obtained by concatenating the unique path in $\mathcal{P}$ from its root to  $\vec{\beta}_{\vec{y}_i}$ with an arbitrary maximal path in $\mathcal{P}^{\vec{\beta}_{\vec{y}_i}} \upharpoonright \vec{x}_{i+1}$.
By Definition~\ref{def:w-pruned}, $\pi$ corresponds to a $w$-correct, $\mathcal{T}$-producing assembly sequence  $\vec{\alpha}$ that results in $\alpha$, the unique $w$-correct terminal assembly of $\mathcal{T}$ such that $Y \subseteq \dom{\alpha}$. This means $\vec{\beta}_{\vec{y}_i} \in \dom{\pi}$ and $\mathcal{P}^{\vec{\beta}_{\vec{y}_i}}$ is not empty.
Since $Y \subseteq \dom{\alpha}$, $r > 1$ and $i$ is an integer such that $1 \leq i < r$, Lemmas~\ref{lem:x-i-y-i-ordering-r-1}~and~\ref{lem:x-i-y-i-ordering} together imply that $\vec{x}_{i+1} \in \dom{\alpha}$ and
\begin{equation}
\label{eqn:y-i-x-i-plus-1-y-i-plus-1}
\textmd{index}_{\vec{\alpha}}\left(\vec{y}_i\right) < \textmd{index}_{\vec{\alpha}}\left(\vec{x}_{i+1}\right) <
\textmd{index}_{\vec{\alpha}}\left(\vec{y}_{i+1}\right).
\end{equation}
This in turn implies that there exists $\vec{\beta}_{\vec{x}_{i+1}} \in \dom{\pi}$ that terminates at $\vec{x}_{i+1}$ and is a descendant of $\vec{\beta}_{\vec{y}_i}$ in $\mathcal{P}$.
Thus, $\mathcal{P}^{\vec{\beta}_{\vec{y}_i}} \upharpoonright \vec{x}_{i+1}$ is well-defined.
For any $\vec{y} \in Y \backslash \left\{ \vec{y}_i \right\}$,  let $\vec{\beta}_{\vec{y}} \in \dom{\pi}$ be such that $\vec{\beta}_{\vec{y}}$ terminates at $\vec{y}$.
Then, by condition~\ref{def:snd-3} of Definition~\ref{def:seq-non-deterministic} and (\ref{eqn:y-i-x-i-plus-1-y-i-plus-1}), we have either 
$$
\textmd{index}_{\vec{\alpha}}\left(\vec{y}\,\right) < \textmd{index}_{\vec{\alpha}}\left(\vec{y}_i\right)
$$ 
or 
$$
\textmd{index}_{\vec{\alpha}}\left(\vec{y}\,\right) > \textmd{index}_{\vec{\alpha}}\left(\vec{x}_{i+1}\right).
$$
This means that $\vec{\beta}_{\vec{y}}$ is either an ancestor of $\vec{\beta}_{\vec{y}_i}$ in $\mathcal{P}$ or a descendant of $\vec{\beta}_{\vec{x}_{i+1}}$ in $\mathcal{P}$.
But every descendant $\vec{\beta}$ (terminating at some $\vec{p} \in \Z^2$) of $\vec{\beta}_{\vec{y}_i}$ in $\mathcal{P}^{\vec{\beta}_{\vec{y}_i}} \upharpoonright \vec{x}_{i+1}$ satisfies
$$
\textmd{index}_{\vec{\alpha}}\left( \vec{y}_i \right) < 
\textmd{index}_{\vec{\alpha}}\left( \vec{p} \right) \leq
\textmd{index}_{\vec{\alpha}}\left( \vec{x}_{i+1} \right). 
$$
Thus, $\vec{\beta}$ cannot terminate at an essential POC in $\mathcal{T}$.
In other words, no node of $\mathcal{P}^{\vec{\beta}_{\vec{y}_i}} \upharpoonright \vec{x}_{i+1}$ (except for its root) terminates at an essential POC in $\mathcal{T}$.
Since $\mathcal{P}^{\vec{\beta}_{\vec{y}_i}} \upharpoonright \vec{x}_{i+1}$ is full relative to $\mathcal{P}$, Corollary~\ref{cor:deterministic-spt-probability-1-generalized} implies that $\textmd{Pr}\left[ \mathcal{P}^{\vec{\beta}_{\vec{y}_i}} \upharpoonright \vec{x}_{i+1} \right] =~1$.
\end{proof}

The following lemma says that the operation of defining a subtree of a given tree rooted at one of its nodes distributes over the operation of restricting the tree to a bottleneck set of its nodes.

\begin{lemma} 
\label{lem:p-restricted-u}
If $\mathcal{Q}$ is a rooted tree, $B$ is a bottleneck of $\mathcal{Q}$, and $v$ is an ancestor in $\mathcal{Q}$ of at least one node in $B$, then $\left( \mathcal{Q} \upharpoonright B \right)^{v} = \mathcal{Q}^{v} \upharpoonright B$.
\end{lemma}

\begin{proof}
  We will characterize each side of the equality in turn, starting with
  its left-hand side.
  
  Since $B$ is a bottleneck of $\mathcal{Q}$, no ancestor of a node in $B$ can be a descendant of any node in $B$. This means that every maximal downward path from $v$ in $\mathcal{Q}$ contains a unique element of $B$.
Thus, $\left( \mathcal{Q} \upharpoonright B \right)^{v}$ is a subtree of $\mathcal{Q} \upharpoonright B$ in which every path starts at node $v$.
Now, by Definition~\ref{def:spt-tas-restricted-set-of-nodes}, $\mathcal{Q} \upharpoonright B$ is the unique subtree of $\mathcal{Q}$ comprising all (and only the) paths of nodes in $\mathcal{Q}$ from the root of  $\mathcal{Q}$ to some $b$ in $B$.
Therefore, $\left( \mathcal{Q} \upharpoonright B \right)^v$ is the unique subtree of $\mathcal{Q}$ comprising all (and only the) paths of nodes in $\mathcal{Q}$ from $v$ to some $b$ in $B$.

Next, we consider the right-hand side of the equality. By Definition~\ref{def:spt-tas-restricted-set-of-nodes}, $\mathcal{Q}^v \upharpoonright B$ is the unique subtree of $\mathcal{Q}^v$ comprising all (and only the) paths of nodes in $\mathcal{Q}^v$ from $v$ (i.e., the root of $\mathcal{Q}^v$) to some node $b$ in $B$. Since $Q^v$ is a subtree of $\mathcal{Q}$, $\mathcal{Q}^v \upharpoonright B$ is the unique subtree of $\mathcal{Q}$ comprising all (and only the) paths of nodes in $\mathcal{Q}$ from $v$ to some $b$ in $B$

Finally, since the characterizations of the two sides of the equality are identical, $\left(\mathcal{Q} \upharpoonright B\right)^v = \mathcal{Q}^v \upharpoonright~B$. 
\end{proof}

The following lemma ``collapses'' the level of $\mathcal{P}$ between a POC and the starting point of the next POC.

\begin{lemma} 
\label{lem:probability-y-i-equals-probability-x-i-plus-1}
If $r > 1$ and $i \in \mathbb{Z}^+ \cap [1, r-1]$, then $\mathcal{P} \upharpoonright \vec{y}_{i}$ and $\mathcal{P} \upharpoonright \vec{x}_{i+1}$ are well-defined and $\textmd{Pr}\left[ \mathcal{P} \upharpoonright \vec{y}_{i}  \right] = \textmd{Pr}\left[ \mathcal{P} \upharpoonright \vec{x}_{i+1} \right]$.
\end{lemma}

\begin{proof}
Let $\pi$ be an arbitrary maximal path in $\mathcal{P}$. 
By Definition~\ref{def:w-pruned}, $\pi$ corresponds to a $w$-correct, $\mathcal{T}$-producing assembly sequence $\vec{\alpha}$ that results in $\alpha$, the unique $w$-correct, terminal assembly of $\mathcal{T}$ such that $Y \subseteq \dom{\alpha}$.
This means there is a node in $\dom{\pi}$ that terminates at $\vec{y}_i$.
Thus, $\mathcal{P} \upharpoonright \vec{y}_{i}$ is well-defined.
Since $Y \subseteq \dom{\alpha}$, $r > 1$, and $i$ is an integer such that $1 \leq i < r$, Lemma~\ref{lem:x-i-y-i-ordering} implies that $\vec{x}_{i+1} \in \dom{\alpha}$ and $\textmd{index}_{\vec{\alpha}}\left(\vec{y}_i\right) < \textmd{index}_{\vec{\alpha}}\left(\vec{x}_{i+1}\right)$.
Thus, $\mathcal{P} \upharpoonright \vec{x}_{i+1}$ is well-defined. Furthermore, there is a node in $\dom{\pi}$ that terminates at $\vec{x}_{i+1}$ and  is a descendant in $\mathcal{P}$ of the node in $\dom{\pi}$ that terminates at $\vec{y}_i$.
This means $\mathcal{P} \upharpoonright \vec{y}_{i}$ is a subtree of $\mathcal{P} \upharpoonright \vec{x}_{i+1}$.
If we define $B$ to be the set of nodes of $\mathcal{P}$ that terminate at $\vec{y}_i$, then $B$ is a bottleneck of $\mathcal{P}  \upharpoonright \vec{x}_{i+1}$.
Moreover, Lemma~\ref{lem:restricted-tree-finite} implies that $\mathcal{P} \upharpoonright \vec{y}_i$ is finite and $B$ is the set of its leaf nodes.
Then, we have:
\[
\begin{array}{llll}
\textmd{Pr}\left[\mathcal{P} \upharpoonright \vec{x}_{i+1}\right] & = & \displaystyle\sum_{\vec{\beta}_{\vec{y}_i} \in B}{\left( \textmd{Pr}_{\mathcal{P} \upharpoonright \vec{x}_{i+1}}\left[\vec{\beta}_{\vec{y}_i}\right] \cdot \textmd{Pr}\left[ \left( \mathcal{P} \upharpoonright \vec{x}_{i+1} \right)^{\vec{\beta}_{\vec{y}_i} \;} \right]\right)} & \textmd{ Lemma~\ref{lem:alternative-characterization-of-pr} with } \mathcal{Q} = \mathcal{P} \upharpoonright \vec{x}_{i+1} \\
	 & = & \displaystyle\sum_{\vec{\beta}_{\vec{y}_i} \in B}{\left( \textmd{Pr}_{\mathcal{P} \upharpoonright \vec{x}_{i+1}}\left[\vec{\beta}_{\vec{y}_i}\right] \cdot \textmd{Pr}\left[ \mathcal{P}^{\vec{\beta}_{\vec{y}_i}} \upharpoonright \vec{x}_{i+1} \right]\right)} & \textmd{ Lemma~\ref{lem:p-restricted-u} with } \mathcal{Q} = \mathcal{P} \upharpoonright \vec{x}_{i+1} \textmd{ and } v = \vec{\beta}_{\vec{y}_i}  \\
	 & = & \displaystyle\sum_{\vec{\beta}_{\vec{y}_i} \in B}{\left( \textmd{Pr}_{\mathcal{P} \upharpoonright \vec{x}_{i+1}}\left[\vec{\beta}_{\vec{y}_i}\right] \cdot 1 \right)} & \textmd{ Lemma~\ref{lem:terminate-at-y-i-x-i-plus-1-pr-1} }\\
	 & = & \displaystyle\sum_{\vec{\beta}_{\vec{y}_i} \textmd{ leaf node of } \mathcal{P}\upharpoonright \vec{y}_i}{\textmd{Pr}_{\mathcal{P} \upharpoonright \vec{x}_{i+1}}\left[\vec{\beta}_{\vec{y}_i}\right] } & \textmd{ Definition of } B \\
	 & = & \displaystyle\sum_{\vec{\beta}_{\vec{y}_i} \textmd{ leaf node of } \mathcal{P}\upharpoonright \vec{y}_i}{ \textmd{Pr}_{\mathcal{P}\upharpoonright \vec{y}_i }\left[\vec{\beta}_{\vec{y}_i}\right]  } & \textmd{ Obs.~\ref{obs:prob-subtree-same-root} with } \mathcal{Q} = \mathcal{P} \upharpoonright \vec{x}_{i+1} \textmd{ and } \mathcal{Q}' = \mathcal{P} \upharpoonright \vec{y}_i \\
	 & = & \displaystyle\sum_{\pi \textmd{ maximal path in } \mathcal{P} \upharpoonright \vec{y}_i}{ \textmd{Pr}_{\mathcal{P}\upharpoonright \vec{y}_i}\left[\pi\right]  } & \textmd{ Observation~\ref{obs:leaf-node-maximal-path} with } \mathcal{Q} = \mathcal{P} \upharpoonright \vec{y}_i \\
	& = & \textmd{Pr}\left[ \mathcal{P} \upharpoonright \vec{y}_i \right] & \textmd{ Definition of } \textmd{Pr}\left[ \mathcal{P} \upharpoonright \vec{y}_i \right].
\end{array}
\]
\end{proof}

We now turn our attention to analyzing the quantity $\textmd{Pr}\left[ \mathcal{P}^{\vec{\beta}_{\vec{x}}} \upharpoonright \vec{y} \right]$, which is the probability of a subtree of $\mathcal{P}$ with root $\vec{\beta}_{\vec{x}}$ that corresponds to $\mathcal{C}$.
Our goal is to lower bound it by $\textmd{Pr}\left[ \mathcal{C}\right]$. 
\begin{lemma} 
  \label{lem:b-x-y-well-defined-leaf-nodes}
$\mathcal{P}^{\vec{\beta}_{\vec{x}}} \upharpoonright \vec{y}$ has the following properties:
\begin{enumerate}
	\item \label{lem:b-x-y-well-defined-leaf-nodes-1} $\mathcal{P}^{\vec{\beta}_{\vec{x}}} \upharpoonright \vec{y}$ is well-defined.
	\item \label{lem:b-x-y-well-defined-leaf-nodes-2} All of the leaf nodes of $\mathcal{P}^{\vec{\beta}_{\vec{x}}} \upharpoonright \vec{y}$ are nodes of $\mathcal{P}$ that terminate at $\vec{y}$.
\end{enumerate} 
\end{lemma}

\begin{proof}
Let $\pi$ be an arbitrary maximal path in $\mathcal{P}$.
By Definition~\ref{def:w-pruned}, $\pi$ corresponds to a $w$-correct $\mathcal{T}$-producing assembly sequence $\vec{\alpha}$ that results in $\alpha$, the unique $w$-correct terminal assembly of $\mathcal{T}$ such that $Y \subseteq \dom{\alpha}$.
Therefore, there is a node in $\dom{\pi}$ that terminates at $\vec{y}$.
Since $Y \subseteq \dom{\alpha}$, Lemma~\ref{lem:x-i-y-i-ordering-r-1} says that $\vec{x} \in \dom{\alpha}$ and $\textmd{index}_{\vec{\alpha}}\left(\vec{x}\right) < \textmd{index}_{\vec{\alpha}}\left(\vec{y}\right)$.
Hence, there is a node in $\dom{\pi}$ that terminates at $\vec{x}$, which implies that $\mathcal{P}^{\vec{\beta}_{\vec{x}}}$ is not empty.
Moreover, the previous inequality implies that the node in $\dom{\pi}$ that terminates at $\vec{x}$ is an ancestor of the node in $\dom{\pi}$ that terminates at $\vec{y}$. 
This means $\mathcal{P}^{\vec{\beta}_{\vec{x}}} \upharpoonright \vec{y}$ is well-defined.
Then, Lemma~\ref{lem:restricted-tree-finite} implies that $\mathcal{P}^{\vec{\beta}_{\vec{x}}} \upharpoonright \vec{y}$ is finite and that each one of its leaf nodes is a node in $\mathcal{P}$ that terminates at $\vec{y}$.
\end{proof}

\begin{definition}  
\label{def:competing-tile}
Let $\vec{p} \in \mathbb{Z}^2$ and $\beta \in \mathcal{A}[\mathcal{T}]$. 
We say that the tile $\left( \vec{p}, t \right) \in \mathbb{Z}^2 \times T$ is a \emph{competing tile in} $\beta$ at $\vec{p}$ if either
$\beta = \sigma = \left\{ \left(\vec{p}=\vec{s}=\vec{x}_1,t\right) \right\}$
or
there exists $\vec{y} \in Y$ with corresponding $\vec{x}$, $\piw$, and $\pil$ such that the following four conditions hold:
\begin{enumerate}
	\item \label{def:ct-1} There exists a competing path $\pi \in \left\{ \piw, \pil \right\}$ and an integer $1 \leq l \leq \left| \pi \right|$ such that $\vec{p} = \pi[l]$.
	\item \label{def:ct-2} $ \vec{y}  \not \in \dom{\beta}$
	\item \label{def:ct-3} $\vec{p} \in \partial_t{\beta}$
        \item \label{def:ct-4} Every simple path in $G^{\textmd{b}}_{\beta + \left(\vec{p},t\right)}$ from $\vec{s}$ to $\vec{p}$ goes through $\pi[1] = \vec{x}$.  
\end{enumerate}
\end{definition}

Lemma~\ref{lem:w-correct-no-inessential-poc} implies that $\alpha$ does not place tiles at the inessential POCs, which means that tiles placed on any competing path associated with an inessential POC are not actually competing.
This motivates our decision for Definition~\ref{def:competing-tile} to explicitly require that a competing tile be one that is placed at a point along a competing path associated with an essential POC.
Figure~\ref{fig:def-example-tas-competing-tile} depicts some examples and counter-examples of competing tiles in our running example.
As an additional example, consider the TAS resulting from removing the W$^0_a$, W$^0_b$, and W$^0_c$ tiles from the TAS depicted in Figure~\ref{fig:def1-tas} on page~\pageref{fig:def1-tas}. In the resulting TAS, the only terminal assembly is $\alpha_{1,1}$ shown in Figure~\ref{fig:def1-tas}(c), since the bit value 1 is always written and the losing path $\pi'$ is always blocked. Now, consider the subassembly of $\alpha_{1,1}$ in which only the S and  W$^1_a$ tiles are placed. The tile shown as  W$^1_b$ in Figure~\ref{fig:def1-tas}(c) is not a competing tile in this subassembly, even though it satisfies the first three conditions in Definition~\ref{def:competing-tile}. This tile breaks condition~\ref{def:ct-4} because the assembly sequence containing the  S, W$^1_a$, and  W$^1_b$ tiles follows a path that ends at $\pi'[3]$ without first going through $\pi[1]=\pi'[1]$, namely the point where the R tile will later be placed.

 \begin{figure}[!h]
  \begin{minipage}{\linewidth}
    \centering
        \includegraphics[width=3in]{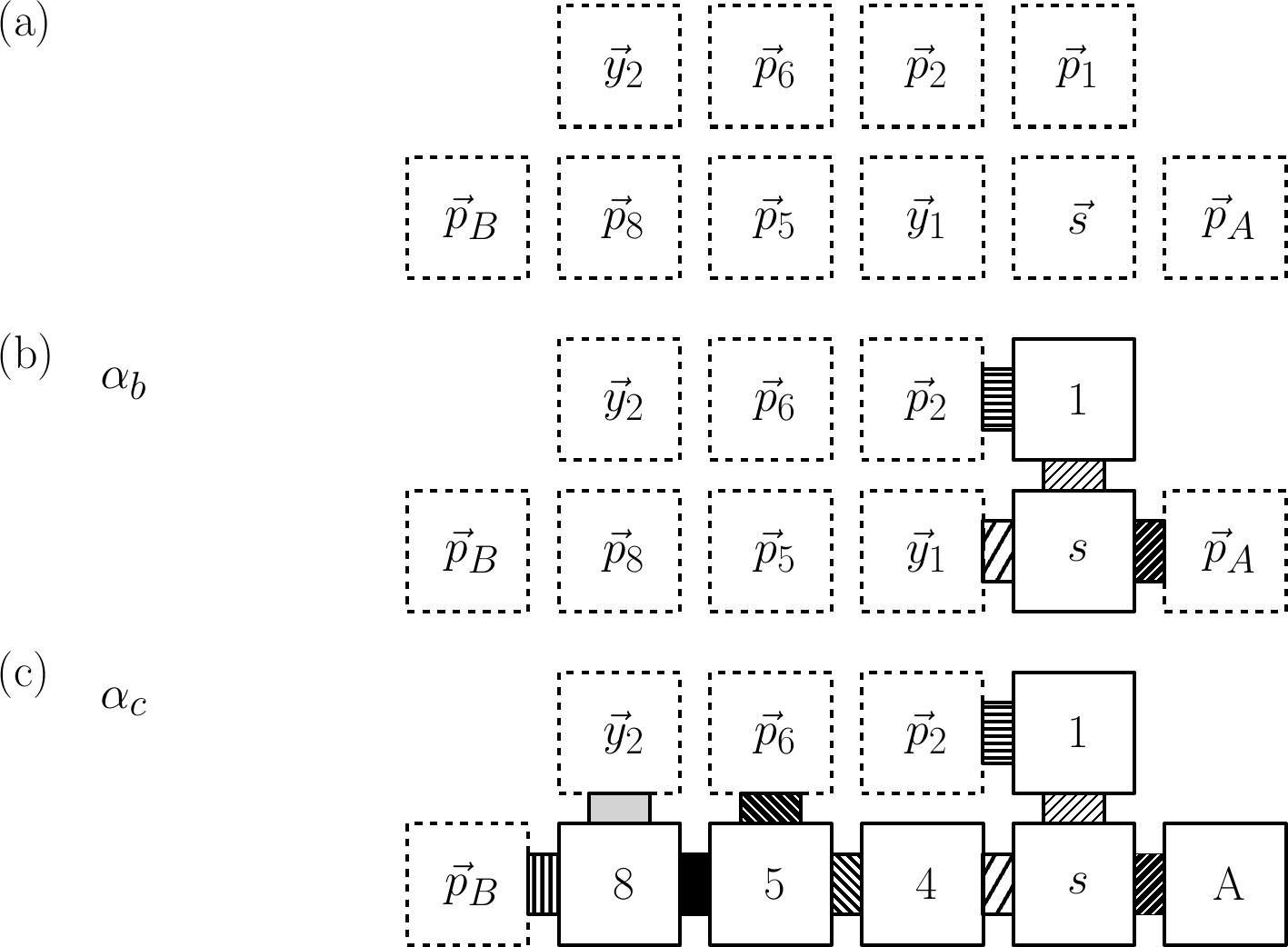}
      \end{minipage}      
      \caption{\label{fig:def-example-tas-competing-tile}Competing
        tiles in our running example.\\
        (a) Full labeling of the points
        that are in the domain of every $\mathcal{T}$-terminal
        assembly.\\
        (b) In this assembly $\alpha_b \not \in
        \mathcal{A}_{\Box}[\mathcal{T}]$, the two tiles
        $(\vec{p}_2,2)$ and $(\vec{y}_1,4)$ are competing tiles
        because each one of them satisfies all four conditions in
        Definition~\ref{def:competing-tile}; note that $(\vec{p}_5,5)$
        is not a competing tile because it does not satisfy
        condition~\ref{def:ct-3} and $(\vec{p}_A,A)$ is not a
        competing tile because it does not satisfy
        condition~\ref{def:ct-1}.\\
        (c) In this assembly $\alpha_c \not
        \in \mathcal{A}_{\Box}[\mathcal{T}]$, the two tiles
        $(\vec{p}_6,6)$ and $(\vec{y}_2,9)$ are competing tiles
        because each one of them satisfies all four conditions in
        Definition~\ref{def:competing-tile}; note that $(\vec{p}_B,B)$
        is not a competing tile because it does not satisfy
        condition~\ref{def:ct-1} and $(\vec{p}_2,2)$ is not a
        competing tile because it does not satisfy
        condition~\ref{def:ct-2}.\\ }
\end{figure}
\begin{definition}  
\label{def:competitive-node}
Let $k \in \mathbb{Z}^+$, $\vec{p} \in \mathbb{Z}^2$, $t \in T$, and $\vec{\beta} = \left( \beta_i \mid 1 \leq i \leq k \right)$ be a node of $\mathcal{M}_{\mathcal{T}}$. We say that $\vec{\beta}$ is \emph{competing in} $\mathcal{M}_{\mathcal{T}}$ if either:
\begin{enumerate}
\item $k = 1$ and the only element of $\beta_1 = \sigma$, namely $\left(\vec{s},s\right)$, is a tile that is competing in $\sigma$, or
\item $\beta_k = \beta_{k-1} + \left( \vec{p}, t \right)$ and $\left( \vec{p}, t \right)$ is a tile competing in $\beta_{k-1}$.
\end{enumerate}
\end{definition} 
In Figure~\ref{fig:def-example-tas-P_t_w'}, the competing nodes are circled.
\begin{figure}[!h]
          \centerline{\includegraphics[width=2.5in]{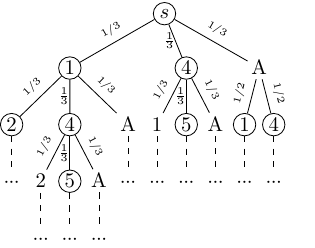}}
          
          \begin{minipage}{\linewidth}
            \caption{\label{fig:def-example-tas-P_t_w'} Partial
              depiction of $\mathcal{P}_{w'}$, where $\mathcal{T}$ is
              our running example TAS and
              $w'~=~\{(\vec{y}_1,4),(\vec{y}_2,9)\}$. The circled
              nodes are the only competing nodes in this portion of
              the tree. The nodes labeled with an A are not competing
              nodes because the A tile type does not attach at any
              point along a competing path. The two non-circled nodes
              labeled 1 and 2, respectively, are not competing nodes
              because, in the assembly that results from the sequences
              to which their respective parent nodes correspond, the
              tile type 4 has already attached at the point of
              competition $\vec{y}_1$, which contradicts the second
              condition of Definition~\ref{def:competing-tile}.  }
          \end{minipage}
        \end{figure}

The next lemma says that if a competing path is going to be blocked, then it must be blocked prior to placing a tile at the starting point for the corresponding POC.

\begin{lemma} 
\label{lem:longest-simple-finite-path-beta-m}
Let:\vspace*{-2mm}
\begin{itemize}
\item $\vec{y} \in Y$ with associated starting point $\vec{x}$ and competing path $\pi \in \{\piw, \pil\}$,
\item $\vec{\beta}_{\vec{x}}$ be a $\mathcal{T}$-producible assembly sequence that terminates at $\vec{x}$ and whose result  $\beta_{\vec{x}}$ is a subassembly of $\alpha$,
\item $\vec{\beta}$ be equal to, or any extension of,  $\vec{\beta}_{\vec{x}}$ with result $\beta$ such that $\vec{y} \notin \dom{\beta}$, 
\item $l$ be the unique integer in $[1..|\pi|-1]$ such that $\pi[1\ldots l]$ is a simple path in $G^{\textmd{b}}_{\beta}$ but $\pi[1\ldots l+1]$ is not, and
\item $\vec{p} \in  \mathbb{Z}^2$.
\end{itemize}
If $\vec{p} \in \dom{\pi\left[l+1\ldots \left|\pi\right|\right]}$ and $\vec{p} \not \in \dom{\beta}_{\vec{x}}$, then $\vec{p} \not \in \dom{\beta}$.
\end{lemma}

\begin{proof}
Suppose, for the sake of obtaining a contradiction, that $\vec{p} \in \dom{\beta}$.
Since $\vec{y} \not \in \dom{\beta}$, it follows that $\vec{p} \ne \vec{y}$.
Now, we consider whether or not there exists a simple path $p'$ from $\vec{s}$ to $\vec{p}$ in $G^{\textmd{b}}_{\beta}$ that does not go through $\vec{x}$. We handle both cases in turn.
\begin{enumerate}
\item Assume that there exists a simple path $p'$ from $\vec{s}$ to $\vec{p}$ in $G^{\textmd{b}}_{\beta}$ that does not go through $\vec{x}$.
Since, by assumption, $\vec{p} \not \in \dom{\beta_{\vec{x}}}$, there is a simple path $p$ from $\vec{s}$ to $\vec{x}$ in $G^{\textmd{b}}_{\beta_{\vec{x}}}$ that does not go through $\vec{p}$. 
Since $\vec{\beta}$ is equal to, or an extension of, $\vec{\beta}_{\vec{x}}$, $p$ is also a simple path from $\vec{s}$ to $\vec{x}$ in $G^{\textmd{b}}_{\beta}$ that does not go through $\vec{p}$.
Condition~\ref{def:cp-4b} of Definition~\ref{def:competing-paths} says that $p'$ is a prefix of every simple path from $\vec{s}$ to $\vec{x}$ in $G^{\textmd{b}}_{\beta}$.
However, $p'$ cannot be a prefix of $p$ because $\vec{p} \not \in \dom{p}$ but $\vec{p} \in \dom{p'}$. 
	\item Assume that every simple path $p'$ from $\vec{s}$ to $\vec{p}$ in $G^{\textmd{b}}_{\beta}$ goes through $\vec{x}$.
Since $\vec{p} \ne \vec{y}$, let $n$ be an integer such that $1 \leq n < \left|\pi\right|$ and $\vec{p} = \pi[n]$. 
Then, condition~\ref{def:cp-4a} of Definition~\ref{def:competing-paths} implies that $\pi\left[1\ldots n\right]$ is a suffix of $p'$, which means $\pi\left[1\ldots n\right]$ is a simple path in $G^{\textmd{b}}_{\beta}$.
Since $\vec{p} \in \dom{\pi\left[l+1\ldots \left| \pi \right|\right]}$, we have $n > l$.
This contradicts the fact that $\pi[1\ldots l+1]$ is not a simple path in $G^{\textmd{b}}_{\beta}$.
\end{enumerate}
Since we reached a contradiction in both cases, it follows that $\vec{p} \not \in \dom{\beta}$. 
\end{proof}

The following lemma says that placing a tile at a point along a competing path cannot create an alternate path to this point that does not go through the starting point of the competing paths. 

\begin{lemma} 
\label{lem:competing-tile-4}
Let:\vspace*{-2mm}
\begin{itemize}
\item $\vec{y} \in Y$ with associated starting point $\vec{x}$ and competing path $\pi \in \{\piw, \pil\}$,
\item $\vec{\beta}_{\vec{x}}$ be a $\mathcal{T}$-producible assembly sequence that terminates at $\vec{x}$ and whose result  $\beta_{\vec{x}}$ is a subassembly of $\alpha$,
\item $\vec{\beta}$ be equal to, or any extension of, $\vec{\beta}_{\vec{x}}$ with result $\beta$ such that $\vec{y} \not \in \dom{\beta}$,
\item $\vec{p} \in \dom{\pi}$, and
\item
  $\beta + \left( \vec{p}, t \right)$ is $\mathcal{T}$-producible.
\end{itemize}
Then every simple path in $G^{\textmd{b}}_{\beta + \left( \vec{p}, t\right)}$ from $\vec{s}$ to $\vec{p}$ goes through $\vec{x}$.
\end{lemma}

\begin{proof}
Assume, for the sake of obtaining a contradiction, that there exists a simple path $p'$ in $G^{\textmd{b}}_{\beta + \left( \vec{p}, t \right)}$ from $\vec{s}$ to $\vec{p}$ that does not go through $\vec{x}$.	
Then, either $\vec{p} = \vec{y}$ or $\vec{p} \ne \vec{y}$. We handle these two cases in turn.
\begin{enumerate}
	\item Case $\vec{p} = \vec{y}$: Condition~\ref{def:cp-5} of Definition~\ref{def:competing-paths} says that every simple path from $\vec{s}$ to $\vec{y}$ in $G^{\textmd{b}}_{\beta + \left( \vec{y}, t \right)}$ has either $\piw$ or $\pil$ as a suffix. 
But $p'$ is a simple path from $\vec{s}$ to $\vec{y}$ in $G^{\textmd{b}}_{\beta + \left( \vec{y}, t \right)}$ that does not go through $\vec{x}$.
Therefore $p'$ cannot have either $\piw$ or $\pil$ as a suffix, since each one of these two competing paths starts from $\vec{x}$ and thus contains that point in its domain.
	\item Case $\vec{p} \ne \vec{y}$: Since $p'$ is a simple path from $\vec{s}$ to $\vec{p}$ in $G^{\textmd{b}}_{\beta + \left( \vec{p}, t\right)}$ that does not go through $\vec{x}$, Condition~\ref{def:cp-4b} of Definition~\ref{def:competing-paths} says $p'$ is a prefix of every simple path from $\vec{s}$ to $\vec{x}$ in $G^{\textmd{b}}_{\beta + \left( \vec{p}, t \right)}$.
By assumption, $\beta + \left( \vec{p}, t \right)$ is $\mathcal{T}$-producible, which implies $\vec{p} \not \in \dom{\beta}$.
Since $\vec{\beta}$ is equal to, or an extension of, $\vec{\beta}_{\vec{x}}$, it follows that there exists a simple path $p$ from $\vec{s}$ to $\vec{x}$ in $G^{\textmd{b}}_{\beta}$ that does not go through $\vec{p}$.
Therefore, $p'$ cannot be a prefix of $p$, since $\vec{p} \in \dom{p'}$ but $\vec{p} \not \in \dom{p}$. 
\end{enumerate}
Since we reach a contradiction in every case, it follows that every simple path in $G^{\textmd{b}}_{\beta + \left( \vec{p}, t \right)}$ from $\vec{s}$ to $\vec{p}$ goes through $\vec{x}$.
\end{proof}

\begin{lemma} 
\label{lem:cannot-terminate-at-SY}
No internal node of $\mathcal{P}^{\vec{\beta}_{\vec{x}}}
\upharpoonright \vec{y}$ has a child node in $\mathcal{M}_{\mathcal{T}}$
that terminates at a point in $S_Y$.
\end{lemma}

\begin{proof}
Let $\vec{\beta}$ be any internal node of
$\mathcal{P}^{\vec{\beta}_{\vec{x}}} \upharpoonright \vec{y}$. Thus
any child node of $\vec{\beta}$ in $\mathcal{M}_{\mathcal{T}}$ must be
different from $\vec{\beta}_{\vec{x}}$ (the root node of
$\mathcal{P}^{\vec{\beta}_{\vec{x}}} \upharpoonright \vec{y}$) and
thus cannot terminate at $\vec{x}$. Let $\vec{x}\,'$ be any point in
$S_Y\backslash\{\vec{x}\}$ with corresponding essential POC
$\vec{y}\,'$. We now show that no child node of $\vec{\beta}$ may
terminate at $\vec{x}\,'$. Assume, for the sake of obtaining a
contradiction, that $\vec{\beta}'$ is a child node of $\vec{\beta}$ in
$\mathcal{M}_{\mathcal{T}}$ that terminates at $\vec{x}\,'$. By
Observation~\ref{obs:alpha-w-correct-iff-node-in-p-corr-alpha},
$\vec{\beta}$ is a finite, $w$-correct, $\mathcal{T}$-producing
assembly sequence. Since $\vec{\beta}'$ is a child node of
$\vec{\beta}$ in $\mathcal{M}_{\mathcal{T}}$ that terminates at
$\vec{x}\,'\in S_Y \subseteq S_P$ and $\vec{x}\,'\notin Y \subseteq P$
(by condition~\ref{def:snd-4} of
Definition~\ref{def:seq-non-deterministic}), it follows that
$\vec{\beta}'$ is also $w$-correct and thus a node in
$\mathcal{P}^{\vec{\beta}_{\vec{x}}} \upharpoonright \vec{y}$. Note
that 1) $\mathcal{P}^{\vec{\beta}_{\vec{x}}} \upharpoonright \vec{y}$ is
full relative to $\mathcal{P}$ and 2) the only way for $\vec{\beta}'$ to
be $w$-incorrect would be for its last tile attachment step to place a
$w$-incorrect tile at some essential POC. Since $\vec{\beta}'$ is a
node in $\mathcal{P}^{\vec{\beta}_{\vec{x}}} \upharpoonright \vec{y}$
that terminates at $\vec{x}\,'$ and, by
Lemma~\ref{lem:b-x-y-well-defined-leaf-nodes}, all of the leaf nodes
of $\mathcal{P}^{\vec{\beta}_{\vec{x}}} \upharpoonright \vec{y}$ are
nodes of $\mathcal{P}$ that terminate at $\vec{y}$, it follows that
$\vec{y} \notin \dom{\res{\vec{\beta}'}}$.

We now show that $\res{\vec{\beta}'} \not \in
\mathcal{A}_{\Box}\left[\mathcal{T}\right]$. Suppose, for the sake of
obtaining a contradiction, that $\res{\vec{\beta}'} \in
\mathcal{A}_{\Box}\left[\mathcal{T}\right]$. Then, by
Corollary~\ref{cor:unique-w-correct-assembly},
$\res{\vec{\beta}'}=\alpha$. Since $Y \subseteq \dom{\alpha}$, we must
have $\vec{y}\,' \in \dom \res{\vec{\beta}'}$. Therefore, by
Lemma~\ref{lem:x-i-y-i-ordering-r-1}, $\textmd{index}_{\vec{\beta}'
}\left( \vec{x}\,' \right) < \textmd{index}_{\vec{\beta}'}\left(
\vec{y}\,' \right)$, which contradicts the fact that $\vec{\beta}'$
terminates at $\vec{x}\,'$.

Since $\res{\vec{\beta}'} \not \in
\mathcal{A}_{\Box}\left[\mathcal{T}\right]$,
Lemma~\ref{lem:w-correct-extended-to-w-correct-terminal} implies that
there exists an extension $\vec{\varepsilon}^{\; \prime}$ of
$\vec{\beta}'$ resulting in $\alpha$, with $Y \subseteq \dom{\alpha}$.
Since $\vec{\beta}_{\vec{x}}$ terminates at $\vec{x}$, $\vec{\beta}'$
terminates at $\vec{x}\,'$ and is a descendant of
$\vec{\beta}_{\vec{x}}$ (via its parent node $\vec{\beta}$), it
follows that $\textmd{index}_{\vec{\varepsilon}^{\, \prime}}\left(
\vec{x}\, \right) < \textmd{index}_{\vec{\varepsilon}^{\,
    \prime}}\left( \vec{x}\,' \right)$. Furthermore, $\vec{y} \notin \dom{\res{\vec{\beta}'}}$ implies that
$\textmd{index}_{\vec{\varepsilon}^{\, \prime}}\left( \vec{x}\,'
\right) < \textmd{index}_{\vec{\varepsilon}^{\, \prime}}\left(
\vec{y}\, \right)$. Chaining the last two inequalities yields
$\textmd{index}_{\vec{\varepsilon}^{\, \prime}}\left( \vec{x}\,
\right) < \textmd{index}_{\vec{\varepsilon}^{\, \prime}}\left(
\vec{x}\,' \right) < \textmd{index}_{\vec{\varepsilon}^{\,
    \prime}}\left( \vec{y}\, \right)$, which contradicts the result of
applying Lemma~\ref{lem:x-i-y-i-ordering} to $\vec{\varepsilon}^{\;
  \prime}$.
\end{proof}

If the children of an internal node of $\mathcal{P}^{\vec{\beta}_{\vec{x}}} \upharpoonright \vec{y}$ are partitioned into competing and non-competing nodes, then the following lemma says that the former set contains at least one node and at most two nodes, and if the corresponding competition is not rigged, then it contains exactly two nodes. 
\begin{lemma} 
\label{lem:number-of-competing-nodes}
Any internal node of $\mathcal{P}^{\vec{\beta}_{\vec{x}}} \upharpoonright \vec{y}$ has exactly one or two children  in $\mathcal{M}_{\mathcal{T}}$ that are both $w$-correct and competing in $\mathcal{M}_{\mathcal{T}}$. Furthermore, if $\mathcal{C}$ is not rigged by $\vec{\beta}_{\vec{x}}$, then any internal node of $\mathcal{P}^{\vec{\beta}_{\vec{x}}} \upharpoonright \vec{y}$ has exactly two children in $\mathcal{M}_{\mathcal{T}}$ that are competing in $\mathcal{M}_{\mathcal{T}}$.
\end{lemma}

Lemma~\ref{lem:number-of-competing-nodes} will be useful in two ways going forward.
We will first use it to show that for an internal node $\vec{\beta}$ of $\mathcal{P}^{\vec{\beta}_{\vec{x}}} \upharpoonright \vec{y}$, every non-competing child of $\vec{\beta}$ in $\mathcal{M}_{\mathcal{T}}$ is a child of $\vec{\beta}$ in a certain kind of subtree $\mathcal{Q}$ of $\mathcal{P}^{\vec{\beta}_{\vec{x}}} \upharpoonright \vec{y}$. 
We will then use it to show that the probability of the edge into a node in $\mathcal{Q}$ that is competing in $\mathcal{M}_{\mathcal{T}}$ is equal to $\frac{1}{2}$.

\begin{proofsketch}
We proceed as follows:
\begin{enumerate}
	\item Prove that $\vec{\beta}$ always has at least one child in $\mathcal{M}_{\mathcal{T}}$ that is competing in $\mathcal{M}_{\mathcal{T}}$ and $w$-correct, and then use the fact that $\mathcal{C}$ is not rigged by $\vec{\beta}_{\vec{x}}$ to prove that $\vec{\beta}$ has at least one other child in $\mathcal{M}_{\mathcal{T}}$ that is competing in $\mathcal{M}_{\mathcal{T}}$, but need not be $w$-correct. 
	\item Prove (by contradiction) that $\vec{\beta}$ has no more than two competing children as follows:
	\begin{enumerate}[label=\theenumi\alph*.,itemsep=0mm]
  		\item Assume that there exists a third competing child of $\vec{\beta}$ whose resulting assembly places a tile at, say, $\vec{p}^{\; \prime\prime\prime}$.
		\item Prove that $\vec{p}^{\; \prime\prime\prime}$ is different from any point on either one of the competing paths associated with $\vec{y}$. 
		We prove this by considering the following two cases:
		\begin{enumerate}[itemsep=0mm,labelsep=2mm]   
    		\item[2bi.] Case  $\vec{p}^{\; \prime\prime\prime} \neq \vec{y}$
    		\item[2bii.] Case  $\vec{p}^{\; \prime\prime\prime} = \vec{y}$
    	\end{enumerate}
  		\item Prove that $\vec{p}^{\; \prime\prime\prime}$ is different from any point on either one of the competing paths associated with any other essential POC, i.e., an essential POC different from $\vec{y}$.
	\end{enumerate}
\end{enumerate}\vspace*{-7mm}
\end{proofsketch}
\begin{proof}
Let $m\in\Z^+$ and $\vec{\beta} = \left( \beta_i \mid 1 \leq i \leq m \right)$ be an internal node of $\mathcal{P}^{\vec{\beta}_{\vec{x}}} \upharpoonright \vec{y}$.
By Observation~\ref{obs:alpha-w-correct-iff-node-in-p-corr-alpha}, $\vec{\beta}$ is a finite, $w$-correct, $\mathcal{T}$-producing assembly sequence.
{\em Step 1.} 
Let $\vec{\beta}_{\vec{y}}$  be any leaf node of $\mathcal{P}^{\vec{\beta}_{\vec{x}}} \upharpoonright \vec{y}$ with  $\beta_{\vec{y}} = \res{\vec{\beta}_{\vec{y}}}$. Since the latter is a subtree of $\mathcal{P}$, $\vec{\beta}_{\vec{y}}$  is $w$-correct. Furthermore, by part~\ref{lem:b-x-y-well-defined-leaf-nodes-2} of Lemma~\ref{lem:b-x-y-well-defined-leaf-nodes}, $\vec{\beta}_{\vec{y}}$  terminates at $\vec{y}$ which, together with the fact that $\vec{\beta}$ is an internal node of $\mathcal{P}^{\vec{\beta}_{\vec{x}}} \upharpoonright \vec{y}$, implies that $\vec{y} \notin \dom{\beta_m}$.
By condition~\ref{def:cp-5} of Definition~\ref{def:competing-paths}, $\piw$ is a suffix of every simple path from $\vec{s}$ to $\vec{y}$ in $G^{\textmd{b}}_{\beta_{\vec{y}}}$. 
Note that at least one leaf node of $\mathcal{P}^{\vec{\beta}_{\vec{x}}} \upharpoonright \vec{y}$ is an extension of $\vec{\beta}$ that terminates at $\vec{y}$.
Let $\vec{\beta}_{\vec{y}}$ be such a leaf node.
Then, we have $\beta_m \sqsubseteq \beta_{\vec{y}}$.
Since $\vec{y} \not \in \dom{\beta_m}$, $\beta_m \sqsubseteq \beta_{\vec{y}}$, and $\piw$ is a suffix of every simple path from $\vec{s}$ to $\vec{y}$ in $G^{\textmd{b}}_{\beta_{\vec{y}}}$, there exists $1 \leq l' < \left| \piw \right|$ such that  $\piw\left[1\ldots l' \right]$ is a simple path in $G^{\textmd{b}}_{\beta_m}$ but $\piw\left[l\ldots l'+1\right]$ is not.
Note that $l' < \left| \piw \right|$ because $\vec{y} \not \in \dom{\beta_m}$. 
Let $\vec{p}^{\; \prime} = \piw\left[l'+1\right]$.
We now show that for some $t' \in T$, $\left( \vec{p}^{\; \prime}, t' \right)$ is a competing tile in $\beta_m$.
To that end, first note that $\vec{p}^{\; \prime} \not \in \dom{\beta_m}$ because $\vec{\beta}_{\vec{y}}$ is an extension of $\vec{\beta}$, and $\piw$ is a suffix of every simple path from $\vec{s}$ to $\vec{y}$ in $G^{\textmd{b}}_{\beta_{\vec{y}}}$.
Let $t' = \beta_{\vec{y}}\left( \vec{p}^{\; \prime} \right) = \alpha\left( \vec{p}^{\; \prime} \right)$, i.e., the type of tile that the unique $w$-correct $\alpha$ places at $\vec{p}^{\; \prime}$. 
Then, we have $\vec{p}^{\;\prime} \in \partial_{t'}^{\mathcal{T}} \beta_m$, because $\vec{p}^{\; \prime} \not \in \dom{\beta_m}$ and $\piw\left[1\ldots l'+1\right]$ is (a prefix of) a simple path in $G^{\textmd{b}}_{\beta_{\vec{y}}}$.
Moreover, Lemma~\ref{lem:competing-tile-4} says that condition~\ref{def:ct-4} of Definition~\ref{def:competing-tile} is satisfied.
Thus, all of the conditions of Definition~\ref{def:competing-tile} are satisfied and  $\left( \vec{p}^{\; \prime},  t'=\alpha\left( \vec{p}^{\; \prime} \right)  \right)$ is a competing tile in $\beta_m$.
Let $\vec{\beta}' = \vec{\beta} + \left( \vec{p}^{\; \prime}, t' \right)$
and $\beta' = \res{\vec{\beta}'}$.
Then, $\vec{\beta}'$ is a child node of $\vec{\beta}$ in $\mathcal{M}_{\mathcal{T}}$ that is competing in $\mathcal{M}_{\mathcal{T}}$.
Moreover, by our choice of $t' = \beta_{\vec{y}}\left( \vec{p}^{\; \prime} \right) = \alpha\left(\vec{p}^{\; \prime} \right)$, it follows that $\beta' \sqsubseteq \beta_{\vec{y}}$, which implies that $\beta'$ and $\vec{\beta}'$ are $w$-correct. 
We now show that, if $\mathcal{C}$ is not rigged by $\vec{\beta}_{\vec{x}}$,  $\vec{\beta}$ has another child node in $\mathcal{M}_{\mathcal{T}}$ that is competing in $\mathcal{M}_{\mathcal{T}}$. 
Let $l''$ be an integer such that $1 \leq l'' < \left| \pil \right|$ and $\pil\left[1\ldots l''\right]$ is a simple path in $G^{\textmd{b}}_{\beta_m}$, but $\pil\left[1\ldots l''+1\right]$ is not.
In other words, $l''$ is the length of the longest prefix of $\pil$ that occurs as a simple path in the binding graph of $\beta_m$.
Let $\vec{p}^{\; \prime\prime} = \pil\left[l''+1\right]$.
We now show that for some $t'' \in T$, $\left( \vec{p}^{\; \prime\prime}, t'' \right)$ is a competing tile in $\beta_m$, i.e., the four parts of Definition~\ref{def:competing-tile} are satisfied:
\begin{enumerate}
	\item Condition~\ref{def:ct-1} is satisfied because $\vec{p}^{\; \prime\prime} = \pil\left[ l'' + 1 \right]$ where $1 \leq l'' < \left| \pil \right|$.
	\item Condition~\ref{def:ct-2} is satisfied because $\vec{y} \not \in \dom{\beta_m}$.
	\item 
By the definition of $l''$, $\vec{p}^{\; \prime\prime} \ne \vec{x}$.
Since $\mathcal{C}$ is not rigged by $\vec{\beta}_{\vec{x}}$, Lemma~\ref{lem:not-rigged} implies  that $\dom{\pil}~\cap~\dom{\beta_{\vec{x}}} = \left\{ \vec{x} \right\}$. 
This means $\vec{p}^{\; \prime\prime} \not \in \dom{\beta}_{\vec{x}}$.
Then, by Lemma~\ref{lem:longest-simple-finite-path-beta-m}, we have $\vec{p}^{\; \prime\prime} \not \in \dom{\beta_m}$.
In other words, $\beta_m$ has yet to place a tile at $\vec{p}^{\; \prime\prime}$. So it suffices to show that for some $t'' \in T$, $t''$ can attach to $\beta_m$ at $\vec{p}^{\; \prime\prime}$. 
	Since either $\vec{p}^{\; \prime\prime} \ne \vec{y}$ or $\vec{p}^{\; \prime\prime} = \vec{y}$, we handle these two cases in turn.
	\begin{enumerate}
		\item Case $\vec{p}^{\; \prime\prime} \ne \vec{y}$:
	By Observation~\ref{obs:alpha-w-correct-iff-node-in-p-corr-alpha}, $\beta_m \sqsubseteq \alpha$, and thus $\beta_m\left( \pil\left[l''\right] \right) = \alpha\left( \pil\left[ l'' \right] \right)$.
	Then, if $t'' = \alpha\left( \vec{p}^{\; \prime\prime} \right)$, we have $\vec{p}^{\; \prime\prime} \in \partial^{\mathcal{T}}_{t''} \beta_m$.
	In this case, $\vec{\beta}$ has two competing children in $\mathcal{M}_{\mathcal{T}}$ that are both $w$-correct. 
      \item Case $\vec{p}^{\; \prime\prime} = \vec{y}$:
	By condition~\ref{def:cp-6a} of Definition~\ref{def:competing-paths}, there exists a $\mathcal{T}$-assembly sequence $\vec{\gamma} = \left( \gamma_i \mid 1 \leq i \leq \left| \pil \right| \right)$ such that $\gamma_1 = \left\{ \vec{x}, \alpha\left(\vec{x}\right)\right \}$ and, for $\gamma = \res{\vec{\gamma}}$, $\dom{\gamma} = \dom{\pil}$. 
	By the definition of $\vec{\gamma}$, $\gamma\left(\pil[1]\right) = \alpha\left(\pil[1]\right)$. 
	Then, assuming for all $1 \leq i < \left| \pil \right|-1$, $\gamma\left(\pil[i]\right) = \alpha\left(\pil[i]\right)$, the 
	facts
	that $\mathcal{T}$ is directionally deterministic 
	and $\pil[i]$ is not a POC imply (via the contrapositive of condition~\ref{def:dd-ad-pt-1} of Definition~\ref{def:dd-ad-pt})
	that $\gamma\left(\pil[i+1]\right) = \alpha\left(\pil[i+1]\right)$.
	Thus, for all $\vec{p} \in \dom{\gamma} \backslash \left\{ \vec{y} \right\}$, $\gamma\left( \vec{p\,} \right) = \alpha\left(\vec{p\,}\right)$.
	In particular, $\gamma\left( \pil\left[\left| \pil \right|-1\right] \right) = \alpha\left(\pil\left[\left| \pil \right|-1\right]\right)$.
	By Observation~\ref{obs:alpha-w-correct-iff-node-in-p-corr-alpha}, $\beta_m \sqsubseteq \alpha$, and it follows that $\gamma\left( \pil\left[\left| \pil \right| - 1\right] \right) = \beta_m\left(\pil\left[\left| \pil \right| - 1\right]\right)$.
	Let $t'' = \gamma\left( \vec{y\,} \right)$.
	Then, since $\vec{y} \in \partial^{\mathcal{T}}_{t''} \gamma_{\left| \pil \right|-1}$ and $\beta_m\left(\pil\left[\left| \pil \right| - 1\right]\right) = \gamma\left( \pil\left[\left| \pil \right| - 1\right] \right)$, we have $\vec{y} = \vec{p}^{\; \prime\prime} \in \partial^{\mathcal{T}}_{t''} \beta_m$.
	In this case, $\vec{\beta}$ has two competing children in $\mathcal{M}_{\mathcal{T}}$, but one is not $w$-correct and therefore is not a child of $\vec{\beta}$ in $\mathcal{P}$. 
	\end{enumerate}
	In both cases, we showed that for some $t'' \in T$, $\vec{p}^{\; \prime\prime} \in \partial^{\mathcal{T}}_{t''} \beta_m$, which establishes condition~\ref{def:ct-3}.
	\item Lemma~\ref{lem:competing-tile-4} says that every simple path in $G^{\textmd{b}}_{\beta_m + \left( \vec{p}^{\; \prime\prime}, t''\right)}$ from $\vec{s}$ to $\vec{p}^{\; \prime\prime}$ goes through $\vec{x}$, establishing condition~\ref{def:ct-4}.	
\end{enumerate}
Since Definition~\ref{def:competing-tile} holds for $\left( \vec{p}^{\; \prime\prime}, t'' \right)$ and $\beta_m$, it follows that $\vec{\beta}$ has at least two child nodes in $\mathcal{M}_{\mathcal{T}}$ that are competing in $\mathcal{M}_{\mathcal{T}}$.
{\em Step 2.} We now show that $\vec{\beta}$ has no more than two competing children. 
{\em Step 2a.} Assume, for the sake of obtaining a contradiction, that $\vec{\beta}$ has a third competing child, which we will denote as $\vec{\beta}'''$, that satisfies $\vec{\beta}''' = \vec{\beta} + \left(\vec{p}^{\; \prime\prime\prime}, t'''\right)$ for some $\vec{p}^{\; \prime\prime\prime} \in \Z^2$ and $t''' \in T$,  where $\left( \vec{p}^{\; \prime\prime\prime}, t'''\right)$ is a competing tile in $\beta_m$.
{\em Step 2b.}
We will first rule out the possibility that $\vec{p}^{\; \prime\prime\prime} \in \dom{\piw} \cup \dom{\pil}$.
Assume, for the sake of obtaining a contradiction, that $\vec{p}^{\; \prime\prime\prime} \in \dom{\piw} \cup \dom{\pil}$.
Then either $\vec{p}^{\; \prime\prime\prime} \neq \vec{y}$ or  $\vec{p}^{\; \prime\prime\prime} = \vec{y}$.
We handle both cases in turn.
\begin{enumerate}
\item[{\em 2bi.}] Assume $\vec{p}^{\; \prime\prime\prime} \ne \vec{y}$. 
	WLOG further assume that $\vec{p}^{\; \prime\prime\prime} \in \dom{\piw}$.
	Then, there is an integer $1 \leq l''' < \left| \piw \right|-1$ such that $\piw\left[ l'''+1 \right] = \vec{p}^{\; \prime\prime\prime}$.
	Note that $l''' \geq l'$ because $\dom{\piw\left[ 1\ldots l'\right]} \subseteq \dom{\beta_m}$.
	By Definition~\ref{def:competing-tile}, every simple path in $G^{\textmd{b}}_{\vec{\beta}'}$ from $\vec{s}$ to $\vec{p}^{\; \prime}$ goes through $\vec{x}$.
	Similarly, every simple path in $G^{\textmd{b}}_{\vec{\beta}'''}$ from $\vec{s}$ to $\vec{p}^{\; \prime\prime\prime}$ goes through $\vec{x}$.
	Since $\vec{p}^{\; \prime\prime\prime} \ne \vec{y}$,  condition~\ref{def:cp-4a} of Definition~\ref{def:competing-paths} implies that:
	\begin{enumerate}
	\item every simple path from $\vec{s}$ to $\vec{p}^{\; \prime\prime\prime}$ in $G^{\textmd{b}}_{\res{\vec{\beta}'''}}$
          has  $\piw[1\ldots l''']$ as a suffix.
        \end{enumerate}
        	Since $\vec{p}^{\; \prime}$ may or may not be equal to
                $\vec{y}$, either condition~\ref{def:cp-5} or
                condition~\ref{def:cp-4a} of
                Definition~\ref{def:competing-paths} implies that:
        \begin{enumerate}        \setcounter{enumii}{1}
	\item every simple path from $\vec{s}$ to $\vec{p}^{\; \prime}$ in $G^{\textmd{b}}_{\res{\vec{\beta}'}}$ has  $\piw[1\ldots l']$ as a suffix.
	\end{enumerate}
	Since $l''' \geq l'$, if $p'''$ is a path satisfying (a) and $p'$ is a path satisfying (b), then $\vec{p}^{\; \prime} = \vec{p}^{\; \prime\prime\prime}$ and $l' = l'''$.
        Such paths $p'''$ and $p'$ exist because $\vec{p}^{\; \prime\prime\prime} \in \dom{\res{\vec{\beta}'''}}$ and $\vec{p}^{\; \prime} \in \dom{\res{\vec{\beta}'}}$.
        Since  $\vec{\beta}'$ and $\vec{\beta}'''$ are distinct nodes in $\mathcal{M}_{\mathcal{T}}$, it must be the case that $t' \ne t'''$.
        By the definition of this case, we have $\vec{p}^{\; \prime\prime\prime} \ne \vec{y}$.
        Then, the contrapositive of condition~\ref{def:dd-ad-pt-1} of Definition~\ref{def:dd-ad-pt}
        applied to $\vec{p}^{\; \prime\prime\prime}$ says that $t' = t'''$, which is a contradiction to the fact that $t' \ne t'''$.
	The case where $\vec{p}^{\; \prime\prime\prime} \in \dom{\pil}$ can be handled in a similar manner. 
	In that case, $\vec{\beta}$ would necessarily have two competing children, one that terminates at a point in $\dom{\piw}$ and the other that terminates at another point in $\dom{\pil}$. 			

	\item[{\em 2bii.}] Assume $\vec{p}^{\; \prime\prime\prime} = \vec{y}$. Let $p$ be any simple path from $\vec{s}$ to $\vec{y}$ in $G^{\textmd{b}}_{\res{\vec{\beta}'''}}$.
	By condition~\ref{def:cp-5} of Definition~\ref{def:competing-paths}, there exists $\pi \in \left\{ \piw, \pil \right\}$ such that $\pi$ is a suffix of $p$.
	Assume WLOG that $\pi = \piw$. 
	Then, we have $\piw\left[ \left| \piw \right| \right] = \vec{p}^{\; \prime\prime\prime} = \vec{p}^{\; \prime}$ because $\dom{\piw\left[ 1\ldots l'\right]} \subseteq \dom{\beta_m}$.
	Since  $\vec{\beta}'$ and $\vec{\beta}'''$ are distinct nodes in $\mathcal{M}_{\mathcal{T}}$, it must be the case that $t' \ne t'''$.
	Furthermore, if we let $\vec{p} = \vec{p}^{\; \prime} = \vec{p}^{\; \prime\prime\prime}= \vec{y}$, the following two statements must both be true.

	\begin{enumerate}        \setcounter{enumii}{4}
	\item There exists a unique unit vector $\vec{u} \in \left\{ (0,1), (1,0), (0,-1), (-1,0) \right\}$ such that every simple path from $\vec{s}$ to $\vec{p}$ in $G^{\textmd{b}}_{\res{\vec{\beta}'''}}$ has $\left(\vec{p}+\vec{u}, \vec{p}\,\right)$ as a suffix.
		\item $\beta_m + \left( \vec{p}, t' \right)$ and $\beta_m + \left( \vec{p}, t''' \right)$ are both 
		valid tile attachment steps.
	\end{enumerate}
	But (e) and (f) together imply that $\vec{p}$ cannot be a directionally deterministic point.
	In other words, 
	condition~\ref{def:dd-ad-pt-2} of Definition~\ref{def:dd-ad-pt}
	is violated because two distinct types of tiles are attaching at $\vec{p}$ in the same direction and to the same type of tile.

	Thus, $\vec{p}$ cannot be a directionally deterministic point in $\mathcal{T}$, which contradicts condition~\ref{def:snd-1} of Definition~\ref{def:seq-non-deterministic}.
	The case where $\vec{p}^{\; \prime\prime\prime} \in \dom{\pil}$ can be handled in a similar manner.
	In that case, $\vec{\beta}$ would necessarily have two competing children that both terminate at $\vec{y}$. 
\end{enumerate}
At this point, we have shown that $\vec{p}^{\; \prime\prime\prime} \not \in  \dom{\piw} \cup \dom{\pil}$.

{\em Step 2c.} We now show that $\vec{p}^{\; \prime\prime\prime}$ cannot be on any competing path other than $\piw$ and  $\pil$.
Since $\left(\vec{p}^{\; \prime\prime\prime}, t'''\right)$ is a competing tile in $\beta_m$ and $\vec{p}^{\; \prime\prime\prime} \not \in \dom{\piw} \cup \dom{\pil}$, assume, for the sake of obtaining a contradiction, that there exist a competing path $\pi''' \not \in \left\{ \piw, \pil \right\}$, competing for an essential POC $\vec{y}^{\; \prime\prime\prime}$ from starting point $\vec{x}^{\; \prime\prime\prime}$ in $\mathcal{T}$ and an integer $1 <  l''' \leq \left| \pi''' \right|$ such that $\vec{p}^{\; \prime\prime\prime} = \pi'''\left[ l''' \right]$, $\vec{y}^{\; \prime\prime\prime} \not \in \dom{\beta_m}$, $\vec{p}^{\; \prime\prime\prime} \in \partial^{\mathcal{T}}_{t'''}\beta_m$, and every simple path in $G^{\textmd{b}}_{\beta_m + \left( \vec{p}^{\; \prime\prime\prime}, t''' \right)}$ goes through $\vec{x}^{\; \prime\prime\prime}$.
Since Lemma~\ref{lem:only-two-competing-paths} states that $\piw$ and $\pil$ are the only two paths competing for $\vec{y}$, from $\vec{x}$ in $\mathcal{T}$, $\vec{y}^{\; \prime\prime\prime} \ne \vec{y}$. As a result, $r > 1$.
Since $\vec{p}^{\; \prime\prime\prime} \in \partial^{\mathcal{T}}_{t'''} \beta_m$, $\vec{p}^{\; \prime\prime\prime} \not \in \dom{\piw} \cup \dom{\pil}$, and $\vec{\beta}' = \vec{\beta} + \left(\vec{p}^{\; \prime}, t' \right)$, it follows that $\res{\vec{\beta}'} \not \in \mathcal{A}_{\Box}\left[\mathcal{T}\right]$.
Moreover, we are assuming that $\vec{\beta}'$ is a finite, $w$-correct $\mathcal{T}$-producing assembly sequence.
Therefore, by Lemma~\ref{lem:w-correct-extended-to-w-correct-terminal}, there exists an extension $\vec{\varepsilon}^{\; \prime}$ of $\vec{\beta}'$ resulting in $\alpha$, with $Y \subseteq \dom{\alpha}$.

Now the fact that $\vec{x}^{\; \prime\prime\prime} \in \dom{\beta_m}$ follows from the conjunction of the following facts:
\begin{itemize}
\item $\vec{\beta}'''$  is a child of $\vec{\beta}$ whose result places $t'''$ at $\vec{p}^{\; \prime\prime\prime}$,
\item  $\vec{p}^{\; \prime\prime\prime} \neq \vec{x}^{\; \prime\prime\prime}$ (by Lemma~\ref{lem:cannot-terminate-at-SY}),
\item $\vec{p}^{\; \prime\prime\prime}$ is on a competing path from  $\vec{x}^{\; \prime\prime\prime}$ to $\vec{y}^{\; \prime\prime\prime}$, and
\item the path from $\vec{s}$ to $\vec{p}^{\; \prime\prime\prime}$ in $\bindinggraph_{\beta_m + \left( \vec{p}^{\; \prime\prime\prime}, t''' \right) }$ goes through  $\vec{x}^{\; \prime\prime\prime}$ (by condition~\ref{def:ct-4} of Definition~\ref{def:competing-tile}).
\end{itemize}

Since we also have $\vec{x}\in \dom{\beta_m}$ (by definition of
$\vec{\beta}$), it follows that

\begin{equation}
  \label{eqn:m-as-upperbound}
\max\left\{\textmd{index}_{\vec{\varepsilon}^{\, \prime}}\left( \vec{x}\, \right), \textmd{index}_{\vec{\varepsilon}^{\, \prime}}\left(\vec{x}^{\; \prime\prime\prime} \right) \right\} \leq m.
\end{equation}

Similarly, since $\vec{y} \notin \dom{\beta_m}$ (as established in
Step 1 above) and $\vec{y}^{\; \prime\prime\prime} \notin
\dom{\beta_m}$ (as we assume in this step), it follows that

\begin{equation}
  \label{eqn:m-as-lowerbound}
m <
\min\left\{\textmd{index}_{\vec{\varepsilon}^{\, \prime}}\left(
\vec{y}\, \right), \textmd{index}_{\vec{\varepsilon}^{\,
    \prime}}\left(\vec{y}^{\; \prime\prime\prime} \right) \right\}.
\end{equation}

Since, by Observation~\ref{obs:distinct-starting-points}, $\vec{x}
\neq \vec{x}^{\; \prime\prime\prime}$, we only have the following two
cases to consider based on the possible orderings of the index values in
$\vec{\varepsilon}^{\, \prime}$ of the two starting points.

 \begin{itemize}
 \item Case $\textmd{index}_{\vec{\varepsilon}^{\, \prime}}\left(
   \vec{x}\, \right) < \textmd{index}_{\vec{\varepsilon}^{\,
       \prime}}\left( \vec{x}^{\; \prime\prime\prime} \right)$: In
   this case, Equations~\ref{eqn:m-as-upperbound}
   and~\ref{eqn:m-as-lowerbound} together imply that
   $\textmd{index}_{\vec{\varepsilon}^{\, \prime}}\left( \vec{x}\,
   \right) < \textmd{index}_{\vec{\varepsilon}^{\, \prime}}\left(
   \vec{x}^{\; \prime\prime\prime} \right) <
   \textmd{index}_{\vec{\varepsilon}^{\, \prime}}\left( \vec{y}\,
   \right)$, which contradicts Lemma~\ref{lem:x-i-y-i-ordering}.

 \item Case $\textmd{index}_{\vec{\varepsilon}^{\, \prime}}\left(
   \vec{x}^{\; \prime\prime\prime} \right) <
   \textmd{index}_{\vec{\varepsilon}^{\, \prime}}\left( \vec{x}\,
   \right)$: 
 In
   this case, Equations~\ref{eqn:m-as-upperbound}
   and~\ref{eqn:m-as-lowerbound} together imply that
   $\textmd{index}_{\vec{\varepsilon}^{\, \prime}}\left( \vec{x}^{\; \prime\prime\prime}
   \right) < \textmd{index}_{\vec{\varepsilon}^{\, \prime}}\left(
   \vec{x}\, \right) <
   \textmd{index}_{\vec{\varepsilon}^{\, \prime}}\left( \vec{y}^{\; \prime\prime\prime}
   \right)$, which contradicts Lemma~\ref{lem:x-i-y-i-ordering}.

 \end{itemize}
Since we obtained a contradiction in both cases, we have shown that  $\vec{p}^{\; \prime\prime\prime}$ cannot be on any competing path other than $\piw$ and  $\pil$.
In conclusion, $\left( \vec{p}^{\; \prime\prime\prime}, t''' \right)$ cannot satisfy Definition~\ref{def:competing-tile}, which contradicts our initial assumption that $\vec{\beta}$ has a third competing child.
\end{proof}

The next lemma has two conclusions regarding the extension of a node $\vec{\beta}_{\vec{x}}$ of $\mathcal{P}$ that terminates at the starting point $\vec{x}$.
The first part assumes the competition $\mathcal{C}$ associated with $\vec{x}$ is not rigged and concludes that it is always possible to extend $\vec{\beta}_{\vec{x}}$ using an arbitrarily-chosen winning assembly sequence for $\mathcal{C}$ such that the extension is $w$-correct and terminates at the POC $\vec{y}$ corresponding to $\vec{x}$.
In other words, the first part is saying that it is always possible to extend $\vec{\beta}_{\vec{x}}$ in a winning manner. 
The second part concludes, regardless of whether $\mathcal{C}$ is assumed to be rigged, that any $w$-correct extension of $\vec{\beta}_{\vec{x}}$ that terminates at $\vec{y}$ must have embedded within it some winning assembly sequence for $\mathcal{C}$.
In other words, the second part is saying that the only way to extend $\vec{\beta}_{\vec{x}}$ to $\vec{y}$ is in a winning manner. 

\begin{lemma} 
\label{lem:extend-by-winning-and-only-winning}
Let $k \in \mathbb{Z}^+$ such that $\vec{\beta}_{\vec{x}} = \left( \beta_i \mid 1 \leq i \leq k\right)$ is a $w$-correct, $\mathcal{T}$-producing assembly sequence that terminates at $\vec{x}$. 
Then:
\begin{enumerate}
		\item \label{lem:extend-by-winning-and-only-winning-1} If $\mathcal{C}$ is not rigged by $\vec{\beta}_{\vec{x}}$, then for every winning assembly sequence for competition $\mathcal{C}$, say $\vec{\gamma} = \left( \gamma_i \mid 1 \leq i \leq m\right)$ with $m \in \mathbb{Z}^+\backslash\{1\}$, there exists a $\mathcal{T}$-producing assembly sequence 
		$$
		\vec{\varepsilon} = \left( \underbrace{\varepsilon_1, \varepsilon_2, \ldots, \varepsilon_k}_{\vec{\beta}_{\vec{x}}}, \underbrace{\varepsilon_{k+1}, \ldots, \varepsilon_{k+m-1}}_{\substack{\textmd{same tile attachment}\\ \textmd{steps as in\ }\vec{\gamma}[2\ldots m]}} \right),
		$$
	which satisfies the following conditions:
	\begin{enumerate}
	\item $\vec{\varepsilon}\,[1\ldots k] = \vec{\beta}_{\vec{x}}$,
	\item
          for all integers $k+1 \leq i \leq k+m-1$, $\varepsilon_i = \varepsilon_{i-1} +\left( \gamma_{i-k+1} \backslash \gamma_{i-k} \right)$,
          
		\item $\vec{\varepsilon}$ is a $w$-correct extension of $\vec{\beta}_{\vec{x}}$ that terminates at $\vec{y}$, and
		\item $\vec{\gamma}$ is the unique longest winning assembly sequence for $\mathcal{C}$ embedded in $\vec{\varepsilon}$.
	\end{enumerate}
	\item \label{lem:extend-by-winning-and-only-winning-2} For an arbitrary $\mathcal{T}$-producing assembly sequence $\vec{\varepsilon}$, if $\vec{\varepsilon}$ is a $w$-correct extension of $\vec{\beta}_{\vec{x}}$ that terminates at $\vec{y}$, then there exists a unique longest winning assembly sequence for $\mathcal{C}$ that is embedded in $\vec{\varepsilon}$.
	\end{enumerate}
\end{lemma}

We will use Lemma~\ref{lem:extend-by-winning-and-only-winning} to construct a partition of the leaf nodes of $\mathcal{P}^{\vec{\beta}_{\vec{x}}} \upharpoonright \vec{y}$ in terms of winning assembly sequences for a non-rigged competition.

\begin{proof}
Let $\piw$ and $\pil$ be the winning and losing paths, respectively, that are competing for $\vec{y}$.
By Lemma~\ref{lem:only-two-competing-paths}, $\piw$ and $\pil$ are the only two such paths.
Let $\beta_{\vec{x}} = \res{\vec{\beta}_{\vec{x}}}= \beta_k$.
We will now prove the two conclusions of the lemma.
\begin{enumerate}
	\item Let $\gamma = \res{\vec{\gamma}}$. 
Since $\vec{\gamma}$ is a winning assembly sequence for $\mathcal{C}$, Definition~\ref{def:winning-path-assembly-competition} implies that $\piw$ is a path in $G^{\textmd{b}}_{\gamma}$, $\gamma\left(\vec{y}\,\right) = w\left(\vec{y}\,\right)$, and $\dom{\gamma} \subseteq \dom{\piw} \cup \dom{\pil}$.
Let $\vec{p} \in \left\{\piw[2], \pil[2]\right\}$
and $t \in T$ such that $\gamma_2 = \gamma_1 + \left( \vec{p}, t \right)$.
Since $\mathcal{C}$ is not rigged by $\vec{\beta}_{\vec{x}}$, then by Lemma~\ref{lem:not-rigged}, we have $\dom{\gamma} \cap \dom{\beta_{\vec{x}}} = \left\{ \vec{x} \right\}$.
This implies that $\vec{p} \not \in \dom{\beta_{\vec{x}}}$.
Since $\vec{\beta}_{\vec{x}}$ is $w$-correct, i.e., $\beta_{\vec{x}}\left( \vec{x}\, \right) = \alpha\left( \vec{x}\, \right) = \gamma_1\left(\vec{x}\,\right) = \varepsilon_k\left(\vec{x}\,\right)$, it follows that $\vec{p} \in \partial^{\mathcal{T}}_t \beta_{\vec{x}}$, and we have $\beta_{\vec{x}} \rightarrow_1^{\mathcal{T}} \varepsilon_{k+1}$.
Then, $\vec{\varepsilon}$ is a $\mathcal{T}$-producing assembly sequence because 1) $\vec{\beta}_{\vec{x}}$ is a $\mathcal{T}$-producing assembly sequence, 2) $\vec{\gamma}$ is a $\mathcal{C}$-producing assembly sequence, which means $\vec{\gamma}$ is a $\mathcal{T}$-assembly sequence (although not necessarily a $\mathcal{T}$-producing one), and 3) Lemma~\ref{lem:not-rigged} implies that none of the tile attachment steps in $\vec{\gamma}$ is blocked in  $\vec{\varepsilon}$.
Thus, $\vec{\varepsilon}$ is an extension of $\vec{\beta}_{\vec{x}}$.
To see that $\vec{\gamma}$ is embedded in $\vec{\varepsilon}$, we consider the cases $k = 1$ and $k > 1$ in turn.
If $k = 1$, then $\dom{\beta_1} = \left\{ \vec{x} \right\}$, $\vec{\varepsilon} = \vec{\gamma}$, and the latter is trivially embedded in the former.
If $k > 1$, then the function $f:\left\{ 1, \ldots, m\right\} \rightarrow \left\{ 1, \ldots, k+m-1 \right\}$, defined as $f(1) = k$, and for $i > 1$, $f(i) = k+i-1$, testifies to $\vec{\gamma}$ being embedded in $\vec{\varepsilon}$.
Note that the first tile that attaches in any winning assembly sequence for $\mathcal{C}$ must attach at $\vec{x}$.
Since $\vec{\beta}_{\vec{x}}$ terminates at $\vec{x}$, it follows that $\vec{\gamma}$ is the longest winning assembly sequence of $\mathcal{C}$ embedded in $\vec{\varepsilon}$. 
Furthermore, $\vec{\varepsilon}$ terminates at $\vec{y}$ because $\vec{\gamma}$ is a winning assembly sequence for $\mathcal{C}$.
Finally,  $\vec{\varepsilon}$ is $w$-correct because both $\vec{\beta}_{\vec{x}}$ and $\vec{\gamma}$ are $w$-correct.
\item Let $\vec{\varepsilon}$ be a $w$-correct extension of $\vec{\beta}_{\vec{x}}$ that terminates at $\vec{y}$ such that $\varepsilon = \res{\vec{\varepsilon}}$. 
Since $\vec{\varepsilon}$ terminates at $\vec{y}$, there exists a simple path $\pi'$ from $\vec{s}$ to $\vec{y}$ in $G^{\textmd{b}}_{\varepsilon}$.
By condition~\ref{def:cp-5} of Definition~\ref{def:competing-paths}, there exists $\pi \in \left\{ \piw, \pil \right\}$ such that
$\pi$ is a suffix of $\pi'$.
Since $\vec{\varepsilon}$ is $w$-correct, $\pi = \piw$. 
Therefore $\piw$ is a simple path in $G^{\textmd{b}}_{\varepsilon}$ and thus $\dom{\piw} \subseteq \dom{\varepsilon}$.
Since $\vec{x}$ is the starting point for both $\piw$ and $\pil$, and $\vec{\beta}_{\vec{x}}$ terminates at $\vec{x}$, let $l$ be such that $1 \leq l < \left| \pil \right|$, and $\pil[1\ldots l]$ is a simple path in $ G^{\textmd{b}}_{\varepsilon}$ but $\pil[1\ldots l+1]$ is not.
Note that $l$ could equal $1$ if $\vec{x}$ is the only point of $\pil$ contained in $\dom{\varepsilon}$, but the largest value for $l$ is $\left| \pil \right| - 1$ because $\pi \ne \pil$.
In other words, $l$ is the length of the longest prefix of $\pil$ in $G^{\textmd{b}}_{\varepsilon}$.
By the definition of $l$, we have $\dom{\piw} \cup \dom{\pil[1\ldots l]} \subseteq \dom{\varepsilon}$.
Let $m=\left|\dom{\piw} \cup \dom{\pil[1\ldots l]}\right|$.
By condition~\ref{def:snd-4} of Definition~\ref{def:seq-non-deterministic}, which implies that $\vec{x} \ne \vec{y}$, we have $m > 1$. 
Let $\vec{\varepsilon} = \left( \varepsilon_i \mid 1 \leq i \leq n \right)$ for some $n \in \mathbb{Z}^+$ such that $n > \left|\vec{\beta}_{\vec{x}}\right|$.
Then, there exist $m$ indices 
$$
1 \leq i_1 < i_2 < \cdots < i_m \leq n
$$ 
such that, if $\dom{\varepsilon_{1}} = \left\{ \vec{x} \right\}$, then $i_1 = 1$
and for all $2 \leq j \leq m$, $i_j$ is such that $\dom{\varepsilon_{i_j}} \backslash \dom{\varepsilon_{i_j-1}} \subseteq \dom{\piw} \cup \dom{\pil[1\ldots l]}$; otherwise for all $1 \leq j \leq m$, $i_j$ is such that $\dom{\varepsilon_{i_j}} \backslash \dom{\varepsilon_{i_j-1}} \subseteq \dom{\piw} \cup \dom{\pil[1\ldots l]}$.
These indices must exist because $\dom{\piw} \cup \dom{\pil[1\ldots l]} \subseteq \dom{\varepsilon}$.
Note that, since $\vec{\varepsilon}$ terminates at $\vec{y}$, we have $i_m = n$ and $\dom{\varepsilon_{i_m}} \backslash \dom{\varepsilon_{i_m-1}} = \left\{ \vec{y} \right\}$. 
Let $f~:~\left\{1,\ldots, m\right\} \rightarrow \left\{ 1, \ldots, n\right\}$ be such that for all $1 \leq j \leq m$, $f(j) = i_j$, and for all $1 \leq j < k \leq m$, $f\left(j\right)<f\left( k\right)$.
We will now define the winning assembly sequence $\vec{\gamma}$ for $\mathcal{C}$ such that it is the unique longest winning assembly sequence for $\mathcal{C}$ embedded in $\vec{\varepsilon}$. 
To that end, define the assemblies $\gamma_1, \ldots, \gamma_m$ as follows:
\begin{enumerate}
	\item if $i_1 = 1$, then let $\gamma_1 = \varepsilon_1$, else let $\gamma_1 = \{ \varepsilon_{i_1} \backslash \varepsilon_{i_1-1}\}$, and
	\item for all $2 \leq j \leq m$, let $\gamma_j = \gamma_{j-1} + \left( \varepsilon_{f(j)} \backslash \varepsilon_{f(j)-1} \right)$.
\end{enumerate}
Let $\vec{\gamma} = \left( \gamma_1, \gamma_2, \ldots, \gamma_m \right)$.
We now show that $\vec{\gamma}$ is a $\mathcal{T}$-assembly sequence.
First, we have $\dom{\gamma_1} = \left\{ \vec{x} \right\}$. Furthermore, since $\vec{\beta}_{\vec{x}}$ is $w$-correct, we have: if  $i_1 = 1$, then  $\gamma_1\left(\vec{x}\,\right) = \varepsilon_1\left(\vec{x}\,\right) = \alpha\left(\vec{x}\,\right)$, else $\gamma_1\left(\vec{x}\,\right) = \varepsilon_{i_1}\left(\vec{x}\,\right) = \alpha\left(\vec{x}\,\right)$.
Then, for all $2 \leq j \leq m$, $\gamma_j = \gamma_{j-1} + \left( \varepsilon_{f(j)} \backslash \varepsilon_{f(j)-1} \right)$ is a valid tile attachment.
This is because:
\begin{enumerate}
	\item $\piw$ and $\pil[1\ldots l]$ are simple paths in $G^{\textmd{b}}_{\varepsilon}$, which means that
	if $\varepsilon_{f(j)} \backslash \varepsilon_{f(j)-1} = \left( \vec{p}, t \right) \in \mathbb{Z}^2 \times T$, then $\vec{p} \in \dom{\piw} \cup \dom{\pil[1\ldots l]}$
	and
	\item $1 \leq i_1 < i_2 < \cdots < i_m \leq n$ means that tiles attach in $\vec{\gamma}$ to locations in $\dom{\piw} \cup \dom{\pil}$ in the correct relative orders, i.e.,
	$$
	\textmd{index}_{\vec{\gamma}}\left( \piw[1] \right) < \cdots < \textmd{index}_{\vec{\gamma}}\left( \piw\left[ \left| \piw \right| \right] \right)
	$$
	and
	$$
	\textmd{index}_{\vec{\gamma}}\left( \pil[1] \right) < \cdots < \textmd{index}_{\vec{\gamma}}\left( \pil\left[ l \right] \right).
	$$
\end{enumerate}
Thus, $f$ testifies to $\vec{\gamma}$ being a $\mathcal{T}$-assembly sequence embedded in $\vec{\varepsilon}$.
Moreover, if we let $\gamma = \res{\vec{\gamma}} = \gamma_m$, then
$$
\dom{\gamma_m} = \bigcup_{i=1}^m{\dom{\gamma_i}} = \dom{\piw} \cup \dom{\pil[1\ldots l]}
$$
and, for all $\vec{p} \in \dom{\piw} \cup \dom{\pil[1\ldots l]}$, 
$$
\gamma\left(\vec{p}\,\right) = \varepsilon\left(\vec{p}\,\right) = \alpha\left(\vec{p}\,\right).
$$ 
Up to this point, we have established that $\gamma$ satisfies both
condition~\ref{def:wp-1} (i.e., $\piw$ is a path in
$G^{\textmd{b}}_{\gamma}$) and condition~\ref{def:wp-2} (i.e., $\gamma
\sqsubseteq \alpha$) of
Definition~\ref{def:winning-path-assembly-competition}.
We now prove that  $\gamma$ also satisfies condition~\ref{def:wp-3} of Definition~\ref{def:winning-path-assembly-competition}.

Assume, for the sake of obtaining a contradiction, that $p$ is a
simple path in $G^{\textmd{b}}_{\gamma}$ that starts at $\vec{x}$ and
is not a prefix of either $\piw$ or
$\pil[1\ldots l]$.  Since $\dom{\gamma} =
\dom{\piw} \cup \dom{\pil[1\ldots l]}$,
$p$ must be a simple path containing a prefix of the form
$(p_1=\vec{x}, \ldots, p_k,p_{k+1})$ where $1<k
<|\piw|$, $p[1\ldots k]$ is a prefix of
$\piw$, and $p_{k+1} = \pil[i]$ for some
integer $i$ such that $1<i\leq l$. This scenario is possible when
$p_k$ and $p_{k+1}$ belong to different competing paths but happen to
be adjacent points in $\mathbb{Z}^2$. Note that a similar reasoning
could be carried out if instead $p$ shared a prefix with
$\pil$ and then ``crossed over'' to some point in
$\dom{\piw}$. Now, by the way $\vec{\gamma}$ was built
from assemblies in $\vec{\varepsilon}$, $p$ is a simple path from
$\vec{x}$ to $p[|p|]$ in the binding graph of the
$\mathcal{T}$-producible assembly sequence $\vec{\varepsilon}$. Thus
$p$ is also a suffix of a simple path $p'$ in
$G^{\textmd{b}}_{\varepsilon}$ from $\vec{s}$ to $p[|p|]$ that goes
through $\vec{x}$. We now consider two cases. If $p[|p|] \neq
\vec{y}$, then condition~\ref{def:cp-4a} of
Definition~\ref{def:competing-paths} says $\pil[1\ldots
  i]$ is a suffix of $p$, which contradicts the fact that $p_2 \in
\dom{\piw}$. If $p[|p|] = \vec{y}$, then
condition~\ref{def:cp-5} of Definition~\ref{def:competing-paths} says
either $\piw$ or $\pil$ is a suffix of
$p$, which contradicts the constraint placed on $p$ when it was
initially selected.  Since we obtained a contradiction in all possible
cases, we have established that $\gamma$ satisfies
condition~\ref{def:wp-3} of
Definition~\ref{def:winning-path-assembly-competition} and is thus a
winning assembly for $\mathcal{C}$.

Since $\pil\left[1\ldots l \right]$ and $\piw$ are simple paths in $G^{\textmd{b}}_{\varepsilon}$ such that $\dom{\piw} \cap \dom{\pil\left[1 \ldots l \right]} = \left\{ \vec{x} \right\}$, it follows that $\vec{\gamma}$ is a $\mathcal{C}$-producing assembly sequence such that $\gamma_m$ is $\mathcal{C}$-producible from $\gamma_1$. 
Then, we have:
$$
\gamma_m = \gamma_{m-1} + \left( \varepsilon_{f(m)} \backslash \varepsilon_{f(m)-1}\right) = \gamma_{m-1} + \left( \varepsilon_{i_m} \backslash \varepsilon_{i_m-1}\right) = \gamma_{m-1} + \left( \vec{y}, \alpha\left( \vec{y}\, \right) \right),
$$
which means that $\vec{\gamma}$ terminates at $\vec{y}$. 
Thus, $\vec{\gamma}$ is a winning assembly sequence for $\mathcal{C}$. 
Finally, $\vec{\gamma}$ is the unique longest winning assembly sequence for $\mathcal{C}$ that is embedded in $\vec{\varepsilon}$, because all the points in $\dom{\piw}$ are in $\dom{\varepsilon}$, $l$ was 
chosen to pick out the longest prefix of $\pil$ that belongs to  $G^{\textmd{b}}_{\varepsilon}$, and the order of the assemblies in $\vec{\gamma}$ is uniquely determined by the fixed order of the assemblies in $\vec{\varepsilon}$. 
\end{enumerate}
\end{proof}

\begin{definition}  
\label{def:longest-winning-assembly-sequences}
Let $\vec{\gamma}$ be a winning assembly sequence for $\mathcal{C}$.
Then, we define $B_{\vec{\gamma}}$ to be the set of all leaf nodes $\vec{\beta}_{\vec{y}}$ of $\mathcal{P}^{\vec{\beta}_{\vec{x}}} \upharpoonright \vec{y}$ such that $\vec{\beta}_{\vec{y}}$ terminates at $\vec{y}$ and $\vec{\gamma}$ is the unique longest winning assembly sequence for $\mathcal{C}$ that is embedded in $\vec{\beta}_{\vec{y}}$.
\end{definition}

Note that if $\mathcal{C}$ is rigged by $\vec{\beta}_{\vec{x}}$, then it could be the case that $B_{\vec{\gamma}} = \emptyset$, depending on $\vec{\gamma}$. 

However, if $\mathcal{C}$ is not rigged by $\vec{\beta}_{\vec{x}}$, then, by conclusion~\ref{lem:extend-by-winning-and-only-winning-1} of Lemma~\ref{lem:extend-by-winning-and-only-winning}, for every winning assembly sequence $\vec{\gamma}$ for $\mathcal{C}$, $B_{\vec{\gamma}} \ne \emptyset$. 
We now prove, using Lemma~\ref{lem:extend-by-winning-and-only-winning}, that the set of all $B_{\vec{\gamma}}$ is a partition of the set of leaf nodes of $\mathcal{P}^{\vec{\beta}_{\vec{x}}} \upharpoonright \vec{y}$ in the case $\mathcal{C}$ is not rigged by $\vec{\beta}_{\vec{x}}$.

\begin{lemma} 
\label{lem:partition} If $\mathcal{C}$ is not rigged by $\vec{\beta}_{\vec{x}}$, then $\displaystyle\bigcup_{\substack{\vec{\gamma} \textmd{ winning assembly}\\\textmd{ sequence for }\mathcal{C}}}{ \{B_{\vec{\gamma}}\} }$ is a partition of the set of leaf nodes of $\mathcal{P}^{\vec{\beta}_{\vec{x}}} \upharpoonright \vec{y}$.
\end{lemma}

\begin{proof}

Since $\mathcal{C}$ is not rigged by $\vec{\beta}_{\vec{x}}$, then by conclusion~\ref{lem:extend-by-winning-and-only-winning-1} of Lemma~\ref{lem:extend-by-winning-and-only-winning}, $B_{\vec{\gamma}} \ne \emptyset$ for every winning assembly sequence $\vec{\gamma}$ for $\mathcal{C}$.

Let $U$ denote the union of sets defined in the lemma. We first show that the set of leaf nodes of $\mathcal{P}^{\vec{\beta}_{\vec{x}}} \upharpoonright \vec{y}$ is equal to $U$.
Let $\vec{\varepsilon}$ be a leaf node of $\mathcal{P}^{\vec{\beta}_{\vec{x}}} \upharpoonright \vec{y}$.
Then, $\vec{\varepsilon}$ is an extension of $\vec{\beta}_{\vec{x}}$ that terminates at $\vec{y}$.
By conclusion~\ref{lem:extend-by-winning-and-only-winning-2} of Lemma~\ref{lem:extend-by-winning-and-only-winning}, there is a  winning assembly sequence $\vec{\gamma}$ for $\mathcal{C}$ such that $\vec{\gamma}$ is the unique longest winning assembly sequence for $\mathcal{C}$ embedded in $\vec{\varepsilon}$. 
This means $\vec{\varepsilon} \in B_{\vec{\gamma}}$ and thus  $\vec{\varepsilon} \in U$.
By Definition~\ref{def:longest-winning-assembly-sequences}, if $\vec{\gamma}$ is a winning assembly sequence for $\mathcal{C}$, then $B_{\vec{\gamma}}$ is a subset of the set of leaf nodes of $\mathcal{P}^{\vec{\beta}_{\vec{x}}} \upharpoonright \vec{y}$.
Thus, the set of leaf nodes of $\mathcal{P}^{\vec{\beta}_{\vec{x}}} \upharpoonright \vec{y}$ is equal to $U$.
We now show that, for distinct winning assembly sequences $\vec{\gamma}_1, \vec{\gamma}_2$ for $\mathcal{C}$, $B_{\vec{\gamma}_1} \cap B_{\vec{\gamma}_2} = \emptyset$.
Let $\vec{\gamma}_1$ be a winning assembly sequence for $\mathcal{C}$ and $\vec{\beta}_1 \in B_{\vec{\gamma}_1}$.
By Definition~\ref{def:longest-winning-assembly-sequences}, $\vec{\beta}_1$ is a leaf node of $\mathcal{P}^{\vec{\beta}_{\vec{x}}} \upharpoonright \vec{y}$ and $\vec{\gamma}_1$ is the unique longest assembly sequence embedded in $\vec{\beta}_1$.
Let $\vec{\gamma}_2$ be a winning assembly sequence for $\mathcal{C}$ such that $\vec{\gamma}_1 \ne \vec{\gamma}_2$.
Since $\vec{\gamma}_1$ is the unique longest assembly sequence embedded in $\vec{\beta}_1$ and $\vec{\gamma}_1 \ne \vec{\gamma}_2$, it follows that $\vec{\beta}_1 \not \in B_{\vec{\gamma}_2}$.
Thus, $B_{\vec{\gamma}_1} \cap B_{\vec{\gamma}_2} = \emptyset$. 
\end{proof}

The following lemma says that 1) a non-competing descendant of $\vec{\beta}_{\vec{x}}$ cannot place a tile at a POC or along either competing path associated with $\vec{x}$, but 2) a competing descendant must place a tile along one of the competing paths associated with $\vec{x}$.

\begin{lemma} 
  \label{lem:location-of-p'}
  Let:\vspace*{-2mm}
\begin{itemize}
\item $\vec{y} \in Y$ with associated starting point $\vec{x}$ and competing path $\pi \in \{\piw, \pil\}$,
\item $\vec{\beta}_{\vec{x}}$ be a $\mathcal{T}$-producing assembly sequence that terminates at $\vec{x}$ and whose result  $\beta_{\vec{x}}$ is a subassembly of $\alpha$,
  \item $\vec{\beta}$ be an internal node of $\mathcal{P}^{\vec{\beta}_{\vec{x}}} \upharpoonright \vec{y}$, and
  \item  $\vec{\beta}'$ be a child node of $\vec{\beta}$ in $\mathcal{M}_{\mathcal{T}}$ such that $\vec{\beta}' = \vec{\beta} + \left( \vec{p}^{\; \prime}, t' \right)$ for some  $\vec{p}^{\; \prime} \in  \mathbb{Z}^2$ and  $t' \in T$.

\end{itemize}
Then we have:
\begin{enumerate}
\item \label{lem:loc-p'-1}  If $\vec{\beta}'$ is not competing in $\mathcal{M}_{\mathcal{T}}$, then $\vec{p}^{\; \prime} \notin (P\cup\dom{\piw}\cup\dom{\pil})$.
\item \label{lem:loc-p'-2} If $\vec{\beta}'$ is competing in $\mathcal{M}_{\mathcal{T}}$, then  $\vec{p}^{\; \prime} \in (\dom{\piw}\cup\dom{\pil})$.
\end{enumerate}
\end{lemma}

\begin{proof}
  Let $\beta = \res{\vec{\beta}}$ and $\beta' = \res{\vec{\beta}'}$.
  Let $i$ be the unique integer such that $1\leq i\leq r$ and $\vec{y}
  = \vec{y}_i$. Let $Y_{<i} =\{\vec{y}_j\in Y~|~j<i\}$ and $Y_{>i}
  =\{\vec{y}_j\in Y~|~j>i\}$.  Then $P =
  Y_{<i}~\cup~\{\vec{y}\}~\cup~Y_{>i}~\cup~(P\backslash Y)$.
  We prove the two parts of the lemma in turn.
  \begin{enumerate}
\item If $\vec{\beta}'$ is not competing in
  $\mathcal{M}_{\mathcal{T}}$, then $\left( \vec{p}^{\; \prime}, t'
  \right)$ is not a competing tile in $\beta$, which in turns implies that
  \begin{equation}
    \label{eqn:not-a-competing-tile}
    \forall \vec{p} \in Y\textmd{, at least one of the four conditions of
    Definition~\ref{def:competing-tile} is not satisfied by } \beta, \left( \vec{p}^{\; \prime}, t' \right)\textmd{, and } \vec{p}.
  \end{equation}
  
   \noindent Assume, for the sake of obtaining a contradiction, that
   $\vec{p}^{\; \prime} \in (P\cup\dom{\piw}\cup\dom{\pil})$. We now
   partition this set into five subsets and prove that a contradiction
   arises if $\vec{p}^{\;\prime}$ belongs to any one of these subsets.
 Since $\vec{p}^{\;
     \prime} \in (P\cup\dom{\piw}\cup\dom{\pil})$ and $\vec{y} \in (P\cap(
   \dom{\piw}\cup\dom{\pil}))$, we consider each one of the five
   possible cases in turn.
   \begin{itemize}
     
  \item Case $\vec{p}^{\;\prime} \in Y_{<i}$: Let
    $\vec{p}^{\;\prime}=\vec{y}_j$ for some integer $1 \leq j < i \leq
    r$. By Lemma~\ref{lem:w-correct-extended-to-w-correct-terminal},
    there exists an extension $\vec{\alpha}$ of $\vec{\beta}$ by some
    $\mathcal{T}$-assembly sequence such that $\vec{\alpha}$ results
    in $\alpha$. Since $r>1$, $\vec{x}$ is the starting point
    corresponding to $\vec{y}=\vec{y}_i$, and $j<i$, by
    Lemma~\ref{lem:x-i-y-i-ordering}, for all $\vec{p} \in Y_{<i}$,
    $\textmd{index}_{\vec{\alpha}}\left( \vec{p}\, \right) <
    \textmd{index}_{\vec{\alpha}}\left( \vec{x}\, \right)$. Therefore,
    for all $\vec{p} \in Y_{<i}$, $\vec{p}\in\dom{\beta_{\vec{x}}}
    \subseteq \dom{\beta}$ and thus
    $\vec{p}^{\;\prime}\in\dom{\beta}$, which implies that
    $\vec{p}^{\;\prime} \not \in \partial^{\mathcal{T}}\beta$, which
    in turn contradicts the fact that $\vec{\beta}' = \vec{\beta} +
    \left( \vec{p}^{\; \prime}, t' \right)$ is a
    $\mathcal{T}$-producing assembly sequence.
    
  \item Case $\vec{p}^{\;\prime} = \vec{y}$: In this case,
    $\vec{p}^{\;\prime}$ directly satisfies the first three conditions
    of Definition~\ref{def:competing-tile} applied to $\beta$, $\left(
    \vec{p}^{\; \prime}, t' \right)$, and $\vec{y}$, since
    $\vec{p}^{\;\prime} = \vec{y}$ implies that $\vec{p}^{\;\prime} =
    \piw[|\piw|]$, $\vec{y} \notin \dom{\beta}$, and
    $\vec{p}^{\;\prime} \in \partial^{\mathcal{T}}\beta$. Furthermore,
    condition~\ref{def:cp-5} of Definition~\ref{def:competing-paths}
    implies that $\vec{p}^{\;\prime}$ also satisfies condition~\ref{def:ct-4} (i.e., the fourth and
    last condition) of Definition~\ref{def:competing-tile} applied to
    $\beta$, $\left( \vec{p}^{\; \prime}, t' \right)$, and
    $\vec{y}$. However, the fact that $\vec{p}^{\;\prime}$ satisfies
    the four conditions of Definition~\ref{def:competing-tile} applied
    to $\beta$, $\left( \vec{p}^{\; \prime}, t' \right)$, and
    $\vec{y}$ contradicts the fact
    labeled~\ref{eqn:not-a-competing-tile} for $\vec{p} = \vec{y}$.
  \item Case $\vec{p}^{\;\prime} \in Y_{>i}$: Let
    $\vec{p}^{\;\prime}=\vec{y}_j$ for some integer $1\leq i< j \leq
    r$ and $\vec{x}_j$ be the corresponding starting point.
    Lemma~\ref{lem:competing-tile-4} implies $\vec{x}_j \in \dom{\beta'}$. Furthermore, by
    condition~\ref{def:snd-4} of
    Definition~\ref{def:seq-non-deterministic}, $\vec{x}_j
    \neq\vec{y}_j$. Thus $\vec{x}_j \in \dom{\beta}$.  Now, since
    $\vec{\beta}$ is $w$-correct and $\beta \notin
    \termasm{\mathcal{T}}$,
    Lemma~\ref{lem:w-correct-extended-to-w-correct-terminal} implies
    that there exists an extension $\vec{\alpha}$ of $\vec{\beta}$
    that results in $\alpha$.  If $\vec{p}_{\vec{\beta}}$ is the point
    at which $\vec{\beta}$ terminates and we let
    $k=\textmd{index}_{\vec{\alpha}}\left( \vec{p}_{\vec{\beta}}
    \right)$, then $\textmd{index}_{\vec{\alpha}}\left( \vec{x}_j
    \right)\leq k$. On the other hand, since $\vec{\beta}$ is an
    internal node of $\mathcal{P}^{\vec{\beta}_{\vec{x}}}
    \upharpoonright \vec{y}$, $\vec{y} \notin \dom{\beta}$ and thus
    $\textmd{index}_{\vec{\alpha}}\left( \vec{y}\, \right) > k$.  In
    conclusion, $\textmd{index}_{\vec{\alpha}}\left( \vec{x}_j\,
    \right) \leq k < \textmd{index}_{\vec{\alpha}}\left( \vec{y}\,
    \right) = \textmd{index}_{\vec{\alpha}}\left( \vec{y}_i\,
    \right)$. Now the facts that $r>1$, $i<j$, and
    $\textmd{index}_{\vec{\alpha}}\left( \vec{x}_j\, \right) <
    \textmd{index}_{\vec{\alpha}}\left( \vec{y}_i\, \right)$ together
    contradict the required linear ordering of tile placements at
    essential POCs and associated starting points that is implied by
    the combination of condition~\ref{def:snd-4} of
    Definition~\ref{def:seq-non-deterministic} and
    Lemma~\ref{lem:x-i-y-i-ordering}, namely the requirement that
    $\textmd{index}_{\vec{\alpha}}\left( \vec{x}_i\, \right) <
    \textmd{index}_{\vec{\alpha}}\left( \vec{y}_i\, \right) <
    \textmd{index}_{\vec{\alpha}}\left( \vec{x}_j\, \right) <
    \textmd{index}_{\vec{\alpha}}\left( \vec{y}_j\, \right)$.
    
  \item Case $\vec{p}^{\;\prime} \in (P\backslash Y)$: In this case,
    Lemma~\ref{lem:w-correct-no-inessential-poc} implies that
    $\vec{p}^{\;\prime} \notin \partial^{\mathcal{T}} \beta$,
    which contradicts the fact that $\vec{\beta}' = \vec{\beta} +
    \left( \vec{p}^{\; \prime}, t' \right)$ is a
    $\mathcal{T}$-producing assembly sequence.
    
  \item Case $\vec{p}^{\;\prime} \in
    \left((\dom{\piw}\cup\dom{\pil})\backslash\{\vec{y}\}\right)$: In
    this case, $\vec{p}^{\;\prime}$ directly satisfies the first and
    third conditions of Definition~\ref{def:competing-tile} applied to
    $\beta$, $\left( \vec{p}^{\; \prime}, t' \right)$, and
    $\vec{y}$. Furthermore, $\vec{p}^{\;\prime}$ satisfies
    condition~\ref{def:ct-2} of Definition~\ref{def:competing-tile}
    because $\vec{\beta}$ is an internal node of
    $\mathcal{P}^{\vec{\beta}_{\vec{x}}} \upharpoonright
    \vec{y}$. Finally, the facts that $\vec{\beta}^{\;\prime}$ is a child
    node of $\vec{\beta}$ in $\mathcal{M}_{\mathcal{T}}$ and
    $\vec{\beta}$ is equal to, or a descendant of,
    $\vec{\beta}_{\vec{x}}$ together imply that $\vec{p}^{\;\prime}$
    also satisfies condition~\ref{def:ct-4} (i.e., the fourth and last condition) of
    Definition~\ref{def:competing-tile} applied to $\beta$, $\left(
    \vec{p}^{\; \prime}, t' \right)$, and $\vec{y}$.  However, the
    fact that $\vec{p}^{\;\prime}$ satisfies the four conditions of
    Definition~\ref{def:competing-tile} applied to $\beta$, $\left(
    \vec{p}^{\; \prime}, t' \right)$, and $\vec{y}$ contradicts the
    fact labeled~\ref{eqn:not-a-competing-tile} for $\vec{p} =
    \vec{y}$.
  \end{itemize}
Since we obtain a contradiction in all cases, $\vec{p}^{\; \prime} \notin
(P\cup\dom{\piw}\cup\dom{\pil})$.

\item If $\vec{\beta}'$ is competing in $\mathcal{M}_{\mathcal{T}}$,
  then $\left( \vec{p}^{\; \prime}, t' \right)$ is a competing tile in
  $\beta$, which in turns implies that there exists $\vec{p} \in Y$
  such that all four conditions of Definition~\ref{def:competing-tile}
  are satisfied by $\beta$, $\left( \vec{p}^{\; \prime},t'\right)$,
  and $\vec{p}$. We now prove that $\vec{p}=\vec{y}$. Assume, for the
  sake of obtaining a contradiction, that $\vec{p}\neq\vec{y}$, i.e.,
  $\vec{p} \in Y\backslash\{\vec{y}\}$, which implies $\vec{p} \in
  \left(Y_{<i} \cup Y_{>i}\right)$. We consider each case in turn.

  \begin{itemize}
  \item Case $\vec{p} \in Y_{<i}$: Let $\vec{p}=\vec{y}_j$ for some
    integer $1 \leq j < i \leq r$. By
    Lemma~\ref{lem:w-correct-extended-to-w-correct-terminal}, there
    exists an extension $\vec{\alpha}$ of $\vec{\beta}$ by some
    $\mathcal{T}$-assembly sequence such that $\vec{\alpha}$ results
    in $\alpha$. Since $j<i$, $1<r$, and $\vec{x}$ is the starting
    point corresponding to $\vec{y} = \vec{y}_i$,
    Lemma~\ref{lem:x-i-y-i-ordering} and condition~\ref{def:snd-3} of
    Definition~\ref{def:seq-non-deterministic} together imply that
    $\textmd{index}_{\vec{\alpha}}\left( \vec{y}_j \right) =
    \textmd{index}_{\vec{\alpha}}\left( \vec{p}\, \right) <
    \textmd{index}_{\vec{\alpha}}\left( \vec{x}\, \right)$ and thus
    $\vec{y}_j\in\dom{\beta_{\vec{x}}}$. Since $\dom{\beta_{\vec{x}}} \subseteq
    \dom{\beta}$, we have $\vec{y}_j\in\dom{\beta}$, which contradicts
    condition~\ref{def:ct-2} of Definition~\ref{def:competing-tile}
    applied to $\beta$, $\left( \vec{p}^{\; \prime}, t' \right)$, and
    $\vec{p}=\vec{y}_j$, and thus the fact that $\left( \vec{p}^{\;
      \prime}, t' \right)$ is a competing tile in $\beta$.

  \item Case $\vec{p} \in Y_{>i}$: Let $\vec{p}=\vec{y}_j$ for some
    integer $1 \leq i < j \leq r$, $\vec{x}_j$ be the corresponding
    starting point, and $\pi_j$ and $\pi'_j$ be the corresponding winning
    and losing paths, respectively. By condition~\ref{def:ct-1} of
    Definition~\ref{def:competing-tile}, $\vec{p}^{\;\prime} \in
    (\dom{\pi_j}\cup\dom{\pi'_j})$. Since either $\vec{p}^{\;\prime}=\vec{x}_j$ or  $\vec{p}^{\;\prime}\neq\vec{x}_j$, we consider both sub-cases in turn.
    \begin{itemize}
    \item Sub-case $\vec{p}^{\;\prime}=\vec{x}_j$:
      Since $S_Y
      \subseteq S_P$ and condition~\ref{def:snd-4} of
      Definition~\ref{def:seq-non-deterministic} together imply $S_Y
      \cap P = \emptyset$, we have $\vec{p}^{\;\prime} \notin P$.
      Since $\alpha$ and $\beta$ agree and $\vec{p}^{\;\prime} \notin
      P$, 
      the contrapositive of condition~\ref{def:dd-ad-pt-1} of Definition~\ref{def:dd-ad-pt}
      implies that
      $t'=\beta'(\vec{p}^{\;\prime}) =
      \alpha(\vec{p}^{\;\prime})$. 
      Therefore, $\beta'$ attaches, at
      each point in its domain, the same tile as $\alpha$
      does. Furthermore, since $\vec{p}^{\;\prime} \notin P$,
      $\dom{\beta'}\cap (P\backslash Y)=\emptyset$, $\beta'$ must be
      $w$-correct, by Definition~\ref{def:correct-assembly}. Then,
      Lemma~\ref{lem:extend-to-alpha} implies that $\vec{\beta}'$ is a
      child node of $\vec{\beta}$ in $\mathcal{P}$. Since
      $\vec{\beta}'$ is a node in $\mathcal{P}$ that extends
      $\vec{\beta}_{\vec{x}}$ (whose domain does not contain $\vec{y}\,$),
      and $\vec{p}^{\;\prime} \neq \vec{y}$, $\vec{\beta}'$ must be an
      internal node of $\mathcal{P}^{\vec{\beta}_{\vec{x}}}
      \upharpoonright \vec{y}$. This further implies that
      $\beta' \notin \termasm{\mathcal{T}}$.  Then, by
      Lemma~\ref{lem:w-correct-extended-to-w-correct-terminal}, there
      exists an extension $\vec{\alpha}$ of $\vec{\beta}'$ by some
      $\mathcal{T}$-assembly sequence such that $\vec{\alpha}$ results
      in $\alpha$. Finally, the facts that $i<j$, $1<r$,
      $\vec{p}^{\;\prime}=\vec{x}_j$ is the starting point
      corresponding to $\vec{y}_j$, and $\vec{y}_i=\vec{y}\notin \dom{\beta'}$,
      together imply that $\textmd{index}_{\vec{\alpha}}\left(
      \vec{x}_j \right) < \textmd{index}_{\vec{\alpha}}\left(
      \vec{y}_i \right)$. This inequality contradicts the required
      linear ordering of tile placements at essential POCs and
      associated starting points that is implied by the combination of
      condition~\ref{def:snd-3} of
      Definition~\ref{def:seq-non-deterministic} and
      Lemma~\ref{lem:x-i-y-i-ordering}.  
      
    \item Sub-case $\vec{p}^{\;\prime} \in
      \left((\dom{\pi_j}\cup\dom{\pi'_j})\backslash\{\vec{x}_j\}\right)$:
      Lemma~\ref{lem:competing-tile-4} 
      implies $\vec{x}_j \in \dom{\beta'}$.
      In fact, since $\vec{p}^{\;\prime}\neq
      \vec{x}_j$, we must have $\vec{x}_j \in \dom{\beta}$. Now, since
      $\vec{\beta}$ is $w$-correct and $\beta \notin
      \termasm{\mathcal{T}}$,
      Lemma~\ref{lem:w-correct-extended-to-w-correct-terminal} implies
      that there exists an extension $\vec{\alpha}$ of $\vec{\beta}$
      that results in $\alpha$.  If $\vec{p}_{\vec{\beta}}$ is the
      point at which $\vec{\beta}$ terminates and we let
      $k=\textmd{index}_{\vec{\alpha}}\left( \vec{p}_{\vec{\beta}}
      \right)$, then $\textmd{index}_{\vec{\alpha}}\left( \vec{x}_j
      \right)\leq k$. On the other hand, since $\vec{\beta}$ is an
      internal node of $\mathcal{P}^{\vec{\beta}_{\vec{x}}}
      \upharpoonright \vec{y}$, $\vec{y} \notin \dom{\beta}$ and thus
      $\textmd{index}_{\vec{\alpha}}\left( \vec{y}\, \right) > k$.  In
      conclusion, $\textmd{index}_{\vec{\alpha}}\left( \vec{x}_j\,
      \right) \leq k < \textmd{index}_{\vec{\alpha}}\left( \vec{y}\,
      \right) = \textmd{index}_{\vec{\alpha}}\left( \vec{y}_i\,
      \right)$. Now the facts that $r>1$, $i<j$, and
      $\textmd{index}_{\vec{\alpha}}\left( \vec{x}_j\, \right) <
      \textmd{index}_{\vec{\alpha}}\left( \vec{y}_i\, \right)$
      together contradict the required linear ordering of tile
      placements at essential POCs and associated starting points that
      is implied by the combination of condition~\ref{def:snd-4} of
      Definition~\ref{def:seq-non-deterministic} and
      Lemma~\ref{lem:x-i-y-i-ordering}.
    \end{itemize}
  \end{itemize}
  Since we obtain a contradiction in both cases, we have $\vec{p} = \vec{y}$.
  Finally, by condition~\ref{def:ct-1} of
  Definition~\ref{def:competing-tile} applied to $\beta$, $\left(
  \vec{p}^{\; \prime}, t' \right)$, and $\vec{y}$, $\vec{p}^{\;
    \prime} \in (\dom{\piw}\cup\dom{\pil})$.
  \end{enumerate}
\end{proof}

In the following lemma, we define $\mathcal{Q}$ to be the unique subtree of $\mathcal{P}^{\vec{\beta}_{\vec{x}}} \upharpoonright \vec{y}$ such that embedded in all of its leaf nodes is the same winning assembly sequence $\mathcal{C}$, which is assumed to not be rigged.
Then, the lemma concludes that every internal node of $\mathcal{Q}$ has at most one child in $\mathcal{Q}$ that is competing in $\mathcal{M}_{\mathcal{T}}$ and the rest of its children is the set of its children in $\mathcal{M}_{\mathcal{T}}$ that are non-competing. 
\begin{lemma} 
\label{lem:number-of-competing-and-non-competing-children}
Let $\vec{\gamma}$ be a winning assembly sequence for $\mathcal{C}$.
If $\mathcal{C}$ is not rigged by $\vec{\beta}_{\vec{x}}$ and $\mathcal{Q}$ is the unique subtree of $\mathcal{P}^{\vec{\beta}_{\vec{x}}} \upharpoonright \vec{y}$ whose set of leaf nodes is equal to $B_{\vec{\gamma}}$, then, for every internal node $\vec{\beta}$ of $\mathcal{Q}$:
\begin{enumerate}
\item \label{lem:number-of-competing-and-non-competing-children-1} $\vec{\beta}$ has at most one child in $\mathcal{Q}$ that is competing in $\mathcal{M}_{\mathcal{T}}$, and
	\item \label{lem:number-of-competing-and-non-competing-children-2} every non-competing child of $\vec{\beta}$ in $\mathcal{M}_{\mathcal{T}}$ is a child of $\vec{\beta}$ in $\mathcal{Q}$.
\end{enumerate}
\end{lemma}

Since every node of $\mathcal{Q}$ is either competing or not, Lemma~\ref{lem:number-of-competing-and-non-competing-children} is saying that every internal node of $\mathcal{Q}$ is nearly normalized.
We will later construct a modified version $\mathcal{Q}'$ of $\mathcal{Q}$ where each internal node $\vec{\beta}$ of $\mathcal{Q}$ has its children partitioned into competing and non-competing children.
To that end, we make all competing children of $\vec{\beta}$ children of a new node $\Diamond$, and all of its non-competing nodes children of another new node $\bigcirc$.
Then, we make $\Diamond$ and $\bigcirc$ the only children of $\vec{\beta}$, each with probability ``$\frac{1}{2}$'', which means $\vec{\beta}$ is normalized in $\mathcal{Q}'$.
We use conclusion~\ref{lem:number-of-competing-and-non-competing-children-2} of Lemma~\ref{lem:number-of-competing-and-non-competing-children} to conclude that if $\bigcirc$ has any children, then it is normalized.
Moreover, if $\Diamond$ has any children, then Lemma~\ref{lem:number-of-competing-nodes} says that the probability on the edge from $\vec{\beta}$ to $\Diamond$ will be ``$\frac{1}{2}$''.
Note that the ``1'' in the numerator comes from conclusion~\ref{lem:number-of-competing-and-non-competing-children-1} of Lemma~\ref{lem:number-of-competing-and-non-competing-children}, and the ``2'' is justified by combining the fact that $\mathcal{C}$ is not rigged with an application of Lemma~\ref{lem:number-of-competing-nodes}.
This conversion process is defined formally in Definition~\ref{def:split-nodes}.

\begin{proofsketch}
  \begin{enumerate}
  \item To prove this part of the lemma, we start by assuming
    $\vec{\beta}$ has at least one child $\vec{\beta}'$ in
    $\mathcal{Q}$ that is competing in $\mathcal{M}_{\mathcal{T}}$ and
    attaches a tile at, say, $\vec{p}^{\;\prime}$.  Then, we prove by
    contradiction that $\vec{\beta}$ cannot have another child in
    $\mathcal{Q}$, say $\vec{\beta}''$, that is also competing in
    $\mathcal{M}_{\mathcal{T}}$. Assuming $\vec{\beta}''$ does exist
    and attaches a tile at, say, $\vec{p}^{\;\prime\prime}$, we first
    prove that $\vec{p}^{\;\prime} \neq \vec{p}^{\;\prime\prime}$ and
    then that $\vec{\beta}''$ could not be (the prefix of) a leaf node
    of $\mathcal{Q}$ such that $\vec{\gamma}$ is embedded in it.

  \item To prove this part of the lemma, we use
    conclusion~\ref{lem:extend-by-winning-and-only-winning-1} of
    Lemma~\ref{lem:extend-by-winning-and-only-winning} to conclude
    that $\vec{\beta}_{\vec{x}}$ can be extended with all (but the
    first) attachment steps in $\vec{\gamma}$ to yield
    $\vec{\varepsilon}$. Then we prove that:
    \begin{itemize}
    \item no tile attachment step in (a suffix of) $\vec{\varepsilon}$
      can be blocked by any tile attachment step in $\vec{\beta}'$,
    \item any child node $\vec{\beta}''$ of
      $\vec{\beta}$ that is not competing in $\mathcal{M}_{\mathcal{T}}$ can be extended
      with a suffix of $\vec{\varepsilon}$ to yield  $\vec{\varepsilon}^{\;\prime\prime}$,
    \item $\vec{\varepsilon}^{\;\prime\prime}$ is a
      $\mathcal{T}$-producing assembly sequence (see parts 2a.\ through
      2e.\ of the proof), and
    \item $\vec{\varepsilon}^{\;\prime\prime}$ is a leaf node of
      $\mathcal{Q}$, thereby establishing that $\vec{\beta}''$ is also
      in $\mathcal{Q}$.
    \end{itemize}
  \end{enumerate}\vspace*{-7mm}
\end{proofsketch}

\begin{proof}
For some $m \in \mathbb{Z}^+$, let $\vec{\beta} = \left( \beta_i \mid 1 \leq i \leq m \right)$ such that $\vec{y} \not \in \dom{\beta_m}$.
Recall that $\piw$ and $\pil$ are the two competing paths associated with competition $\mathcal{C}$. Let $\gamma = \res{\vec{\gamma}}$.
We prove the two conclusions of the lemma as follows.
\begin{enumerate}
	\item Assume $\vec{\beta}$ has at least one child in $\mathcal{Q}$ that is competing in $\mathcal{M}_{\mathcal{T}}$.
	Let $\vec{\beta}'$ be this child node of $\vec{\beta}$ such that, for some $\vec{p}^{\; \prime} \in \Z^2$ and $t' \in T$, $\vec{\beta}' = \vec{\beta} + \left(\vec{p}^{\; \prime}, t' \right)$ is a competing tile in $\beta_m$.
Moreover, since $\left( \vec{p}^{\; \prime}, t' \right)$ is a competing tile in $\beta_m$, $\vec{p}^{\; \prime} \in \partial^{\mathcal{T}}_{t'} \beta_m$, which means $\vec{p}^{\; \prime} \not \in \dom{\beta_m}$.
We now show that $\vec{\beta}$ does not have another child node in $\mathcal{Q}$ that is competing in $\mathcal{M}_{\mathcal{T}}$.
Assume, for the sake of obtaining a contradiction, that $\vec{\beta}$ has a second child node $\vec{\beta}''$ in $\mathcal{Q}$ that is competing in $\mathcal{M}_{\mathcal{T}}$.
Let $\vec{\beta}'' = \vec{\beta} +  \left( \vec{p}^{\; \prime\prime}, t'' \right)$,
where $\left( \vec{p}^{\; \prime\prime}, t'' \right)$ is a competing tile in $\beta_m$ that is distinct from $\left( \vec{p}^{\; \prime}, t' \right)$.
Let  $\beta' = \res{\vec{\beta}'}$ and $\beta'' = \res{\vec{\beta}''}$.
Note that, if  $\vec{p}^{\; \prime} = \vec{p}^{\; \prime\prime}=\vec{p}$, then  $\beta'$ and  $\beta''$ would place distinct tile types at $\vec{p}$, which would contradict the fact that  $\vec{\beta}'$ and $\vec{\beta}''$ are both nodes in $\mathcal{Q}$ that must place the same tile type at any  point in $\dom{\gamma}$, to which $\vec{p}$ belongs.
Going forward, we assume $\vec{p}^{\; \prime} \neq \vec{p}^{\; \prime\prime}$.
  We now prove that  $\vec{\beta}''$ cannot exist. Since $\vec{\beta}'$ is a node of $\mathcal{Q}$, it is (or is a prefix of) some leaf node $\vec{\varepsilon}^{\; \prime}$ of $\mathcal{Q}$ such that $\vec{\gamma}$ is embedded in $\vec{\varepsilon}^{\; \prime}$.
Similarly $\vec{\beta}''$ is (or is a prefix of) some leaf node $\vec{\varepsilon}^{\; \prime\prime}$ of $\mathcal{Q}$ such that $\vec{\gamma}$ is embedded in $\vec{\varepsilon}^{\; \prime\prime}$.
Note that if $\vec{\beta}'$ is a leaf node, then $\vec{\beta}' = \vec{\varepsilon}^{\; \prime}$. 
Similarly, if $\vec{\beta}''$ is a leaf node, then $\vec{\beta}'' = \vec{\varepsilon}^{\; \prime\prime}$.
However, it is impossible for both $\vec{\beta}'$ and $\vec{\beta}''$ to be leaf nodes, because we previously ruled out the possibility of $\vec{p}^{\; \prime} = \vec{p}^{\; \prime\prime} = \vec{y}$. 
Nevertheless, there exist embedding functions $f'$ of $\vec{\gamma}$ in $\vec{\varepsilon}^{\; \prime}$ and $f''$ of $\vec{\gamma}$ in $\vec{\varepsilon}^{\; \prime\prime}$.
Let $j \in \left\{ 1,\ldots, \left| \vec{\gamma} \right| \right\}$ be such that $f'(j) = m+1$.
In other words, $j$ is the index in  $\vec{\gamma}$ of the assembly that results from attaching $\left(\vec{p}^{\; \prime},t'\right)$, since $m+1$ is the index in $\vec{\varepsilon}^{\; \prime}$ of the assembly $\beta'$ that results from attaching $\left(\vec{p}^{\; \prime},t'\right)$ to $\beta_m$.
Since $\vec{\beta}$ is a prefix of both $\vec{\varepsilon}^{\; \prime}$ and $\vec{\varepsilon}^{\; \prime\prime}$, we have $f'(i) = f''(i)$ for all $1 \leq i < j$.
However, $\vec{p}^{\; \prime} \ne \vec{p}^{\; \prime\prime}$, which we already established, implies that $f''(j) > m+1$ and for some $k > j$, $f''\left(k\right) = m+1$.
This means $\vec{\gamma}$ cannot be embedded in $\vec{\varepsilon}^{\; \prime\prime}$ because, even though  $\vec{\gamma}$ is the unique longest winning assembly sequence for $\mathcal{C}$ embedded in every leaf node of $\mathcal{Q}$, we have both $j < k$ and  $f''(j) > m +1 = f''(k)$. In other words, the relative order of the tile attachments of $\vec{\gamma}$ is not preserved in $\vec{\varepsilon}\;''$, which contradicts condition~(4) in the definition of $f''$ as an embedding function.
Therefore, $\vec{\beta}''$ cannot exist and $\vec{\beta}$ can only have one child in $\mathcal{Q}$ that is competing in $\mathcal{M}_{\mathcal{T}}$.
%

	
	\item Assume  $\vec{\beta}_{\vec{x}} = \left( \beta_i \mid 1 \leq i \leq k \right)$ for some $k \in \mathbb{Z}^+$.
By conclusion~\ref{lem:extend-by-winning-and-only-winning-1} of Lemma~\ref{lem:extend-by-winning-and-only-winning}, there exists a $\mathcal{T}$-producing assembly sequence $\vec{\varepsilon}$ such that $\varepsilon = \res{\vec{\varepsilon}\,}$, $\vec{\varepsilon}$ is a $w$-correct extension of $\vec{\beta}_{\vec{x}}$ that terminates at $\vec{y}$, and $\vec{\gamma}$ is the unique longest winning assembly sequence for $\mathcal{C}$ embedded in $\vec{\varepsilon}$.
In particular, if $\vec{\gamma} = \left( \gamma_i \mid 1 \leq i \leq g \right)$ for some $g \in \Z^+\backslash\{1\}$, then 
$$
\vec{\varepsilon} = \left( \underbrace{\varepsilon_1, \varepsilon_2, \ldots, \varepsilon_k}_{\vec{\beta}_{\vec{x}}}, \underbrace{\varepsilon_{k+1}, \ldots \varepsilon_{k+g-1}}_{\substack{\textmd{same tile attachment}\\ \textmd{steps as in\ }\vec{\gamma}[2\ldots g]}} \right),
$$
where, for all integers $1 \leq i \leq k$, $\varepsilon_i = \beta_i$, and for all integers $k+1 \leq i \leq k+g-1$, $\varepsilon_i = \varepsilon_{i-1} +\left( \gamma_{i-k+1} \backslash \gamma_{i-k} \right)$.
By definition of $\mathcal{Q}$, $\vec{\varepsilon}$ is one of its leaf nodes.
Recall that $\vec{\beta}$ is an internal node of $\mathcal{Q}$, and thus can be written, for some integer $l \geq 0$, as
$$
\vec{\beta} = \left( \underbrace{\beta_1, \beta_2, \ldots, \beta_k}_{\vec{\beta}_{\vec{x}}}, \ldots, \beta_{k + l} \right).
$$
Note that $\vec{\beta}_{\vec{x}}$ could be equal to $\vec{\beta}$, in which case $l=0$.

We now show, in several steps outlined in the proof sketch above, that
every non-competing child of $\vec{\beta}$ in
$\mathcal{M}_{\mathcal{T}}$ is a child of $\vec{\beta}$ in
$\mathcal{Q}$.

\hspace*{5mm}$\bullet$\hspace*{1mm} By
Lemma~\ref{lem:number-of-competing-nodes}, $\vec{\beta}$ has a child
$\vec{\beta}'$ in $\mathcal{M}_{\mathcal{T}}$ that is $w$-correct and
competing in $\mathcal{M}_{\mathcal{T}}$.  The first step is to prove
that no tile attachment step in (a suffix of) $\vec{\varepsilon}$ can be blocked by
any tile attachment step in $\vec{\beta}'$.
Since  $\vec{\beta}'$ is a child of  $\vec{\beta}$ that is competing in  $\mathcal{M}_{\mathcal{T}}$, $\vec{\beta}'$ terminates at some
$\vec{p}^{\; \prime} \in \dom{\gamma}$. Therefore, there exists some integer $j$ such that $k+1 \leq j \leq k+g-1$ and $\left(\vec{p}^{\; \prime}, t'\right) = \varepsilon_{j} \backslash \varepsilon_{j-1}$, and  we have:
$$
\vec{\beta}' = \left( \underbrace{\beta_1, \beta_2, \ldots, \beta_{k+l}}_{\vec{\beta}},  \underbrace{\beta_{k+l}+\left(\vec{p}^{\; \prime}, t'\right)}_{\beta'_{k+l+1}}\right).
$$
Let $\beta'=\res{\vec{\beta'}} =\beta'_{k+l+1}$.
Note that the last tile attachment step in $\vec{\beta}'$ places a tile at the same point $\vec{p}^{\; \prime}$ along either $\piw$ or $\pil$ where $\varepsilon_j$ also places a tile.
Since $\dom{\varepsilon} \backslash\dom{\varepsilon_j}$ is the set of points at which the suffix $(\varepsilon_{j+1}, \ldots, \varepsilon_{k+g-1})$  of $\vec{\varepsilon}$ places tiles, establishing that $\dom{\beta'} \cap \left(\dom{\varepsilon} \backslash\dom{\varepsilon_j}\right)$ is empty is sufficient to complete this first step.

If $\vec{p}^{\; \prime} = \vec{y}$, then $j = k+g-1$ and it follows that $\dom{\beta'} \cap \left(\dom{\varepsilon} \backslash\dom{\varepsilon_j}\right) = \dom{\beta'} \cap \emptyset = \emptyset$.
Going forward, assume $\vec{p}^{\; \prime} \ne \vec{y}$, making $\vec{\beta}'$ a node of $\mathcal{M}_{\mathcal{T}}$ that extends $\vec{\beta}_{\vec{x}}$ but does not terminate at $\vec{y}$. 
Let $m_w,m_l$ be integers such that:
\begin{enumerate}
	\item $1 \leq m_w < \left| \piw \right|$ such that $\piw[1\ldots m_w]$ is a simple path in $G^{\textmd{b}}_{\beta'}$ but $\piw[1\ldots m_w+1]$ is not, and 
	\item $1 \leq m_l < \left| \pil \right|$ such that $\pil[1\ldots m_l]$ is a simple path in $G^{\textmd{b}}_{\beta'}$ but $\pil[1\ldots m_l+1]$ is not.
\end{enumerate}
Note that either $m_w > 1$ or $m_l > 1$, and that $\vec{p}^{\; \prime} \in \left\{ \piw[m_w], \pil[m_l] \right\}$.
Suppose, for the sake of obtaining a contradiction, that $\dom{\beta'} \cap \left(\dom{\varepsilon} \backslash\dom{\varepsilon_j}\right) \ne \emptyset$.
Let $\vec{p}^{\; \prime\prime} \in \left(\dom{\varepsilon} \backslash\dom{\varepsilon_j}\right)$.
Such an element exists because $\vec{p}^{\; \prime} \ne \vec{y}$.
Then $\vec{p}^{\; \prime\prime} \in \left( \dom{\piw\left[m_w+1\ldots \left| \piw\right|\right]} \cup \dom{\pil\left[ m_l+1\ldots \left| \pil \right| \right]} \right)$, which implies $\vec{p}^{\; \prime\prime} \ne \vec{x}$.
Since we assume that $\mathcal{C}$ is not rigged by $\vec{\beta}_{\vec{x}}$, Lemma~\ref{lem:not-rigged} says that $\dom{\beta_k} \cap \left( \dom{\piw} \cup \dom{\pil} \right) = \left\{ \vec{x} \right\}$, and thus $\vec{p}^{\; \prime\prime} \notin \dom{\beta_k}$.
Finally, by Lemma~\ref{lem:longest-simple-finite-path-beta-m}, it follows that $\vec{p}^{\; \prime\prime} \not \in \dom{\beta}'$.
Therefore, $\dom{\beta'} \cap \left(\dom{\varepsilon} \backslash\dom{\varepsilon_j}\right) = \emptyset$.
In other words, recalling that $\beta' = \beta'_{k+l+1}$:
\begin{equation*}
  \label{eqn:number-of-competing-and-non-competing-children}
    \textmd{for all integers\ }  0 < i \leq k + g - j - 1,\ (\dom{\varepsilon_{i+j}}\backslash \dom{\varepsilon_j})  \cap \dom{\beta'_{k+l+1}} = \emptyset. 
\end{equation*}
From this equation, together with the fact that
$(\dom{\varepsilon_{i+j}}\backslash \dom{\varepsilon_{i+j-1}}) \subseteq (\dom{\varepsilon_{i+j}}\backslash \dom{\varepsilon_j})$, it immediately follows that:
\begin{equation}
  \label{eqn:number-of-competing-and-non-competing-children'}
    \textmd{for all integers\ }  0 < i \leq k + g - j - 1,\ (\dom{\varepsilon_{i+j}}\backslash \dom{\varepsilon_{i+j-1}})  \cap \dom{\beta'_{k+l+1}} = \emptyset. 
\end{equation}

\hspace*{5mm}$\bullet$\hspace*{1mm}
In this second step, we prove that any child node
$\vec{\beta}''$ of $\vec{\beta}$ that is not competing in
$\mathcal{M}_{\mathcal{T}}$ can be extended with a suffix of
$\vec{\varepsilon}$ to yield $\vec{\varepsilon}^{\;\prime\prime}$.
Assume $\vec{\beta}$ has a non-competing child node $\vec{\beta}''$ in $\mathcal{M}_{\mathcal{T}}$, say
$$
\vec{\beta}'' = \left( \underbrace{\beta_1, \beta_2, \ldots, \beta_{k+l}}_{\vec{\beta}},  \underbrace{\beta_{k+l} + \left( \vec{p}^{\; \prime\prime}, t'' \right)}_{\beta''_{k+l+1}} \right),
$$
where $\left( \vec{p}^{\; \prime\prime}, t'' \right)$ is a non-competing tile in $\beta_{k+l}=\res{\vec{\beta}}=\beta$.
Let $\beta'' = \res{\vec{\beta}''} =  \beta''_{k+l+1} = \beta_{k+l} + \left(\vec{p}^{\; \prime\prime}, t''\right)$.
We first show that $\vec{\beta}''$ is $w$-correct.
First, since $\res{\vec{\beta}}$ is $w$-correct, for all $\vec{p} \in \dom{\beta_{k+l}} \cap Y$, $\beta_{k+l}\left(\vec{p}\,\right) = w\left(\vec{p}\,\right)$.
Second, applying part~\ref{lem:loc-p'-1} of Lemma~\ref{lem:location-of-p'} with $\vec{\beta}'= \vec{\beta}''$ yields $\vec{p}^{\; \prime\prime} \notin P$ and thus $\vec{p}^{\; \prime\prime} \notin Y$. Therefore,  $\dom{\beta_{k+l}} \cap Y = \dom{\left(\beta_{k+l} + \left(\vec{p}^{\; \prime\prime}, t''\right)\right)}  \cap Y=\dom{\beta''_{k+l+1}} \cap Y$. As a result, for all $\vec{p} \in \dom{\beta''_{k+l+1}} \cap Y$, $\beta''_{k+l+1}\left(\vec{p}\,\right) = w\left(\vec{p}\,\right)$. In other words, $\beta''_{k+l+1}$ satisfies condition~\ref{def:correct-assembly-1} of  Definition~\ref{def:correct-assembly}. Since $\vec{p}^{\; \prime\prime} \notin P$ implies $\vec{p}^{\; \prime\prime} \notin (P\backslash Y)$ and thus that $\beta''_{k+l+1}=\left(\beta_{k+l} + \left(\vec{p}^{\; \prime\prime}, t''\right)\right)$ satisfies the second condition of  Definition~\ref{def:correct-assembly}, it follows that $\vec{\beta}''$ is $w$-correct.

To complete this step, it suffices to show that there exists an extension of $\vec{\beta}''$ by a suffix of $\vec{\varepsilon}$.
By part~\ref{lem:loc-p'-1} of Lemma~\ref{lem:location-of-p'},  $\vec{p}^{\; \prime\prime} \not \in \left(\dom{\piw}\cup\dom{\pil}\right)$, and thus $\vec{p}^{\; \prime\prime} \not \in \dom{\gamma}$.
Let 
$$
\vec{\varepsilon}^{\; \prime\prime} = \left( \underbrace{\varepsilon''_1 = \beta_1, \ldots, \varepsilon''_{k+l}= \beta_{k+l}, \varepsilon''_{k+l+1}=\beta''_{k+l+1}}_{\vec{\beta}''}, \underbrace{\varepsilon''_{k+l+2}, \ldots, \varepsilon''_{k+l+2+(k+g-j-1)}}_{\substack{\textmd{same tile attachment steps}\\ \textmd{as in\ }\vec{\gamma}[j-k+1\ldots g]}}\right),
$$
where, for all integers $0 \leq i \leq k+g-j-1$, 
\begin{eqnarray}
  \label{eqn:number-of-competing-and-non-competing-children2}
\varepsilon''_{k+l+2+i} & = & \varepsilon''_{k+l+2+i-1} + \left( \varepsilon_{i+j} \backslash \varepsilon_{i+j-1} \right) \\
	& = & \varepsilon''_{k+l+2+i-1} + \left( \gamma_{j-k+i+1} \backslash \gamma_{j-k+i} \right).\nonumber
\end{eqnarray}
Note that $\vec{\varepsilon}^{\; \prime\prime}$ is the extension of $\vec{\beta}''$ where all (and only) the remaining competing tiles of $\vec{\varepsilon}$ that have not yet attached to  $\vec{\beta}_{\vec{x}}$, namely $\varepsilon_j\backslash \varepsilon_{j-1},\ \varepsilon_{j+1}\backslash\varepsilon_{j},\ \ldots,\ \varepsilon_{k+g-1}\backslash\varepsilon_{k+g-2}$,
are attached to form the suffix of $\vec{\varepsilon}^{\; \prime\prime}$.

\hspace*{5mm}$\bullet$\hspace*{1mm}
In this third step, we use the following facts to prove that $\vec{\varepsilon}^{\; \prime\prime}$ is a $\mathcal{T}$-producing assembly sequence:
\begin{enumerate}[label=\theenumi\alph*.,itemsep=0mm]
	\item $\vec{\beta}$ is a $\mathcal{T}$-producing assembly sequence.
	\item $\varepsilon''_{k+l+1} = \beta_{k+l} + \left( \vec{p}^{\; \prime\prime}, t''\right)$ is a valid tile attachment step because $\vec{\beta}''$ is a child of $\vec{\beta}$ in $\mathcal{M}_{\mathcal{T}}$.
	\item $\varepsilon''_{k+l+2} = \beta''_{k+l+1} + \left( \varepsilon_{j} \backslash \varepsilon_{j-1} \right)$ is a valid tile attachment step because $\vec{p}^{\; \prime} \in \partial^{\mathcal{T}}_{t'} \beta_{k+l}$ and $\vec{p}^{\; \prime} \ne \vec{p}^{\; \prime\prime}$, which means $\vec{p}^{\; \prime} \in \partial^{\mathcal{T}}_{t'} \beta''_{k+l+1}$. 
        \item $\vec{\varepsilon}^{\;*} = \left( \varepsilon''_{k+l+2}, \ldots, \varepsilon''_{k+l+2+(k+g-j-1)} \right)$ is a $\mathcal{T}$-assembly sequence because $\vec{\varepsilon}$ is a $\mathcal{T}$-producing assembly sequence.
	\item Each point where a tile attaches in $\vec{\varepsilon}^{\;*}$ is empty just
          before the corresponding tile attachment step, as we now
          establish.
          First, since $\vec{p}^{\;  \prime\prime} \not \in \dom{\gamma}$, placing tile $t''$ at $\vec{p}^{\; \prime\prime}$
          could not have blocked any tile attachment step in
          $\vec{\varepsilon}^{\;*}$. Second, we already established above that the last tile attachment
          step resulting in $\varepsilon''_{k+l+2}$ could not have been blocked. Third,
          for all integers $0 < i \leq k+g-j-1$,
          equation (\ref{eqn:number-of-competing-and-non-competing-children2}) yields:
          \begin{equation}
            \label{eqn:number-of-competing-and-non-competing-children3}
          \dom{\varepsilon''_{k+l+2+i}} \backslash  \dom{\varepsilon''_{k+l+2+i-1}} = \dom{\varepsilon_{i+j}} \backslash \dom{\varepsilon_{i+j-1}}
          \end{equation}
          which, combined with (\ref{eqn:number-of-competing-and-non-competing-children'}), yields:
          $$
          (\dom{\varepsilon''_{k+l+2+i}}\backslash \dom{\varepsilon''_{k+l+2+i-1}}) \cap \dom{\beta'_{k+l+1}} =
          \emptyset.
          $$
          This equation, together with the facts that $\dom{\beta''_{k+l+1}} = \left(\dom{\beta'_{k+l+1}} \backslash \left\{\vec{p}^{\; \prime}\right\}\right) \cup  \left\{\vec{p}^{\; \prime\prime}\right\}$ and $\vec{p}^{\; \prime\prime} \in  \dom{\varepsilon''_{k+l+2+i-1}}$, imply that
          $$
          (\dom{\varepsilon''_{k+l+2+i}}\backslash \dom{\varepsilon''_{k+l+2+i-1}}) \cap \dom{\beta''_{k+l+1}} =\emptyset.
          $$
          In conclusion, no tile attachment step in
          $\vec{\varepsilon}^{\;*}$ may be blocked by a prior tile
          attachment step.

\end{enumerate}

\hspace*{5mm}$\bullet$\hspace*{1mm}
In this fourth and last step, we prove that  $\vec{\varepsilon}^{\;\prime\prime}$ is a leaf node of
      $\mathcal{Q}$, thereby establishing that $\vec{\beta}''$ is also
in $\mathcal{Q}$. Note that $\vec{\varepsilon}$ terminates at $\vec{y}$ because $\vec{\gamma}$ is a winning assembly sequence for $\mathcal{C}$. Therefore, using (\ref{eqn:number-of-competing-and-non-competing-children'}), we have:
\begin{eqnarray*}
  \dom{\varepsilon''_{k+l+2+(k+g-j-1)}} \backslash \dom{\varepsilon''_{k+l+2+(k+g-j-2)}}
  & = & \dom{\varepsilon_{k+g-1}} \backslash \dom{\varepsilon_{k+g-2}}\\
  & = &\dom{\gamma_g} \backslash \dom{\gamma_{g-1}}\\
  & = & \left\{ \vec{y} \right\}.
\end{eqnarray*}
It is also true that $\vec{\varepsilon}^{\; \prime\prime}$ is
$w$-correct because $\vec{\beta}''$ is $w$-correct and every
subsequent assembly in $\vec{\varepsilon}^{\; \prime\prime}$ is the
result of an assembly sequence that is a node of $\mathcal{Q}$, since
each one of them results from a tile attachment step that is part of
$\vec{\gamma}$. Therefore, every subsequent assembly in
$\vec{\varepsilon}^{\; \prime\prime}$ is $w$-correct.
Moreover, by the definition of $\vec{\varepsilon}^{\;\prime\prime}$ and the fact that $\vec{\gamma}$ is the unique longest winning assembly sequence embedded in $\vec{\varepsilon}$, it follows that $\vec{\gamma}$ is the longest winning assembly sequence embedded in $\vec{\varepsilon}^{\;\prime\prime}$.
As a result, $\vec{\varepsilon}^{\; \prime\prime}$ must be  a leaf node of $\mathcal{Q}$, making $\vec{\beta}''$ a child of $\vec{\beta}$ in $\mathcal{Q}$.
Since we originally chose $\vec{\beta}''$ to be an arbitrary non-competing child node of $\vec{\beta}$ in $\mathcal{M}_{\mathcal{T}}$, it follows that every non-competing child node of $\vec{\beta}$ in $\mathcal{M}_{\mathcal{T}}$  is a child of $\vec{\beta}$ in $\mathcal{Q}$.
\end{enumerate}

\end{proof}

\begin{definition}
\label{def:split-nodes}
Let $\mathcal{Q}$ be a finite subtree of $\mathcal{M}_{\mathcal{T}}$. Then $\texttt{split}(\mathcal{Q})$ is the \emph{split of} $\mathcal{Q}$, that is, the tree returned by the {\tt split} algorithm shown in Figure~\ref{fig:split-algorithm}.
\end{definition}

Intuitively, the {\tt split} algorithm recursively transforms a given finite subtree $\mathcal{Q}$ of $\mathcal{M}_{\mathcal{T}}$ into a corresponding SPT $\texttt{split}( \mathcal{Q} )$, which has the same probability as the input SPT but with extra ``diamond'' and ``circular'' nodes inserted under each internal node of $\mathcal{Q}$ and whose purpose is to separate competitive nodes from their non-competitive siblings, respectively. The idea is that these extra nodes simulate an assembly sequence that first chooses whether a competitive or non-competitive tile will attach in the next tile attachment step and then chooses which particular tile attachment step of that type (competitive or non-competitive) to perform. The {\tt split} algorithm ensures that the probability of traversing from parent to child via an extra node in $\texttt{split}(\mathcal{Q})$ is equal to the probability of traversing from parent to child directly in $\mathcal{Q}$. Introducing this kind of artificial structure will allow us to more easily  manipulate the quantity $\textmd{Pr}[\mathcal{Q}]$ algebraically. See Figure~\ref{fig:split-high-level} for a sample run of {\tt split} being applied to the node $\vec{\beta}$ of $\mathcal{Q}$.

\begin{figure}[!h]
	\centerline{\includegraphics[width=\linewidth]{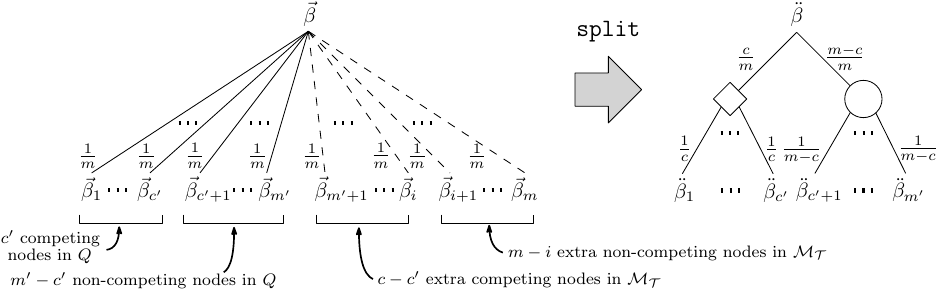}}
	\caption{\label{fig:split-high-level} A sample run of {\tt split} applied to the root node $\vec{\beta}$ of $\mathcal{Q}$. While the input to {\tt split} is $\mathcal{Q}$, the additional children of $\vec{\beta}$ in $\mathcal{M}_{\mathcal{T}}$ that are not in  $\mathcal{Q}$ (shown under dashed lines in the figure) are also accessible to the algorithm. As shown in the sub-figure on the left, node $\vec{\beta}$ has $m$ child nodes in $\mathcal{M}_{\mathcal{T}}$, $m'$ of which also belong to $\mathcal{Q}$. Among the $c$ child nodes of $\vec{\beta}$ that are competing in  $\mathcal{M}_{\mathcal{T}}$, $c'$ of them also belong to $\mathcal{Q}$. The number of children of $\vec{\beta}$ in each one of the four categories (i.e., competing or not, in $\mathcal{Q}$ or not) are indicated in the figure, with $i=m'+(c-c')$. In the output of {\tt split}, we decorate a non-diamond, non-circular node  with an umlaut. For instance, $\ddot{\beta}$ in {\tt split($Q$)} corresponds to $\vec{\beta}$ in $\mathcal{Q}$. Note that, for each integer $j$ in $[1,c']\cup[c'+1,m']$,  the probability of traversing from $\ddot{\beta}$ to $\ddot{\beta}_j$ in {\tt split($Q$)} is equal to the probability of traversing from $\vec{\beta}$ to $\vec{\beta}_j$ in $\mathcal{Q}$, namely $\frac{1}{m}$. }
\end{figure}

Note that the {\tt split} algorithm shown in Figure~\ref{fig:split-algorithm} is defined (implicitly) relative to $\mathcal{T}$, which we have been assuming to be $w$-sequentially non-deterministic. 
However, the algorithm works relative to an arbitrary TAS, sequentially non-deterministic or otherwise. 

\begin{figure}[!h]
  \centering
  \begin{minipage}[t]{\linewidth}
    \begin{algorithm}[H]
      \DontPrintSemicolon
      \SetKwFunction{Fmain}{split}
      \SetKwFunction{FgetChildren}{getChildrenIn$\mathcal{Q}$}
      \SetKwFunction{FnumChildMT}{getNumberOfChildrenIn$\mathcal{M}_{\mathcal{T}}$}
      \SetKwFunction{FnumChildQ}{getNumberOfChildrenIn$\mathcal{Q}$}
      \SetKwFunction{FnumCompChildQ}{getNumberOfCompetingChildrenIn$\mathcal{Q}$}         
      \SetKwFunction{FnumCompChildMT}{getNumberOfCompetingChildrenIn$\mathcal{M}_{\mathcal{T}}$}      
      \SetKwFunction{Fclone}{clone}     
      \SetKwFunction{FnewDiamond}{makeNewDiamondNode}
      \SetKwFunction{FnewCircle}{makeNewCircularNode}      
      \SetKwFunction{FisCompetitive}{isNodeCompeting}
      \SetKwFunction{FaddChild}{addChild} 
      \SetKwFunction{FgetRoot}{getRoot}                      
      \SetKwProg{Fn}{function}{:}{end}
      
      \Fn(\tcp*[f]{$\mathcal{Q}$ is any finite subtree of $\mathcal{M}_{\mathcal{T}}$}){\Fmain{$\mathcal{Q}$}}{
        
        $\vec{\beta} \leftarrow \FgetRoot{$\mathcal{Q}$}$ 
        \tcp*[f]{$\vec{\beta}$ is the root of $\mathcal{Q}$}
        
        $m' \leftarrow \FnumChildQ{$\vec{\beta}$}$ 
        
        \If(\tcp*[f]{determine whether $\vec{\beta}$ is a leaf node in $\mathcal{Q}$}){$m' = 0$}{
          \KwRet\ $\Fclone{$\vec{\beta}$}$
          \tcp*[f]{return a copy of this leaf node}
        }
        \Else(\tcp*[f]{$\vec{\beta}$ is an internal node}){
          $c \leftarrow \FnumCompChildMT{$\vec{\beta}$}$
          
          $m \leftarrow \FnumChildMT{$\vec{\beta}$}$

          $\Diamond \leftarrow \FnewDiamond{}$
          \tcp*[f]{create a new diamond child node}
          
          $\bigcirc \leftarrow \FnewCircle{}$
          \tcp*[f]{create a new circular child node}          
          
          \ForEach(\tcp*[f]{loop over the child nodes of $\vec{\beta}$} in $\mathcal{Q}$)
                  {$\vec{\beta}^{\prime} \in \FgetChildren{$\vec{\beta}$}$}{

            $\ddot{\beta'} \leftarrow \Fmain{$Q^{\vec{\beta}^{\prime}}$}$
            \tcp*[f]{recursively process the subtree rooted at $\vec{\beta}^{\prime}$}
            
            \If(\tcp*[f]{determine whether $\vec{\beta}^{\prime}$ is competing})
               {$\FisCompetitive{$\vec{\beta}^{\prime}$}$}{
              $\FaddChild{ $\Diamond$, $\ddot{\beta'}$,  $1/c$ }$
              \tcp*[f]{make $\ddot{\beta'}$ a new child of the diamond node}          
            }
            \Else{
              $\FaddChild{ $\bigcirc$, $\ddot{\beta'}$, $1/(m-c)$ }$
              \tcp*[f]{make $\ddot{\beta'}$ a new child of the circular node}
            }
          }
          $\ddot{\beta} \leftarrow \Fclone{$\vec{\beta}$}$
          \tcp*[f]{create a new root node identical to $\vec{\beta}$ but with no children}
          
          $c' \leftarrow \FnumCompChildQ{$\vec{\beta}$}$
          
          \If(\tcp*[f]{if needed, make diamond node a new child of
            new root node}){$c'>0$}{
            $\FaddChild{ $\ddot{\beta}$, $\Diamond$, $c/m$ }$
          }
          \If(\tcp*[f]{if needed, make circular node a new child of 
            new root node}){$m'-c'>0$}{
            $\FaddChild{ $\ddot{\beta}$, $\bigcirc$, $(m-c)/m$ }$
          }          
          \KwRet\ $\ddot{\beta}$  \hfill\tcp{return the new root node
            (and its fully-processed subtree)}
        }
      }
\end{algorithm}
\end{minipage}
\caption{\label{fig:split-algorithm} Algorithm to split a given finite subtree of $\mathcal{M}_{\mathcal{T}}$. A diamond node only has  competing children, whereas a circular node only has children that are not competing. Note that the input tree $\mathcal{Q}$ need not be full relative to $\mathcal{M}_{\mathcal{T}}$ (which is an additional, implicit argument of this algorithm).}
\end{figure}

\begin{observation}
\label{obs:prime-nodes-split-bijection}
There is a one-to-one correspondence between the set of all non-circular and non-diamond nodes of $\texttt{split}\left(\mathcal{Q}\right)$ and the set of all nodes of $\mathcal{Q}$.
\end{observation}

\begin{definition}
\label{def:prime-node-split}
If $\mathcal{Q}$ is a finite subtree of $\mathcal{P}$ and $\vec{\beta}$ is a node of $\mathcal{Q}$, then we denote as $\ddot{\beta}$ the corresponding (non-circular and non-diamond) \emph{split} node of $\vec{\beta}$ in $\texttt{split}\left(\mathcal{Q}\right)$.
\end{definition}

Observation~\ref{obs:prime-nodes-split-bijection} ensures that $\ddot{\beta}$ is well-defined in Definition~\ref{def:prime-node-split}. The following lemma proves some properties of  the output of the {\tt split} algorithm.

\begin{lemma} 
\label{lem:split-probability}
If $\mathcal{Q}$ is a finite subtree of $\mathcal{M}_{\mathcal{T}}$, then the following \emph{lemma conclusions} are true:
\begin{enumerate}
	\item \label{split-conclusion-1} $\texttt{split}\left(\mathcal{Q}\right)$ is an SPT.
	\item \label{split-conclusion-2} For every diamond node $\Diamond$ of $\texttt{split}\left(\mathcal{Q}\right)$ with parent $\ddot{\rho}$, if $\vec{\rho}$ has $c'$ children in $\mathcal{Q}$ that are competing in $\mathcal{M}_{\mathcal{T}}$ and $c$  competing children in $\mathcal{M}_{\mathcal{T}}$, then $S_{\Diamond} = \frac{c'}{c}$.
	\item \label{split-conclusion-3} For every circular node $\bigcirc$ of $\texttt{split}\left(\mathcal{Q}\right)$ with parent $\ddot{\rho}$, if $\vec{\rho}$ has $n'$ children in $\mathcal{Q}$ that are not competing in $\mathcal{M}_{\mathcal{T}}$ and $n$ children in $\mathcal{M}_{\mathcal{T}}$ that are not  competing, then $S_{\bigcirc} = \frac{n'}{n}$.
	\item \label{split-conclusion-4} Every internal node of $\texttt{split}\left(\mathcal{Q}\right)$ that is neither a diamond nor circular is normalized.
	\item \label{split-conclusion-5} Every leaf node of $\texttt{split}\left(\mathcal{Q}\right)$ is neither a circular nor a diamond node.
	\item \label{split-conclusion-6} $\textmd{Pr}\left[ \mathcal{Q} \right] = \textmd{Pr}\left[ \texttt{split}(\mathcal{Q})\right]$.
\end{enumerate}
\end{lemma}

\begin{proofsketch}
  In the non-vacuous, we proceed by structural induction.  We let
  $\vec{\beta}$ be an arbitrary internal node of $\mathcal{Q}$ and
  assume that the lemma's conclusions hold for $\texttt{split}\left(
  \mathcal{Q}^{\vec{\beta}'} \right)$, where $\vec{\beta}'$ is an
  arbitrary child node of $\vec{\beta}$ in $\mathcal{Q}$. Then we
  complete the following steps:
\begin{enumerate}[label={\bf Step \theenumi.},itemindent=6mm]
\item We show that $\texttt{split}\left( \mathcal{Q}^{\vec{\beta}}
  \right)$ is a rooted, weighted tree such that every one of its leaf
  nodes is neither a circular nor a diamond node, thereby partially
  establishing the lemma's conclusion~\ref{split-conclusion-1} and
  fully establishing the lemma's conclusion~\ref{split-conclusion-5}
  for $\texttt{split}\left( \mathcal{Q}^{\vec{\beta}} \right)$.

  \item We prove that the lemma's conclusion~\ref{split-conclusion-2}
    holds for $\texttt{split}\left( \mathcal{Q}^{\vec{\beta}}
    \right)$.
  \item We prove that the lemma's conclusion~\ref{split-conclusion-3}
    holds for $\texttt{split}\left( \mathcal{Q}^{\vec{\beta}}
    \right)$.
  \item We complete the proof of the lemma's conclusion~\ref{split-conclusion-1}
    for $\texttt{split}\left( \mathcal{Q}^{\vec{\beta}}
    \right)$.
  \item We prove that the lemma's conclusion~\ref{split-conclusion-4}
    holds for $\texttt{split}\left( \mathcal{Q}^{\vec{\beta}}
    \right)$.
  \item We prove that the lemma's conclusion~\ref{split-conclusion-6}
    holds for $\texttt{split}\left( \mathcal{Q}^{\vec{\beta}}
    \right)$.
\end{enumerate}\vspace*{-5mm}
\end{proofsketch}

\begin{proof}
  If $\mathcal{Q}$ is a single node, then line 5 of $\texttt{split}$ returns a copy of $\mathcal{Q}$, and the lemma trivially holds.
Going forward, assume that $\mathcal{Q}$ comprises more than one node.
We proceed by structural induction.
Let $\vec{\beta}$ be  an arbitrary internal node of $\mathcal{Q}$  and $B_{\vec{\beta}}$ be the non-empty set of children of $\vec{\beta}$ in $\mathcal{Q}$.
For all $\vec{\beta}' \in B_{\vec{\beta}}$, we assume:
\begin{enumerate}
	\item \label{split-ih-1} $\texttt{split}\left(\mathcal{Q}^{\vec{\beta}'}\right)$ is an SPT.
		
	\item \label{split-ih-2} For every diamond node $\Diamond$ of $\texttt{split}\left(\mathcal{Q}^{\vec{\beta}'}\right)$ with parent $\ddot{\rho}$, if $\vec{\rho}$ has $c'$ children in $\mathcal{Q}$ that are competing in $\mathcal{M}_{\mathcal{T}}$ and $c$  competing children in $\mathcal{M}_{\mathcal{T}}$, then $S_{\Diamond} = \frac{c'}{c}$.

	\item \label{split-ih-3} For every circular node $\bigcirc$ of $\texttt{split}\left(\mathcal{Q}^{\vec{\beta}'}\right)$ with parent $\ddot{\rho}$, if $\vec{\rho}$ has $n'$ children in $\mathcal{Q}$ that are not competing in $\mathcal{M}_{\mathcal{T}}$ and $n$ children in $\mathcal{M}_{\mathcal{T}}$ that are not competing, then $S_{\bigcirc} = \frac{n'}{n}$.
	\item \label{split-ih-4} Every internal node of $\texttt{split}\left(\mathcal{Q}^{\vec{\beta}'}\right)$ that is neither a diamond nor circular is normalized.
	\item \label{split-ih-5} Every leaf node of $\texttt{split}\left(\mathcal{Q}^{\vec{\beta}'}\right)$ is neither a circular nor a diamond node.
	\item \label{split-ih-6} $\textmd{Pr}\left[ \mathcal{Q}^{\vec{\beta}'} \right] = \textmd{Pr}\left[ \texttt{split}\left(\mathcal{Q}^{\vec{\beta}'}\right) \right]$.
\end{enumerate}
The six previous items make up the inductive hypothesis.
{\bf Step 1.} We show that $\texttt{split}\left( \mathcal{Q}^{\vec{\beta}} \right)$ is a rooted, weighted tree such that every one of its leaf nodes is neither a circular nor a diamond node.
By inductive hypothesis items~\ref{split-ih-1} and~\ref{split-ih-5}, for all $\vec{\beta}' \in B_{\vec{\beta}}$, $\texttt{split}\left( \mathcal{Q}^{\vec{\beta}'} \right)$ is an SPT such that every one of its leaf nodes is neither a circular nor a diamond node.
Since $\vec{\beta}'$ is either a competing node or not, lines 13 through 16 of \texttt{split} ensure that $\ddot{\beta'}$ is either a child of $\Diamond$ or $\bigcirc$, but not both.
Then, lines 19 through 22 of $\texttt{split}$ make either $\Diamond$ or $\bigcirc$ (or both) children of $\ddot{\beta}$ in  $\texttt{split}\left( \mathcal{Q}^{\vec{\beta}} \right)$, which is the root of  $\texttt{split}\left( \mathcal{Q}^{\vec{\beta}} \right)$.
Moreover, if no child is added to $\Diamond$ (resp., $\bigcirc$) on line 14 (resp., 16), then $\Diamond$ (resp., $\bigcirc$) is not added as a child of $\ddot{\beta}$ on line 20 (resp., 22). 
Thus,  $\texttt{split}\left( \mathcal{Q}^{\vec{\beta}} \right)$ is a rooted, weighted tree such that every one of its leaf nodes is neither a circular nor a diamond node.
{\bf Step 2.} Let $m'_{\texttt{split}}$, $c_{\texttt{split}}$, $m_{\texttt{split}}$, $\Diamond_{\texttt{split}}$, $\bigcirc_{\texttt{split}}$, and $c'_{\texttt{split}}$ denote the variables $m'$, $c$, $m$, $\Diamond$, $\bigcirc$, and $c'$ that are declared on lines 3, 7, 8, 9, 10, and 18 of {\tt split}, respectively.
Since $\mathcal{Q}$ comprises more than one node, and every one of them is either competing or not, either $c'_{\texttt{split}} > 0$ or $m'_{\texttt{split}} - c'_{\texttt{split}} > 0$.
Additionally, $c_{\texttt{split}} \leq m_{\texttt{split}}$ and $c'_{\texttt{split}} \leq m'_{\texttt{split}}$.
Since $\mathcal{Q}$ is a subtree of $\mathcal{M}_{\mathcal{T}}$, we also have $m'_{\texttt{split}} \leq m_{\texttt{split}}$,  $c'_{\texttt{split}} \leq c_{\texttt{split}}$, and $m'_{\texttt{split}} -  c'_{\texttt{split}} \leq m_{\texttt{split}} -  c_{\texttt{split}}$. 

We now prove the lemma's conclusion~\ref{split-conclusion-2} for  $\texttt{split}\left(\mathcal{Q}^{\vec{\beta}}\right)$.
Assume $c'_{\texttt{split}} > 0$, which means that the number of times line 14 of {\tt split} executes is $c'_{\texttt{split}}$.
By inductive hypothesis item~\ref{split-ih-2}, for every diamond node $\Diamond$ of $\texttt{split}\left( \mathcal{Q}^{\vec{\beta}'} \right)$, with parent $\ddot{\rho}$, if $\vec{\rho}$ has $c'$ children in $\mathcal{Q}$ that are competing in $\mathcal{M}_{\mathcal{T}}$ and $c$ nodes in $\mathcal{M}_{\mathcal{T}}$ that are  competing, then $S_\Diamond = \frac{c'}{c}$.
Then, after the foreach loop on line 11 of {\tt split} executes, we have
$$
S_{\Diamond_{\texttt{split}}} = c'_{\texttt{split}} \cdot \frac{1}{c_{\texttt{split}}} = \frac{c'_{\texttt{split}}}{c_{\texttt{split}}}.
$$ 
Since $c'_{\texttt{split}} > 0$, line 20 of {\tt split} adds $\Diamond_{\texttt{split}}$ to $\texttt{split}\left(\mathcal{Q}^{\vec{\beta}}\right)$ as a child of $\ddot{\beta}$. 
Thus, when line 23 of {\tt split} executes, the lemma's conclusion~\ref{split-conclusion-2} is established for $\texttt{split}\left( \mathcal{Q}^{\vec{\beta}} \right)$.
{\bf Step 3.} We now prove the lemma's conclusion~\ref{split-conclusion-3}  for $\texttt{split}\left( \mathcal{Q}^{\vec{\beta}} \right)$.
Assume $m'_{\texttt{split}} - c'_{\texttt{split}} > 0$, which means that the number of times line 16 of {\tt split} executes is $m'_{\texttt{split}} - c'_{\texttt{split}}$.
Let $n'_{\texttt{split}}$ denote the number of child nodes of $\vec{\beta}$ in $\mathcal{Q}$ that are not competing  and $n_{\texttt{split}}$ denote the number of  child nodes of $\vec{\beta}$ in $\mathcal{M}_{\mathcal{T}}$ that are not competing.
Since a node in $\mathcal{M}_{\mathcal{T}}$ is either competing or not, it follows that
$n'_{\texttt{split}} = m'_{\texttt{split}} - c'_{\texttt{split}}$ and $n_{\texttt{split}} = m_{\texttt{split}} - c_{\texttt{split}}$. 
By inductive hypothesis item~\ref{split-ih-3}, for every circular node $\bigcirc$ of $\texttt{split}\left(\mathcal{Q}^{\vec{\beta}'}\right)$, with parent $\ddot{\rho}$, if $\vec{\rho}$ has $n'$ children in $\mathcal{Q}$ that are not  competing in $\mathcal{M}_{\mathcal{T}}$, and $n$ children in $\mathcal{M}_{\mathcal{T}}$ that are not competing, then $S_\bigcirc = \frac{n'}{n}$.
Then, after the foreach loop on line 11 of {\tt split} executes, we have 
$$
S_{\bigcirc} = \left(m'_{\texttt{split}} - c'_{\texttt{split}}\right) \cdot \frac{1}{(m_{\texttt{split}} - c_{\texttt{split}})} = \frac{n'_{\texttt{split}}}{n_{\texttt{split}}}.
$$
Since $m'_{\texttt{split}} - c'_{\texttt{split}} > 0$, line 22 of {\tt split} adds $\bigcirc$ to $\texttt{split}\left(\mathcal{Q}^{\vec{\beta}}\right)$ as a child of $\ddot{\beta}$.
Thus, when line 23 of {\tt split} executes, lemma conclusion~\ref{split-conclusion-3} is established  for $\texttt{split}\left( \mathcal{Q}^{\vec{\beta}} \right)$.
{\bf Step 4.} We now prove the lemma's conclusion~\ref{split-conclusion-1}  for $\texttt{split}\left( \mathcal{Q}^{\vec{\beta}} \right)$.
We first establish that the labels of all edges in
$\texttt{split}\left(\mathcal{Q}^{\vec{\beta}}\right)$ are in $(0,1]$.

If the diamond node $\Diamond_{\texttt{split}}$ created on line 9 is added as a child of $\ddot{\beta}$ on line 20 (which implies $c_{\texttt{split}}\geq c'_{\texttt{split}}\geq 1$), then the label of the corresponding edge is $\frac{c_{\texttt{split}}}{m_{\texttt{split}}}\leq 1$. 
Furthermore, the label of the edge from $\Diamond_{\texttt{split}}$ to any one of its children is equal to $\frac{1}{c_{\texttt{split}}}$, which is also less than or equal to 1, since $c_{\texttt{split}}\geq 1$.
Finally, each child of  $\Diamond_{\texttt{split}}$ added on line 14 is,  by inductive hypothesis item~\ref{split-ih-1}, an SPT and thus satisfies the required condition on all of its edge labels.

Similarly, if the circular node $\bigcirc_{\texttt{split}}$ created on line 10 is added as a child of $\ddot{\beta}$ on line 22 (which implies $m_{\texttt{split}}-c_{\texttt{split}} \geq m'_{\texttt{split}}-c'_{\texttt{split}}\geq 1$), then the label of the corresponding edge is $\frac{m_{\texttt{split}}-c_{\texttt{split}}}{m_{\texttt{split}}}\leq 1$. 
Furthermore, the label of the edge from $\bigcirc_{\texttt{split}}$ to any one of its children is equal to $\frac{1}{m_{\texttt{split}}-c_{\texttt{split}}}$, which is also less than or equal to 1, since $m_{\texttt{split}}-c_{\texttt{split}}\geq 1$.
Finally, each child of  $\bigcirc_{\texttt{split}}$ added on line 16 is,  by inductive hypothesis item~\ref{split-ih-1}, an SPT and thus satisfies the required condition on all of its edge labels.

Second, we prove that, for each node $v$  in $\texttt{split}\left(\mathcal{Q}^{\vec{\beta}}\right)$, $S_v\leq 1$.

If the diamond node $\Diamond_{\texttt{split}}$ created on line 9 is added as a child of $\ddot{\beta}$ on line 20, then we already established (in the proof of the lemma's conclusion~\ref{split-conclusion-2} above) that $S_{\Diamond_{\texttt{split}}} = \frac{c'_{\texttt{split}}}{c_{\texttt{split}}}$, which is less than or equal to 1.
Furthermore, by inductive hypothesis item~\ref{split-ih-1}, every child node of  $\Diamond_{\texttt{split}}$ is an SPT, which means that, for each one of its nodes $v$, $S_v \leq 1$.
Similarly, if the circular node $\bigcirc_{\texttt{split}}$ created on line 10 is added as a child of $\ddot{\beta}$ on line 22, then we already established (in the proof of the lemma's conclusion~\ref{split-conclusion-3} above) that $S_{\bigcirc_{\texttt{split}}} = \frac{m'_{\texttt{split}}-c'_{\texttt{split}}}{m_{\texttt{split}}-c_{\texttt{split}}}$, which is less than or equal to 1. 
Furthermore, by inductive hypothesis item~\ref{split-ih-1}, every child node of  $\bigcirc_{\texttt{split}}$ is an SPT, which means that, for each one of its nodes $v$, $S_v \leq 1$.
We now only have to consider $\ddot{\beta}$, the root node of
$\texttt{split}\left(\mathcal{Q}^{\vec{\beta}}\right)$. Since
$\ddot{\beta}$ has either one or two children added on lines~20
and/or~22 of the algorithm, when line~23 executes, we have:

$$
S_{\ddot{\beta}} = \frac{c_{\texttt{split}}}{m_{\texttt{split}}} + \frac{m_{\texttt{split}} - c_{\texttt{split}}}{m_{\texttt{split}}} = 1,
$$ 
which means that $\ddot{\beta}$ is normalized, and thus we trivially have $S_{\ddot{\beta}} \leq 1$.
In conclusion, the tree  $\texttt{split}\left(\mathcal{Q}^{\vec{\beta}}\right)$ returned on line 23 of the algorithm is an SPT,  which establishes the lemma's conclusion~\ref{split-conclusion-1} for that tree.
{\bf Step 5.} The lemma's conclusion~\ref{split-conclusion-4} holds for $\texttt{split}\left( \mathcal{Q}^{\vec{\beta}} \right)$ because, as we just established, the root node of $\texttt{split}\left( \mathcal{Q}^{\vec{\beta}} \right)$ is normalized and, by inductive hypothesis item~\ref{split-ih-4}, every other internal node of $\texttt{split}\left( \mathcal{Q}^{\vec{\beta}} \right)$ that is neither diamond nor circular is normalized.
{\bf Step 6.} We now prove that $\textmd{Pr}\left[ \mathcal{Q}^{\vec{\beta}} \right] = \textmd{Pr}\left[ \texttt{split}\left(\mathcal{Q}^{\vec{\beta}}\right)\right]$, thereby establishing lemma conclusion~\ref{split-conclusion-6}.
Note that $B_{\vec{\beta}}$ is trivially a bottleneck of $\mathcal{Q}^{\vec{\beta}}$ because it is the set of children of ${\vec{\beta}}$.
Thus, we have:
\[
\begin{array}{llll}
\textmd{Pr}\left[\mathcal{Q}^{\vec{\beta}}\right] & = & \displaystyle\sum_{\vec{\beta}' \in B_{\vec{\beta}}}{\textmd{Pr}_{\mathcal{Q}^{\vec{\beta}}}\left[ \vec{\beta}' \right] \cdot \textmd{Pr}\left[ \mathcal{Q}^{\vec{\beta}'} \right] } & \textmd{ Lemma~\ref{lem:alternative-characterization-of-pr} with } B = B_{\vec{\beta}}    \\
				& = & \displaystyle\sum_{\vec{\beta}' \in B_{\vec{\beta}}}{\textmd{Pr}_{\mathcal{Q}^{\vec{\beta}}}\left[ \vec{\beta}' \right] \cdot \textmd{Pr}\left[ \texttt{split}\left( \mathcal{Q}^{\vec{\beta}'} \right) \right] } & \textmd{ inductive hypothesis item \ref{split-ih-6}}   \\
				& = & \displaystyle\sum_{\vec{\beta}' \in B_{\vec{\beta}}}{\textmd{Pr}_{\mathcal{Q}^{\vec{\beta}}}\left[ \vec{\beta} \right] \cdot \frac{1}{m_{\texttt{split}}} \cdot \textmd{Pr}\left[ \texttt{split}\left( \mathcal{Q}^{\vec{\beta}'} \right) \right] } & \textmd{ Definition of } \mathcal{M}_{\mathcal{T}} \textmd{ and } \textmd{Pr}_{\mathcal{Q}^{\vec{\beta}}}\left[ \vec{\beta}' \right] \\
				& = & \displaystyle\sum_{\vec{\beta}' \in B_{\vec{\beta}}}{1\cdot \frac{1}{m_{\texttt{split}}} \cdot \textmd{Pr}\left[ \texttt{split}\left( \mathcal{Q}^{\vec{\beta}'} \right) \right] } & \textmd{ } \vec{\beta} \textmd{ is the root of } \mathcal{Q}^{\vec{\beta}}.
\end{array}
\]
On the one hand, if $\vec{\beta}'$ is a competing node in $\mathcal{M}_{\mathcal{T}}$, then $c_{\texttt{split}} > 0$, and by lines 14 and 20 of \texttt{split}, we have 
\begin{equation}
\label{eqn:pr-split-3}
\frac{1}{m_{\texttt{split}}} = \frac{c_{\texttt{split}}}{m_{\texttt{split}}} \cdot \frac{1}{c_{\texttt{split}}} = 
\textmd{Pr}_{\texttt{split}\left(\mathcal{Q}^{\vec{\beta}}\right)}\left[ \ddot{\beta}' \right]. 
\end{equation}
On the other hand, if $\vec{\beta}'$ is not a  competing node in $\mathcal{M}_{\mathcal{T}}$, then $m_{\texttt{split}} - c_{\texttt{split}} > 0$, and by lines 16 and 22 of \texttt{split}, we have
\begin{equation}
\label{eqn:pr-split-4}
\frac{1}{m_{\texttt{split}}} = \frac{1}{(m_{\texttt{split}}-c_{\texttt{split}})} \cdot \frac{(m_{\texttt{split}}-c_{\texttt{split}})}{m_{\texttt{split}}} = 
\textmd{Pr}_{\texttt{split}\left(\mathcal{Q}^{\vec{\beta}}\right)}\left[ \ddot{\beta}' \right].
\end{equation}
Moreover, since $\ddot{\beta}'$ is the root of $\texttt{split}\left(\mathcal{Q}^{\vec{\beta}'}\right)$ in $\texttt{split}\left(\mathcal{Q}^{\vec{\beta}}\right)$, it follows that:
\begin{equation}
\label{eqn:pr-split-5}
\textmd{Pr}\left[ \texttt{split}\left( \mathcal{Q}^{\vec{\beta}'} \right) \right]= \textmd{Pr}\left[ \texttt{split}\left( \mathcal{Q}^{\vec{\beta}} \right)^{\ddot{\beta}'} \right].
\end{equation}
Then, since every node in $\mathcal{M}_{\mathcal{T}}$ is either  competing or not, we have:
\[
\begin{array}{llll}
\textmd{Pr}\left[\mathcal{Q}^{\vec{\beta}}\right] & = & \displaystyle\sum_{\vec{\beta}' \in B_{\vec{\beta}}}{\frac{1}{m_{\texttt{split}}} \cdot \textmd{Pr}\left[ \texttt{split}\left( \mathcal{Q}^{\vec{\beta}'} \right) \right]} &    \\
					& = & \displaystyle \sum_{\vec{\beta}' \in B_{\vec{\beta}}}{\textmd{Pr}_{\texttt{split}\left(\mathcal{Q}^{\vec{\beta}}\right)}\left[ \ddot{\beta}' \right] \cdot \textmd{Pr}\left[ \texttt{split}\left( \mathcal{Q}^{\vec{\beta}'} \right) \right]} & \textmd{ Equations (\ref{eqn:pr-split-3}) and (\ref{eqn:pr-split-4}) } \\
					
					& = & \displaystyle \sum_{\vec{\beta}' \in B_{\vec{\beta}}}{\textmd{Pr}_{\texttt{split}\left(\mathcal{Q}^{\vec{\beta}}\right)}\left[ \ddot{\beta}' \right] \cdot \textmd{Pr}\left[ \texttt{split}\left(\mathcal{Q}^{\vec{\beta}}\right)^{\ddot{\beta}'} \right]} & \textmd{ Equation }(\ref{eqn:pr-split-5}) \\
\end{array}
\]

Let $\ddot{B}_{\vec{\beta}} = \left\{ \ddot{\beta}' \mid \vec{\beta}' \in B_{\vec{\beta}} \right\}$ be the set of nodes of ${\tt split}\left( \mathcal{Q}^{\vec{\beta}} \right)$ corresponding to the children of $\vec{\beta}$ in $\mathcal{Q}^{\vec{\beta}}$. 
By Observation~\ref{obs:prime-nodes-split-bijection}, there is a one-to-one correspondence between the elements of $B_{\vec{\beta}}$ and those of $\ddot{B}_{\vec{\beta}}$.
Thus:
\begin{equation}
    \label{eqn:split-probability-lemma-next-to-last}
\displaystyle\textmd{Pr}\left[\mathcal{Q}^{\vec{\beta}}\right] = \sum_{\vec{\beta}' \in B_{\vec{\beta}}}{\textmd{Pr}_{\texttt{split}\left(\mathcal{Q}^{\vec{\beta}}\right)}\left[ \ddot{\beta}' \right] \cdot \textmd{Pr}\left[ \texttt{split}\left( \mathcal{Q}^{\vec{\beta}} \right)^{\ddot{\beta}'} \right]} = \sum_{\ddot{\beta}' \in \ddot{B}_{\vec{\beta}}}{\textmd{Pr}_{\texttt{split}\left(\mathcal{Q}^{\vec{\beta}}\right)}\left[ \ddot{\beta}' \right] \cdot \textmd{Pr}\left[ \texttt{split}\left( \mathcal{Q}^{\vec{\beta}} \right)^{\ddot{\beta}'} \right]}
\end{equation}
We now prove that $\ddot{B}_{\vec{\beta}}$ is a bottleneck of ${\tt split}\left(\mathcal{Q}^{\vec{\beta}}\right)$.
Assume, for the sake of obtaining a contradiction, that $\ddot{B}_{\vec{\beta}}$ is not a bottleneck of ${\tt split}\left(\mathcal{Q}^{\vec{\beta}}\right)$.
Then, there exists a  maximal path $\pi$ in ${\tt split}\left(\mathcal{Q}^{\vec{\beta}}\right)$ such that no $\ddot{\beta}' \in \ddot{B}_{\vec{\beta}}$ belongs to $\dom{\pi}$. 
But, by lines 19 through 22 of {\tt split}, $\Diamond$ and $\bigcirc$ are the only children of $\ddot{\beta}$ in ${\tt split}\left(\mathcal{Q}^{\vec{\beta}}\right)$.
This means that either $\Diamond \in \dom{\pi}$ or $\bigcirc \in \dom{\pi}$.  
Furthermore, by definition, each node $\ddot{\beta}'$ of  $\ddot{B}_{\vec{\beta}}$ corresponds to a child  $\vec{\beta}'$ of $\vec{\beta}$ in $\mathcal{Q}^{\vec{\beta}}$.
Since each $\vec{\beta}'$ is either competing or not in $\mathcal{M}_{\mathcal{T}}$, by lines 13 through 16 of {\tt split}, each $\ddot{\beta}'$ in  $\ddot{B}_{\vec{\beta}}$ is either a child of $\Diamond$ or $\bigcirc$.
Therefore, the domain of $\pi$ must contain one element of $\ddot{B}_{\vec{\beta}}$.
Since we reached a contradiction, it follows that $\ddot{B}_{\vec{\beta}}$ is a bottleneck of ${\tt split}\left(\mathcal{Q}^{\vec{\beta}}\right)$. 
Thus, Lemma~\ref{lem:alternative-characterization-of-pr} applied to ${\tt split}\left(\mathcal{Q^{\vec{\beta}}}\right)$ and $\ddot{B}_{\vec{\beta}}$ gives:
\begin{equation}
  \label{eqn:split-probability-lemma-last}
\displaystyle\sum_{\ddot{\beta}' \in \ddot{B}_{\vec{\beta}}}{\textmd{Pr}_{\texttt{split}\left(\mathcal{Q}^{\vec{\beta}}\right)}\left[ \ddot{\beta}' \right] \cdot \textmd{Pr}\left[ \texttt{split}\left( \mathcal{Q}^{\vec{\beta}} \right)^{\ddot{\beta}'} \right]} = \textmd{Pr}\left[ {\tt split}\left(\mathcal{Q}^{\vec{\beta}}\right) \right]
\end{equation}
Finally, the equality resulting from chaining \ref{eqn:split-probability-lemma-next-to-last} and \ref{eqn:split-probability-lemma-last} together proves the lemma's conclusion~\ref{split-conclusion-6} for $\texttt{split}\left( \mathcal{Q}^{\vec{\beta}} \right)$.
\end{proof}

The next lemma says that the tile attachments that should not affect the probability of a competition in fact do not. 
\begin{lemma} 
\label{lem:subtree-at-x-i-to-y-i-probability}
If $\mathcal{C}$ is not rigged by $\vec{\beta}_{\vec{x}}$, then $\textmd{Pr}\left[ \mathcal{P}^{\vec{\beta}_{\vec{x}}} \upharpoonright \vec{y} \right] = \textmd{Pr}\left[ \mathcal{C} \right]$.
\end{lemma}

Note that the left-hand side of the equality in Lemma~\ref{lem:subtree-at-x-i-to-y-i-probability} is the probability of $\mathcal{C}$ (i.e., the sum of probabilities of all the winning $\mathcal{T}$-assembly sequences that start at $\vec{x}$ and terminate at $\vec{y}$) computed by taking into account all possible tile attachments, including those at points along conmpeting paths corresponding to an ``earlier'' POC (i.e., a POC at which a tile has already been placed) and those points that do not  belong to any competing path.

\begin{proofsketch} To prove this lemma, we
  complete the following steps:
\begin{enumerate}[label={\bf Step \theenumi.},itemindent=6mm]
  \item We build, from $\mathcal{P}^{\vec{\beta}_{\vec{x}}} \upharpoonright \vec{y}$, a forest containing one tree $\mathcal{Q}_{\vec{\gamma}}$ for each winning assembly sequence $\vec{\gamma}$ for $\mathcal{C}$ in such a way that $\textmd{Pr}\left[ \mathcal{P}^{\vec{\beta}_{\vec{x}}} \upharpoonright \vec{y}\right]$ is equal to the sum over all such $\vec{\gamma}$s of $\textmd{Pr}\left[\mathcal{Q}_{\vec{\gamma}}\right]$.
  \item We use conclusion~\ref{split-conclusion-6} of Lemma~\ref{lem:split-probability} to infer that $\textmd{Pr}\left[ \mathcal{P}^{\vec{\beta}_{\vec{x}}} \upharpoonright \vec{y} \right]$ is equal to the sum over all $\vec{\gamma}$s of $\textmd{Pr}\left[ \texttt{split}\left( \mathcal{Q}_{\vec{\gamma}} \right) \right]$.
  \item We show that every internal non-diamond node of \texttt{split}$\left( \mathcal{Q}_{\vec{\gamma}} \right)$ is normalized.
  \item We build the SPT $\mathcal{Q}'_{\vec{\gamma}}$ that results from setting the probability of the outgoing edge on every diamond node in \texttt{split}$\left( \mathcal{Q}_{\vec{\gamma}} \right)$ to 1 and establish that $\textmd{Pr}\left[ \texttt{split}\left( \mathcal{Q}_{\vec{\gamma}} \right) \right] = \displaystyle \frac{1}{2^{\left| \vec{\gamma} \right| - 1}} \cdot \textmd{Pr}\left[ \mathcal{Q}'_{\vec{\gamma}} \right]$.
    \item We prove that every internal node of $\mathcal{Q}'_{\vec{\gamma}}$ is normalized and use Corollary~\ref{cor:normalized-probability-1} to infer that $\textmd{Pr}[ \texttt{split}\left( \mathcal{Q}_{\vec{\gamma}} \right) ] = \frac{1}{2^{\left| \vec{\gamma} \right| - 1}}=\textmd{Pr}_{\mathcal{C}}\left[ \vec{\gamma} \right]$, from which the lemma follows.
  \end{enumerate}
\end{proofsketch}
\begin{proof}
We use the example TAS $\mathcal{T}_{24}$ shown in Figure~\ref{fig:intro-example2-inside-proof} to illustrate the key steps in the proof of this lemma.
\begin{figure}[!h]
  \centering
  \includegraphics[width=.65\linewidth]{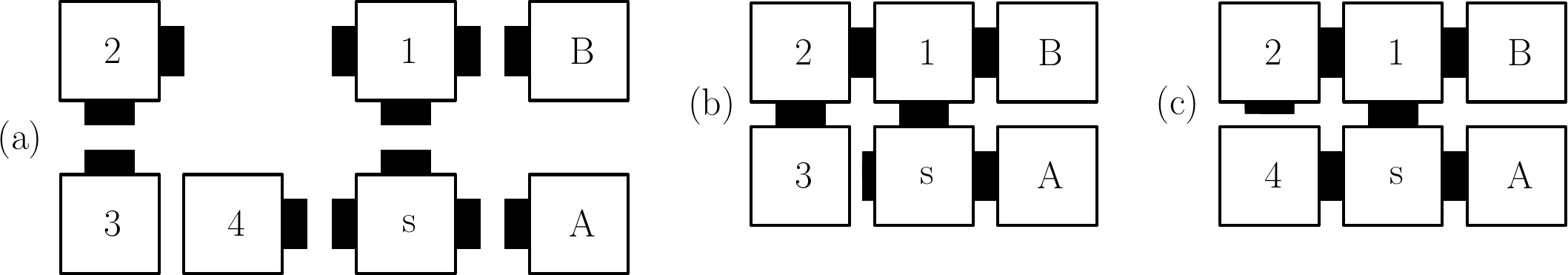}
  \caption{\label{fig:intro-example2-inside-proof} A TAS $\mathcal{T}_{24}$ with the tile set shown in (a), in which the temperature is 1 and s is the seed tile. (b) and (c) show the respective results of two different $\mathcal{T}_{24}$-producing assembly sequences. Note that this is a duplication of Figure~\ref{fig:intro-example2}. }
\end{figure}
Assume that $\alpha_4$ is the assembly depicted in Figure~\ref{fig:intro-example2-inside-proof}c, and $w$ is such that $\alpha_4$ is the unique $w$-correct $\mathcal{T}_{24}$-terminal assembly.
The SPT $\mathcal{P}_{w}$ corresponding to $\mathcal{T}_{24}$ is shown in Figure~\ref{fig:P_4}.
\begin{figure}[h!]
	\centering
	\includegraphics[width=\linewidth]{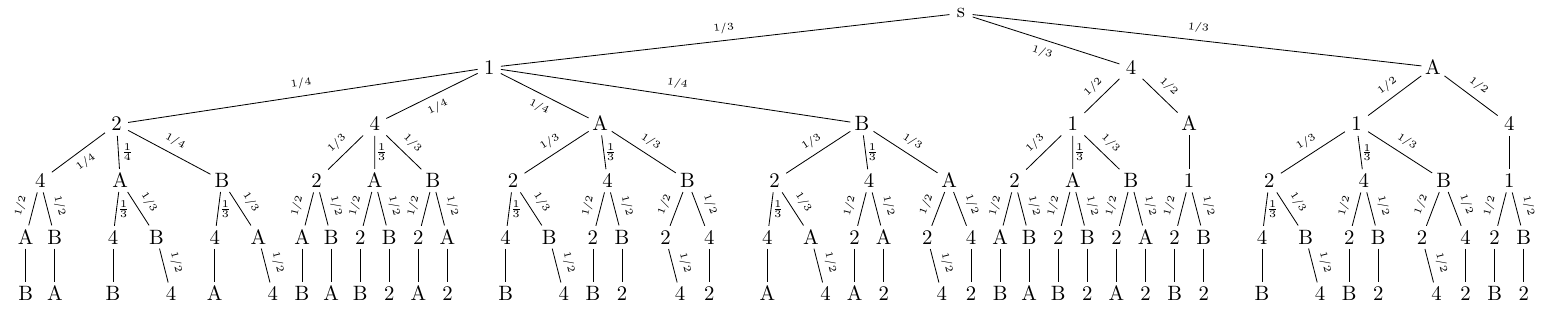}
	\caption{\label{fig:P_4} The SPT $\mathcal{P}_{w}$ that contains all (and only) maximal paths in $\mathcal{M}_{\mathcal{T}_{24}}$ that correspond to $\mathcal{T}_{24}$-producing assembly sequences resulting in the unique $w$-correct $\mathcal{T}_{24}$-terminal assembly $\alpha_4$. }
\end{figure}
By Lemma~\ref{lem:b-x-y-well-defined-leaf-nodes}, $\mathcal{P}^{\vec{\beta}_{\vec{x}}} \upharpoonright \vec{y}$ is well-defined and all of its leaf nodes are  nodes of $\mathcal{P}$ that terminate at $\vec{y}$.
If $\vec{x}$ is the unique point in the domain of the seed assembly of $\mathcal{T}_{24}$ (i.e., the point at which the s tile is placed) and $\vec{y}$ is the point at which the 4 tile is placed by $\alpha_4$, then Figure~\ref{fig:P_4_restricted_to_y} shows $\mathcal{P}^{\vec{\beta}_{\vec{x}}} \upharpoonright \vec{y}$. 
\begin{figure}[h!]
	\centering
	\includegraphics[width=\linewidth]{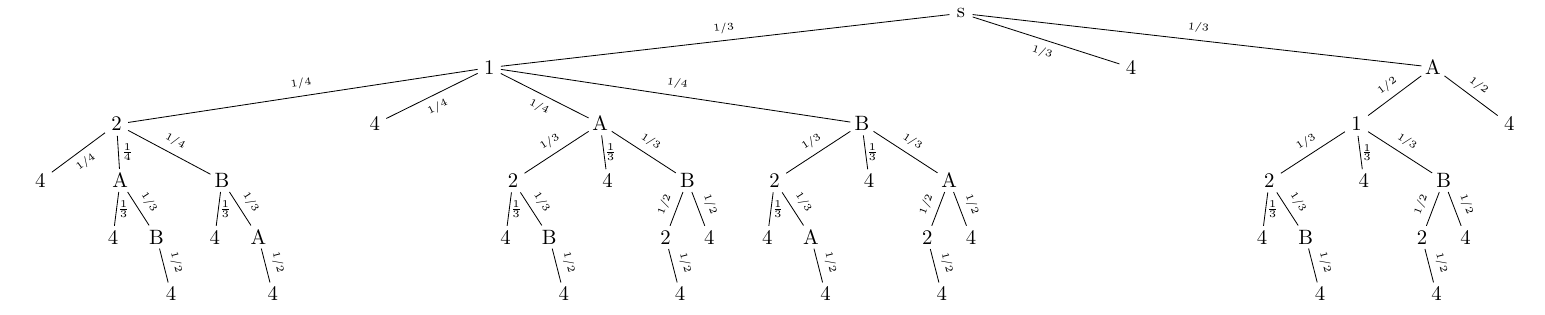}
	\caption{\label{fig:P_4_restricted_to_y} A depiction of $\mathcal{P}^{\vec{\beta}_{\vec{x}}} \upharpoonright \vec{y}$, where $\vec{x}$ is the unique point in the domain of the seed assembly of $\mathcal{T}_{24}$ (at which the s tile is placed) and $\vec{y}$ is the point at which the 4 tile is placed by $\alpha_4$.  }
\end{figure}

{\bf Step 1.} For each winning assembly sequence $\vec{\gamma}$ for $\mathcal{C}$, define $\mathcal{Q}_{\vec{\gamma}}$ to be the unique subtree of $\mathcal{P}^{\vec{\beta}_{\vec{x}}} \upharpoonright \vec{y}$ whose set of leaf nodes is $B_{\vec{\gamma}}$ (see  Definition~\ref{def:longest-winning-assembly-sequences}).
For example, let $\vec{\gamma}_{124}$ be the $\mathcal{T}_{24}$-producing assembly sequence that attaches tiles 1, 2, and 4 in that order.
Figure~\ref{fig:Q_gamma124} depicts $Q_{\vec{\gamma}_{124}}$ in our example. 
\begin{figure}[h!]
	\centering
	\includegraphics[width=\linewidth]{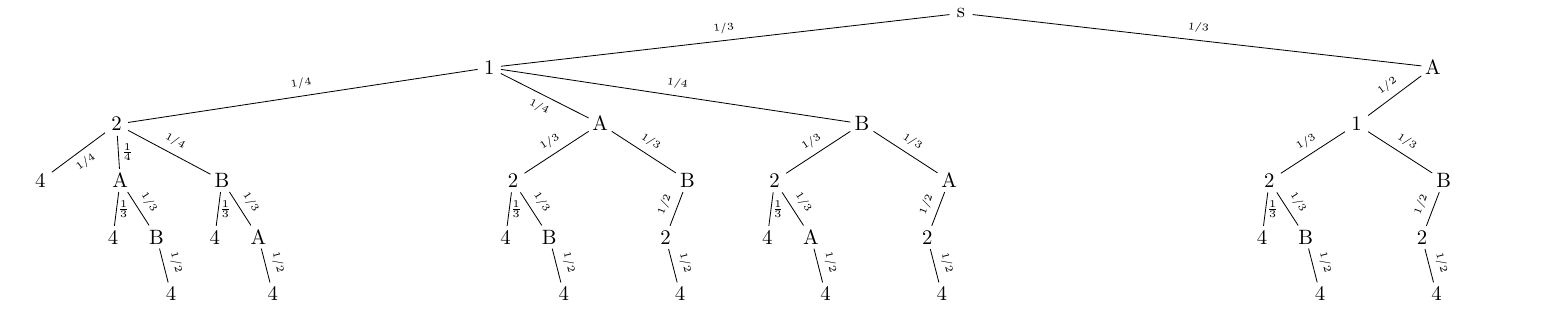}
	\caption{\label{fig:Q_gamma124} A depiction of $Q_{\vec{\gamma}_{124}}$, where $\vec{\gamma}_{124}$ is the $\mathcal{T}_{24}$-producing assembly sequence that attaches tiles 1, 2, and 4 in that order.  }
\end{figure}
Then, we have:
\[
\begin{array}{llll}
\textmd{Pr}\left[ \mathcal{P}^{\vec{\beta}_{\vec{x}}} \upharpoonright \vec{y} \right] & = & \displaystyle\sum_{\vec{\beta} \textmd{ leaf node of } \mathcal{P}^{\vec{\beta}_{\vec{x}}} \upharpoonright \vec{y}}{\textmd{Pr}_{ \mathcal{P}^{\vec{\beta}_{\vec{x}}} \upharpoonright \vec{y} }\left[ \vec{\beta} \right]} & \textmd{ Definition of } \textmd{Pr}\left[ \mathcal{P}^{\vec{\beta}_{\vec{x}}} \upharpoonright \vec{y} \right]  \\

 & = & \displaystyle\sum_{\substack{\vec{\gamma} \textmd{ winning assembly} \\\textmd{sequence for } \mathcal{C}}}{\left(\sum_{\vec{\beta} \in B_{\vec{\gamma}}}{\textmd{Pr}_{ \mathcal{P}^{\vec{\beta}_{\vec{x}}} \upharpoonright \vec{y} }\left[ \vec{\beta} \right]}\right)} & \textmd{ Lemma~\ref{lem:partition}} \\

	& = & \displaystyle\sum_{\substack{\vec{\gamma} \textmd{ winning assembly} \\\textmd{sequence for } \mathcal{C}}}{\textmd{Pr}\left[\mathcal{Q}_{\vec{\gamma}}\right]} & \textmd{ Definition of } \textmd{Pr}\left[\mathcal{Q}_{\vec{\gamma}}\right]. \\

\end{array}
\]
Moreover, Lemma~\ref{lem:restricted-tree-finite} says $\mathcal{P}^{\vec{\beta}_{\vec{x}}} \upharpoonright \vec{y}$ is finite, which means that for each winning assembly sequence $\vec{\gamma}$ for $\mathcal{C}$, $\mathcal{Q}_{\vec{\gamma}}$ is finite.

{\bf Step 2.} By conclusion~\ref{split-conclusion-6} of Lemma~\ref{lem:split-probability}, we have, for each winning assembly sequence $\vec{\gamma}$ for $\mathcal{C}$, $\textmd{Pr}\left[\mathcal{Q}_{\vec{\gamma}}\right] = \textmd{Pr}\left[ \texttt{split}\left( \mathcal{Q}_{\vec{\gamma}} \right) \right]$, and it follows that 
\begin{equation}
\label{eqn:main-eqn-to-prove}
\textmd{Pr}\left[ \mathcal{P}^{\vec{\beta}_{\vec{x}}} \upharpoonright \vec{y} \right] = \sum_{\substack{\vec{\gamma} \textmd{ winning assembly} \\\textmd{sequence for } \mathcal{C}}}{\textmd{Pr}\left[\texttt{split}\left( \mathcal{Q}_{\vec{\gamma}} \right) \right]}.
\end{equation}
It suffices to show that, for each winning assembly sequence $\vec{\gamma}$ for $\mathcal{C}$, $\textmd{Pr}\left[\texttt{split}\left( \mathcal{Q}_{\vec{\gamma}} \right)\right] = \textmd{Pr}_{\mathcal{C}}\left[ \vec{\gamma} \right]$.
Figure~\ref{fig:split_Q_gamma124} shows $\texttt{split}\left(\mathcal{Q}_{\vec{\gamma}_{124}}\right)$.
\begin{figure}[h!]
	\centering
	\includegraphics[width=\linewidth]{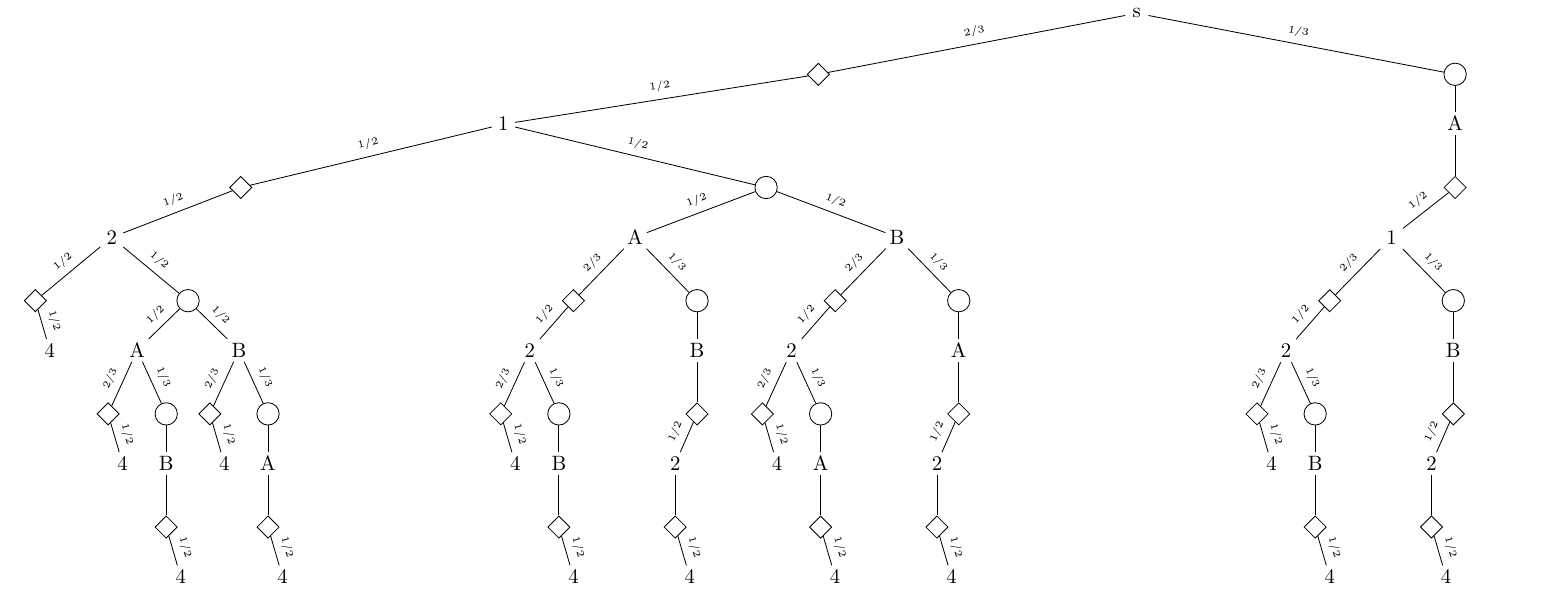}
	\caption{\label{fig:split_Q_gamma124} A depiction of $\texttt{split}\left(\mathcal{Q}_{\vec{\gamma}_{124}}\right)$. It is easy to verify that $\textmd{Pr}\left[ \mathcal{Q}_{\vec{\gamma}_{124}} \right] = \textmd{Pr}\left[\texttt{split}\left(\mathcal{Q}_{\vec{\gamma}_{124}}\right)\right]$ and that every internal non-diamond node is normalized.  }
\end{figure}
{\bf Step 3.}
Let $\vec{\gamma}$ be a winning assembly sequence for $\mathcal{C}$.
Lemma~\ref{lem:number-of-competing-and-non-competing-children} applied to $\mathcal{Q}_{\vec{\gamma}}$ says that for every internal node $\vec{\beta}$ of $\mathcal{Q}_{\vec{\gamma}}$:
\begin{enumerate}
	\item $\vec{\beta}$ has at most $c' = 1$ child in $\mathcal{Q}_{\vec{\gamma}}$ that is competing in $\mathcal{M}_{\mathcal{T}}$ corresponding to a tile attachment step in $\vec{\gamma}$ in which a competing tile attaches, and
	\item the number of children $n'$ of $\vec{\beta}$ in $\mathcal{Q}_{\vec{\gamma}}$ that are non-competing in $\mathcal{M}_{\mathcal{T}}$ is equal to the number of non-competing children $n$ of $\vec{\beta}$ in $\mathcal{M}_{\mathcal{T}}$, i.e., $n' = n$.
\end{enumerate}
Note that no internal node of $\mathcal{Q}_{\vec{\gamma}}$ has a child that terminates at $\vec{x}$ because $\mathcal{Q}_{\vec{\gamma}}$ is a subtree of $\mathcal{P}^{\vec{\beta}_{\vec{x}}} \upharpoonright \vec{y}$.
By Lemma~\ref{lem:number-of-competing-nodes}, for every internal node $\vec{\beta}$ of $\mathcal{Q}_{\vec{\gamma}}$, $\vec{\beta}$ has $c = 2$ children in $\mathcal{M}_{\mathcal{T}}$ that are competing in $\mathcal{M}_{\mathcal{T}}$.
Then: 
\begin{enumerate}
	\item conclusion~\ref{split-conclusion-2} of Lemma~\ref{lem:split-probability} says that for every diamond node $\Diamond$ of $\texttt{split}\left( \mathcal{Q}_{\vec{\gamma}} \right)$, $\Diamond$ has one outgoing edge, which means $S_{\Diamond} = \frac{c'}{c} = \frac{1}{2}$,
	\item conclusion~\ref{split-conclusion-3} of Lemma~\ref{lem:split-probability} says that for every circular node $\bigcirc$ of $\texttt{split}\left( \mathcal{Q}_{\vec{\gamma}} \right)$, $S_{\bigcirc} = \frac{n'}{n} = 1$, meaning every such node is normalized, and
	\item conclusion~\ref{split-conclusion-4} of Lemma~\ref{lem:split-probability} says that for every internal node $\ddot{\beta}$ of $\texttt{split}\left( \mathcal{Q}_{\vec{\gamma}} \right)$ that is neither diamond nor circular, $S_{\ddot{\beta}} = 1$, meaning every such node is normalized.
\end{enumerate}
This means every internal non-diamond node of $\texttt{split}\left( \mathcal{Q}_{\vec{\gamma}} \right)$ is normalized. 

{\bf Step 4.} We now compute $\textmd{Pr}\left[ \texttt{split}\left( \mathcal{Q}_{\vec{\gamma}} \right) \right]$, the definition of which is:
$$
\textmd{Pr}\left[ \texttt{split}\left( \mathcal{Q}_{\vec{\gamma}} \right) \right] = \sum_{\ddot{\beta} \textmd{ leaf node of } \texttt{split}\left( \mathcal{Q}_{\vec{\gamma}} \right)}{\textmd{Pr}_{\texttt{split}\left( \mathcal{Q}_{\vec{\gamma}} \right)}\left[ \ddot{\beta} \right]}. 
$$
Let $\ddot{\beta}$ be a leaf node of $\texttt{split}\left( \mathcal{Q}_{\vec{\gamma}} \right)$.
Then, by conclusion~\ref{split-conclusion-5} of Lemma~\ref{lem:split-probability}, $\ddot{\beta}$ is neither a circular nor a diamond node in $\texttt{split}\left( \mathcal{Q}_{\vec{\gamma}} \right)$.
By Observation~\ref{obs:prime-nodes-split-bijection}, $\ddot{\beta}$ corresponds to some $\vec{\beta} \in B_{\vec{\gamma}}$.
Thus, $\vec{\beta}$ terminates at $\vec{y}$ and $\vec{\gamma}$ is the unique longest winning assembly sequence for $\mathcal{C}$ embedded in $\vec{\beta}$.
As a result, $\vec{\beta}$ comprises exactly $\left| \vec{\gamma} \right| - 1$ tile attachment steps in which a competing tile attaches after a tile has attached at $\vec{x}$. 
Since $\vec{\gamma}$ is the unique longest winning assembly sequence for $\mathcal{C}$ embedded in $\vec{\beta}$ and line 14 of {\tt split} ensures that children of $\Diamond$ are competing in $\mathcal{M}_{\mathcal{T}}$, it follows that the path in $\texttt{split}\left( \mathcal{Q}_{\vec{\gamma}} \right)$ from the root to $\ddot{\beta}$ will contain exactly as many diamond nodes as there are assemblies in $\vec{\gamma}$ that result from the attachment of a competing tile.
Thus, the path in $\texttt{split}\left( \mathcal{Q}_{\vec{\gamma}} \right)$ from the root to $\ddot{\beta}$ will contain exactly $\left| \vec{\gamma} \right| - 1$ diamond nodes and each such node has one outgoing edge with probability $\frac{1}{2}$.
Let $\mathcal{Q}'_{\vec{\gamma}}$ be the SPT that results from setting the probability of the outgoing edge on every diamond node in $\texttt{split}\left( \mathcal{Q}_{\vec{\gamma}} \right)$ to 1.
Figure~\ref{fig:split_Q_gamma124_prime} shows $\mathcal{Q}'_{\vec{\gamma}_{124}}$.
\begin{figure}[h!]
	\centering
	\includegraphics[width=\linewidth]{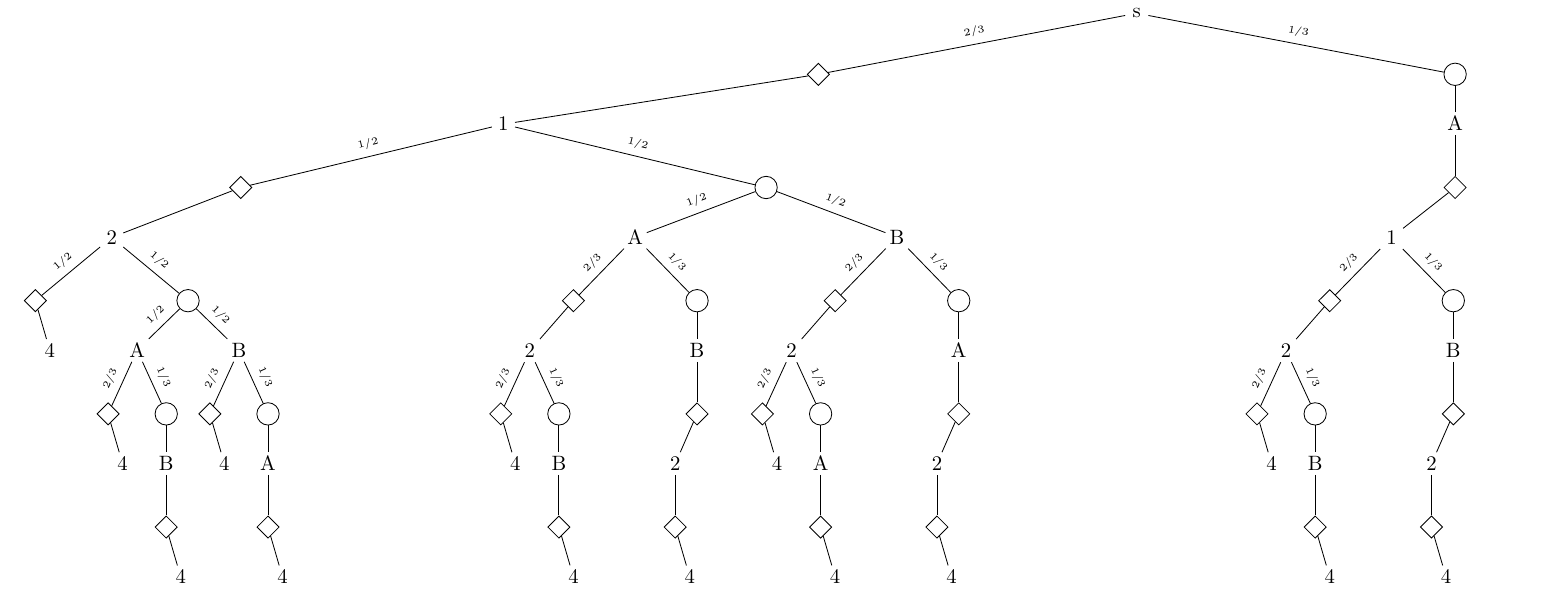}
	\caption{\label{fig:split_Q_gamma124_prime} A depiction of $\mathcal{Q}'_{\vec{\gamma}_{124}}$. Note that the probability on the outgoing edge on every diamond node is set to 1, making $\mathcal{Q}'_{\vec{\gamma}_{124}}$ normalized. It is easy to verify that $\textmd{Pr}\left[ \texttt{split}\left( \mathcal{Q}_{\vec{\gamma}_{124}} \right) \right] = \frac{1}{2^3} \cdot \textmd{Pr}\left[ \mathcal{Q}'_{\vec{\gamma}_{124}} \right]$. }
\end{figure}
Since $\mathcal{Q}'_{\vec{\gamma}}$ and $\texttt{split}\left( \mathcal{Q}_{\vec{\gamma}} \right)$ are isomorphic, it follows that
\begin{equation}
\label{eqn:split-q-prime}
\textmd{Pr}_{\texttt{split}\left( \mathcal{Q}_{\vec{\gamma}} \right)}\left[ \ddot{\beta} \right] = \frac{1}{2^{\left| \vec{\gamma}\right| - 1}} \cdot \textmd{Pr}_{\mathcal{Q}'_{\vec{\gamma}}}\left[ \ddot{\beta} \right].
\end{equation}
Then, we have:
\[
\begin{array}{llll}
\textmd{Pr}\left[ \texttt{split}\left( \mathcal{Q}_{\vec{\gamma}} \right) \right] & = & \displaystyle\sum_{\ddot{\beta} \textmd{ leaf node of } \texttt{split}\left( \mathcal{Q}_{\vec{\gamma}} \right)}{\textmd{Pr}_{\texttt{split}\left( \mathcal{Q}_{\vec{\gamma}} \right)}\left[ \ddot{\beta} \right]} & \textmd{ Definition of } \textmd{Pr}\left[ \texttt{split}\left( \mathcal{Q}_{\vec{\gamma}} \right)\right]\\

	& = & \displaystyle\sum_{\ddot{\beta} \textmd{ leaf node of } \mathcal{Q}'_{\vec{\gamma}}}{\frac{1}{2^{\left| \vec{\gamma}\right|-1}} \cdot \textmd{Pr}_{\mathcal{Q}'_{\vec{\gamma}}}\left[ \ddot{\beta} \right]} & \textmd{ Equation~(\ref{eqn:split-q-prime})} \\
	
	& = & \displaystyle \frac{1}{2^{\left| \vec{\gamma}\right|-1}} \sum_{\ddot{\beta} \textmd{ leaf node of } \mathcal{Q}'_{\vec{\gamma}}}{\textmd{Pr}_{\mathcal{Q}'_{\vec{\gamma}}}\left[ \ddot{\beta} \right]} & \\
	
	& = &  \displaystyle \frac{1}{2^{\left| \vec{\gamma}\right|-1}} \cdot \textmd{Pr}\left[ \mathcal{Q}'_{\vec{\gamma}} \right]. & \textmd{ Definition of } \textmd{Pr}\left[ \mathcal{Q}'_{\vec{\gamma}} \right]

\end{array}
\]

{\bf Step 5.}
Every internal node of $\mathcal{Q}'_{\vec{\gamma}}$ is normalized because:\vspace*{-2mm}
\begin{enumerate}[itemsep=0mm]
	\item every internal non-diamond node in $\texttt{split}\left( \mathcal{Q}_{\vec{\gamma}} \right)$ is normalized, and the  probabilities on the corresponding outgoing edges are unchanged in $\mathcal{Q}'_{\vec{\gamma}}$, and
		\item every diamond node in $\texttt{split}\left( \mathcal{Q}_{\vec{\gamma}} \right)$ has exactly one child and the probability of the corresponding edge is $\frac{1}{2}$, which gets set to $1$ in $\mathcal{Q}'_{\vec{\gamma}}$. 
\end{enumerate}

Since every internal node of $\mathcal{Q}'_{\vec{\gamma}}$ is normalized, we have:
      \[\begin{array}{llll}
        \textmd{Pr}[ \texttt{split}\left( \mathcal{Q}_{\vec{\gamma}} \right) ]
        & = & \displaystyle\frac{1}{2^{\left| \vec{\gamma}\right| - 1}} \cdot \textmd{Pr}[ \mathcal{Q}'_{\vec{\gamma}} ] & \textmd{Previous equation}\\
        & = & \displaystyle\frac{1}{2^{\left| \vec{\gamma}\right| - 1}} & \textmd{Corollary~\ref{cor:normalized-probability-1} applied to\ }  \mathcal{Q}'_{\vec{\gamma}},\\
      \end{array}
      \]
which, by conclusion~\ref{lem:dc-2} of Lemma~\ref{lem:domain-competition}, implies:

      \begin{equation}
        \label{eqn:PrSplitgammaEqualsPrCgamma}
        \textmd{Pr}[ \texttt{split}\left( \mathcal{Q}_{\vec{\gamma}} \right) ] =
       \textmd{Pr}_{\mathcal{C}}\left[\vec{\gamma}\right].
      \end{equation}

Finally, we have:
      \[
      \begin{array}{llll}
        \textmd{Pr}\left[ \mathcal{P}^{\vec{\beta}_{\vec{x}}} \upharpoonright \vec{y} \right]
        & = & \displaystyle\sum_{\substack{\vec{\gamma} \textmd{ winning assembly} \\\textmd{sequence for } \mathcal{C}}} 
         {\textmd{Pr}\left[\texttt{split}\left( \mathcal{Q}_{\vec{\gamma}} \right) \right]}
               & \textmd{Equation (\ref{eqn:main-eqn-to-prove})} \\
        & = & \displaystyle\sum_{\substack{\vec{\gamma} \textmd{ winning assembly} \\\textmd{sequence for } \mathcal{C}}}
         {\textmd{Pr}_{\mathcal{C}}\left[ \vec{\gamma} \right]} & \textmd{Equation (\ref{eqn:PrSplitgammaEqualsPrCgamma})}\\
         & = & \textmd{Pr}\left[ \mathcal{C} \right] &
         \textmd{Definition~\ref{def:competition-probability}},         \\
      \end{array}
      \]
      which establishes the lemma.

\end{proof}

The following lemma says that if $\mathcal{C}$ is rigged, then it is impossible to place the incorrect type of tile at the corresponding POC.
\begin{lemma} 
\label{lem:p-x-y-subtree-normalized}
If $\mathcal{C}$ is rigged by $\vec{\beta}_{\vec{x}}$, then $\textmd{Pr}\left[ \mathcal{P}^{\vec{\beta}_{\vec{x}}} \upharpoonright \vec{y}\, \right] = 1$.
\end{lemma} 

\begin{proofsketch}
  We proceed as follows:
  \begin{itemize}
	\item Prove that, for an arbitrary internal node $\vec{\beta}$ of $\mathcal{P}^{\vec{\beta}_{\vec{x}}} \upharpoonright \vec{y}$, if $\vec{\beta}'$ is an arbitrary child of $\vec{\beta}$ in $\mathcal{M}_{\mathcal{T}}$, then $\vec{\beta}'$ is $w$-correct.
	\item Invoke Lemma~\ref{lem:extend-to-alpha} applied to $\vec{\beta}$ and $\vec{\beta}'$ to conclude that $\vec{\beta}'$ is a child of $\vec{\beta}$ in $\mathcal{P}$.
	\item Observe that since every child of $\vec{\beta}$ in $\mathcal{M}_{\mathcal{T}}$ is a child of $\vec{\beta}$ in $\mathcal{P}$ and $\mathcal{P}^{\vec{\beta}_{\vec{x}}} \upharpoonright \vec{y}$ is full relative to $\mathcal{P}$, it follows that $\mathcal{P}^{\vec{\beta}_{\vec{x}}} \upharpoonright \vec{y}$ is normalized.
	\item Conclude via Corollary~\ref{cor:normalized-probability-1} that $\textmd{Pr}\left[ \mathcal{P}^{\vec{\beta}_{\vec{x}}} \upharpoonright \vec{y} \right] = 1$.
  \end{itemize}
\end{proofsketch}

\begin{proof}
Let $\vec{\beta}$ be an arbitrary internal node of $\mathcal{P}^{\vec{\beta}_{\vec{x}}} \upharpoonright \vec{y}$ and $\beta = \res{\vec{\beta}}$.
Then, $\beta \sqsubseteq \alpha$ and $\beta \not \in \mathcal{A}_{\Box}[\mathcal{T}]$ because $\beta$ is an internal node of $\mathcal{P}^{\vec{\beta}_{\vec{x}}} \upharpoonright \vec{y}$.
Let $\vec{\beta}'$ be an arbitrary child of $\vec{\beta}$ in $\mathcal{M}_{\mathcal{T}}$ such that, for some tile $\left( \vec{p}^{\; \prime}, t' \right)$, we have $\vec{\beta}' = \vec{\beta} + \left( \vec{p}^{\; \prime}, t' \right)$.
Either $\vec{\beta}'$ is competing in $\mathcal{M}_{\mathcal{T}}$ or it is not. We handle these two cases in turn.
\begin{enumerate}
\item Assume $\vec{\beta}'$ is not competing in $\mathcal{M}_{\mathcal{T}}$:
By part~\ref{lem:loc-p'-1} of Lemma~\ref{lem:location-of-p'}, we have $\vec{p}^{\; \prime} \not \in  \left( P  \cup \dom{\piw} \cup \dom{\pil} \right)$.
Since $\vec{p}^{\; \prime} \not \in P$ and $\alpha$ and $\beta$ agree, it follows by  
the contrapositive of condition~\ref{def:dd-ad-pt-1} of Definition~\ref{def:dd-ad-pt}
that $t' = \alpha\left(\vec{p}^{\; \prime}\right)$, which means $\vec{\beta}'$ is $w$-correct.
\item Assume $\vec{\beta}'$ is competing in $\mathcal{M}_{\mathcal{T}}$:
  By  part~\ref{lem:loc-p'-2} of Lemma~\ref{lem:location-of-p'}, we have $\vec{p}^{\; \prime} \in (\dom{\piw} \cup \dom{\pil})$.
Either $\vec{p}^{\; \prime} \ne \vec{y}$ or $\vec{p}^{\; \prime} = \vec{y}$. We handle these two sub-cases as follows.
\begin{enumerate}
	\item 
	Case $\vec{p}^{\; \prime} \ne \vec{y}$: Since $\vec{p}^{\; \prime} \not \in P$ and $\alpha$ and $\beta$ agree, it follows by the contrapositive of
	condition~\ref{def:dd-ad-pt-1} of Definition~\ref{def:dd-ad-pt} that $t' = \alpha\left(\vec{p}^{\; \prime}\right)$, which means $\vec{\beta}'$ is $w$-correct.

	\item Case $\vec{p}^{\; \prime} = \vec{y}$: 
	Let $p'$ be an arbitrary simple path from $\vec{s}$ to $\vec{y}$ in $G^{\textmd{b}}_{\res{\vec{\beta}'}}$.
By condition~\ref{def:cp-5} of Definition~\ref{def:competing-paths}, there exists $\pi \in \left\{ \piw, \pil \right\}$ such that $\pi$ is a suffix of $p'$.
Since $\mathcal{C}$ is rigged by $\vec{\beta}_{\vec{x}}$ and $\vec{\beta}'$ is a $\mathcal{T}$-producing extension of $\vec{\beta}_{\vec{x}}$, Lemma~\ref{lem:rigged} says that $\pil$ is not a simple path in $G^{\textmd{b}}_{\res{\vec{\beta}'}}$.
This means $\pi = \piw$.
In other words, $\piw$ is a suffix of every simple path from $\vec{s}$ to $\vec{y}$ in $G^{\textmd{b}}_{\res{\vec{\beta}'}}$.
This implies the existence of a unique unit vector $\vec{u} \in \left\{ (0,1), (1,0), (0,-1), (-1,0) \right\}$ such that every simple path from $\vec{s}$ to $\vec{y}$ in $G^{\textmd{b}}_{\res{\vec{\beta}'}}$ has $\left( \vec{y} + \vec{u}, \vec{y}\, \right)$ as a suffix.
In other words, condition~\ref{def:dd-ad-pt-2a} of Definition~\ref{def:dd-ad-pt} is ruled out.
Moreover, if $t'' \in T$ such that $\beta + \left( \vec{y}, t'' \right)$ is $\mathcal{T}$-producible, then condition~\ref{def:dd-ad-pt-2b} of Definition~\ref{def:dd-ad-pt} is ruled out.
Thus, by the contrapositive of condition~\ref{def:dd-ad-pt-2} of Definition~\ref{def:dd-ad-pt}, we have $t'' = t'$.
Note that $\vec{y} \in \partial^{\mathcal{T}}_{\alpha\left(\vec{y}\,\right)} \beta$ because $\piw\left[ \left| \piw \right| - 1 \right] \in \dom{\beta}$ and $\beta \sqsubseteq \alpha$. 
Therefore, we have $t' = \alpha\left(\vec{y}\,\right)$ and it follows that $\vec{\beta}'$ is $w$-correct.
\end{enumerate}
\end{enumerate}

At this point, we have shown that if $\vec{\beta}'$ is a child of $\vec{\beta}$ in $\mathcal{M}_{\mathcal{T}}$, then $\vec{\beta}'$ is $w$-correct.
Then, Lemma~\ref{lem:extend-to-alpha} applied to $\vec{\beta}$ and $\vec{\beta}'$ says that $\vec{\beta}'$ is a child of $\vec{\beta}$ in $\mathcal{P}$.
Note that $\mathcal{P}^{\vec{\beta}_{\vec{x}}} \upharpoonright \vec{y}$ is full relative to $\mathcal{P}$, which means $\vec{\beta}'$ is a child of $\vec{\beta}$ in $\mathcal{P}^{\vec{\beta}_{\vec{x}}} \upharpoonright \vec{y}$.
Since every child of $\vec{\beta}$ in $\mathcal{M}_{\mathcal{T}}$ is a child of $\vec{\beta}$ in $\mathcal{P}^{\vec{\beta}_{\vec{x}}} \upharpoonright \vec{y}$, it follows that every internal node of $\mathcal{P}^{\vec{\beta}_{\vec{x}}} \upharpoonright \vec{y}$ is normalized.
Thus, by Corollary~\ref{cor:normalized-probability-1}, we have $\textmd{Pr}\left[ \mathcal{P}^{\vec{\beta}_{\vec{x}}} \upharpoonright \vec{y}\, \right] = 1$. 
\end{proof}

\begin{lemma} 
\label{lem:proability-lower-bound-probability-gamma}
$\textmd{Pr}\left[ \mathcal{P}^{\vec{\beta}_{\vec{x}}} \upharpoonright \vec{y}\, \right] \geq \textmd{Pr}\left[ \mathcal{C} \right]$
\end{lemma}

\begin{proof}

If $\mathcal{C}$ is not rigged by $\vec{\beta}_{\vec{x}}$, then by Lemma~\ref{lem:subtree-at-x-i-to-y-i-probability}, we have $\textmd{Pr}\left[ \mathcal{P}^{\vec{\beta}_{\vec{x}}} \upharpoonright \vec{y}\, \right] = \textmd{Pr}\left[ \mathcal{C} \right]$.
Otherwise, by Lemma~\ref{lem:p-x-y-subtree-normalized}, we have $\textmd{Pr}\left[ \mathcal{P}^{\vec{\beta}_{\vec{x}}} \upharpoonright \vec{y}\, \right] = 1 \geq \textmd{Pr}\left[ \mathcal{C} \right]$.
\end{proof}

If  $\vec{x}$ is the starting point corresponding to some POC $\vec{y}$ and
we interpret $\frac{\textmd{Pr}\left[ \mathcal{P} \upharpoonright \vec{y}\, \right]}{\textmd{Pr}\left[ \mathcal{P} \upharpoonright \vec{x}\, \right]}$ as the
probability of the ``level of $\mathcal{P}$ between $\vec{x}$ and $\vec{y}\;$'', then the following lemma combines the previous machinery to prove that the probability of this level of $\mathcal{P}$ is bounded from below by the probability of the competition associated with $\vec{y}$.

\begin{lemma} 
\label{lem:probability-y-i-equals-x-i}
$\mathcal{P} \upharpoonright \vec{y}$ and $\mathcal{P} \upharpoonright \vec{x}$ are well-defined and
$\textmd{Pr}\left[ \mathcal{P} \upharpoonright \vec{y}\, \right] \geq \textmd{Pr}\left[ \mathcal{P} \upharpoonright \vec{x}\, \right] \cdot \textmd{Pr}\left[ \mathcal{C} \right]$.
\end{lemma}

\begin{proof}
Let $\pi$ be an arbitrary maximal path in $\mathcal{P}$. 
By Definition~\ref{def:w-pruned}, $\pi$ corresponds to a $w$-correct $\mathcal{T}$-producing assembly sequence $\vec{\alpha}$ that results in $\alpha$, the unique $w$-correct terminal assembly of $\mathcal{T}$ such that $Y \subseteq \dom{\alpha}$.
This means there is a node in $\dom{\pi}$ that terminates at $\vec{y}$.
Thus, $\mathcal{P} \upharpoonright \vec{y}$ is well-defined.
Since $Y \subseteq \dom{\alpha}$, Lemma~\ref{lem:x-i-y-i-ordering-r-1} implies that $\vec{x} \in \dom{\alpha}$ and $\textmd{index}_{\vec{\alpha}}\left( \vec{x}\, \right) < \textmd{index}_{\vec{\alpha}}\left(\vec{y}\,\right)$.
Thus, $\mathcal{P} \upharpoonright \vec{x}$ is well-defined and there is a node in $\dom{\pi}$ that terminates at $\vec{x}$ and is an ancestor in $\mathcal{P}$ of the node in $\dom{\pi}$ that terminates at $\vec{y}$.
This means $\mathcal{P} \upharpoonright \vec{x}$ is a subtree of $\mathcal{P} \upharpoonright \vec{y}$.
If we define $B$ to be the set of nodes of $\mathcal{P}$ that terminate at $\vec{x}$, then $B$ is a bottleneck of $\mathcal{P}  \upharpoonright \vec{y}$.
Moreover, Lemma~\ref{lem:restricted-tree-finite} implies that $\mathcal{P} \upharpoonright \vec{x}$ is finite and the set of its leaf nodes  is $B$.
Then, we have:
\[
\begin{array}{llll}
\textmd{Pr}\left[\mathcal{P} \upharpoonright \vec{y}\,\right] & = & \displaystyle\sum_{\vec{\beta}_{\vec{x}} \in B}{\left( \textmd{Pr}_{\mathcal{P} \upharpoonright \vec{y}}\left[\vec{\beta}_{\vec{x}}\right] \cdot \textmd{Pr}\left[ \left( \mathcal{P} \upharpoonright \vec{y}\, \right)^{\vec{\beta}_{\vec{x}}} \right]\right)} & \textmd{ Lemma~\ref{lem:alternative-characterization-of-pr}, with } \mathcal{Q} = \mathcal{P} \upharpoonright \vec{y} \\

	 & = & \displaystyle\sum_{\vec{\beta}_{\vec{x}} \in B}{\left( \textmd{Pr}_{\mathcal{P} \upharpoonright \vec{y}}\left[\vec{\beta}_{\vec{x}}\right] \cdot \textmd{Pr}\left[ \mathcal{P}^{\vec{\beta}_{\vec{x}}} \upharpoonright \vec{y} \right]\right)} & \textmd{ Lemma~\ref{lem:p-restricted-u} with } \mathcal{Q} = \mathcal{P} \upharpoonright \vec{y} \textmd{ and } v = \vec{\beta}_{\vec{x}}  \\
	 
	 & \geq & \displaystyle\sum_{\vec{\beta}_{\vec{x}} \in B}{\left( \textmd{Pr}_{\mathcal{P} \upharpoonright \vec{y}}\left[\vec{\beta}_{\vec{x}}\right] \cdot \textmd{Pr}\left[ \mathcal{C} \right] \right)} & \textmd{ Lemma~\ref{lem:proability-lower-bound-probability-gamma} } \\
	 
	 & = & \displaystyle\left( \sum_{\vec{\beta}_{\vec{x}} \in B}{ \textmd{Pr}_{\mathcal{P} \upharpoonright \vec{y}}\left[\vec{\beta}_{\vec{x}}\right]  }  \right) \cdot \textmd{Pr}\left[ \mathcal{C} \right]&  \\
	 
	 & = & \displaystyle\left( \sum_{\vec{\beta}_{\vec{x}} \textmd{ leaf node of } \mathcal{P}\upharpoonright \vec{x}}{ \textmd{Pr}_{\mathcal{P} \upharpoonright \vec{y}}\left[\vec{\beta}_{\vec{x}}\right]  }  \right) \cdot \textmd{Pr}\left[ \mathcal{C} \right] & \textmd{ Definition of } B \\

	  & = & \displaystyle\left( \sum_{\vec{\beta}_{\vec{x}} \textmd{ leaf node of } \mathcal{P}\upharpoonright \vec{x}}{ \textmd{Pr}_{\mathcal{P} \upharpoonright \vec{x}}\left[\vec{\beta}_{\vec{x}}\right]  }  \right) \cdot \textmd{Pr}\left[ \mathcal{C} \right] & \textmd{ Obs.~\ref{obs:prob-subtree-same-root} with } \mathcal{Q} = \mathcal{P} \upharpoonright \vec{y} \textmd{, } \mathcal{Q}' = \mathcal{P} \upharpoonright \vec{x} \textmd{, and } v = \vec{\beta}_{\vec{x}} \\
	 
	& = & \displaystyle\left( \sum_{\pi \textmd{ maximal path in } \mathcal{P}\upharpoonright \vec{x}}{ \textmd{Pr}_{\mathcal{P} \upharpoonright \vec{x}}[\pi]  }  \right) \cdot \textmd{Pr}\left[ \mathcal{C} \right] & \textmd{ Observation~\ref{obs:leaf-node-maximal-path} with } \mathcal{Q} = \mathcal{P} \upharpoonright \vec{x} \\

	& = & \textmd{Pr}\left[ \mathcal{P} \upharpoonright \vec{x}\, \right] \cdot \textmd{Pr}\left[ \mathcal{C} \right]  & \textmd{ Definition of } \textmd{Pr}\left[ \mathcal{P} \upharpoonright \vec{x}\, \right].
\end{array}
\]
\end{proof}

We now have sufficient machinery to prove our main result in this section, namely:
\setcounter{theorem}{0}
\begin{theorem}
If $\mathcal{C}_1, \ldots, \mathcal{C}_r$ are the competitions in $\mathcal{T}$, then $\mathcal{T}$ strictly self-assembles $\dom{\alpha}$ with probability at least $\displaystyle\prod_{i=1}^{r}{\textmd{Pr}\left[ \mathcal{C}_i \right]}$.
\end{theorem}

\begin{proof}
Let $A = \dom{\alpha}$.
To show that $\mathcal{T}$ strictly self-assembles $A$ with probability at least $\displaystyle\prod_{i=1}^{r}{\textmd{Pr}\left[ \mathcal{C}_i \right]}$, we must show that
$\displaystyle\sum_{\substack{\alpha' \in \mathcal{A}_{\Box}[\mathcal{T}]  \\ \dom{\alpha'}=A}}{\textmd{Pr}_{\mathcal{T}}[\alpha']} \geq \prod_{i=1}^{r}{\textmd{Pr}\left[ \mathcal{C}_i \right]}$.
First, note that:
\[
\begin{array}{llll}
\displaystyle\sum_{\substack{\alpha' \in \mathcal{A}_{\Box}[\mathcal{T}] \\ \dom{\alpha'}=A}}{\textmd{Pr}_{\mathcal{T}}[\alpha']} & = & \displaystyle\sum_{\substack{\alpha' \in \mathcal{A}_{\Box}[\mathcal{T}] \\ \dom{\alpha'}=A}}{\left(\sum_{\substack{\vec{\alpha} \textmd{ is a } \mathcal{T}\textmd{-producing assembly}\\ \textmd{sequence and } \res{\vec{\alpha}}=\alpha'}}{\textmd{Pr}_{\mathcal{M}_{\mathcal{T}}}[\vec{\alpha}]}\right)} & \textmd{ Definition of } \textmd{Pr}_{\mathcal{T}}\left[ \alpha' \right] \\

	& \geq & \displaystyle\sum_{\substack{\vec{\alpha} \textmd{ is a } w\textmd{-correct } \mathcal{T}\textmd{-producing}\\ \textmd{assembly sequence and } \res{\vec{\alpha}}=\alpha}}{\textmd{Pr}_{\mathcal{M}_{\mathcal{T}}}[\vec{\alpha}]} & \\
	
	& = & \textmd{Pr}[\mathcal{P}] & \textmd{ Definition~\ref{def:w-pruned}}, \\
\end{array}
\]
where the inequality follows from the fact that a $\mathcal{T}$-terminal assembly having domain $A$ need not be the unique $w$-correct $\mathcal{T}$-terminal assembly.
Therefore, it suffices to show that
$\displaystyle\textmd{Pr}[ \mathcal{P} ] \geq \prod_{i=1}^{r}{\textmd{Pr}\left[ \mathcal{C}_i\right]}$.
To that end, we now show that, for all $1 \leq j \leq r$,
\begin{equation}
\label{eqn:ih}
\textmd{Pr}\left[ \mathcal{P} \upharpoonright \vec{y}_j \right] \geq \prod_{i=1}^{j}{\textmd{Pr}\left[ \mathcal{C}_i \right]}.
\end{equation}
For any $r \in \mathbb{Z}^+$, we will prove (\ref{eqn:ih}) by induction on $j$.
First, note that, for any $r \in \mathbb{Z}^+$, and $j = 1$, we have
\[
\begin{array}{llll}
\textmd{Pr}\left[ \mathcal{P} \upharpoonright \vec{y}_1 \right] & \geq & \textmd{Pr}\left[ \mathcal{P} \upharpoonright \vec{x}_1 \right] \cdot \textmd{Pr}\left[ \mathcal{C}_1 \right] & \textmd{ by Lemma~\ref{lem:probability-y-i-equals-x-i}} \\
	& = & \textmd{Pr}\left[ \mathcal{C}_1 \right] & \textmd{ by Lemma~\ref{lem:probability-x-1}}.
\end{array}
\]
Assume that $r > 1$ and (\ref{eqn:ih}) holds for $j = k$ such that $1 < k < r$.
Then, we have
\[
\begin{array}{llll}

\textmd{Pr}\left[ \mathcal{P} \upharpoonright \vec{y}_{k+1} \right] & \geq & \textmd{Pr}\left[ \mathcal{P} \upharpoonright \vec{x}_{k+1} \right] \cdot \textmd{Pr}\left[ \mathcal{C}_{k+1} \right] & \textmd{ Lemma~\ref{lem:probability-y-i-equals-x-i}} \\
	& = & \textmd{Pr}\left[ \mathcal{P} \upharpoonright \vec{y}_k \right] \cdot \textmd{Pr}\left[ \mathcal{C}_{k+1} \right] & \textmd{ Lemma~\ref{lem:probability-y-i-equals-probability-x-i-plus-1}} \\
	& \geq & \displaystyle\left( \prod_{i=1}^{k}{\textmd{Pr}\left[ \mathcal{C}_i \right]} \right) \cdot \textmd{Pr}\left[ \mathcal{C}_{k+1} \right] & \textmd{ Inequality~(\ref{eqn:ih}), for $j=k$} \\
	& = & \displaystyle\prod_{i=1}^{k+1}{\textmd{Pr}\left[ \mathcal{C}_i \right]}. &
\end{array}
\]
Thus,
\[
\begin{array}{llll}

 \textmd{Pr}\left[ \mathcal{P} \right] 	& = & \textmd{Pr}\left[ \mathcal{P} \upharpoonright \vec{y}_r \right]  & \textmd{ Lemma~\ref{lem:probability-y-r}}  \\
	& \geq & \displaystyle\prod_{i=1}^{r}{\textmd{Pr}\left[ \mathcal{C}_i \right]} & \textmd{ Inequality~(\ref{eqn:ih}), for } j = r.
\end{array}
\]
\end{proof}

\section{Efficient, sequentially non-deterministic self-assembly of an $N \times N$ square with high probability}
\label{sec:square}
In this section, we prove our ``main theorem'' constructively namely
that, for every $N \in \mathbb{Z}^+$ and real $\delta > 0$, there
exists a TAS that strictly self-assembles an $N \times N$ square with
probability at least $1-\delta$ and uses only $O\left(\log N + \log
\frac{1}{\delta}\right)$ tile types.

This section is broken down into six subsections. The first one
provides an overview of our square construction, describes the counter
template that is the main building block of this construction, and
defines the parameters of this counter construction. The second
subsection both states and provides a proof sketch for our ``main lemma''
pertaining to the  counter construction. The third subsection contains the
first two steps in the proof of the main lemma. The fourth subsection
describes the types of incorrect assemblies that may result from our
counter construction. The fifth subsection contains the last four
steps in the proof of the main lemma. Finally, the sixth subsection uses
the main lemma to prove our main theorem.

\subsection{Overview of the construction}

Our square construction uses three modified copies of a single
template for a zigzag binary counter. This counter template is $K$
tiles wide and $n$ tiles high, with $n \leq N$. It is built out of
``gadgets'', which are groups of tiles that self-assemble within a
fixed region of space to accomplish a specific task, e.g.,
non-deterministically read (i.e., guess) a single unknown,
geometrically-represented binary value. Since our counter assembles in
a zigzag pattern, each one of its values is first incremented, going
from right to left, and then copied from left to right. Since one row
of tiles is used to hardcode the starting value of the counter and all
of the copy and increment gadgets are $h$ tiles high, with $h = \left
\lceil 2 \log N + \log \frac{1}{\delta} \right\rceil + 5$, the number
of values taken on by the counter is equal to $m = \left \lfloor
\frac{n-1}{2h} \right \rfloor$. Then $e = (n-1) \mod 2h$ is the number
of extra tiles (i.e., rows) needed for the height of our counter
template to be exactly equal to $n$. The counter template is P-shaped.
The stem of the P is an upside-down L whose domain we denote by $L_e =
\{ 0 \} \times \{ 0, \ldots, e-1\} \cup \{(1,e-1),(2,e-1)\}$ and
depict in dark gray in Figure~\ref{fig:template-overview}. The
``loop'' of the P is a $K \times (n-e)$ rectangular counter whose
domain we denote by $R_{K,n,e} = \{ 0,\ldots, K-1\} \times \{ e,
\ldots, n-1 \}$ and depict in light gray in
Figure~\ref{fig:template-overview}.

The last parameters we need pertain to the binary representation of
the counter's values. We use $k$ to denote the number of bits in the
binary expansion of $m$. Each bit value will be geometrically encoded
in a 6-tile-wide gadget. Since two tiles are added on each side of the
counter's value and each bit is actually encoded by two gadgets,
namely the bit's value itself and an additional indicator bit used to
distinguish the leftmost and rightmost value bits from the other ones,
$K$ is equal to $12k+4$. Finally, the starting value of our counter is
set to $s = 2^{\left\lfloor \log m \right\rfloor + 1} - m$.
Table~\ref{tab:variables} lists the definitions for our variables.

\begin{figure}[h!]
  \centering
  \includegraphics[width=05in]{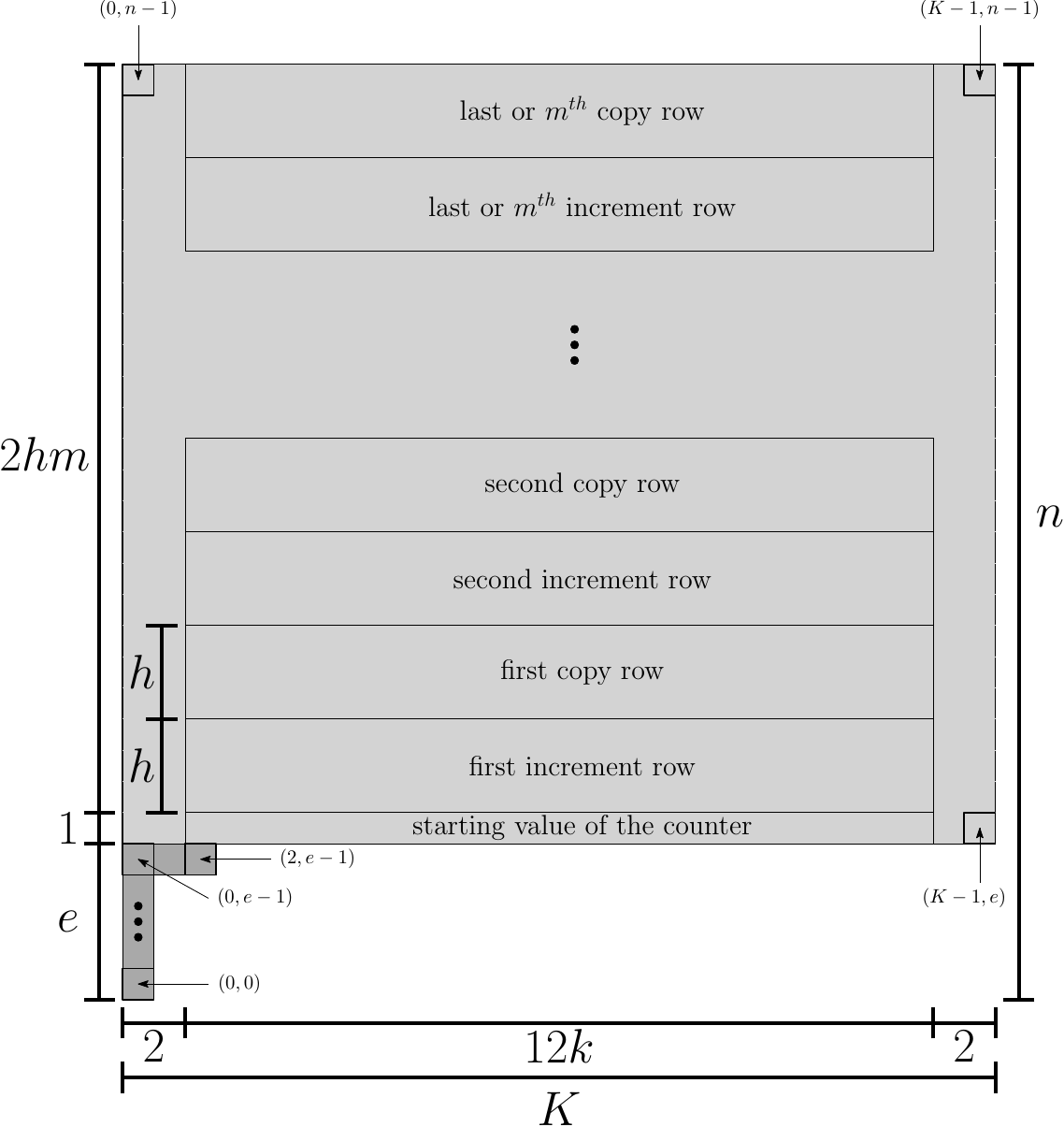}
  \caption{\label{fig:template-overview} Shape of, and parameters for,
    our counter template}
\end{figure}

\begin{table}[h!]
  \centering
  \begin{tabular}{|c|l|}\hline
    $\delta$ & error probability\\\hline
    $N$ & number of tiles on each side of the final square\\ \hline
    $n$ & number of tiles defining the height of the counter template\\ \hline
    $K$ & number of tiles defining the width of the counter template\\ \hline
    $h$ & number of tiles defining the height of an increment or copy
    gadget\\\hline
    $m$ & number of values taken on by the counter\\\hline
    $e$ & extra tiles needed so that the height of the template $e+1+2hm$ is
    equal to $n$ \\\hline
    $s$ & the starting value of the counter\\\hline
    $k$ & the number of bits in the binary expansion of $m$\\\hline
  \end{tabular}
\caption{\label{tab:variables}Definition of the variables used to
  describe our construction}
\end{table}

\subsection{Statement and proof sketch of the main lemma}

To prove this section's theorem, we will use the following lemma:
\begin{lemma} 
\label{lem:high-probability-counter}
Let $N \in \mathbb{Z}^+$, 
$\delta \in \mathbb{R}^+$ such that $0<\delta<1$, 
$h = \left \lceil 2 \log N + \log \frac{1}{\delta} \right\rceil + 5$, 
$n \in \mathbb{Z}^+$ such that $n \leq N$, and
$m = \left \lfloor \frac{n-1}{2h} \right \rfloor$.
Then there exists a TAS $\mathcal{T} = \left(T,\sigma,1\right)$ satisfying the following conditions:
\begin{enumerate}
	\item If $m > 0$, $s = 2^{\left\lfloor \log m \right\rfloor + 1} - m$, $k$ is the number of bits in the binary expansion of $m$, $K = 12k+4$, and $e = (n-1) \mod 2h$, then there exists a finite set $P \subseteq \{ 0,\ldots K-1\} \times \{ 0, \ldots, n-1 \}$, a finite set $Y \subseteq P$ and $w:Y\rightarrow T$ such that:
	\begin{enumerate}
		\item $P$ is the set of POCs in $\mathcal{T}$, $Y$ is the set of essential POCs in $\mathcal{T}$ and $Y = P$,
		\item $\mathcal{T}$ is $w$-sequentially non-deterministic, 
		\item $\alpha$ is the unique $w$-correct $\mathcal{T}$-terminal assembly and $\dom{\alpha} = L_e \cup R_{K,n,e}$,

                \item in every $w$-correct $\mathcal{T}$-producing assembly sequence resulting in $\alpha$, no tile placed at any point on the perimeter of $\alpha$ has a positive-strength glue on its outward-facing side, the last one among $\alpha$'s  perimeter points at which a tile is placed is  $(K-1, n-1)$, and this point is the only one in $\dom{\alpha}$ where any tile of this type is placed.
		\item for every competition $\mathcal{C}$ of $\mathcal{T}$, if $\pi$ and $\pi'$ are the corresponding winning and losing paths, respectively, then $|\pi| = 4$ and $\left| \pi' \right| = 4h - 10$ (or vice versa),	
	\end{enumerate}
      \item $\mathcal{T}$ strictly self-assembles       $\dom{\alpha} = L_e \cup R_{K,n,e}$ with probability at least $1 - \delta$, and
	\item $|T| = O\left( \log N + \log \frac{1}{\delta} \right)$.
	
\end{enumerate}
\end{lemma}
The reason we set $h$ so high is to allow the P-shaped assembly produced by the TAS in Lemma~\ref{lem:high-probability-counter} to be used in a larger construction, e.g., one that strictly self-assembles an $N \times N$ square, which will require a higher value for $h$ in order to guarantee a sufficient probability of correctness. 
We chose the P-shape of $\dom{\alpha}$ to be able to invoke Lemma~\ref{lem:high-probability-counter} multiple times when assembling an $N \times N$ square. 
\begin{proofsketch}
To prove Lemma~\ref{lem:high-probability-counter}, we proceed as follows:
\begin{enumerate}[itemsep=0mm]
\item Prove that the lemma holds in the corner cases when $n$ is too small, i.e., $n < 13$, or when $m = 0$, i.e., when $h$ is too large for a given value of even a sufficiently large $n$.
Assume going forward that $n \geq 13$ and $m > 0$.
\item Define $T$ in terms of ``gadgets'' of tiles and compute an asymptotic upper bound on its cardinality.
The tile types for $T$ will simulate a zigzag binary counter that is initialized with the starting value $s$, and then alternates between increment steps in which the counter value is incremented by 1 going from right to left, followed by copy steps in which the counter value is copied going from left to right.
The counter stops after a copy step in which the copied value is all 0s. 
\begin{enumerate}[label=\theenumi\alph*.,itemsep=0mm]
	\item Create one seed gadget of size $O(e)$, which contains the unique seed tile of $\mathcal{T}$.
	\item Create $O(k)$ gadgets, each of size $O(1)$, to initialize the counter.
	\item Create $O(1)$ gadgets, each of size $O(h)$, to increment the counter.
	\item Create $O(1)$ gadgets, each of size $O(h)$, to copy the value of the counter.
        \item Consolidate the outomes of the four previous gadget creation steps to infer that $\left| T \right| = O\left( \log N + \log \frac{1}{\delta} \right)$.
\end{enumerate}
\item Define $P$, $Y$ and $w$.
\item Prove that $\mathcal{T}$ is $w$-sequentially non-deterministic such that $\alpha$ is the unique $w$-correct $\mathcal{T}$-terminal assembly and 
  $\dom{\alpha} = L_e \cup R_{K,n,e}$.
\item Prove that every $w$-correct $\mathcal{T}$-producing assembly sequence resulting in $\alpha$ attaches the last tile at the upper-right corner of $\alpha$ and that no tile on the perimeter of $\alpha$ has an outward-facing positive-strength glue.
\item Prove that $\mathcal{T}$ strictly self-assembles $\dom{\alpha}$ with probability at least $1 - \delta$.
	\begin{enumerate}[label=\theenumi\alph*.,itemsep=0mm]
		\item Use Corollary~\ref{cor:competition-probability-sum} to prove that, for every competition $\mathcal{C}$ of $\mathcal{T}$, $\textmd{Pr}[\mathcal{C}] \geq 1 - \frac{\delta}{N^2}$.
		\item Apply Theorem~\ref{thm:local-non-determinism-theorem} and use Bernoulli's inequality to bound the product of the probabilities of the $2km \leq N^2$ essential competitions of $\mathcal{T}$. 
	\end{enumerate}
\end{enumerate}\vspace*{-7mm}

\end{proofsketch}

\subsection{First two steps in the proof of the main lemma}
\label{sec:first-two-steps}

This subsection contains steps 1 and 2 in our proof of
Lemma~\ref{lem:high-probability-counter}.
\emph{Step 1.}
In this step, we handle two corner cases.
First, since 1) the starting value of the counter takes up one row of
tiles, 2) each copy and increment gadget, to be described below, has a
minimum height of six tiles, and 3) the counter must contain at least
one increment row and one copy row, the smallest possible value for
$n$ is 13. If $n < 13$ and thus too small, then $e = n$ and
$R_{K,n,e}$ is empty.
In this case, we construct $\mathcal{T}$ to be a TAS with $|T| = O(e) = O(1)$ that uniquely self-assembles $L_e$.
Note that such a TAS strictly self-assembles $L_e$ with probability 1. 
Going forward, assume $n \geq 13$.
Second, even with a large enough $n$, it is possible for $h$ to be set too high, resulting in $m$ being equal to 0. This happens when $n-1 < 2h$. 
In this case, we have 
\begin{eqnarray*}
            n+1 & = & (n-1)+2	\\
  	        & < 	& 2h +2\\
		& =		& 2\left\lceil 2 \log N + \log \frac{1}{\delta} \right\rceil + 12 \\
		& \leq 	& 2\left( 2 \log N + \log \frac{1}{\delta} + 1\right) + 12 \\
		& =		& 4\log N + 2\log \frac{1}{\delta} + 14.
\end{eqnarray*}
Thus, in the case when $m = 0$, we construct $\mathcal{T}$ to be a TAS that uniquely self-assembles $R_{K,n,e} \cup L_e$ and whose tile set contains $e+2$ distinct tile types to cover $L_e$ and $(n-e)+ (K-1)$ distinct tile types to cover $R_{K,n,e}$ using a comb-like construction (see Figure 2b in \cite{RotWin00}). Therefore, $\left| T \right| = (e+2)+(n-e)+(K-1) = (n+1)+K$. Since we just proved that $n+1 = O\left( \log N + \log \frac{1}{\delta}\right)$ and  $K=\Theta(k) = \Theta(\log m) = O(\log N)$, it follows that $|T| = O\left( \log N + \log \frac{1}{\delta}\right)$.
Going forward, assume $m \geq 1$. 

\emph{Step 2.}
We will define the tile set $T$ in terms of gadgets, i.e., groups of tile types that are disjoint from one another, and self-assemble within a fixed region of space into one of possibly several assemblies to accomplish a specific task, e.g., non-deterministically guess a single unknown, geometrically-represented binary value.
A {\em generic gadget} is an assembly in which every positive strength glue of every placed tile has a strength of 1 and some label. 
We depict each generic gadget in its own figure with an arrow pointing toward the {\em input glue} of its {\em input tile} and another arrow pointing away from the {\em output glue} of its {\em output tile}.
For a generic gadget {\tt Gadget}, we use the notation ${\tt Gadget}\left( {\tt a}, {\tt b} \right)$ to represent the {\em creation} of the {\em actual gadget} (or simply {\tt gadget} {\tt Gadget}) with input glue {\tt a}, output glue {\tt b}, and such that the resulting assembly is producible in a corresponding singly-seeded, temperature-1, directionally deterministic TAS.
If a gadget has two possible assemblies, because it makes a non-deterministic choice, then each assembly is depicted in a corresponding sub-figure. In this case,  the corresponding TAS must be directionally deterministic, the gadget's two assemblies must both have the same input glue, but they may have different output glues.
We will use the notation ${\tt Gadget}\left( {\tt a}, {\tt b}, {\tt c} \right)$ to denote the actual version of {\tt Gadget}, where {\tt a} is the input glue, and {\tt b} and {\tt c} are the two possible output glues.
When the input or output glue of a gadget encodes multiple pieces of information or symbols, we use the notation $\left\langle s_1, \ldots, s_n \right\rangle$ to denote some standard encoding of the concatenation of the list of symbols.
In our construction, all input and output glues are positioned on either the west side or the east side of a tile. 
We say that a gadget self-assembles {\em to the right} if its input glue is on the west side of its input tile. 
If such a gadget has an input glue, it must be on the west side of its input tile. 
In the special case where the gadget does not have an input glue, of which there will be only one in our construction, then we say that it self-assembles {\em to the right} if its output glue is on the east side of its output tile.
In all other cases, we say that the gadget self-assembles {\em to the left} if its input glue is on the east side of its input tile.
We define  $T$ in terms of four types of gadgets, each one corresponding to a type of logical step executed by the counter, namely {\tt Seed}, {\tt Init}, {\tt Inc}, or {\tt Copy}.
The {\tt Seed} gadget places a tile at every point in $L_e$, the {\tt Init} gadgets initialize the counter, the {\tt Inc} gadgets increment the value of the counter, and the {\tt Copy} gadgets copy the value of the counter.

\emph{Step 2a.}
%
%
We define three generic {\tt Seed} gadgets shown in Figure~\ref{fig:seed-gadgets}.
\begin{figure}[h!]
    \centering
    \begin{subfigure}[t]{0.3\textwidth}
        \centering
        \includegraphics[width=0.33in]{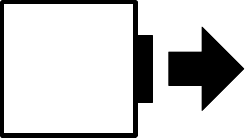}
        \caption{\label{fig:seed-e-zero} {\tt Seed\_e\_equal\_to\_0}}
    \end{subfigure}
    ~
        \begin{subfigure}[t]{0.3\textwidth}
        \centering
        \includegraphics[width=0.8in]{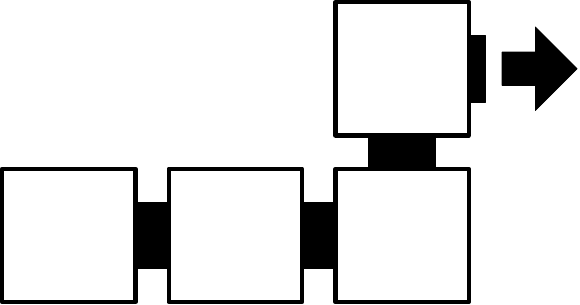}
        \caption{\label{fig:seed-e-one} {\tt Seed\_e\_equal\_to\_1}}
    \end{subfigure}
    ~
    \begin{subfigure}[t]{0.3\textwidth}
        \centering
        \includegraphics[width=0.8in]{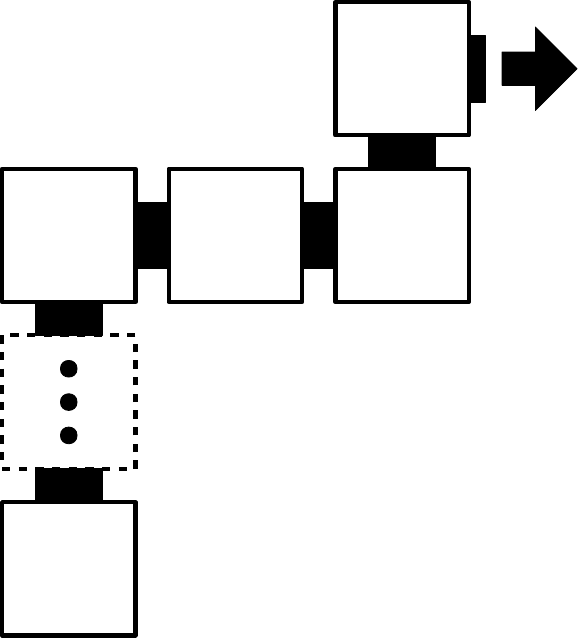}
        \caption{\label{fig:seed-e-gt-1} {\tt Seed\_e\_greater\_than\_1}}
    \end{subfigure}
    \caption{\label{fig:seed-gadgets} The assemblies for the generic {\tt Seed} gadgets. (a) The gadget used when  $e=0$. (b) The gadget used when  $e=1$. (c) The gadget used when $e > 1$, in which case the vertical dots represent $e-2$ tiles. In all cases, the height of the gadget is $e+1$.} 
\end{figure}
Which assembly is created in our construction depends on the value of  $e$.
This gadget is the only one in our construction that does not have an input glue, and in all cases, we denote the absence of an input glue with $\emptyset$ as the first argument in the gadget creation statement.
Its output glue initiates the initialization of the counter. 
If $e = 0$, then create, from the generic gadget in Figure~\ref{fig:seed-e-zero}, the gadget: 
\begin{align*}
{\tt Seed\_e\_equal\_to\_0}( & \emptyset,  {\tt init} ).
\end{align*}
If $e = 1$, then create, from the generic gadget in Figure~\ref{fig:seed-e-one}, the gadget: 
\begin{align*}
{\tt Seed\_e\_equal\_to\_1}( & \emptyset,  {\tt init} ).
\end{align*}
Otherwise, create, from the generic gadget in Figure~\ref{fig:seed-e-gt-1}, the gadget: 
\begin{align*}
{\tt Seed\_e\_greater\_than\_1}( & \emptyset, {\tt init}).
\end{align*}
We then define the seed assembly $\sigma$ of $\mathcal{T}$ as the placement, at the origin, of the bottom-left tile in the created {\tt Seed} gadget.
Note that we created a unique {\tt Seed} gadget of size $O(e) = O(h)$.
Thus, the total number of tile types contributed to $T$ by the created {\tt Seed} gadget is $O(h) = O\left( \log N + \log \frac{1}{\delta} \right)$.

\emph{Step 2b.}
%
%
The {\tt Seed} gadget, via its unique output glue, initiates the initialization of the counter. 
The {\tt Init} gadgets assemble the row of tiles that represent the starting value of the counter.
The five different generic {\tt Init} gadgets are shown in Figure~\ref{fig:init-gadgets}.
\begin{figure}[h!]
    \centering
    \begin{subfigure}[t]{0.3\textwidth}
        \centering
        \includegraphics[width=1.9in]{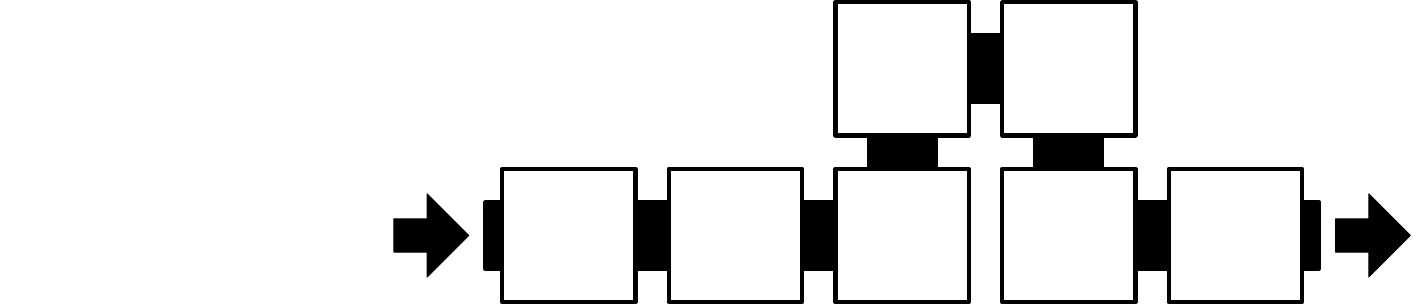}
        \caption{\label{fig:seed} {\tt Init\_left}}
    \end{subfigure}
    ~
    \hfill
    \begin{subfigure}[t]{0.3\textwidth}
        \centering
        \includegraphics[width=1.9in]{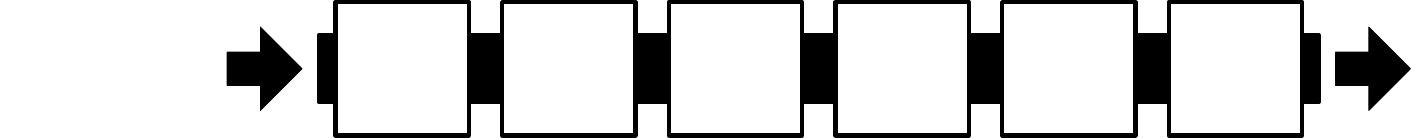}
        \caption{\label{fig:seed-middle-0} {\tt Init\_middle\_0}}
    \end{subfigure}
    ~
    \begin{subfigure}[t]{0.3\textwidth}
        \centering
        \includegraphics[width=1.9in]{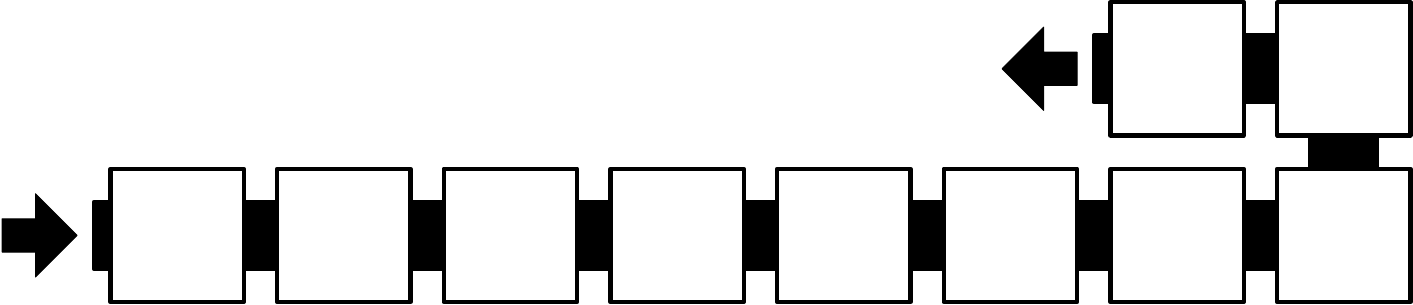}
        \caption{\label{fig:seed-right-0} {\tt Init\_right\_0}}
    \end{subfigure} \\
    \vspace{10pt}
  \hfill  
  \begin{subfigure}[t]{0.3\textwidth}
        \centering
        \includegraphics[width=1.9in]{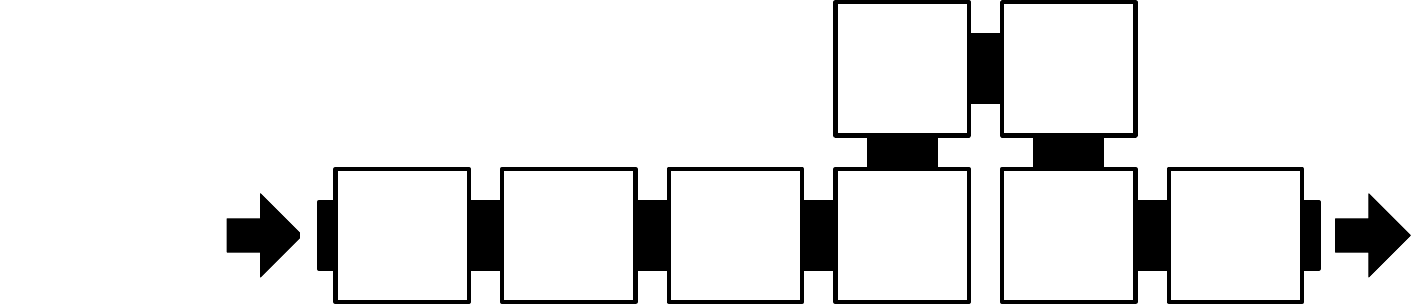}
        \caption{\label{fig:seed-middle-1} {\tt Init\_middle\_1}}
    \end{subfigure}
    ~
  \begin{subfigure}[t]{0.3\textwidth}
        \centering
        \includegraphics[width=1.9in]{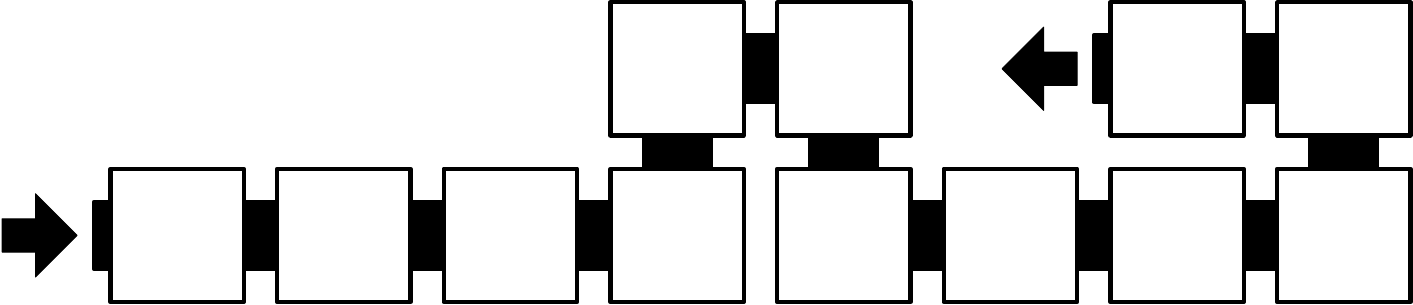}
        \caption{\label{fig:seed-right-1} {\tt Init\_right\_1}}
    \end{subfigure}
    \caption{\label{fig:init-gadgets} The generic {\tt Init} gadgets. Each {\tt Init} gadget geometrically represents one bit value. The bit value 1 is represented by the path of tiles protruding up one point and then back down, whereas the bit value 0 is represented by the absence of such a protuberance.}
\end{figure}
The {\tt Init} gadgets all self-assemble to the right, resulting in the geometric representation of the $k$-bit binary expansion of the starting value $s$ of the counter, but modified to include, for each integer $1 \leq i \leq k$, an {\em indicator} bit $b'_i$ inserted to the left of the {\em value} bit $b_i$ such that the only indicator bit having the value 1 is $b'_k$. More generally, the indicator bits that assemble in any row of the counter will be used to identify the last value bit to be read when the row above it assembles. Since the row above the starting value is an increment row that assembles to the left, the leftmost indicator bit is set to 1. In contrast, when the row above the current one is a copy row, which assembles to the right, the only indicator bit set to 1 in the current row will be the rightmost one. 
Figure~\ref{fig:example-seed-n-5} is a high-level depiction of the result of the self-assembly of a series of {\tt Init} gadgets for $s=5$.
In this and subsequent high-level example figures:
\begin{itemize}
	\item thin black lines outline each gadget, 
	\item the grey and white lines within the outlined gadgets indicate locations at which tiles are placed and a possible order in which the tiles attach, 
	\item the value bits $b_1$, $b_2$ and $b_3$ are highlighted in black, where the binary value 1 is geometrically represented with an upward protruding bump and the binary value 0 is geometrically represented with the absence of such a bump,
	\item the indicator bits $b'_1$, $b'_2$ and $b'_3$ are highlighted in grey and geometrically represented in the same fashion as the value bits, and 
	\item the circled numbers indicate the relative order in which the gadgets self-assemble and always appear at the point of the first tile of the gadget to attach.
\end{itemize}
\begin{figure}[h!]
    \centering
        \centering
        \includegraphics[width=\linewidth]{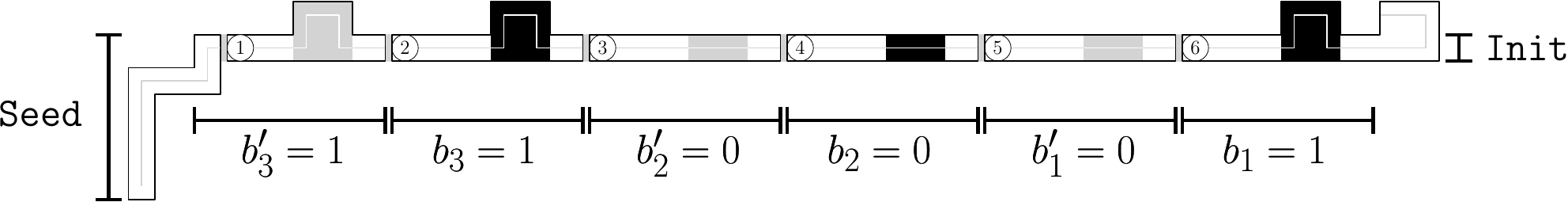}

        \caption{\label{fig:example-seed-n-5} A high-level example showing the result of the self-assembly of a series of {\tt Init} gadgets for the (unrealistic) case when $s=5$, $k=3$ and $e=5$. Note that the black regions correspond, from left to right, to the value bits $b_3 = 1$, $b_2 = 0$ and $b_1 = 1$, which is the binary representation of 5. Moreover, the grey regions correspond, from left to right, to the indicator bits $b'_3 = 1$, $b'_2 = 0$, and $b'_1 = 0$. }
\end{figure}
We now create the gadgets for the initialization step.
The first gadget to assemble in the counter initialization step is the unique {\tt Init\_left} gadget that geometrically represents the value for $b'_k = 1$.
Its input glue matches the output glue of the unique {\tt Seed} gadget, whereas its output glue ensures that the next gadget to assemble is an {\tt Init\_middle\_}$b_i$ gadget uniquely representing $b_k$. 
Therefore, create from the generic gadget in Figure~\ref{fig:seed} the gadget: 
\begin{align*}
{\tt Init\_left}( & {\tt init}, \\
			 	  & {\tt value\_bit\_}k ).
\end{align*}
For example, gadget 1 in Figure~\ref{fig:example-seed-n-5} is the unique {\tt Init\_left} gadget representing $b'_k$ in the $s = 5$ example. 

The {\tt Init\_left} gadget representing $b'_k$ is the first gadget representing an indicator bit associated with a value bit that is not the least significant bit of $s$. After each one of the indicator-bit gadgets $b'_i$ assembles, with $k \geq i \geq 2$, a unique {\tt Init\_middle\_}$b_i$ gadget corresponding to $b_i$ assembles. Its input glue matches the output glue of the unique {\tt Init} gadget representing $b'_i$, whereas its output glue ensures that the next gadget to assemble is an {\tt Init\_middle\_}$b'_{i-1}$ gadget representing $b'_{i-1}$.

Therefore, for each integer $k \geq i \geq 2$, if $b_i = 0$ (resp., $b_i = 1$), then create from the generic gadget in Figure~\ref{fig:seed-middle-0} (resp., Figure~\ref{fig:seed-middle-1}) the gadget: 
\begin{align*}
{\tt Init\_middle\_}b_i( & {\tt value\_bit\_}i, \\
			 	         & {\tt indicator\_bit\_}(i - 1).
\end{align*}
For example, gadgets 2 and 4 in Figure~\ref{fig:example-seed-n-5} are the unique {\tt Init\_middle} gadgets representing $b_3$ and $b_2$, respectively, in the $s = 5$ example.

Any one of the {\tt Init\_middle\_}$b_i$ gadgets representing a value bit $b_i$ that were just created is an example of an {\tt Init\_middle} gadget representing a value bit that is not the least significant bit of $s$. After each one of these value-bit gadgets $b_i$ assembles, with $k \geq i \geq 2$, a unique {\tt Init\_middle\_}$0$ gadget corresponding to the indicator-bit $b'_{i-1}=0$ assembles. Its input glue matches the output glue of the unique {\tt Init\_middle\_}$b_i$ gadget representing $b_i$, whereas its output glue ensures that the next gadget to assemble is the unique {\tt Init} gadget representing $b_{i-1}$. Note that all these indicator-bit gadgets must geometrically represent the value 0 to conform to the modified binary expansion of $s$ in which the only indicator bit that is set to 1 is the one associated with the most significant bit of $s$, because the row above the current row will assemble to the left.

Therefore, for each $k > i \geq 1$, create from the generic gadget in Figure~\ref{fig:seed-middle-0} the gadget:
\begin{align*}
{\tt Init\_middle\_}0( & {\tt indicator\_bit\_}i, \\
			 	         & {\tt value\_bit\_}i ).
\end{align*}
For example, gadgets 3 and 5 in Figure~\ref{fig:example-seed-n-5} are the unique {\tt Init\_middle\_}$0$ gadgets representing $b'_2$ and $b'_1$, respectively, in the $s = 5$ example.

The last gadget to assemble in the initialization step is the unique {\tt Init\_right\_}$b_1$ gadget representing $b_1$.
Its input glue matches the output glue of the unique {\tt Init\_middle\_}$0$ gadget representing $b'_1$. 

Therefore, if $b_1 = 0$ (resp., $b_1 = 1$), create from the generic gadget in Figure~\ref{fig:seed-right-0} (resp., Figure~\ref{fig:seed-right-1}) the gadget:
\begin{align*}
{\tt Init\_right\_}b_1( & {\tt value\_bit\_}1, \\
			 	         & \langle {\tt least\_significant\_value\_bit}, {\tt carry} = 1 \rangle ).
\end{align*}
Its output glue will initiate the first increment step.
For example, gadget 6 in Figure~\ref{fig:example-seed-n-5} is the unique {\tt Init\_right\_}$b_1$ gadget representing $b_1=1$ in the $s = 5$ example.
Note that, for each $1 \leq i \leq k$, one unique gadget represents each one of the $b_i$ and $b'_i$ bits. 
Thus, the number of {\tt Init} gadgets created in this step is $\Theta\left(k\right)$.
Moreover, since each such gadget has a fixed size, the total number of tile types contributed to $T$ by the creation of all the {\tt Init} gadgets is $\Theta\left(k\right) = O(h) = O\left( \log N + \log \frac{1}{\delta} \right)$.
Following the initialization step, the counter alternates between increment and copy steps.
Next, we create the gadgets for increment steps. 

\emph{Step 2c.}
%
%
After the initialization of the counter, its value is incremented in the subsequent increment step.
An increment step always follows the initialization step and, in general, will follow a copy step in which the value of the counter is not 0.
Assume that, prior to each increment step, $b'_k = 1$ and all other indicator bits are 0.  
An increment step proceeds via the self-assembly of a series of {\tt Inc} gadgets that all self-assemble to the left while carrying out the standard binary addition algorithm. 
All the {\tt Inc} gadgets have a fixed width and a height proportional to $h$, which, as we will later show, is sufficient to ensure $\mathcal{T}$ strictly self-assembles the shape of $\alpha$ with probability at least $1-\delta$.
We use four types of generic {\tt Inc} gadgets in our construction, namely {\tt Inc\_read}, {\tt Inc\_write\_}$0$, {\tt Inc\_write\_}$1$, and {\tt Inc\_to\_copy} gadgets.
An {\tt Inc\_read} gadget non-deterministically guesses the unknown binary value of the geometric representation of a bit, and therefore has two gadget assemblies, which are depicted in Figures~\ref{fig:Inc-read-0-good} and~\ref{fig:Inc-read-0-bad} and always correspond to the binary values 0 and 1, respectively.
\begin{figure}[h!]
    \centering
    \begin{subfigure}[t]{0.3\textwidth}
        \centering
        \includegraphics[width=1.2in]{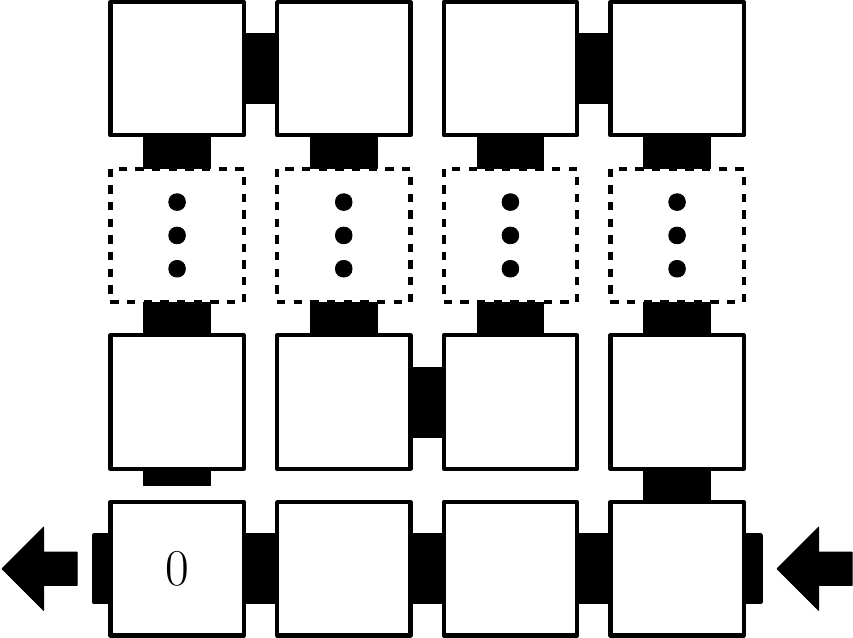}
        \caption{\label{fig:Inc-read-0-good} Guess a 0 bit correctly.}
    \end{subfigure}
    ~
    \begin{subfigure}[t]{0.3\textwidth}
        \centering
        \includegraphics[width=1.2in]{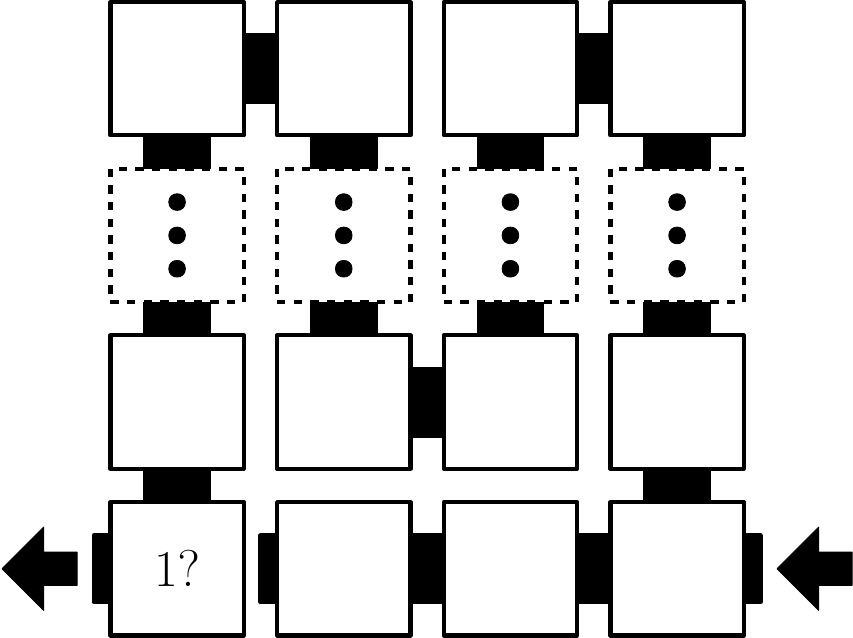}
        \caption{\label{fig:Inc-read-0-bad} Incorrectly guess a 0 bit as a 1. The idea is that $h$ is chosen such that this scenario is unlikely to occur. }
    \end{subfigure}
    ~
    \begin{subfigure}[t]{0.3\textwidth}
        \centering
        \includegraphics[width=1.2in]{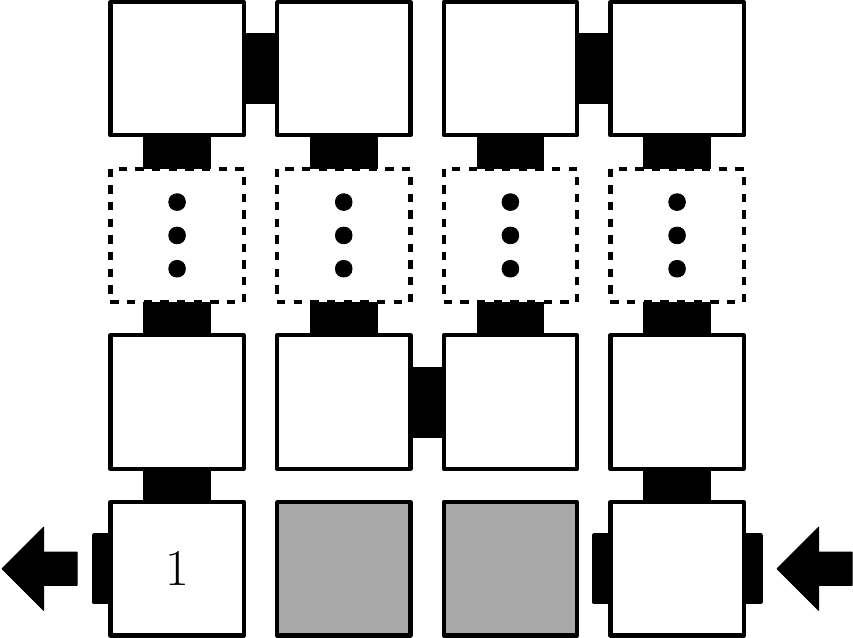}
        \caption{\label{fig:Inc-read-1-good} Always guess a 1 bit correctly as a 1. The grey tiles represent tiles that are guaranteed to be placed before the input tile of the {\tt Inc\_read} gadget is placed. }
    \end{subfigure}
    \caption{\label{fig:Inc-read} The assemblies for the generic {\tt Inc\_read} gadget are shown in parts~(\subref{fig:Inc-read-0-good}) and~(\subref{fig:Inc-read-0-bad}). The vertical dotted tiles in this and subsequent gadget figures represent columns of $h-5$ tiles such that $h-5\geq 1$. Part (\subref{fig:Inc-read-1-good}) is not a gadget assembly but rather the unique subassembly of the gadget assembly in part (\subref{fig:Inc-read-0-bad}) that corresponds to correctly guessing a 1 bit. Note that the point at which the lower-left tile is placed is a POC in $\mathcal{T}$ and the grey tiles in part~(\subref{fig:Inc-read-1-good}) represent a portion of some assembly that rigs the corresponding competition.
    }
\end{figure}
We use two generic {\tt Inc\_write} gadgets in our construction, namely {\tt Inc\_write\_}$0$ and {\tt Inc\_write\_}$1$, shown in Figures~\ref{fig:Inc-write-0} and~\ref{fig:Inc-write-1}, respectively.
\begin{figure}[h!]
    \centering
    \begin{subfigure}[t]{0.4\textwidth}
        \centering
        \includegraphics[width=1.5in]{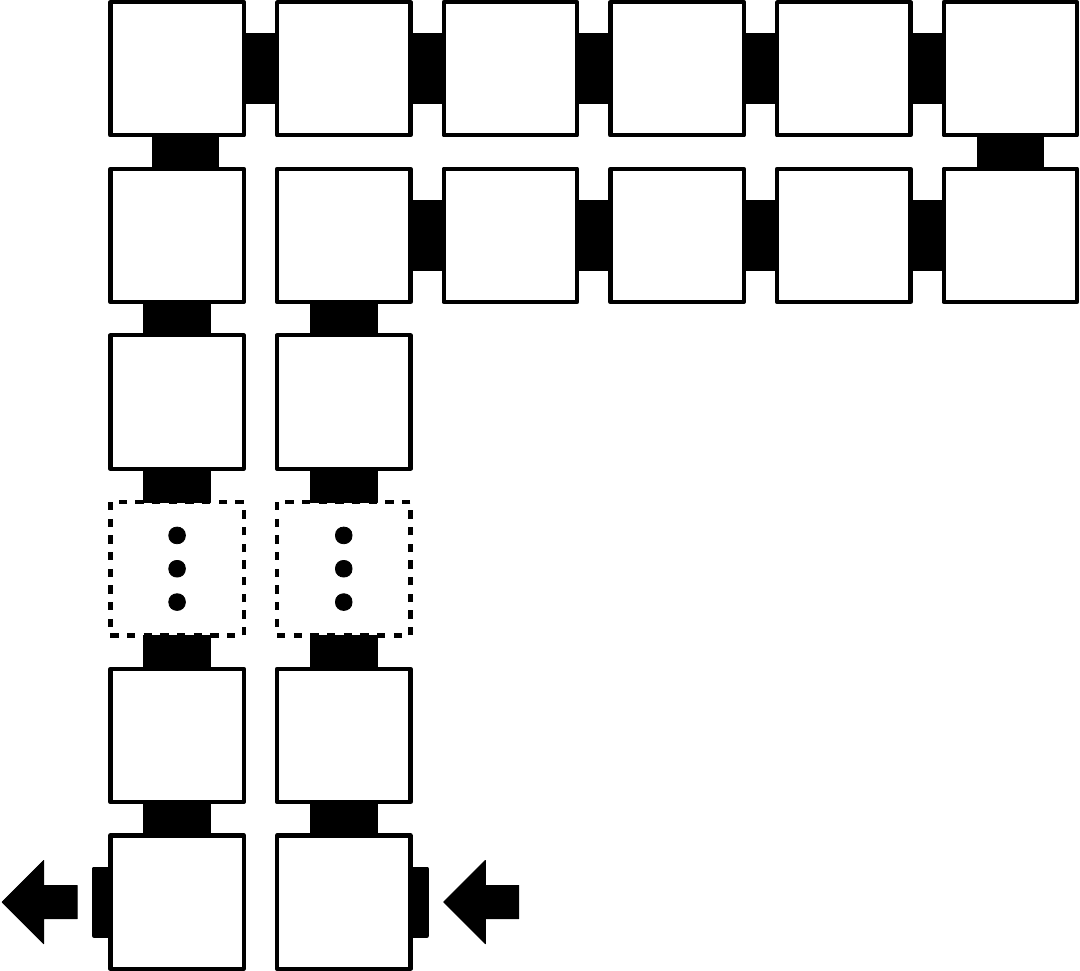}
        \caption{\label{fig:Inc-write-0} {\tt Inc\_write\_0} }
    \end{subfigure}
    ~
    \begin{subfigure}[t]{0.4\textwidth}
        \centering
        \includegraphics[width=1.5in]{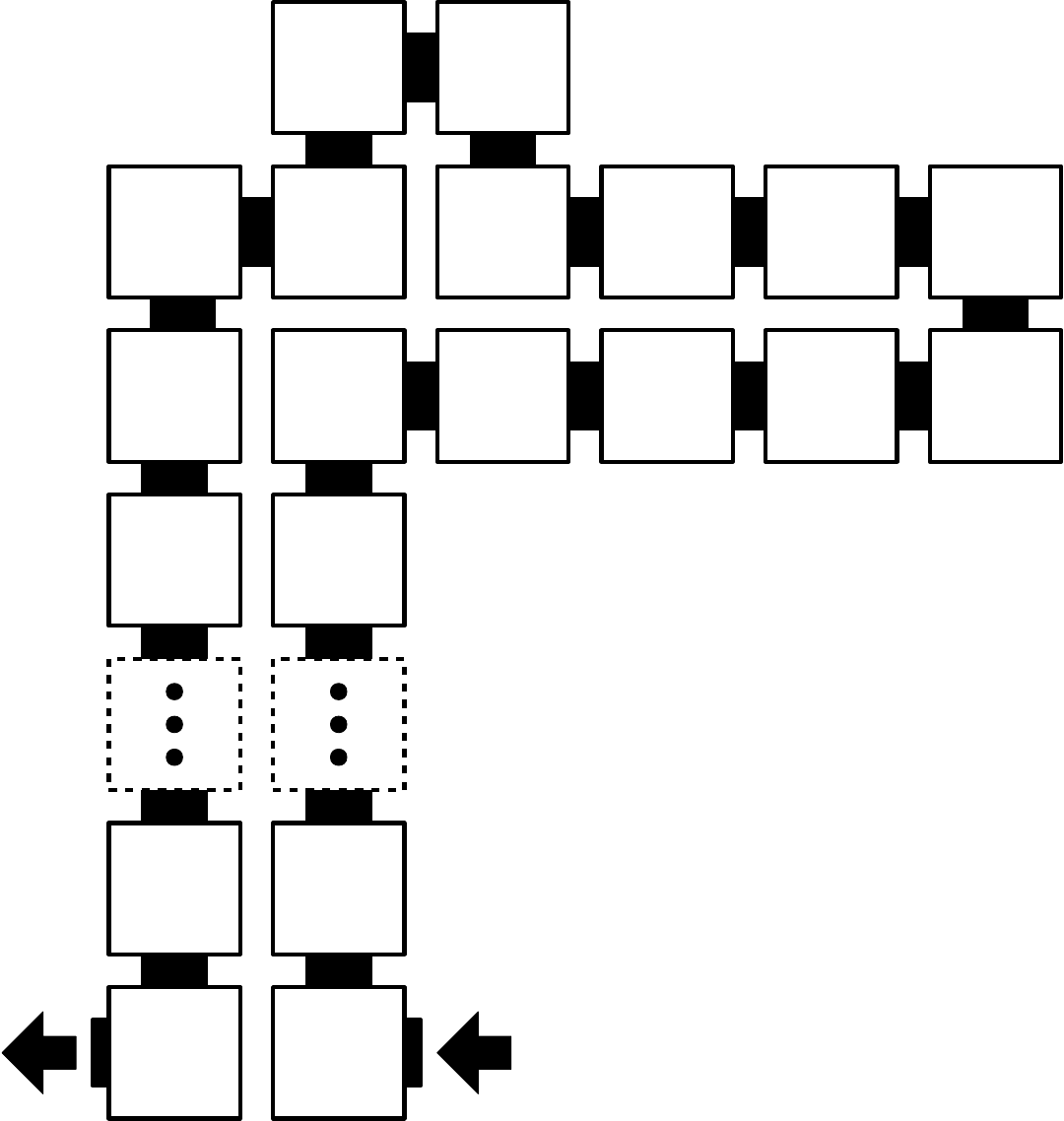}
        \caption{\label{fig:Inc-write-1} {\tt Inc\_write\_1}}
    \end{subfigure}
    \caption{\label{fig:Inc-write} The assemblies for the generic {\tt Inc\_write} gadgets. }
\end{figure}
An {\tt Inc\_write} gadget has a unique assembly that corresponds to the geometric representation of a binary value in an increment step. 

The transition from a completed increment step to the subsequent copy step is facilitated by an {\tt Inc\_to\_copy} gadget, whose generic version is shown in Figure~\ref{fig:Inc-to-copy-gen}.

\begin{figure}[h!]
  \centering
  \begin{subfigure}[t]{0.3\textwidth}
    \centering
    \includegraphics[width=0.6in]{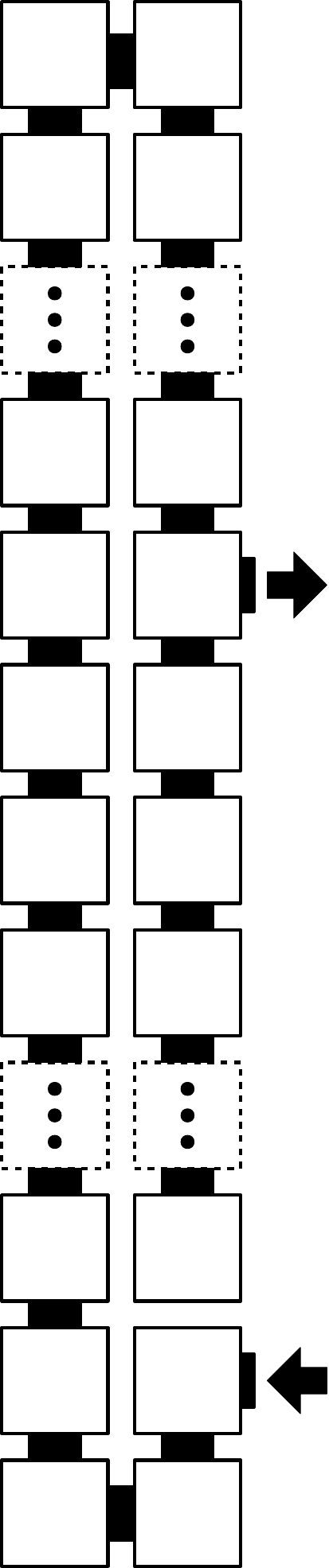}
    \caption{\label{fig:Inc-to-copy-gen} {\tt Inc\_to\_copy}}
    \end{subfigure}
    ~
  \begin{subfigure}[t]{0.3\textwidth}
    \centering
    \includegraphics[width=0.6in]{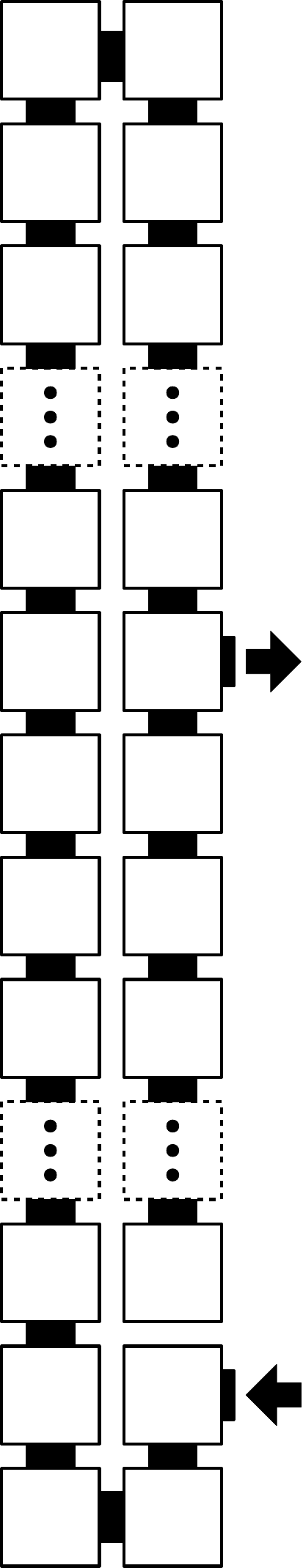}
    \caption{\label{fig:Inc-to-copy-last} {\tt Inc\_to\_copy\_last}}
  \end{subfigure}
      \caption{\label{fig:Inc-to-copy} The assemblies for the generic {\tt Inc\_to\_copy} gadget and the {\tt Inc\_to\_copy\_last} gadget.}
\end{figure}

Figure~\ref{fig:example-inc-n-5} is a continuation of Figure~\ref{fig:example-seed-n-5} and contains a high-level depiction of the self-assembly of the gadgets used in the first increment step of the counter in the $s=5$ example. In that figure, the value of the counter is incremented from 5 to 6.
For the sake of clarity, in Figure~\ref{fig:example-inc-n-5} and subsequent high-level figures, we assume $h = 6$ and use circled checkmarks ($\surd$) to indicate the locations of essential POCs in $\mathcal{T}$. 
\begin{figure}[h!]
    \centering
        \centering
        \includegraphics[width=\linewidth]{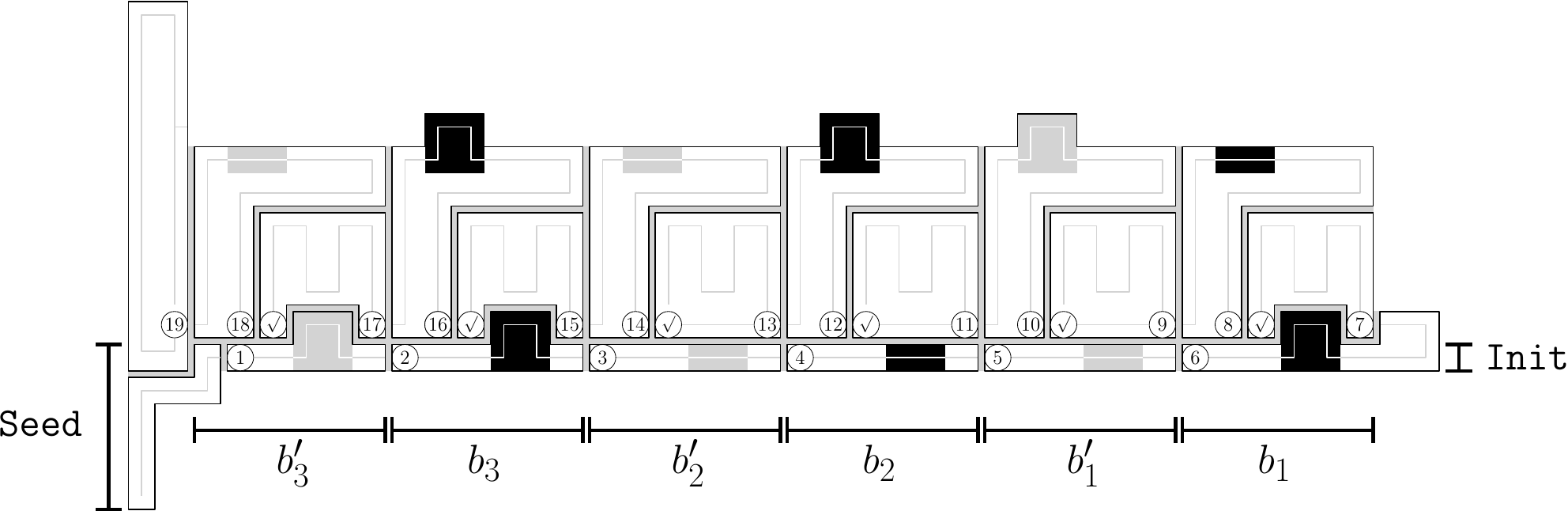}
        \caption{\label{fig:example-inc-n-5} A high-level example showing the first increment step of the counter whose value is incremented from 5 to 6. Note that the top row of black regions correspond, from left to right, to the value bits $b_3 = 1$, $b_2 = 1$, and $b_1 = 0$, which is the binary representation of 6. The top row of grey regions correspond, from left to right, to the indicator bits $b'_3 = 0$, $b'_2 = 0$, and $b'_1 = 1$. }
\end{figure}
Each {\tt Inc} gadget in an increment step is associated either with a single, specific bit or with a group of bits in the value of the counter.
For example, a gadget associated with a specific bit, e.g., an {\tt Inc\_read} gadget associated with $b_1$, ``knows'' the position of the bit with which it is associated in the sense that the index of the bit is explicitly encoded in the glues of its tile types.
Note that, for every bit (indicator or value), there is a unique {\tt Init} gadget associated with it. 
However, a gadget associated with a group of bits, e.g., an {\tt Inc\_write} gadget associated with any indicator bit other than $b'_k$, is ``oblivious'' to the position of the bit with which it is associated in the sense that the glues of its tile types are not a function of the index of the bit.
The idea is that a gadget associated with a group of bits may assemble multiple times within an increment step. 
To summarize, we create:
\begin{itemize}
	\item one {\tt Inc\_read} gadget associated with, and only with, $b_1$,
	\item one {\tt Inc\_write\_}$0$ gadget and one {\tt Inc\_write\_}$1$ gadget, both  associated with, and only with, $b_1$,
	\item two {\tt Inc\_read} gadgets associated with, and only with, $b'_1$,
	\item four {\tt Inc\_write} gadgets associated with any non-most significant indicator bit $b'_i$, for all integers $1~\leq~i~<~k$, 
	\item two {\tt Inc\_read} gadgets associated with any non-least significant value bit,
	\item three {\tt Inc\_write} gadgets associated with any non-least significant value bit,
	\item two {\tt Inc\_read} gadgets associated with any non-least significant indicator bit, $b'_i$,
	\item two {\tt Inc\_write\_}$0$ gadgets associated with, and only with, $b'_k$,
	\item one {\tt Inc\_to\_copy} gadget associated with, and only with, $b'_k$, and
        \item a unique {\tt Inc\_to\_copy\_last} gadget associated with, and only with, $b'_k$.
\end{itemize}
Moreover, all of these gadgets are created so that:
\begin{enumerate}
	\item they all self-assemble to the left, 	
		\item at the conclusion of the increment step, it is the case that $b'_1 = 1$ and all other indicator bits are set to $0$,
	\item for each output glue of every {\tt Inc} gadget except {\tt Inc\_to\_copy} and  {\tt Inc\_to\_copy\_last}, there exists a unique {\tt Inc} gadget with a matching input glue, and
        \item for each output glue of the {\tt Inc\_to\_copy} and  {\tt Inc\_to\_copy\_last} gadgets, there exists a unique {\tt Copy} gadget with a matching input glue.
\end{enumerate} 
Prior to every increment step, e.g., after the initialization step, it must be the case that $b'_k = 1$ and all other indicator bits are set to 0. 
An {\tt Inc\_read} gadget associated with any non-least significant indicator bit, e.g., $b'_k$, can determine if the index of the indicator bit with which it is associated is $k$ by correctly guessing a 1. This {\tt Inc\_read} gadget then exposes its output glue, which only the input glue of the {\tt Inc\_write\_}$0$ gadget associated with $b'_k$ matches.

Therefore, unlike in the initialization step, we need not associate with each bit (indicator or value) a corresponding {\tt Inc} gadget. It follows that $O(1)$ {\tt Inc} gadgets suffice for our construction.
Since the size of each {\tt Inc} gadget is $\Theta(h)$ and the set of tile types for a gadget must be disjoint from the set of tile types for every other gadget, the number of tile types contributed to $T$ by all the {\tt Inc} gadgets is $\Theta(h)$.
We now create the {\tt Inc} gadgets.
%

%
%
The first gadget in an increment step is a unique {\tt Inc\_read} gadget that is  associated with, and only with, $b_1$ and non-deterministically guesses its unknown binary value.
The input glue of this gadget indicates that it is associated with $b_1$ and that the value of the incoming carry bit is 1.
Moreover, the input glue of this gadget is the unique glue over all the gadgets in our construction that matches the output glue  of the unique {\tt Init\_right} gadget, thus ensuring that this gadget is the first {\tt Inc} gadget to assemble in the first increment step. 
The output glues of this gadget indicate that the next gadget must be an {\tt Inc\_write} gadget associated with $b_1$. They also propagate the incremented value of $b_1$, along with the outgoing carry bit.
Note that the latter two values are computed, from the incoming carry bit and the guessed value of $b_1$, in the spirit of a half adder. 
Therefore, create from the generic gadget in Figure~\ref{fig:Inc-read} the gadget:
\begin{align*}
{\tt Inc\_read}( & \langle {\tt least\_significant\_value\_bit}, {\tt carry} = 1 \rangle, \\
			 	& \langle {\tt least\_significant\_value\_bit} = 1, {\tt carry} = 0 \rangle, \\
				& \langle {\tt least\_significant\_value\_bit} = 0,{\tt carry} = 1 \rangle  ).
\end{align*}
For example, gadget 7 in Figure~\ref{fig:example-inc-n-5} is the unique {\tt Inc\_read} gadget associated with $b_1$ in the $s=5$ example. The corresponding gadget assembly is shown in Figure~\ref{fig:Inc-read-1-good}.
Note that, if $b_1 = 1$, then its output glue is $\langle {\tt least\_significant\_value\_bit} = 0, {\tt carry} = 1 \rangle$, which is always correct.
However, if $b_1 = 0$, then its output glue is not guaranteed to be correct, because the assembly in either Figure~\ref{fig:Inc-read-0-good} or Figure~\ref{fig:Inc-read-0-bad} may assemble.
%

%
%
The second gadget in an increment step is an {\tt Inc\_write} gadget associated with, and only with, $b_1$.
We create two such gadgets, one for each binary value $b$ that was propagated via the output glue of the unique {\tt Inc\_read} gadget associated with $b_1$.
For each gadget, its input glue matches one of the output glues of the {\tt Inc\_read} gadget associated with $b_1$, whereas its output glue indicates that the next gadget must be an {\tt Inc\_read} gadget associated with $b'_1$. The output glue also propagates the value of the incoming carry bit.
Therefore, for each $b \in \{0,1\}$, if $b = 0$ (resp., $b = 1$), create from the generic gadget in Figure~\ref{fig:Inc-write-0} (resp., Figure~\ref{fig:Inc-write-1}) the gadget:
	\begin{align*}
{\tt Inc\_write\_}b( & \langle {\tt least\_significant\_value\_bit} = b, {\tt carry} = 1 - b \rangle, \\
			 	& \langle {\tt least\_significant\_indicator\_bit}, {\tt carry}=1-b\rangle  ).
\end{align*}
For example, gadget 8 in Figure~\ref{fig:example-inc-n-5} is the {\tt Inc\_write\_}$0$ gadget for $b_1$ in the $s=5$ example.
%

%
%
The third gadget in an increment step is an {\tt Inc\_read} gadget that is associated with, and only with, $b'_1$ and non-deterministically guesses its unknown binary value.
We create two such gadgets, one for each binary value of the carry bit propagated via the output glue of each {\tt Inc\_write} gadget associated with $b_1$.
For each gadget, its input glue matches the output glue of an {\tt Inc\_write} gadget associated with $b_1$, whereas its output glues indicate that the next gadget must be an {\tt Inc\_write\_}$1$ gadget associated with any non-most significant indicator bit. The output glues also propagate the incoming carry bit.
Therefore, for each $c \in \{0,1\}$, create from the generic gadget in Figure~\ref{fig:Inc-read} the gadget:
\begin{align*}
{\tt Inc\_read}( & \langle {\tt least\_significant\_indicator\_bit}, {\tt carry}=c \rangle, \\
			 	& \langle {\tt non\_most\_significant\_indicator\_bit} = 1, {\tt carry}=c \rangle, \\
				& \langle {\tt non\_most\_significant\_indicator\_bit} = 1, {\tt carry}=c \rangle  ).
\end{align*}
For example, gadget 9 in Figure~\ref{fig:example-inc-n-5} is the {\tt Inc\_read} gadget for $b'_1$ in the $s=5$ example.
Note that the output glues of this gadget ensure that the next gadget is an {\tt Inc\_write\_}$1$ gadget. Therefore it will be the case that $b'_1 = 1$ at the conclusion of the increment step.
%

%
%
The fourth gadget in an increment step is an {\tt Inc\_write} gadget that is associated with any non-most significant indicator bit. 
We create four such gadgets, one for each pair of binary values $b$ and $c$ such that:
\begin{enumerate}
	\item $b$ is the value of the indicator bit that was propagated via the output glue of the previous {\tt Inc\_read} gadget associated with the same non-most significant indicator bit, e.g., the {\tt Inc\_read} gadget associated with $b'_1$, and 	
	\item $c$ is the carry bit propagated via the output glue of the {\tt Inc\_read} gadget associated with the same non-most significant indicator bit.
\end{enumerate}
For each gadget, its input glue matches a corresponding output glue of the {\tt Inc\_read} gadget associated with the same indicator bit, whereas its output glue indicates that the next gadget must be an {\tt Inc\_read} gadget associated with any non-least significant value bit. The output glue also propagates the incoming carry bit.
Therefore, for each $b,c \in \{0,1\}$, if $b = 0$ (resp., $b = 1$), then create from the generic gadget in Figure~\ref{fig:Inc-write-0} (resp., Figure~\ref{fig:Inc-write-1}) the gadget:
\begin{align*}
{\tt Inc\_write\_}b( & \langle {\tt non\_most\_significant\_indicator\_bit} = b, {\tt carry}=c\rangle, \\
			 	&  \langle {\tt non\_least\_significant\_value\_bit}, {\tt carry}=c \rangle  ).
\end{align*}
For example, gadgets 10 and 14 in Figure~\ref{fig:example-inc-n-5} are the {\tt Inc\_write} gadgets associated with non-most significant indicator bits in the $s=5$ example.
%

%
%
After an {\tt Inc\_write} gadget associated with any non-most significant indicator bit, e.g., an {\tt Inc\_write} gadget associated with $b'_1$, the next gadget in an increment step is an {\tt Inc\_read} gadget that is associated with, and non-deterministically guesses, the unknown binary value of any non-least significant value bit.
We create two such gadgets, one for each carry bit propagated via the output glue of an {\tt Inc\_write} gadget associated with any non-most significant indicator bit.
For each gadget, its input glue matches the output glue of a unique {\tt Inc\_write} gadget associated with any non-most significant indicator bit, whereas its output glue indicates that the next gadget must be an {\tt Inc\_write} gadget associated with any non-least significant value bit. Its output glue also propagates the incoming {\tt carry} bit.
Therefore, for each $c \in \{0,1\}$, create from the generic gadget in Figure~\ref{fig:Inc-read} the following gadget:
\begin{align*}
{\tt Inc\_read}( & \langle {\tt non\_least\_significant\_value\_bit}, {\tt carry}=c \rangle, \\
			 	& \langle {\tt non\_least\_significant\_value\_bit} = c, {\tt carry}= 0 \rangle, \\
				& \langle {\tt non\_least\_significant\_value\_bit} = 1 - c, {\tt carry}= c \rangle  ).
\end{align*}
For example, gadgets 11 and 15 in Figure~\ref{fig:example-inc-n-5} are the {\tt Inc\_read} gadgets for all non-most significant value bits in the $s=5$ example.
Note that the output glues of these gadgets are configured in the spirit of a half adder as follows.
Suppose, for some integer $2 \leq i \leq k$, that these gadgets are associated with $b_i$, and $c$ is the carry bit propagated via the output glue of an {\tt Inc\_write} gadget associated with $b'_{i-1}$.
If these gadgets non-deterministically guess the binary value of $b_i$ as $b \in \{0,1\}$, then their output glues propagate the corresponding sum bit $b \oplus c$ and carry bit $b \wedge c$.
%

%
%
%
After an {\tt Inc\_read} gadget associated with any non-least significant value bit, the next gadget in an increment step is an {\tt Inc\_write} gadget that is associated with any non-least significant value bit.
We create three such gadgets, one for each pair of binary values $b$ and $c$ such that:
\begin{enumerate}
	\item $b$ is the binary value of the bit that was propagated via the output glue of a previous {\tt Inc\_read} gadget associated with the same non-least significant value bit, and
	\item $c$ is the carry bit propagated via the output glue of a previous {\tt Inc\_read} gadget associated with the same non-least significant value bit, but
	\item $b$ and $c$ are not both 1.
\end{enumerate}
We exclude the combination $b = c = 1$ because it does not correspond to an output glue of an {\tt Inc\_read} gadget associated with any non-least significant value bit.
For each gadget, its input glue matches a corresponding output glue of an {\tt Inc\_read} gadget associated with the same non-least significant value bit, whereas its output glue indicates that the next gadget must be an {\tt Inc\_read} gadget associated with any non-least significant indicator bit. Its output glue also  propagates the incoming {\tt carry} bit.
Therefore, for each $(b,c)\in\left\{ (0,0), (0,1), (1,0) \right\}$, if $b = 0$ (resp., $b = 1$), then create from the generic gadget in Figure~\ref{fig:Inc-write-0} (resp., Figure~\ref{fig:Inc-write-1}) the gadget:
	\begin{align*}
{\tt Inc\_write\_}b( & \langle {\tt non\_least\_significant\_value\_bit} = b, {\tt carry}=c \rangle, \\
			 	&  \langle {\tt non\_least\_significant\_indicator\_bit}, {\tt carry}=c \rangle  ).
\end{align*}
For example, gadgets 12 and 16 in Figure~\ref{fig:example-inc-n-5} are the {\tt Inc\_write} gadgets associated with non-least significant bits in the $s=5$ example.
%

%
%
After an {\tt Inc\_write} gadget associated with any non-least significant value bit, the next gadget in an increment step is an {\tt Inc\_read} gadget that is associated with, and non-deterministically guesses, the unknown binary value of any non-least significant indicator bit.
We create two such gadgets, one for each carry bit value propagated via the output glue of an {\tt Inc\_write} gadget associated with any non-least significant value bit. 
For each newly created {\tt Inc\_read} gadget, its input glue matches the output glue of the {\tt Inc\_write} gadget associated with the preceding non-least significant value bit, whereas its output glue depends on the non-deterministically guessed value of the indicator bit with which the gadget is associated:
\begin{enumerate}
	\item If the guessed value is $0$, then the corresponding output glue of the gadget indicates that the next gadget must be an {\tt Inc\_write\_}$0$ gadget associated with any non-least significant indicator bit but $b'_k$.
	\item Otherwise, the corresponding output glue of the gadget indicates that the next gadget must be an {\tt Inc\_write\_}$0$ gadget associated with, and only with, $b'_k$.
\end{enumerate}
Regardless of the non-deterministically guessed value, both output glues of this gadget propagate the incoming carry bit.
Therefore, for each $c \in \{0,1\}$, create from the generic gadget in Figure~\ref{fig:Inc-read} the gadget:

\begin{align*}
{\tt Inc\_read}( & \langle {\tt non\_least\_significant\_indicator\_bit}, {\tt carry}=c \rangle, \\
			 	& \langle {\tt non\_most\_significant\_indicator\_bit} = 0, {\tt carry}=c \rangle, \\
				& \langle {\tt most\_significant\_indicator\_bit} = 0, {\tt carry}=c \rangle  ).
\end{align*}
For example, gadgets 13 and 17 in Figure~\ref{fig:example-inc-n-5} are the {\tt Inc\_read} gadgets for $b'_2$ and $b'_3$, respectively, in the $s=5$ example.
These gadgets effectively determine whether to terminate the increment step.
Suppose that for some integer $2 \leq i \leq k$, the indicator bit with which these gadgets are associated is $b'_i$ and $c \in \{0,1\}$ is the carry bit propagated via the output glue of the {\tt Inc\_write} gadget associated with $b_i$.
If either gadget non-deterministically guesses $b'_i$ correctly as $0$, then the increment step continues.
If $b'_i = 1$ and $c = 0$, then the transition to the subsequent copy step is initiated.
If $b'_i = 1$ and $c = 1$, then the process of terminating the counter is initiated. 
In every case, the next gadget will be an {\tt Inc\_write\_}$0$ gadget, thus ensuring that all indicator bits other than $b'_1$ are $0$ after the current increment step.
%

%
%

%
The penultimate gadget in an increment step is an {\tt Inc\_write\_}$0$ gadget associated with, and only with, $b'_k$. 
We create two such gadgets as follows.

One of these gadgets has an input glue that matches the second output glue of an {\tt Inc\_read} gadget associated with any non-least significant indicator bit that propagates a carry of 0. Its output glue initiates the transition to the subsequent, non-final copy step. Therefore, create from the generic gadget in Figure~\ref{fig:Inc-write-0} the gadget:
\begin{align*}
{\tt Inc\_write\_}0( & \langle {\tt most\_significant\_indicator\_bit} = 0, {\tt carry} = 0 \rangle, \\ & {\tt inc\_to\_copy}  ).
\end{align*}

The other one of these gadgets has an input glue that matches the second output glue of an {\tt Inc\_read} gadget associated with any non-least significant indicator bit that propagates a carry of 1. Its output glue initiates the transition to the subsequent copy step and also marks it as the final copy step, thereby initiating the  termination of the counter. Therefore, create from the generic gadget in Figure~\ref{fig:Inc-write-0} the gadget:
\begin{align*}
{\tt Inc\_write\_}0( & \langle {\tt most\_significant\_indicator\_bit} = 0, {\tt carry} = 1 \rangle, \\ & {\tt inc\_to\_copy\_done}  ).
\end{align*}
For example, gadget 18 in Figure~\ref{fig:example-inc-n-5} is the {\tt Inc\_write\_}$0$ gadget for $b'_3$ in the $s=5$ example. This gadget initiates the transition to the subsequent, non-final copy step.
%

%
%
If the output glue of the {\tt Inc\_write\_}$0$ gadget associated with $b'_k$ is {\tt inc\_to\_copy}, then the last gadget in an increment step is an {\tt Inc\_to\_copy} gadget associated with, and only with, $b'_k$, which initiates the transition to the subsequent, non-final copy step.
Its input glue matches the output glue of the {\tt Inc\_write\_}$0$ gadget associated with $b'_k$ and an incoming carry bit of 0, whereas its output glue indicates that a {\tt Copy\_read} gadget (to be described below) is the next gadget and that the next copy step is not the final step in the counter.
Therefore, create from the generic gadget in Figure~\ref{fig:Inc-to-copy-gen} the gadget:
\begin{align*}
{\tt Inc\_to\_copy}( & {\tt inc\_to\_copy}, \\ & \langle {\tt most\_significant\_indicator\_bit}, {\tt done} = 0 \rangle  ).
\end{align*}
For example, gadget 19 in Figure~\ref{fig:example-inc-n-5} is the {\tt Inc\_to\_copy} gadget in the $s=5$ example.
%

%
%
If the output glue of the {\tt Inc\_write\_}$0$ gadget associated with $b'_k$ is {\tt inc\_to\_copy\_done}, then the last gadget in an increment step is an {\tt Inc\_to\_copy\_last} gadget associated with, and only with, $b'_k$, which initiates the transition to the final copy step.
Its input glue matches the output glue of the {\tt Inc\_write\_}$0$ gadget associated with $b'_k$ and an incoming carry bit of 1, whereas its output glue indicates that a {\tt Copy\_read} gadget is the next gadget and that the subsequent copy step is the final step in the counter.
Therefore, create from the generic gadget in Figure~\ref{fig:Inc-to-copy-last} the gadget:
\begin{align*}
{\tt Inc\_to\_copy\_last}( & {\tt inc\_to\_copy\_done}, \\ & \langle {\tt most\_significant\_indicator\_bit}, {\tt done} = 1 \rangle  ).
\end{align*}
Note that at the conclusion of the final increment step, every value bit is equal to  0.

\emph{Step 2d.}
%
%
We now define the gadgets for copy steps.
In general, after every increment step that ends with an {\tt Inc\_to\_copy} gadget, a copy step will follow.
Assume that, prior to each copy step, $b'_1 = 1$ and all other indicator bits are $0$, and that, if this is the final copy step, all of the value bits are $0$.
In a copy step, a row of {\tt Copy} gadgets self-assemble to the right without changing any of the value bits of the counter.
Each {\tt Copy} gadget, like its {\tt Inc} counterpart, has a fixed width and a height proportional to $h$, but it self-assembles to the right.
We use four types of {\tt Copy} gadgets in our construction, namely {\tt Copy\_read}, {\tt Copy\_write\_}$0$, {\tt Copy\_write\_}$1$, and {\tt Copy\_to\_inc}.
The assemblies for the {\tt Copy\_read} gadget are shown in Figure~\ref{fig:Copy-read}.
\begin{figure}[h!]
    \centering
    \begin{subfigure}[t]{0.3\textwidth}
        \centering
        \includegraphics[width=1.2in]{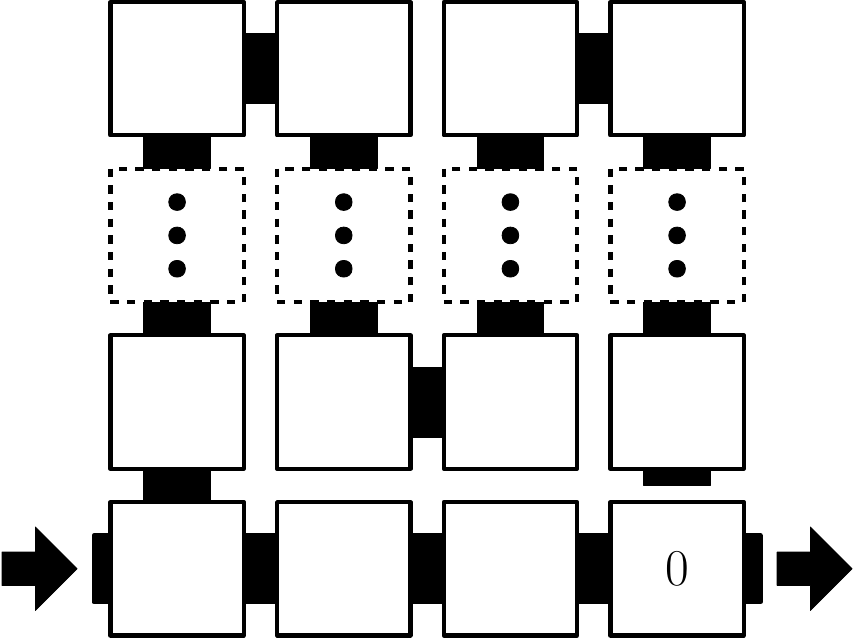}
        \caption{\label{fig:Copy-read-0-good} Guess a 0 bit correctly. }
    \end{subfigure}
    ~
    \begin{subfigure}[t]{0.3\textwidth}
        \centering
        \includegraphics[width=1.2in]{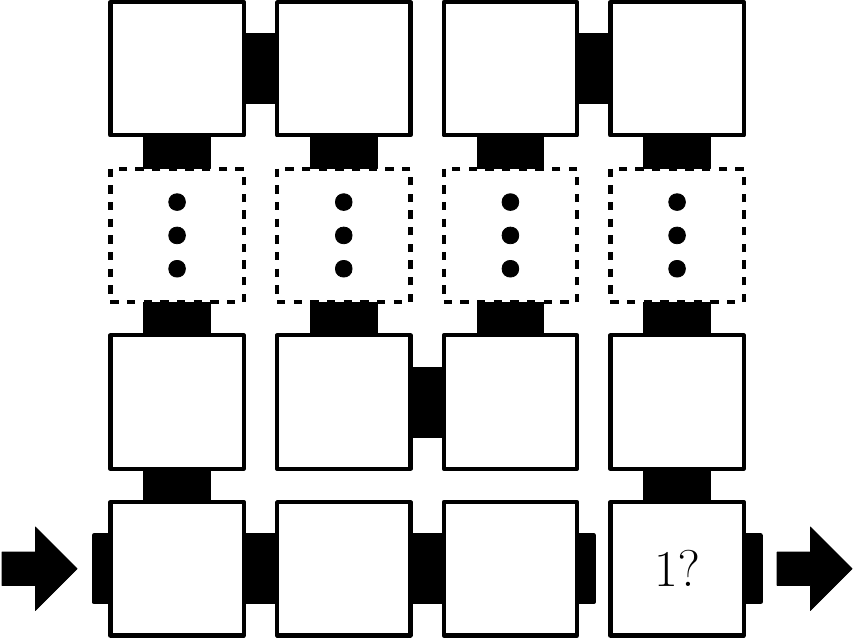}
        \caption{\label{fig:Copy-read-0-bad} Incorrectly guess a 0 bit as a 1. The idea is that $h$ is chosen such that this scenario is unlikely to occur. }
    \end{subfigure}
    ~
    \begin{subfigure}[t]{0.3\textwidth}
        \centering
        \includegraphics[width=1.2in]{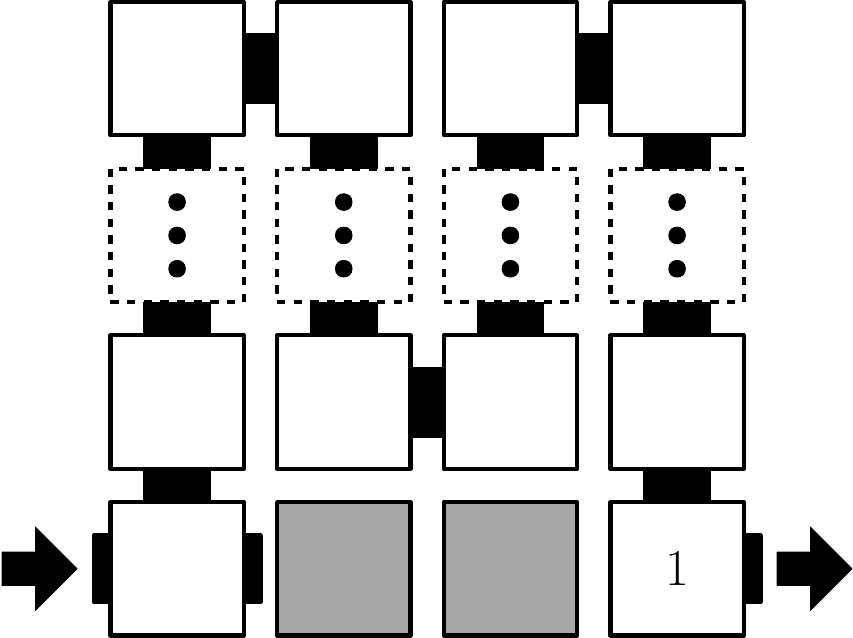}
        \caption{\label{fig:Copy-read-1-good} Always guess a 1 bit correctly as a 1. The grey tiles represent tiles that are guaranteed to be placed before the input tile of the {\tt Copy\_read} is placed.}
    \end{subfigure}
    \caption{\label{fig:Copy-read} The assemblies for the generic {\tt Copy\_read} gadget are shown in parts (\subref{fig:Copy-read-0-good}) and~(\subref{fig:Copy-read-0-bad}). Part (\subref{fig:Copy-read-1-good}) is not a gadget assembly but rather the unique subassembly of the gadget assembly in part (\subref{fig:Copy-read-0-bad}) that corresponds to correctly guessing a 1 bit. Note that the point at which the lower-right tile is placed is a POC in $\mathcal{T}$ and the grey tiles in part~(\subref{fig:Copy-read-1-good}) represent a portion of some assembly that rigs the corresponding competition.  }
\end{figure}
We use two generic {\tt Copy\_write} gadgets in our construction, namely  {\tt Copy\_write\_}$0$ and {\tt Copy\_write\_}$1$, shown in Figures~\ref{fig:Copy-write-0} and~\ref{fig:Copy-write-1}, respectively.
\begin{figure}[h!]
    \centering
    \begin{subfigure}[t]{0.4\textwidth}
        \centering
        \includegraphics[width=1.5in]{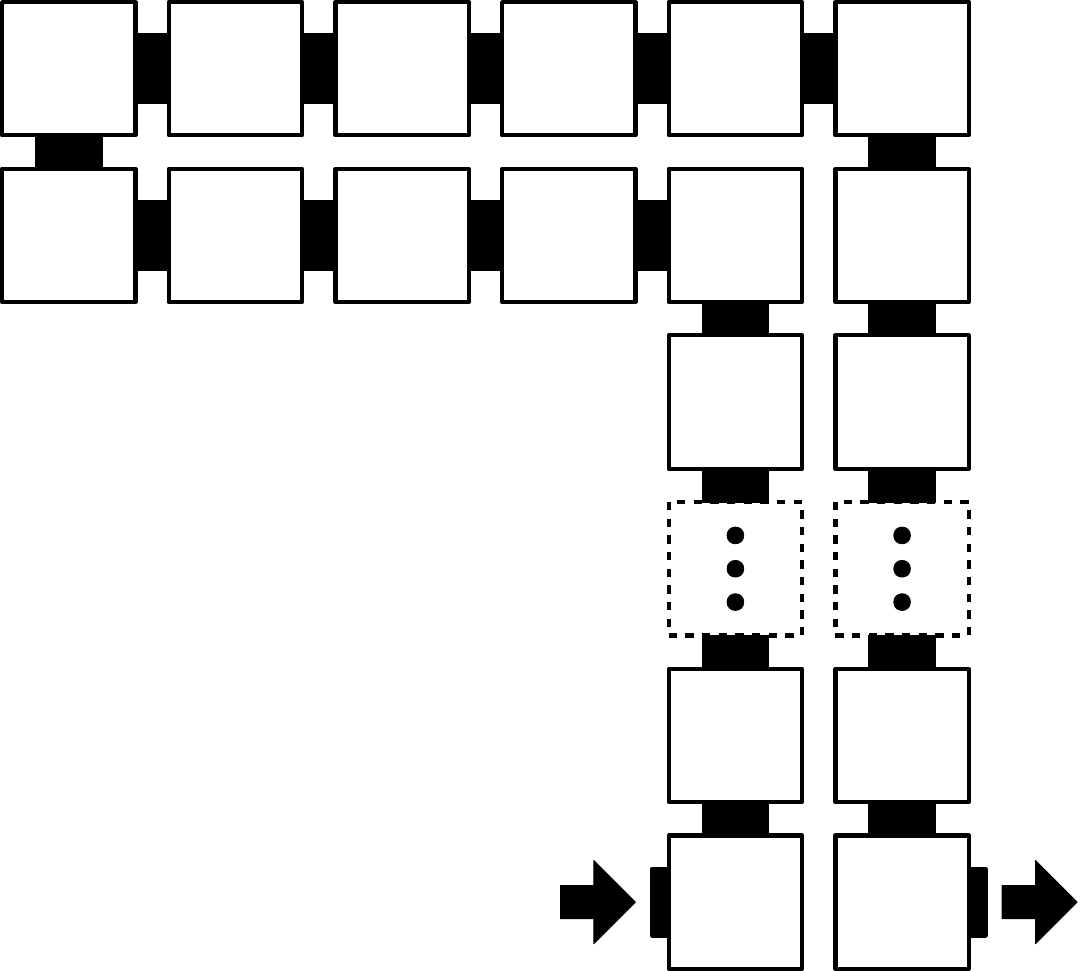}
        \caption{\label{fig:Copy-write-0} {\tt Copy\_write\_0} }
    \end{subfigure}
    ~
    \begin{subfigure}[t]{0.4\textwidth}
        \centering
        \includegraphics[width=1.5in]{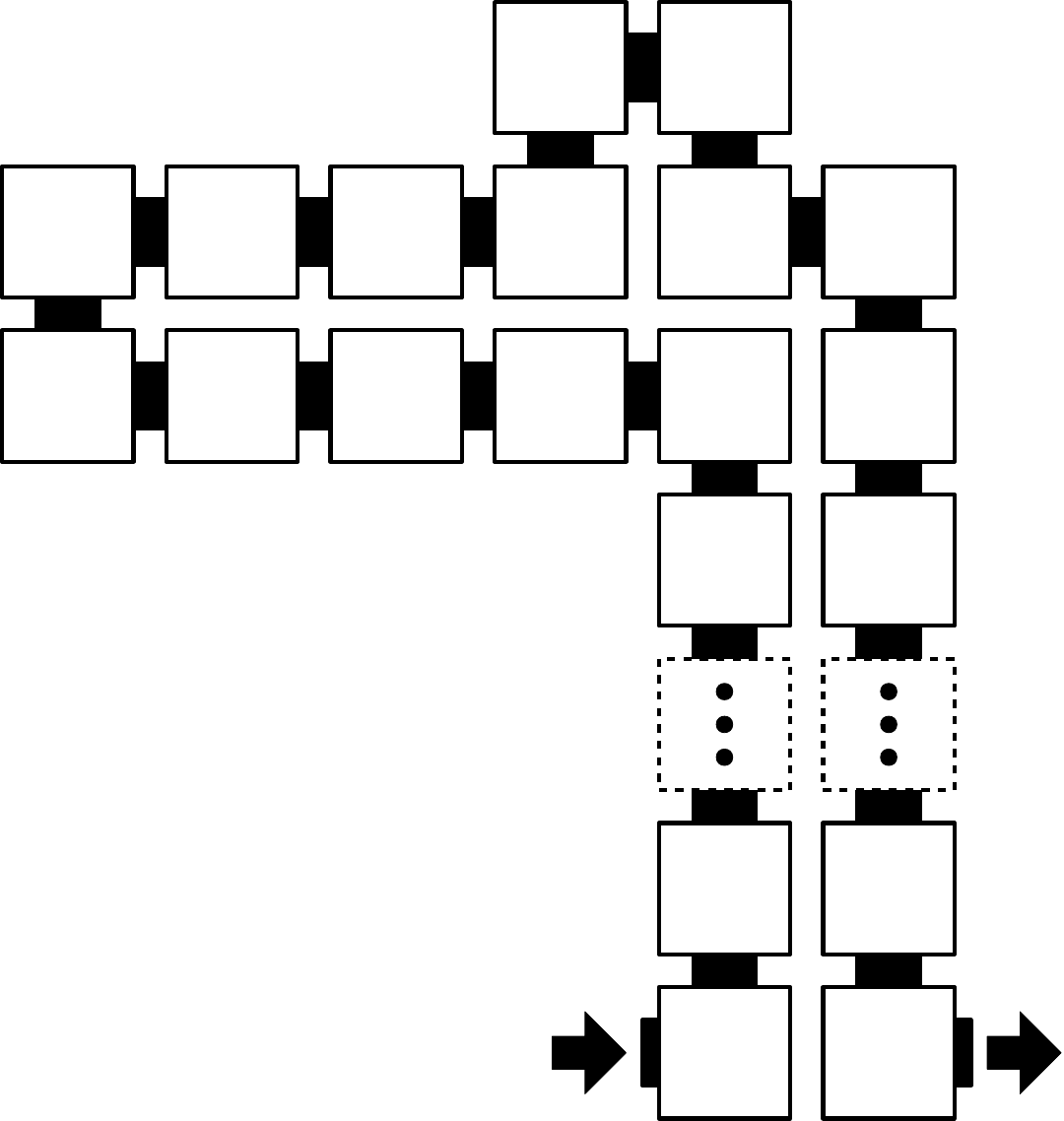}
        \caption{\label{fig:Copy-write-1} {\tt Copy\_write\_1} }
    \end{subfigure}
    \caption{\label{fig:Copy-write} The assemblies for the generic {\tt Copy\_write} gadgets. }
\end{figure}
The transition from a completed copy step to the subsequent increment step is facilitated by a {\tt Copy\_to\_inc} gadget, whose generic version is shown in Figure~\ref{fig:copy-to-inc}.
\begin{figure}[h!]
    \centering
        \begin{subfigure}[t]{0.4\textwidth}
        	\centering
        	\includegraphics[width=0.6in]{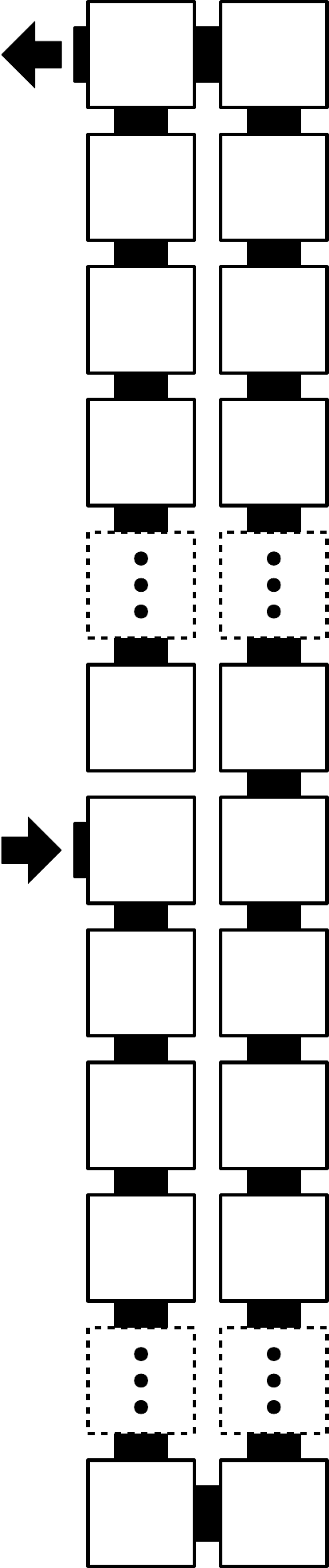}
        	\caption{\label{fig:copy-to-inc} {\tt Copy\_to\_inc}}
        \end{subfigure}
        ~
        \begin{subfigure}[t]{0.4\textwidth}
        	\centering
        	\includegraphics[width=0.75in]{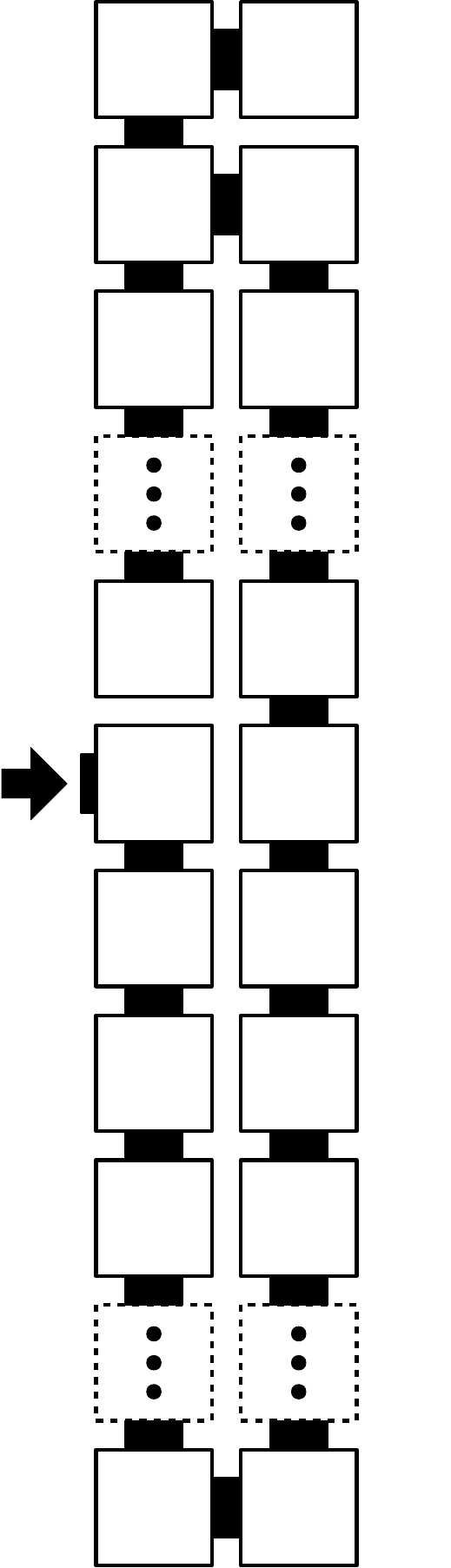}
        	\caption{\label{fig:last} {\tt Last}}
        \end{subfigure}
        \caption{\label{fig:copy-to-inc-and-last} The assemblies for the generic {\tt Copy\_to\_inc} gadget and the {\tt Last} gadget.  }
\end{figure}

Figure~\ref{fig:example-copy-n-5} is a continuation of Figure~\ref{fig:example-inc-n-5} and depicts, at a high-level, the first copy step for the counter in the  $s=5$ example.
\begin{figure}[h!]
    \centering
        \centering
        \includegraphics[width=\linewidth]{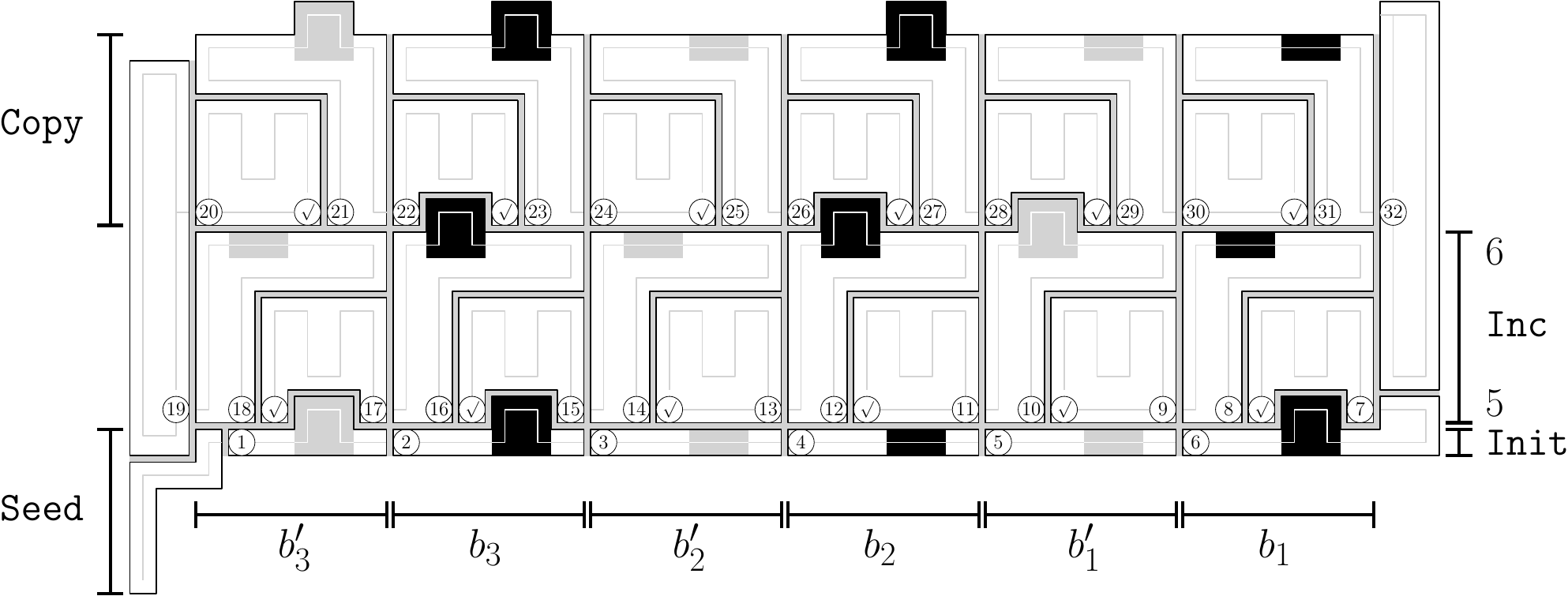}

        \caption{\label{fig:example-copy-n-5} A high-level example showing the self-assembly of the gadgets that comprise the first copy step of the counter.  Note that the top row of black regions correspond, from left to right, to the value bits $b_3 = 1$, $b_2 = 1$, and $b_1 = 0$, which is the binary representation of 6, just like in the increment row below it. The top row of grey regions correspond, from left to right, to the indicator bits $b'_3 = 1$, $b'_2 = 0$, and $b'_1 = 0$.}
\end{figure}
Like the {\tt Inc} gadgets created for an increment step, each {\tt Copy} gadget is associated either with a single, specific bit or with a group of bits in the value of the counter.
To summarize, we create:
\begin{itemize}
	\item two {\tt Copy\_read} gadgets associated with, and only with, $b'_k$,
	\item three {\tt Copy\_write} gadgets associated with any non-least significant indicator bit,
	\item two {\tt Copy\_read} gadgets associated with any non-least significant value bit,
	\item three {\tt Copy\_write} gadgets associated with any non-least significant value bit,
	\item two {\tt Copy\_read} gadgets associated with any non-most significant indicator bit,
	\item two {\tt Copy\_write\_0} gadgets associated with, and only with, $b'_1$,
	\item two {\tt Copy\_read} gadgets associated with, and only with, $b_1$,
	\item three {\tt Copy\_write} gadgets associated with, and only with, $b_1$,
	\item one {\tt Copy\_to\_inc} gadget associated with, and only with, $b_1$, and
	\item a unique {\tt Last} gadget associated with, and only with, $b_1$. 
\end{itemize}
Moreover, all of these gadgets are created so that:
\begin{enumerate}
	\item they all self-assemble to the right,
	\item at the conclusion of any copy step, it is the case that all indicator bits, except for $b'_k$, are equal to $0$, and that $b'_k$ is equal to $1$  if and only if this is not the final copy step,
	\item at the conclusion of the final copy step, it is the case that all value bits are equal to 0
	\item for each output glue of every {\tt Copy} gadget except {\tt Copy\_to\_inc} and {\tt Last}, there exists a unique {\tt Copy} gadget with a matching input glue, and
        \item for the output glue of the {\tt Copy\_to\_inc} gadget, there exists a unique {\tt Inc} gadget with a matching input glue.
\end{enumerate}
Prior to every copy step, i.e., after the previous increment step, it must be the case that $b'_1 = 1$ and all other indicator bits are $0$. The {\tt Copy\_read} gadget associated with any non-most significant indicator bit can determine if the index of the indicator bit with which it is associated is $1$ by correctly guessing a $1$.  This {\tt Copy\_read} gadget then exposes its output glue, which only the input glue of the {\tt Copy\_write\_0} gadget associated with $b'_1$ matches.
Therefore, we need not associate with each bit (indicator or value) a corresponding {\tt Copy} gadget. It follows that $O(1)$ {\tt Copy} gadgets suffice for our construction.
Since the size of each {\tt Copy} gadget is $\Theta(h)$ and the set of tile types for a gadget must be disjoint from the set of tile types for every other gadget, the number of tile types contributed to $T$ by all the {\tt Copy} gadgets is $\Theta(h)$.
We now create the {\tt Copy} gadgets.
%

%
%
The first gadget in a copy step is a {\tt Copy\_read} gadget associated with, and only with, $b'_k$, which non-deterministically guesses its unknown binary value. 
We create two such gadgets, one for each binary value $d$ propagated by the previous gadget that indicates whether this is the final copy step.
For each gadget, its input glue indicates that it is associated with $b'_k$ and whether this is the final copy step. It is the unique glue over all gadgets in our construction that matches the output glue of the unique {\tt Inc\_to\_copy} gadget, thus ensuring that this gadget is the first {\tt Copy} gadget in any copy step. 
The output glues of each gadget indicate that the next gadget must be a {\tt Copy\_write\_1} gadget associated with any non-least significant indicator bit and whether this is the final copy step.
Therefore, for each $d \in \{0,1\}$, create from the generic gadget in Figure~\ref{fig:Copy-read} the gadget:
\begin{align*}
{\tt Copy\_read}( & \langle {\tt most\_significant\_indicator\_bit}, {\tt done} = d\rangle, \\
			 	&  \langle {\tt non\_least\_significant\_indicator\_bit} = 1-d, {\tt done} = d \rangle, 
				\\
				&  \langle {\tt non\_least\_significant\_indicator\_bit} = 1-d, {\tt done} = d \rangle  
			  ).
\end{align*}
For example, gadget 20 in Figure~\ref{fig:example-copy-n-5} is the unique {\tt Copy\_read} gadget associated with $b'_k$ in the $s=5$ example. The corresponding gadget assembly is shown in Figure~\ref{fig:Copy-read-0-good}.
Note that the output glues of these gadgets ensure that the next gadget is a {\tt Copy\_write\_}$(1-d)$ gadget and thus that $b'_k = 1-d$ at the conclusion of the copy step. 
Furthermore, the gadget that assembles in the last copy row (i.e., the one whose input glue contains {\tt done} = 1) never makes a read error since its output glue always encodes the fact that a bit value of 0 was read in this case.

%
%
The second gadget in a copy step is a {\tt Copy\_write} gadget that is associated with any non-least significant indicator bit.
We create three such gadgets, one for each pair of binary values $b$ and $d$ such that:
\begin{enumerate}
	\item $b$ is the value that was propagated via the output glue of the previous {\tt Copy\_read} gadget associated with the same non-least significant indicator bit, e.g., the {\tt Copy\_read} gadget associated with $b'_k$,
	\item $d$ is the value propagated by the previous gadget that indicates whether this is the final copy step, but
	\item $b$ and $d$ cannot both be equal to 1.
\end{enumerate} 
We exclude the combination where $b=d=1$ because, if this is the last copy step (i.e., $d=1$), then all the indicator bits must be 0 at the conclusion of the step.
For each gadget created in this step, its input glue matches a corresponding output glue of the {\tt Copy\_read} gadget associated with the same indicator bit. Its input glue also indicates whether this is the final copy step.
The output glue of each created gadget indicates that the next gadget must be a {\tt Copy\_read} gadget associated with any non-least significant value bit. Its output glue also indicates whether this is the final copy step.
Therefore, for each $(b,d) \in \{(0,0), (0,1), (1,0)\}$, if $b = 0$ (resp., $b = 1$), then create from the generic gadget in Figure~\ref{fig:Copy-write-0} (resp., Figure~\ref{fig:Copy-write-1}) the gadget:
\begin{align*}
{\tt Copy\_write\_}b( 	& \langle {\tt non\_least\_significant\_indicator\_bit} = b, {\tt done} = d \rangle, \\
			 			& \langle {\tt non\_least\_significant\_value\_bit}, {\tt done} = d ).
\end{align*}
For example, gadgets 21 and 25 in Figure~\ref{fig:example-copy-n-5} are the {\tt Copy\_write} gadgets associated with non-least significant indicator bits in the $s=5$ example.
%

%
%
After any {\tt Copy\_write} gadget associated with any non-least significant indicator bit, e.g., the {\tt Copy\_write\_0} gadget associated with $b'_k$ in the final copy step, the next gadget in a copy step is a {\tt Copy\_read} gadget associated with, and that non-deterministically guesses, the unknown binary value of any non-least significant value bit. 
We create two such gadgets, one for each binary value $d$ propagated by the previous gadget that indicates whether this is the final copy step.
For each gadget, its input glue matches the output glue of a {\tt Copy\_write} gadget associated with any non-least significant indicator bit.
The output glues of each gadget indicate that the next gadget must be a {\tt Copy\_write} gadget associated with any non-least significant value bit. The output glues also indicate whether this is the final copy step. 
Therefore, for $d=0$, create from the generic gadget in Figure~\ref{fig:Copy-read} the gadget:
\begin{align*}
{\tt Copy\_read}( & \langle {\tt non\_least\_significant\_value\_bit}, {\tt done} = 0\rangle, \\
			 	&  \langle {\tt non\_least\_significant\_value\_bit} = 0, {\tt done} = 0\rangle, 
				\\
				&  \langle {\tt non\_least\_significant\_value\_bit} = 1, {\tt done} = 0\rangle 
			  ).
\end{align*}
For example, gadgets 22 and 26 in Figure~\ref{fig:example-copy-n-5} are the {\tt Copy\_read} gadgets associated with non-least significant value bits in the non-final copy row in the $s=5$ example.

Then, for $d=1$, create from the generic gadget in Figure~\ref{fig:Copy-read} the gadget:
\begin{align*}
{\tt Copy\_read}( & \langle {\tt non\_least\_significant\_value\_bit}, {\tt done} = 1\rangle, \\
			 	&  \langle {\tt non\_least\_significant\_value\_bit} = 0, {\tt done} = 1\rangle, 
				\\
				&  \langle {\tt non\_least\_significant\_value\_bit} = 0, {\tt done} = 1\rangle 
			  ).
\end{align*}

Note that this gadget, which assembles in the last copy row, never makes a read error since its output glue always encodes the fact that a bit value of 0 was read in this case.

For example, gadgets 74 and 78 in Figure~\ref{fig:high-level-overview} (see below) are the {\tt Copy\_read} gadgets associated with non-least significant value bits in the final copy row in the $s=5$ example.

%

%
%
After any {\tt Copy\_read} gadget associated with any non-least significant value bit, the next gadget in a copy step is a {\tt Copy\_write} gadget that is associated with any non-least significant value bit.
We create three such gadgets, one for each pair of binary values $b$ and $d$ such that:
\begin{enumerate}
	\item $b$ is the value that was propagated via the output glue of the previous {\tt Copy\_read} gadget associated with the same non-least significant value bit, and
	\item $d$ is the value propagated by the previous gadget that indicates whether this is the final copy step, but
	\item $b$ and $d$ cannot both be equal to 1.
\end{enumerate}
We exclude the combination where $b=d=1$ because, if this is the final copy step, then all the value bits must be 0 at the conclusion of the step.
For each gadget, its input glue matches a corresponding output glue of the {\tt Copy\_read} gadget associated with the same value bit, whereas its output glue indicates that the next gadget must be a {\tt Copy\_read} gadget associated with any non-most significant indicator bit. Its output glue also indicates whether this is the final copy step.
Therefore, for each $(b,d) \in \{(0,0),(0,1),(1,0)\}$, if $b = 0$ (resp., $b = 1$), then create from the generic gadget in Figure~\ref{fig:Copy-write-0} (resp., Figure~\ref{fig:Copy-write-1}) the gadget:
\begin{align*}
{\tt Copy\_write\_}b( 	& \langle {\tt non\_least\_significant\_value\_bit} = b, {\tt done} = d \rangle, \\
			 			& \langle {\tt non\_most\_significant\_indicator\_bit}, {\tt done} = d\rangle  ).
\end{align*}
For example, gadgets 23 and 27 in Figure~\ref{fig:example-copy-n-5} are the {\tt Copy\_write} gadgets associated with non-least significant value bits in the $s=5$ example.
%
%

%
%
After any {\tt Copy\_write} gadget associated with any non-least significant value bit, the next gadget in a copy step is a {\tt Copy\_read} gadget associated with, and that non-deterministically guesses, the unknown binary value of any non-most significant indicator bit.
We create two such gadgets, one for each binary value $d$ propagated by the previous gadget that indicates whether this is the final copy step.
For each gadget, its input glue matches the output glue of a {\tt Copy\_write} gadget associated with any non-least significant value bit, whereas its output glue depends on the non-deterministically guessed value of the indicator bit with which the gadget is associated, as follows:
\begin{enumerate}
	\item If the guessed value is $0$, then the corresponding output glue of the gadget indicates that the next gadget must be a {\tt Copy\_write\_}$0$ gadget associated with any non-most significant indicator bit.
	\item Otherwise, the corresponding output glue of the gadget indicates that the next gadget must be a {\tt Copy\_write\_}$0$ gadget associated with, and only, with $b'_1$.
\end{enumerate}
In both cases, the output glue of this gadget indicates whether this is the final copy step.
Therefore, for each $d \in \{0,1\}$, create from the generic gadget in Figure~\ref{fig:Copy-read} the gadget:
\begin{align*}
{\tt Copy\_read}( & \langle {\tt non\_most\_significant\_indicator\_bit}, {\tt done} = d\rangle, \\
			 	&  \langle {\tt non\_most\_significant\_indicator\_bit} = 0, {\tt done} = d\rangle,
				\\
				&  \langle {\tt least\_significant\_indicator\_bit} = 0, {\tt done} = d\rangle 
			  ).
\end{align*}
For example, gadgets 24 and 28 in Figure~\ref{fig:example-copy-n-5} are the {\tt Copy\_read} gadgets associated with non-most significant indicator bits in the $s=5$ example.
Note that these gadgets ensure that all indicator bits other than $b'_k$ have the value 0 at the conclusion of any copy step, and that, if this is the final copy step, $b'_k$ also has  the value 0. 
Furthermore, the gadgets that assemble in the last copy row (i.e., the ones whose input glue contains {\tt done} = 1) never make a read error since their output glue always encodes the fact that a bit value of 0 was read in this case.

%
%
If the {\tt Copy\_read} gadget associated with $b'_1$ non-deterministically guesses $b'_1 = 1$, then the next gadget is a {\tt Copy\_write\_0} gadget associated with, and only with, $b'_1$. 
We create two such gadgets, one for each binary value $d$ propagated by the previous gadget that indicates whether this is the final copy step.
The input glue of each gadget matches the second output glue of the {\tt Copy\_read} gadget associated with $b'_1$, whereas the output glue of each gadget indicates that the next gadget must be a {\tt Copy\_read} gadget associated with, and only with, $b_1$. The output glue also indicates whether this is the final copy step. 
Therefore, for each $d \in \{0,1\}$, create from the generic gadget in Figure~\ref{fig:Copy-write-0} the gadget:
\begin{align*}
{\tt Copy\_write\_0}( 	& \langle {\tt least\_significant\_indicator\_bit} = 0, {\tt done} = d \rangle, \\
			 			& \langle {\tt least\_significant\_value\_bit}, {\tt done} = d \rangle  ).
\end{align*}
For example, gadget 29 in Figure~\ref{fig:example-copy-n-5} is the {\tt Copy\_write\_}$0$ gadget associated with $b'_1$ in the $s=5$ example.
%

%
%
After the {\tt Copy\_write\_0} associated with $b'_1$, the next gadget in a copy step is a {\tt Copy\_read} gadget associated with, and that non-deterministically guesses, the unknown binary value of $b_1$.
We create two such gadgets, one for each binary value $d$ propagated by the previous gadget that indicates whether this is the final copy step.
For each gadget, its input glue matches the output glue of the {\tt Copy\_write\_0} gadget associated with $b'_1$, whereas its output glue indicates the next gadget must be a {\tt Copy\_write\_0} associated with, and only with, $b_1$. The output glue also indicates whether this is the final copy step.
Therefore, for $d=0$, create from the generic gadget in Figure~\ref{fig:Copy-read} the gadget:
\begin{align*}
{\tt Copy\_read}( & \langle {\tt least\_significant\_value\_bit}, {\tt done} = 0\rangle, \\
			 	&  \langle {\tt least\_significant\_value\_bit} = 0, {\tt done} = 0\rangle ,
				\\
				&  \langle {\tt least\_significant\_value\_bit} = 1, {\tt done} = 0\rangle
			  ).
\end{align*}
For example, gadget 30 in Figure~\ref{fig:example-copy-n-5} is the {\tt Copy\_read} gadget associated with $b_1$ in the $s=5$ example.

Then, for $d=1$, create from the generic gadget in Figure~\ref{fig:Copy-read} the gadget:
\begin{align*}
{\tt Copy\_read}( & \langle {\tt least\_significant\_value\_bit}, {\tt done} = 1\rangle, \\
			 	&  \langle {\tt least\_significant\_value\_bit} = 0, {\tt done} = 1\rangle ,
				\\
				&  \langle {\tt least\_significant\_value\_bit} = 0, {\tt done} = 1\rangle
			  ).
\end{align*}
For example, gadget 82 in Figure~\ref{fig:high-level-overview} (see below) is the {\tt Copy\_read} gadget associated with $b_1$ in the final copy row in the $s=5$ example.

%
%
After the {\tt Copy\_read} gadget associated with $b_1$, the next and penultimate gadget in a copy step is a {\tt Copy\_write} gadget associated with, and only with, $b_1$. 
We create three such gadgets, one for each pair of binary values $b$ and $d$ such that:
\begin{enumerate}
	\item $b$ is the value that was propagated via the output glue of the previous {\tt Copy\_read} gadget associated with $b_1$, and
	\item $d$ is the value propagated by the previous gadget that indicates whether this is the final copy step, but
	\item $b$ and $d$ cannot both be equal to 1.
\end{enumerate}
We exclude the combination where $b=d=1$ because, if this is the final copy step, then all the value bits must be 0 at the conclusion of the step.
For each gadget, its input glue matches a corresponding output glue of the {\tt Copy\_read} gadget associated with $b_1$.
If the output glue of the {\tt Copy\_read} gadget associated with $b_1$ indicates that this is not the final copy step, then the output glue of each gadget initiates the transition to the subsequent increment step.
Therefore, for each $b \in \{0,1\}$, if $b = 0$ (resp., $b = 1$) create from the generic gadget in Figure~\ref{fig:Copy-write-0} (resp., Figure~\ref{fig:Copy-write-1}) the gadget:
\begin{align*}
{\tt Copy\_write\_}b( 	& \langle {\tt least\_significant\_value\_bit} = b, {\tt done} = 0 \rangle, \\
			 			& {\tt copy\_to\_inc}  ).
\end{align*}
However, if the output glue of the {\tt Copy\_read} gadget associated with $b_1$ indicates that this is the final copy step, then the output glue of the last {\tt Copy\_write\_0} gadget ensures that the next gadget must be the {\tt Last} gadget. 
Therefore, create from the generic gadget in Figure~\ref{fig:Copy-write-0} the gadget:
\begin{align*}
{\tt Copy\_write\_0}( 	& \langle {\tt least\_significant\_value\_bit} = 0, {\tt done} = 1 \rangle, \\
			 			&  {\tt copy\_to\_last}  ).
\end{align*}
For example, gadget 31 in Figure~\ref{fig:example-copy-n-5} is the {\tt Copy\_write} gadget associated with $b_1$ in the $s=5$ example.
%

%
%
If the output glue of the {\tt Copy\_write} gadget associated with $b_1$ indicates that this is not the final copy step, then the last gadget in a copy step is a unique {\tt Copy\_to\_inc} gadget associated with, and only with, $b_1$, which transitions from the now-concluded, non-final copy step to the subsequent increment step. 
Its input glue matches the output glue of a {\tt Copy\_write} gadget associated with $b_1$, whereas its output glue indicates that the next gadget is the unique {\tt Inc\_read} gadget associated with $b_1$, which is the first gadget in the subsequent increment step.  
Therefore, create from the generic gadget in Figure~\ref{fig:copy-to-inc} the gadget:
\begin{align*}
{\tt Copy\_to\_inc}( 	& {\tt copy\_to\_inc}, \\
						& \langle {\tt least\_significant\_value\_bit}, {\tt carry} = 1 \rangle ).
\end{align*}
For example, gadget 32 in Figure~\ref{fig:example-copy-n-5} is the {\tt Copy\_to\_inc} gadget associated with $b_1$ in the $s=5$ example.
%

%
%
If the output glue of the {\tt Copy\_write} gadget associated with $b_1$ indicates that this is the final copy step, then the last gadget in this copy step, and thus in our construction, is a unique {\tt Last} gadget associated with, and only with, $b_1$.
Its input glue matches the output glue of a {\tt Copy\_write} gadget associated with $b_1$. It is the unique gadget in our construction with no output glue.
Therefore, create from the generic gadget in Figure~\ref{fig:last} the gadget:
\begin{align*}
{\tt Last}( 	& {\tt copy\_to\_last} ).
\end{align*}
Figure~\ref{fig:high-level-overview} is a high-level depiction of the terminal result of our $s=5$ example.
Note that the {\tt Last} gadget is depicted in this example and is gadget number 84.
\begin{figure}[h!]
    \centering
        \centering
        \includegraphics[width=\linewidth]{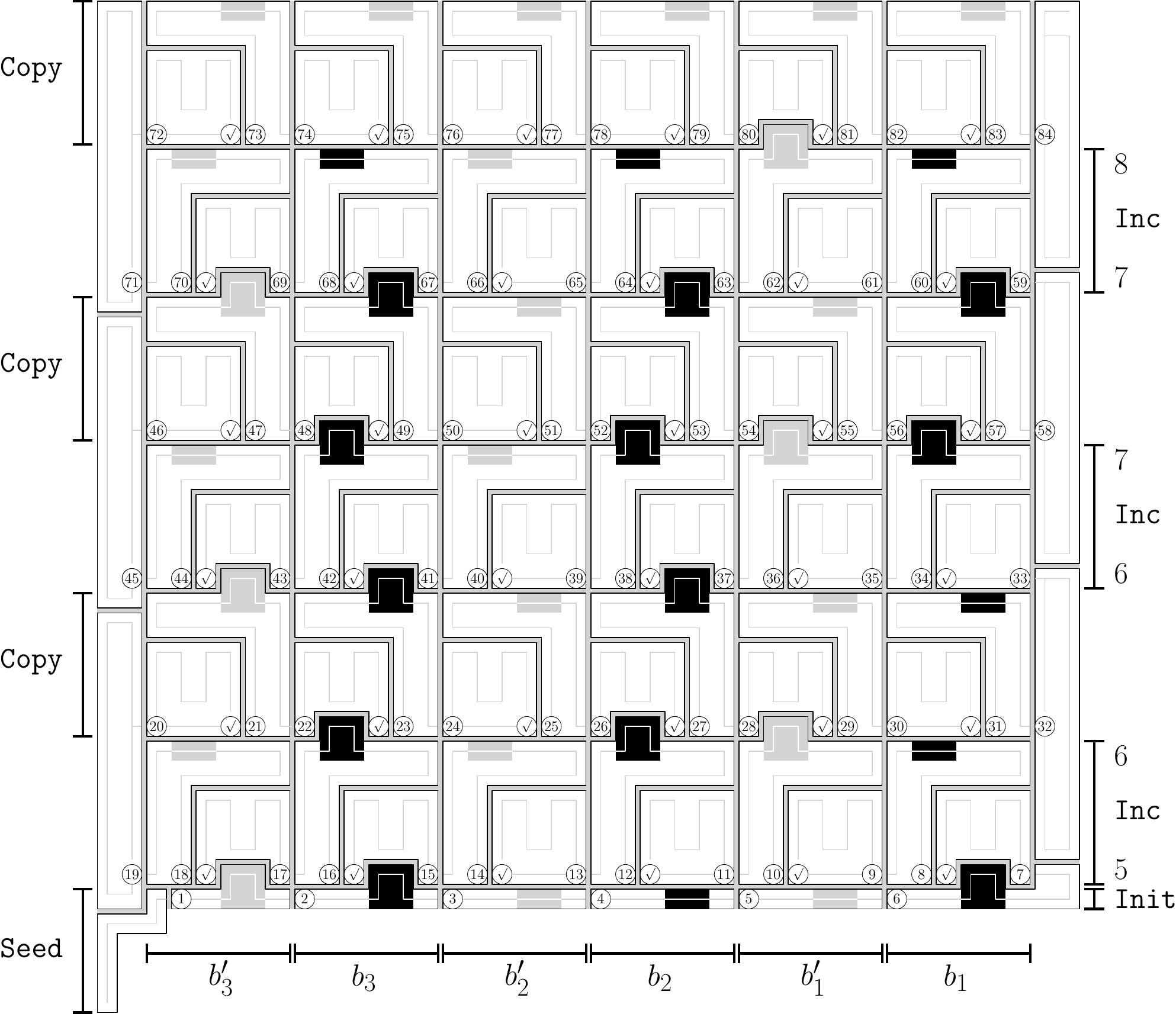}
        \caption{\label{fig:high-level-overview} High-level depiction of the conclusion of our $s=5$ example in which every {\tt Copy\_read} and {\tt Inc\_read} gadget guesses correctly the bit with which it is associated. }
\end{figure}

\emph{Step 2e.}
We now compute the number of tile types in $T$.
Note that $T$ is comprised of exactly the tile types contributed by all of the previously created gadgets. 
The total numbers of tile types contributed by the {\tt Seed}, {\tt Init}, {\tt Inc}, and {\tt Copy} gadgets are $O(e) = O(h)$, $O\left( k \right)$, $O\left( h \right)$, and $O\left( h \right)$, respectively.
Since $k = O\left( h \right)$, we have $\left| T\right| = O\left( h \right) = O\left( \log N + \log \frac{1}{\delta} \right)$.

\subsection{Description of producible assemblies that are not $w$-correct}
\label{sec:incorrect-assemblies}

Now that we have fully described the counter construction from the
perspective of correct assembly, let us consider the possible results
of incorrect assembly. In our counter construction, non-determinism is
exclusively located within the read gadgets that are used to guess the
binary value of both indicator and value bits in copy steps and
increment steps of the counter. By design, these gadgets may make a
single type of error, namely reading a 0 as a 1. The reverse
error is never made. We now discuss the consequences of read-gadget
errors on each type of bits in our construction.

Recall that indicator bits are all zeros except for 1) the leftmost
bit in each right-to-left-assembling increment row and 2) the
rightmost bit in each left-to-right-assembling copy row. Since a 1
will never be incorrectly read as a 0 and 1-valued
indicator bits are used to detect the end of a row and initiate the
beginning of the subsequent one, each row in our counter construction
will never exceed the intended $2k$-bit width. However, if a 0-valued
indicator bit is mistakenly read as a 1, then a row may stop
assembling prematurely. We consider increment and copy rows in turn.

Read errors at indicator bits in increment rows cause the following behaviors.
\begin{itemize}[itemsep=0mm]
\item {\bf[Ind.Inc.1 error]} A read error at $b'_1$ does not lead to an
  incorrect assembly because the two output glues of this read gadget
  are identical, i.e., a 1 bit will be correctly written no matter
  what, in preparation for the subsequent copy row.
\item {\bf[Ind.Inc.i error]} A read error at $b'_i$, for any $1 < i <
  k$, always leads to a 1-valued indicator bit being written, followed
  by the initiation of an {\tt Inc\_to\_copy} or {\tt
    Inc\_to\_copy\_last} gadget. Since both of these gadgets attempt
  to grow downward right after their input tile assembles, this growth
  will be blocked by a previously placed tile in the seed or copy row
  below the current one. Therefore, when the first indicator bit error
  occurs in an increment row, the early termination of that row causes
  the termination of the counter construction in the middle of that
  incomplete increment row.
\item {\bf[Ind.Inc.k error]} A read error at $b'_k$, whose value is
  always 1, can never occur, as discussed above.
\end{itemize}

Read errors at indicator bits in copy rows cause the following behaviors.
\begin{itemize}[itemsep=0mm]
\item {\bf[Ind.Cop.1 error]} A read error at $b'_1$, whose value is
  always 1, can never lead to an incorrect assembly, as discussed
  above.
\item {\bf[Ind.Cop.i error]} A read error at $b'_i$, for any $1 < i <
  k$, always leads to a 0-valued indicator bit being written, followed
  by the (in)correct reading and then coyping of the value bit $b_i$,
  followed in turn by the initiation of a {\tt Copy\_to\_inc} or {\tt
    Last} gadget. Since both of these gadgets attempt
  to grow downward right after their input tile assembles, this growth
  will be blocked by a previously placed tile in the increment row
  below the current one. Therefore, when the first indicator bit error
  occurs in a copy row, the early termination of that row causes the
  termination of the counter construction in the middle of that
  incomplete copy row.
\item {\bf[Ind.Cop.k error]} A read error at $b'_k$ does not lead to
  an incorrect assembly because the two output glues of this read
  gadget are identical, i.e., a 1 bit will be correctly written no
  matter what, in preparation for the subsequent increment row.
\end{itemize}

Note that read errors in the final copy row are not possible because all {\tt Copy\_read} gadgets corresponding to indicator bits and whose input glue contains {\tt done=1} always correctly have the bit value 0 encoded in their output glues.

We now turn our attention to possible errors occurring while reading value bits.

{\bf[Val.Cop error]} In copy rows, a 0 mistakenly read as a 1 would cause the
counter's value to be erroneously incremented by some positive value,
which would cause the counter construction to eventually terminate
with fewer rows, and thus a smaller height, than expected in the
correct assembly.

Note that read errors in the final copy row are not possible because all {\tt Copy\_read} gadgets corresponding to value bits and whose input glue contains {\tt done=1} always correctly have the bit value 0 encoded in their output glues.

Read errors at value bits in increment rows cause the following behaviors.

\begin{itemize}[itemsep=0mm]
\item {\bf[Val.Inc.C0 error]} If the incoming carry is equal to 0,
  then the outgoing carry is still 0 but the value bit at this
  position ends up being incorrectly written as a 1, leading to the
  counter's value being incremented by more than one.
\item {\bf[Val.Inc.C1 error]} If the incoming carry is equal to 1,
  then the value bit at this position ends up being incorrectly
  written as a 0, but the outgoing carry is also incorrectly set to 1,
  again leading to the counter's value being incremented by more than
  one.
\end{itemize}

In both cases, the counter construction would eventually terminate
with fewer rows, and thus a smaller height, than expected in the
correct assembly.

In conclusion, we have identified nine distinct types of read error,
distinguished above with an identifier between square brackets, that
either cannot occur in our construction or else lead to the premature
termination of the assembly process (either by shortening at least one
row or by decreasing the number of rows in the counter).
In other words, all of our terminal assemblies are finite and such
that $P=Y$ is a finite set. However, other design choices could have
led to infinite terminal assemblies such that $P$ would be infinite
and different from $Y$. For example, had we chosen to represent the
bit value 0 with a protruding bump and the bit value 1 with the
absence of such a bump (i.e., the reverse of our current
representation scheme), then a read error could have caused a value
bit of 1 to be read as a 0, leading to the counter's value being
erroneously decremented, which could result in the counter's values
going up and down and thus yield infinite terminal assemblies, with
$Y$ being a strict subset of an infinite set $P$.

\subsection{Last four steps in the proof of the main lemma}

This subsection contains steps 3 through 6 in the proof of Lemma~\ref{lem:high-probability-counter}.

\emph{Step 3.}
We now define the set $P$ of POCs of $\mathcal{T}$ and the set $Y \subseteq P$ of essential POCs of $\mathcal{T}$.
With these sets, we can then define a corresponding winner function $w:Y \rightarrow T$.
To that end, let $P$ be the POC set for $\mathcal{T}$ with essential POC set $Y = \left\{ \vec{y}_1, \ldots, \vec{y}_{4km} \right\}$, where the elements of $Y$ are the locations within the domain of the unique correct terminal assembly of $\mathcal{T}$ (defined below) at which either the lower-right tile of a {\tt Copy\_read} gadget attaches or the lower-left tile of an {\tt Inc\_read} gadget attaches. We order the set $Y$ according to the sequence in which the gadgets self-assemble, which is a fixed order by the way we create the gadgets. 
Based on the argument given in Section~\ref{sec:incorrect-assemblies},
it is not possible for the counter to erroneously decrement its value
and get stuck in an infinite loop. In fact, in our construction, $Y$ is equal to
the finite set $P=P_{\mathrm{inc}} \cup P_{\mathrm{copy}}$ where
$P_{\mathrm{inc}} =
\left\{(4+6i,e+1+2jh)~|~i\in\{0,1,\ldots,k-1\}\ \mathrm{and\ }
j\in\{0,1,\ldots,m-1\}\right\}$ is the set of POCs located within an increment
row and $P_{\mathrm{copy}} =
\left\{\left(5+6i,e+1+(2j+1)h\right)~|~i\in\{0,1,\ldots,k-1\}\ \mathrm{and\ }
j\in\{0,1,\ldots,m-1\}\right\}$ is the set of POCs located within a copy row.
\begin{figure}[h!]
    \centering
        \centering
        \includegraphics[width=.7\linewidth]{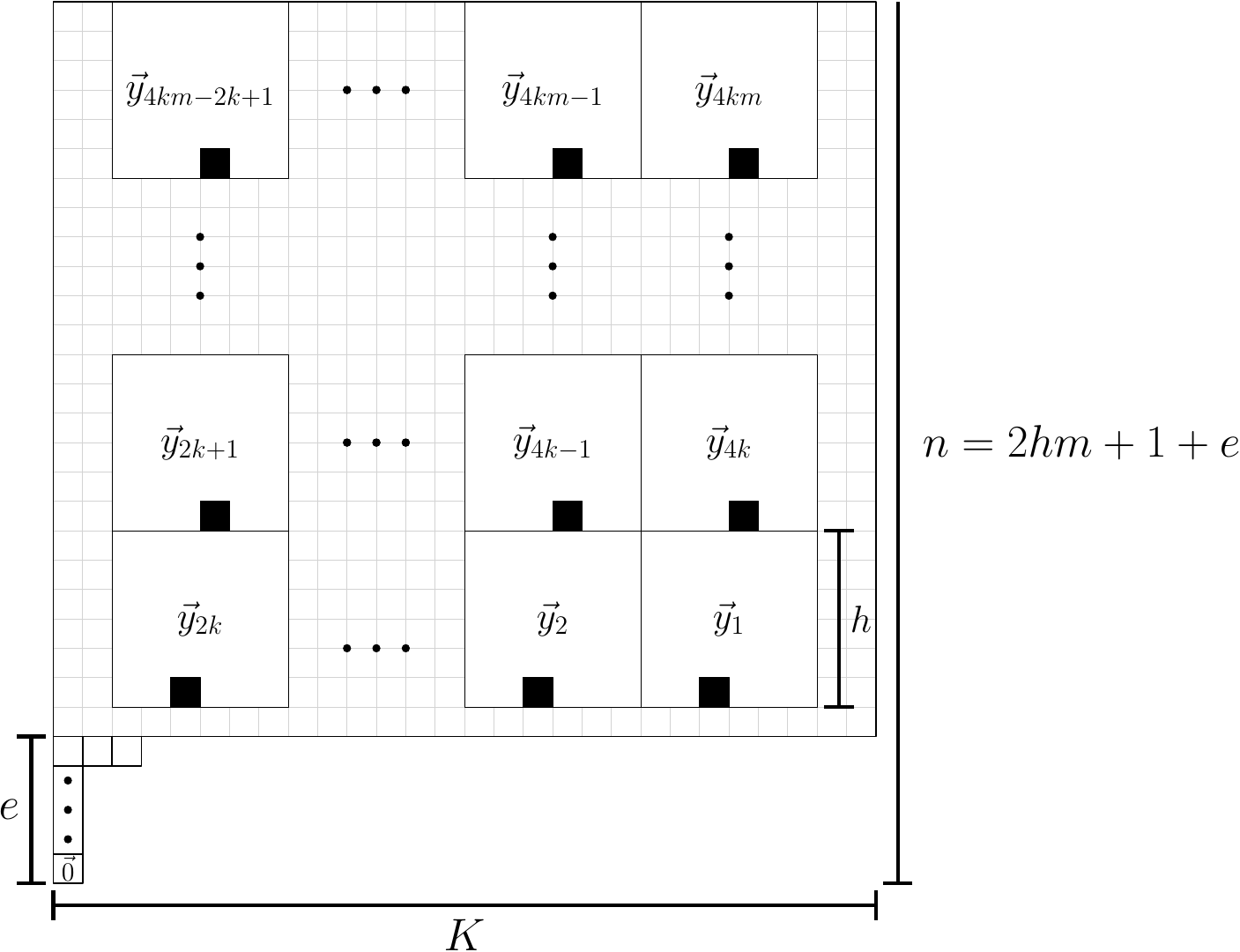}
        \caption{\label{fig:POCs} Depiction of the set of POCs for $\mathcal{T}$. The lower-left tile is assumed to be placed at the origin. The black squares make up $Y=P$. }
\end{figure}
Moreover, every POC of $\mathcal{T}$ is either the point at which the lower-right tile of a {\tt Copy\_read} gadget attaches or the point at which the lower-left tile of an {\tt Inc\_read} gadget attaches.
This means we can define the winner function $w:Y\rightarrow T$ as follows: for every $\vec{y} \in Y$ such that $\vec{y}$ is the point of the lower-right tile of a {\tt Copy\_read} gadget or the lower-left tile of an {\tt Inc\_read} gadget, then map $\vec{y}$ to the unique tile type in $T$ that corresponds to the gadget correctly guessing the binary value of the bit with which it is associated. 

\emph{Step 4.}
We now prove that $\mathcal{T}$ is $w$-sequentially non-deterministic such that $\alpha$ is the unique $w$-correct $\mathcal{T}$-terminal assembly and $\dom{\alpha} = L_e \cup R_{K,n,e}$.

First, we show that $\mathcal{T}$ satisfies all of the assumptions of
Definition~\ref{def:seq-non-deterministic}. Recall that $P$, $Y$, and
$w$ were defined in the previous step. The set $S_P = S_Y$ can be
similarly inferred from this earlier analysis. Furthermore, a
$w$-correct $\mathcal{T}$-producing assembly sequence $\vec{\alpha}$
resulting in $\alpha$ such that $Y \subseteq \dom{\alpha}$ can be
inferred from Figure~\ref{fig:high-level-overview}. By Corollary~\ref{cor:unique-w-correct-assembly},  $\alpha$ is unique.

Second, we now show that $\mathcal{T}$ satisfies the three conditions of Definition~\ref{def:seq-non-deterministic}.
%


\paragraph{\hspace*{5mm}$\bullet$\hspace*{1mm} \bf Condition~\ref{def:snd-1} of Definition~\ref{def:seq-non-deterministic}.}

$\mathcal{T}$ is directionally deterministic because it satisfies Definition~\ref{def:dd-ad-tas}, as we now show. Let $D\subset\mathbb{Z}^2$ denote the finite set of points at which a tile is placed by at least one assembly in $\termasm{\mathcal{T}}$.
If read errors were impossible, then each point in $D$ would have exactly one type of tile placed at it and would thus meet the requirements of Definition~\ref{def:dd-ad-pt} vacuously. In contrast, the possibility of read errors leads to varying numbers of tile types being placed at different points in $D$.

In the following discussion, we say that a set of points is \emph{covered} by a gadget if, in a $w$-correct producible assembly, only that type of gadget places tiles at the points in the set.

Furthermore, we call \emph{bit bump points} the two points that make up a bit bump. If the correct value of this bit is a 1, then the bit bump points are covered by an {\tt Init} gadget (if the bit belongs to the initial value of the counter) or a {\tt Write gadget} that assembled in the previous row (if the bit belongs to any increment or copy row). If the correct value of this bit is a 0, then the bit bump points are covered by a {\tt Read} gadget. In the latter case, bit bump points are the two inner points within the shorter competing path in the {\tt Read} gadget. Therefore, we refer to these two points as ``bit bump points'' whether we are discussing them in the context of the {\tt Write} gadget or the {\tt Read} gadget, even though only one of these gadgets strictly covers these points in the $w$-correct terminal assembly.

Recall (see Section~\ref{sec:incorrect-assemblies}) that one or more read errors may lead to early termination of the assembly process, either in the middle of an incomplete row or with a number of completed rows that is smaller than the number of rows in the $w$-correct terminal assembly. In such cases, some gadgets may not get to assemble in the resulting incorrect terminal assembly. We will use a phrase of the form ``at most one of so many tile types is placed at this point'' in these cases. In all other cases, we will write that ``exactly one of so many tile types is placed at this point''.

To prove that all of the points in $D$ meet the requirements of Definition~\ref{def:dd-ad-pt} in the presence of read errors, we now partition $D$ into subsets of points based on the gadget that covers them, following the order in which the corresponding gadget creation statements appear in Section~\ref{sec:first-two-steps}.
\paragraph{\bf (a) Points covered by gadgets that self-assemble before a tile is placed at the first POC.} These points include:
\begin{itemize}
  \item $\vec{s}$, which satisfies Definition~\ref{def:dd-ad-pt} trivially,
  \item all of the other points covered by the {\tt Seed} gadget,
  \item all of the points covered by the {\tt Init} gadgets, and
  \item the non-POC points covered by the {\tt Inc\_read} gadget associated with the first POC, i.e., the {\tt Inc\_read} gadget associated with $b_1$.
\end{itemize}
Since no read errors may occur before a tile is placed at the first POC, all of the points in this case have exactly one type of tile placed at each one of them in any producible assembly. Therefore, condition~\ref{def:dd-ad-pt-1} of Definition~\ref{def:dd-ad-pt} holds trivially because $t_{\alpha}$ and $t_{\beta}$ cannot be distinct, which means that condition~\ref{def:dd-ad-pt-1} is satisfied vacuously by these points. Note that the bit bump points covered by this read gadget need special attention  because the tiles that are placed there may belong to one of two gadgets, depending on the value of $b_1$ in $s$. If this value is 1, then the last {\tt Init} gadget is an {\tt Init\_right\_1} gadget and its two bit bump points protrude into the {\tt Inc\_read} we are considering here. So, in this case, these points are covered by the {\tt Init\_right\_1} gadget. On the other hand, if the value of $b_1$ in $s$ is 0, then these bit bump points are covered by the {\tt Inc\_read} gadget and belong to its shorter competing path. A similar situation will exist for all {\tt Inc\_read} and {\tt Copy\_read} gadgets. What is different in the first increment row is that no errors happen before all of the bits in $s$ are assembled. Therefore, in this case, exactly one type of tile is placed at each one of them in any producible assembly.
\paragraph{\bf (b) All POCs associated with $b_1$ in increment steps.}
Since the value of the incoming carry bit at the read gadget associated with $b_1$ in an increment step is always 1 and it is possible for two bit values to be read by this gadget (due to the non-deterministic occurrence of a Val.Inc.C1 error), it follows that exactly one of two types of tiles may be placed at this POC in the first increment step and that at most one of two types of tiles may be placed at this POC in later increment steps (since, in this case, any type of read error in an earlier step of the counter may have caused an early termination). Note that both tile types belong to the same read gadget since the incoming carry is always 1. However, the tile type that corresponds to reading the bit value 0 attaches from the east, whereas the tile type that corresponds to reading the bit value 1 attaches from the north. Thus, condition~\ref{def:dd-ad-pt-2a} of Definition~\ref{def:dd-ad-pt} is satisfied by these POCs.
\paragraph{\bf (c) Non-POC points covered by the {\tt Inc\_read} gadget associated with $b_1$ in all (but the first) increment steps.} The POC was handled in case (b) above. Furthermore, since the incoming carry bit is always equal to 1, up to one gadget may assemble in this case. Therefore, the non-POC points belonging to the longer competing path in this gadget have at most one type of tile placed at each one of them in
any producible assembly, in which case condition~\ref{def:dd-ad-pt-1} of Definition~\ref{def:dd-ad-pt} is trivially satisfied. Finally, we consider the bit bump points in this read gadget. 

Since this gadget is associated with $b_1$ after the first increment step, it is possible that a Val.Inc.C1 error may have occurred in a previous increment step at $b_1$, e.g., $b_1 = 0$ in a previous increment step but the associated read gadget erroneously read it as a 1 and the following write gadget  wrote the bit value $b_1 = 0$ because the incoming carry bit is always 1 in this case. 
The absence of a bit bump correctly protruding up from the previous copy step means the shorter competing path of the read gadget may erroneously assemble toward the POC while propagating the carry bit value of 1. 
In this case, it is also possible that a Val.Cop error may have
occurred in a previous copy step at $b_1$ resulting in an incorrect $b_1$ value in this increment step.
Regardless of the type of error that may have occurred, at most one of two
types of tiles may be placed at any one of the two bit bump points.
Now, for each one these points, the two distinct tile types that may be placed there belong to two assemblies that do not agree, because one of them must have made a read error at an earlier point that is also in the domain of the other, correct assembly that is about to place a tile at that same bit bump point. Therefore, condition~\ref{def:dd-ad-pt-1} of Condition~\ref{def:dd-ad-pt} is satisfied by these two points.

\paragraph{\bf (d) Points covered by {\tt Inc\_write} gadgets associated with $b_1$.} First, we consider the {\it input point}, i.e., the point  where the input tile attaches. Since, by case (b) above, exactly two types of tiles may be placed at the preceding POC, exactly one of two types of input tiles may attach at the input point. Since these two tile types belong to two assemblies that do not agree, because one of them must have made a read error at an earlier point that is also in the domain of the other, correct assembly that is about to place a tile at that same input point. Therefore, condition~\ref{def:dd-ad-pt-1} of Condition~\ref{def:dd-ad-pt} is satisfied by the input point of this {\tt Inc\_write} gadget.

Second, we consider the bit bump points that are part of this {\tt
  Inc\_write} gadget in the case $b_1=1$, i.e., the gadget is an {\tt
  Inc\_write\_1} gadget. Then the carry bit value associated with this
gadget must be 0. Thus exactly one type of {\tt Inc\_write\_1} gadget
may assemble here and the tiles at both bit bump points belong to
it. On the other hand, if $b_1=0$, then an {\tt Inc\_write\_0} gadget
assembles here and the two points we are now considering are both
covered by a {\tt Copy\_read} gadget associated with $b_1$ within the
next copy row. Note that one of two types of {\tt Copy\_read} gadget
may assemble in that row, since the value of the {\tt done} bit in
that row may be equal to 0 or 1. So, one of at most two tiles may be
placed by read gadgets at each one of the points under consideration
here. Therefore, given the two possible values of $b_1$, and since all
of the gadgets involved have pairwise-disjoint tile sets, at most one
of three tile types may be placed at each one of the points under
consideration. Finally, we observe that condition~\ref{def:dd-ad-pt-1} of Definition~\ref{def:dd-ad-pt} is
satisfied by both points. For example, consider all of the pairs of
tiles that may be placed at the easternmost one of these points. If
these tile types both belong to a {\tt Copy\_read} gadget, then they
belong to two assemblies that do not agree since at least one previous
read error caused them to assemble two distinct read
gadgets. Similarly, if one of the tile types in the pair belongs to a
{\tt Copy\_read} gadget while the other tile belongs to an {\tt
  Inc\_write\_1} gadget, then they also belong to two assemblies that do
not agree since at least one previous read error caused them to write different bit values. A similar reasoning can be applied to the second bit bump point.

Third, we consider all of the other points covered by these {\tt
  Inc\_write} gadgets.  Since at most two write gadgets may assemble in
this case, one of at most two tile types may attach at each point in this
set. Again, these two tiles must attach to assemblies that do
not agree, and condition~\ref{def:dd-ad-pt-1} of
Definition~\ref{def:dd-ad-pt} is satisfied by all of these points.

\paragraph{\bf (e) Points covered by {\tt Inc\_read} gadgets associated with $b'_1$.} In this case, the value of the bit being read is guaranteed to be a 0. This is because the {\tt Init} gadget corresponding to $b'_1$ always encodes a 0 bit. Then the corresponding read gadget in all increment steps always exposes an output glue that encodes the bit value 1, which means that the corresponding write  gadget in  all increment steps always writes a 1 bit. Finally, since a 1 bit can never be misread and the write gadgets are deterministic, the value of the $b'_1$ bit in the next increment row is guaranteed to be a 0. Now, since the value of the incoming carry bit may be either 0 or 1 (due to the non-deterministic occurrence of a Val.Inc.C1 error at $b_1$),  at most one of two {\tt Inc\_read} gadgets may assemble here, whose tile sets are disjoint. Thus, for any non-POC point covered by these gadgets, condition~\ref{def:dd-ad-pt-1} of Definition~\ref{def:dd-ad-pt} is satisfied by these points (since the assemblies to which the tiles attach cannot agree).

At the POC, at most one of exactly four tile types may attach, corresponding to the four pairs resulting from combining two possible carry bit values with two possible values of the bit being read. Now consider all possible pairs of tile types that may attach at the POC.
\begin{itemize}
\item If the two tiles correspond to reading distinct bit values (i.e., exactly one of them makes a read error), whether they belong to the same gadget or not, then the tiles attach from different sides and condition~\ref{def:dd-ad-pt-2a} in Definition~\ref{def:dd-ad-pt} is satisfied.
\item In the two remaining cases, the two tiles in a pair belong to different gadgets but correspond to reading the same bit value (i.e., they either both make a read error or neither one of them does). Thus they attach from the same side but to distinct tiles and condition~\ref{def:dd-ad-pt-2b} in Definition~\ref{def:dd-ad-pt} is satisfied.
\end{itemize}

\paragraph{\bf (f) Points covered by {\tt Inc\_write} gadgets associated with $b'_i$ for $k > i \geq 1$.} First, note that the case of the bit bump points covered by an {\tt Inc\_write\_1} gadget assembling here can be handled with a similar reasoning to the one used in case (d) above, with at most one of three distinct tile types being placed at each one of the points under consideration. Second, since each one of the {\tt Inc\_read} gadgets that may precede this gadget in the assembly sequence exposes an output glue that always encodes a unique bit value (along with the carry bit value), namely a 1 if $i=1$ and a 0 if $i>1$, only two of the four corresponding {\tt Inc\_write} gadgets may ever assemble in this part of the counter (because of the two possible values of the carry bit). Since all of these {\tt Inc\_write} gadgets are deterministic and have disjoint tile sets, at most one of two distinct tile types may attach at any point covered by this gadget in this second case. These distinct tiles must attach to assemblies that do not agree, because they must have made different types of read errors (and/or different numbers of them) at points in the intersection of their domains. Therefore, condition~\ref{def:dd-ad-pt-1} of Definition~\ref{def:dd-ad-pt} is satisfied by these points.

\paragraph{\bf (g) Points covered by {\tt Inc\_read} gadgets associated with $b_i$ for $k \geq i >1$.} We distinguish two cases.

Case $i = 2$: Ind.Inc.1 errors do not cause early termination. Then, since only two distinct tile types (corresponding to the two possible values of the carry bit) may be placed at the POC preceding this gadget, up to two {\tt Inc\_read} gadgets may assemble in this case, whose tile sets are disjoint. Thus, for any point covered by these gadgets other than the bit bump points and the POC, condition~\ref{def:dd-ad-pt-1} of Definition~\ref{def:dd-ad-pt} is satisfied, because the two tiles are distinct and the assemblies they attach to cannot agree.

Now, at each one of the bit bump points of this gadget, one additional
tile type may attach, namely in the case where the bit bump points of
the corresponding {\tt Copy\_write\_1} gadget propagating the {\tt
  done = 0} flag in the previous copy row protrude into this {\tt
  Inc\_read} gadget. In this case, at most one of three tile types may
attach at each one of the two points under consideration. Following a
reasoning similar to the one used in case (d) above, condition~\ref{def:dd-ad-pt-1} of
Definition~\ref{def:dd-ad-pt} is satisfied by these two points.

Finally,  consider the POC associated with this gadget. Since a Val.Inc.C0 error is possible but the case where both the bit value that is read and the carry bit value that is propagated are equal to 1 is not possible, it follows that at most one of three tile types may attach at this point. Then one can apply a similar reasoning to the one we used for the POC in step (e) above to infer that Definition~\ref{def:dd-ad-pt} is satisfied by the POC.

Case $i>2$: An Ind.Inc.i error may occur at $b'_{i-1}$ that causes the corresponding {\tt Inc\_write} gadget to expose an output glue that allows for the input tile of an {\tt Inc\_to\_copy} gadget or an {\tt Inc\_to\_copy\_last} gadget to attach at the input point for this gadget, in addition to the  input tiles for the two {\tt Inc\_read} gadgets discussed in the previous case. As a result, at most one of four distinct tile types may attach at the input point of this gadget. Since these four tiles belong to distinct gadgets whose tile sets are pairwise disjoint and the assemblies they attach to cannot agree,  condition~\ref{def:dd-ad-pt-1} in Definition~\ref{def:dd-ad-pt} is satisfied by this point. Finally, since both the {\tt Inc\_to\_copy} gadget and {\tt Inc\_to\_copy\_last} gadgets attempt to assemble southward from their input tile and that point is blocked by the previous row in our counter, the rest of the analysis given in the first case for the non-input tiles of this gadget carries over to this case.

\paragraph{\bf (h) Points covered by {\tt Inc\_write} gadgets associated with $b_i$ for $k\geq i >1$.}
First, we consider the input point. Since the {\tt Inc\_read} gadgets associated with $b_i$ expose one of three possible output glues, at most one of three tile types may attach at the input point of this gadget. Since these tiles belong to different gadgets with pairwise disjoint tile sets and they attach to assemblies that cannot agree, condition~\ref{def:dd-ad-pt-1} in Definition~\ref{def:dd-ad-pt} is satisfied by this point.

Second, the fact that the bit bump points (if any) for this gadget  satisfy condition~\ref{def:dd-ad-pt-1} in Definition~\ref{def:dd-ad-pt} follows from a similar reasoning to the one we used in case (d) above, which we omit here.

Third, we consider all of the other points covered by these {\tt
  Inc\_write} gadgets. Since at most one of three distinct write gadgets may assemble in this case, at most one of three tile types may attach at each point in this set.  Since these tiles belong to different gadgets with pairwise disjoint tile sets and they attach to assemblies that cannot agree, condition~\ref{def:dd-ad-pt-1} in Definition~\ref{def:dd-ad-pt} is satisfied by this point.

\paragraph{\bf (i) Points covered by {\tt Inc\_read} gadgets associated with $b'_i$ for $k \geq i > 1$.}
We distinguish two cases. First, if $i < k$, since two distinct output
glues may be exposed by the {\tt Inc\_write} gadget corresponding to
$b_i$, at most one of two distinct gadgets may assemble in this case,
each one encoding a possible value of the incoming carry bit. Since
the tile sets of these {\tt Inc\_read} gadgets are disjoint and the
assemblies to which these tiles attach cannot agree, for any point
covered by these gadgets including the bit bump points but not the POC, condition~\ref{def:dd-ad-pt-1} in
Definition~\ref{def:dd-ad-pt} is satisfied. Note that we include the bit bump points because the output tiles of the {\tt Copy\_read} gadget below the current {\tt Inc\_read} gadget only induce the assembly of a {\tt Copy\_write\_0}, which means there can never be a bit bump protruding up into the shorter competing path of the current {\tt Inc\_read} gadget.

At the POC, at most one of four tile types may attach, corresponding
to the four pairs resulting from combining two possible carry bit
values with two possible values of the bit being read. The same
four-case reasoning we used in step (e) above can be applied here to
infer that Definition~\ref{def:dd-ad-pt} is satisfied by the POC.

Second, if $i=k$, the value of $b'_k$ is guaranteed to be equal to,
and also correctly read as, a 1. While the same reasoning as in the
previous case applies to the non-POC points, the situation at the POC
is different because at most one of only two distinct tile types may
attach in this case, namely the tiles attaching from the north and
reading the indicator bit value 1 in the two {\tt Inc\_read} gadgets
propagating different carry bit values. Since these two tiles belong
to different gadgets with disjoint tile sets and they thus attach to
two different tiles belonging to the same two tile sets, respectively,
condition~\ref{def:dd-ad-pt-2b} in
Definition~\ref{def:dd-ad-pt} is satisfied by the POC of this gadget.

\paragraph{\bf (j) Points covered by the {\tt Inc\_write} gadgets associated with $b'_k$.}
First, we consider the input point. Since the {\tt Inc\_read} gadgets
associated with $b'_k$ expose one of two possible output glues
corresponding to the possible values of the incoming carry bit, at
most one of two distinct tile types may attach at the input point of
this gadget. Since the assemblies to which these tiles attach cannot
agree, condition~\ref{def:dd-ad-pt-1} in
Definition~\ref{def:dd-ad-pt} is satisfied by the input point of this gadget.

Second, we need not discuss the bit bump points that would protrude into the next row because each one of the {\tt Inc\_write} gadgets associated with $b'_k$ always assembles the indicator bit value 0.

Third, we consider all of the other points covered by these {\tt
  Inc\_write} gadgets. Since at most one of two distinct write gadgets
may assemble in this case and those gadgets are deterministic, at most
one of two distinct tile types may attach at each point in this set.
Since the assemblies to which these tiles attach cannot agree, the condition~\ref{def:dd-ad-pt-1} in
Definition~\ref{def:dd-ad-pt} is satisfied by all of the other points
covered by these gadgets.

\paragraph{\bf (k) Points covered by the {\tt Inc\_to\_copy} gadgets.}
Note that, in the event of a Val.Inc.C1 error at $b_k$, the outgoing carry bit is erroneously set to 1.
In this case, the subsequent read and write gadgets associated with $b'_k$ initiate the incorrect assembly of an {\tt Inc\_to\_copy\_last} gadget instead of the correct assembly of an {\tt Inc\_to\_copy} gadget, which would assemble in the absence of such a Val.Inc.C1 error.
The respective input tiles of the {\tt Inc\_to\_copy} and {\tt Inc\_to\_copy\_last} gadgets attach to different types of tiles of the {\tt Inc\_write\_}1 gadget associated with $b'_k$ depending on the value of the incoming carry bit for the associated read gadget. 
This implies that the two distinct input tiles must attach to
assemblies that do not agree and thus that condition~\ref{def:dd-ad-pt-1} in Definition~\ref{def:dd-ad-pt} is
satisfied by the input point of the {\tt Inc\_to\_copy} gadget.
Moreover, the {\tt Inc\_to\_copy} and {\tt Inc\_to\_copy\_last}
gadgets have two disjoint tile sets. Therefore, again, condition~\ref{def:dd-ad-pt-1} in Definition~\ref{def:dd-ad-pt} is
satisfied by every point covered by the tiles in both gadgets that attach after the input tile.

\paragraph{\bf (l) Points covered by the bottom two tiles in the {\tt Inc\_to\_copy\_last} gadget.}
Note that the {\tt Inc\_to\_copy\_last} gadget has an additional top row of tiles, compared to the {\tt Inc\_to\_copy} gadget. These two tiles may potentially attach at  the points in the bottom row of tiles of a subsequent {\tt Inc\_to\_copy} or the {\tt Inc\_to\_copy\_last} gadget, e.g., if the Val.Inc.C1 error occurs in the penultimate increment step. Since the tile sets of these two gadgets are distinct and the two tiles that may attach at any one of these two points must attach to assemblies that cannot agree, condition~\ref{def:dd-ad-pt-1} in Definition~\ref{def:dd-ad-pt} is
satisfied by the points in the bottom row of tiles of an {\tt Inc\_to\_copy} or the {\tt Inc\_to\_copy\_last} gadget.

  \paragraph{\bf (m) Points covered by the {\tt Copy} and {\tt Last} gadgets:}

A similar case-by-case analysis as the foregoing one for the increment
step gadgets can be applied to the points covered by the copy step
gadgets, namely all of the {\tt Copy} gadgets and the {\tt Last}
gadget, in order to establish that the requirements imposed by
Definition~\ref{def:dd-ad-pt} are satisfied in those cases as
well. This analysis is omitted here in the interest of concision.
Even though the copy gadgets do not propagate a carry bit, they do
propagate the {\tt done} bit (to indicate whether or not the current row
is the final one), which can be set erroneously (i.e., prematurely)
due to a read error in the previous increment step. Since this flag
also has two values, like the carry bit, the cases to consider in copy
steps are similar to those in the increment steps. Moreover, it is
possible for a copy step to terminate early, similarly to how an
increment step does so, but via the placement of the input tile of a
{\tt Copy\_to\_inc} after an Ind.Cop.i error.

In conclusion, all of the points covered by these copy step gadgets are
directionally deterministic. Therefore, the TAS $\mathcal{T}$ is directionally
deterministic.

\paragraph{\hspace*{5mm}$\bullet$\hspace*{1mm} \bf Condition~\ref{def:snd-3} of Definition~\ref{def:seq-non-deterministic}.}

We observe that the gadgets of our construction are created such
  that, if a gadget has an input tile, then its input glue binds to, and only
  to, the output glue of some other gadget. Similarly, for each gadget, if
  the gadget has an output tile, then its output glue is such that there is a
  unique gadget with a matching input glue.
In other words, the gadgets self-assemble in a strict sequential order.
It follows that, for all $\mathcal{T}$-producing assembly sequences
$\vec{\alpha}$, if $Y \subseteq \dom{\res{\vec{\alpha}}}$, then for
all integers $1 \leq i < |Y|$, $\textmd{index}_{\vec{\alpha}}\left(
\vec{y}_i \right) < \textmd{index}_{\vec{\alpha}}\left( \vec{y}_{i+1}
\right)$.
Thus, $\mathcal{T}$ satisfies condition~\ref{def:snd-3} of Definition~\ref{def:seq-non-deterministic}.

\paragraph{\hspace*{5mm}$\bullet$\hspace*{1mm} \bf Condition~\ref{def:snd-4} of Definition~\ref{def:seq-non-deterministic}.}
The fact that $S_P \cap P = \emptyset$ directly follows from the construction of the {\tt Copy\_read} and {\tt Inc\_read} gadgets, since 1) a POC is always a point located in one of the bottom corners of such a gadget and the corresponding starting point is always located in the opposite bottom corner of the same gadget and 2) the intersection of the respective sets of points covered by any two distinct read gadgets is always empty.

\emph{Step 5.}  First, since the final gadget to assemble in $\alpha$
is the {\tt Last} gadget (see Figure~\ref{fig:last}) and its
upper-right corner both coincides with the upper-right corner of
$\alpha$ and is the point on the perimeter of $\alpha$ where the last
tile of this gadget attaches, it follows that every $w$-correct
$\mathcal{T}$-producing assembly sequence resulting in $\alpha$
attaches the last tile at the upper-right corner of $\alpha$.

Second, we check for the absence of an outward-facing, positive-strength
glue in all of the tiles placed on the perimeter of $\alpha$ in any
$w$-correct $\mathcal{T}$-producing assembly sequence resulting in
$\alpha$.

\begin{itemize}[itemsep=0mm]
\item {\bf Tiles on the perimeter of $\alpha$ with an exposed south
  side:} None of the tiles with an exposed south side in the {\tt Seed}
  gadget (see Figure~\ref{fig:seed-gadgets}) or {\tt Init} gadgets
  (see Figure~\ref{fig:init-gadgets}) has a positive-strength glue on
  this side.
\item {\bf Tiles on the perimeter of $\alpha$ with an exposed west
  side:} None of the tiles with an exposed west side in the {\tt Seed}
  gadget (see Figure~\ref{fig:seed-gadgets}), {\tt Inc\_to\_copy}
  gadget (see Figure~\ref{fig:Inc-to-copy-gen}), or {\tt
    Inc\_to\_copy\_last} gadget (see
  Figure~\ref{fig:Inc-to-copy-last}) has a positive-strength glue on
  this side.
\item {\bf Tiles on the perimeter of $\alpha$ with an exposed north
  side:} None of the tiles with an exposed north side in the {\tt
  Inc\_to\_copy\_last} gadget (see Figure~\ref{fig:Inc-to-copy-last}),
  {\tt Copy\_write} gadgets (see Figure~\ref{fig:Copy-write}), or {\tt
    Last} gadget (see Figure~\ref{fig:last}) has a positive-strength
  glue on this side.
\item {\bf Tiles on the perimeter of $\alpha$ with an exposed east
  side:} None of the tiles with an exposed east side in the {\tt Last}
  gadget (see Figure~\ref{fig:last}), {\tt Copy\_to\_inc} gadget (see
  Figure~\ref{fig:copy-to-inc}), {\tt Init\_right\_0} gadget (see
  Figure~\ref{fig:init-gadgets}(c)), {\tt Init\_right\_1} gadget (see
  Figure~\ref{fig:init-gadgets}(e)), or {\tt Seed} gadget (see
  Figure~\ref{fig:seed-gadgets}) has a positive-strength glue on this
  side. Note that, while the output tile in the {\tt Seed} gadget has
  a positive-strength, outward-facing glue on its east side, this side
  is not exposed in $\alpha$ because this output tile binds with the
  input tile in the {\tt Init\_left} gadget (see
  Figure~\ref{fig:init-gadgets}(a)).
\end{itemize}

Therefore, in every $w$-correct $\mathcal{T}$-producing assembly
sequence resulting in $\alpha$, no tile placed on the perimeter of
$\alpha$ has an outward-facing, positive-strength glue.

\emph{Step 6.}
We now prove that $\mathcal{T}$ strictly self-assembles $\alpha$ with probability at least $1 - \delta$.
\emph{Step 6a.}
We first show that for every competition $\mathcal{C}$ of $\mathcal{T}$, $\textmd{Pr}[\mathcal{C}] \geq 1 - \frac{\delta}{N^2}$.  
Let $\mathcal{C}$ be a competition in $\mathcal{T}$ with corresponding competing winning and losing paths $\piw$ and $\pil$, respectively.
Note that $\mathcal{C}$ corresponds to either an {\tt Inc\_read} or a {\tt Copy\_read} gadget.
We distinguish two cases based on the correct value of the bit to be
read. If the correct bit value for $\mathcal{C}$ is 1, then the losing
path $\pil$ is blocked by a {\tt Write} gadget that already assembled
in the previous row of the counter. In this case, $\mathcal{C}$ is 
rigged in $\mathcal{T}$ and the {\tt Read} gadget is guaranteed to read the
correct bit value. Thus, $\textmd{Pr}[\mathcal{C}] = 1 \geq 1 -
\frac{\delta}{N^2}$. In contrast, in the case where the correct bit
value for $\mathcal{C}$ is 0, then $\mathcal{C}$ is not rigged in $\mathcal{T}$.  By
the construction of the {\tt Inc\_read} (see
Figure~\ref{fig:Inc-read}) and {\tt Copy\_read} (see
Figure~\ref{fig:Copy-read}) gadgets, the length of $\piw$ and $\pil$
are $l_{\textmd{win}} = 4$ and $l_{\textmd{lose}} = 4h-10$, respectively.  Furthermore, by
construction of the {\tt Inc\_write} (see Figure~\ref{fig:Inc-write})
and {\tt Copy\_write} (see Figure~\ref{fig:Copy-write}) gadgets, each
row has a height $h$ that is greater than or equal to 6.  We now show
that, in this case, $\textmd{Pr}[\mathcal{C}] \geq 1 -
\frac{\delta}{N^2}$ also holds. Our proof will use the following two
lemmas.

\begin{lemma} 
  \label{lem:closed-form-PrC}
  $\forall n \in \mathbb{Z}^+\backslash\{1\}$,
  $\displaystyle\sum_{i=0}^{n-2}  \frac{(i+2)(i+1)}{2^{i+4}} =
  1- \frac{n^2+3n+4}{2^{n+2}}$.
\end{lemma}

\begin{proof}
  We proceed by induction on $n$.

  {\bf Basis step} When $n=2$,
  $\displaystyle\sum_{i=0}^{2-2}  \frac{(i+2)(i+1)}{2^{i+4}} = \frac{1}{8}$ and
  $\displaystyle 1- \frac{2^2+3\cdot 2+4}{2^{2+2}}=1-\frac{14}{16}=\frac{1}{8}$. Thus the lemma holds for $n=2$.

  {\bf Inductive step} Assume that $\displaystyle\sum_{i=0}^{k-2}
  \frac{(i+2)(i+1)}{2^{i+4}} = 1- \frac{k^2+3k+4}{2^{k+2}}$ for an
  arbitrary $k \in \mathbb{Z}^+\backslash\{1\}$. We now prove that the
  lemma holds for $n=k+1$.

\[
\begin{array}{llll}
\displaystyle\sum_{i=0}^{(k+1)-2}
  \frac{(i+2)(i+1)}{2^{i+4}} & = & \displaystyle\sum_{i=0}^{k-2}
  \frac{(i+2)(i+1)}{2^{i+4}}+\frac{(k+1)k}{2^{k+3}} & \\
  & = & \displaystyle 1- \frac{k^2+3k+4}{2^{k+2}} + \frac{(k+1)k}{2^{k+3}} & \textmd{ inductive hypothesis} \vspace*{1mm}\\
  & = & \displaystyle 1- \frac{2(k^2+3k+4)-(k+1)k}{2^{k+3}} &\vspace*{1mm}\\
  & = & \displaystyle 1- \frac{k^2+5k+8}{2^{k+3}} &\vspace*{1mm}\\
  & = & \displaystyle 1- \frac{(k+1)^2+3(k+1)+4}{2^{(k+1)+2}} &\\    
\end{array}
\]
\end{proof}

\begin{lemma} 
  \label{lem:inequality-for-PrC}
$\forall h \in \mathbb{N}$ such that $h\geq 6$, $2^{3h-8} \geq 16h^2 -68h + 74$.
\end{lemma}

\begin{proof}
  We proceed by induction on $h$.

  {\bf Basis step} When $h=6$, $2^{3h-8}=2^{10}=1024$ and $16h^2 -
  68h + 74 = 242$. Thus the lemma holds for $h=6$.

  {\bf Inductive step} Assume that $2^{3k-8} \geq 16k^2 -68k + 74$ for
  an arbitrary $k \in \mathbb{N}$ such that $k\geq 6$. We now prove
  that the lemma holds for $h=k+1$.

\[
\begin{array}{llll}
  2^{3(k+1)-8} & = & 8\cdot  2^{3k-8} &\\
  & \geq & 8\cdot(16k^2 -68k + 74) & \textmd{ inductive hypothesis}\\
  & = & 8\cdot 16k^2 -544k + 592 & \\
  & = & 16k^2 + (7\cdot 16k^2 -544k) + 592 & \\
  & \geq & 16k^2 + (7\cdot 16\cdot 6k -544k) + 592 & \textmd{ } k \geq 6\\
  & = & 16k^2 + 128k + 592 & \\
  & \geq & 16k^2 -36k + 22 &  \textmd{ } k > 0\\
  & = & 16(k+1)^2 - 68(k+1) + 74.&\\
\end{array}
\]
\end{proof}

We are now ready to prove our lower bound on  Pr[$\mathcal{C}$] as follows.

\[
\begin{array}{llll}
  \textmd{Pr}[\mathcal{C}] & = &
  \displaystyle\sum_{i=0}^{l_{\textmd{lose}}-2}
  \left({l_{\textmd{win}}+i-2 \choose l_{\textmd{win}}-2}\cdot \left(\frac{1}{2}\right)^{l_{\textmd{win}}+i-1}\right)
    & \textmd{ Corollary~\ref{cor:competition-probability-sum}} \\
  & = &
  \displaystyle\sum_{i=0}^{l_{\textmd{lose}}-2}\left({i+2 \choose 2}\cdot \left(\frac{1}{2}\right)^{i+3}\right) & \textmd{ } l_{\textmd{win}} = 4 \\

  & = &
  \displaystyle\sum_{i=0}^{l_{\textmd{lose}}-2}
  \frac{(i+2)(i+1)}{2^{i+4}} & \\
  & = &
  \displaystyle 1 - \frac{\left(l_{\textmd{lose}}\right)^2 + 3l_{\textmd{lose}} + 4}{2^{l_{\textmd{lose}}+2}}
  & \textmd{ Lemma~\ref{lem:closed-form-PrC}}\vspace*{1mm}\\
  & = &
  \displaystyle 1 - \frac{(4h-10)^2 + 3(4h-10) + 4}{2^{4h-8}}
  & \textmd{ } l_{\textmd{lose}}=4h-10\vspace*{1mm}\\
  & = &
  \displaystyle 1 - \frac{16h^2 - 68h + 74}{2^{4h-8}}
  & \vspace*{1mm}\\
  & \geq &
  \displaystyle 1 - \frac{2^{3h-8}}{2^{4h-8}}
  & \textmd{  Lemma~\ref{lem:inequality-for-PrC} and } h\geq 6 \vspace*{1mm} \\
  & = & \displaystyle 1 - \frac{1}{2^h}\\
  & = & \displaystyle 1 - \frac{1}{2^{\left\lceil 2\log N + \log \frac{1}{\delta} \right \rceil + 5}}
  & \textmd{ definition of } h\\
  & \geq
  & \displaystyle 1 - \frac{1}{2^{2\log N + \log \frac{1}{\delta}}} & \vspace*{1mm}\\
  & = & \displaystyle 1 - \frac{\delta}{N^2}.&
\end{array}
\]
In conclusion,  for every competition $\mathcal{C}$ of $\mathcal{T}$, $\textmd{Pr}[\mathcal{C}] \geq 1 - \frac{\delta}{N^2}$.

\emph{Step 6b.}
Finally, we invoke Theorem~\ref{thm:local-non-determinism-theorem} and bound the resulting product of probabilities using Bernoulli's inequality.
Let $r = 4km \leq N^2$ and $\mathcal{C}_1,\ldots, \mathcal{C}_r$ be the competitions in $\mathcal{T}$. 
Then, by Theorem~\ref{thm:local-non-determinism-theorem}, $\mathcal{T}$ strictly self-assembles $\dom{\alpha} = L_e \cup R_{K,n,e}$ with probability at least:
\[
\begin{array}{llll}
\displaystyle\prod_{i=1}^{r}{\textmd{Pr}\left[ \mathcal{C}_i \right]} & \geq & \displaystyle \left( 1 - \frac{\delta}{N^2} \right)^r & \\
 & \geq & \displaystyle \left( 1 - \frac{\delta}{r} \right)^r & \textmd{ }r \leq N^2 \\
	& \geq & \displaystyle 1 + r\left( \frac{-\delta}{r} \right) & \textmd{ Bernoulli's inequality with even $r$} \\
	& = & 1 - \delta. &
\end{array}
\]

\subsection{Proof of the main theorem}

Lemma~\ref{lem:high-probability-counter} can be used to prove the following theorem:
\begin{theorem}
\label{thm:high-probability-squares}
For each $N \in \mathbb{Z}^+$ and real $\delta$ such that $0<\delta<1$, there exists a TAS $\mathcal{T}_{N,\delta} = \left(T_{N,\delta}, \sigma_{N,\delta}, 1\right)$ that strictly self-assembles $\{0,\ldots,N-1\}^2$ with probability at least $1 - \delta$ and such that $|T_{N,\delta}| = O\left( \log N + \log \frac{1}{\delta} \right)$.
\end{theorem}
The construction of $\mathcal{T}_{N,\delta}$ follows from three invocations of Lemma~\ref{lem:high-probability-counter}, a high-level depiction of which is shown in Figure~\ref{fig:square_construction_high_level}.
\begin{figure}[h!]
    \centering
        \centering
        \includegraphics[width=.4\linewidth]{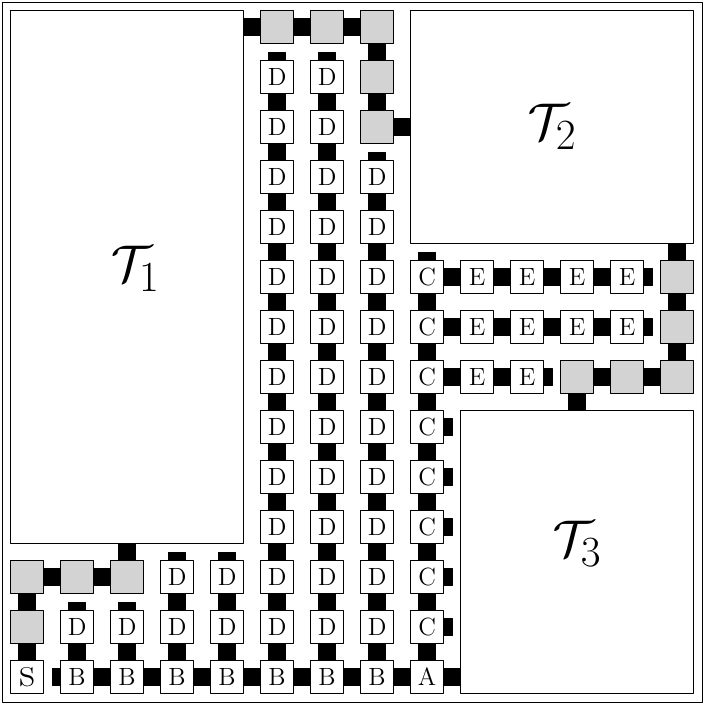}
        \caption{\label{fig:square_construction_high_level} Construction of $\mathcal{T}_{N,\delta}$ via three invocations of  Lemma~\ref{lem:high-probability-counter} resulting in $\mathcal{T}_1$, $\mathcal{T}_2$ and $\mathcal{T}_3$, respectively. The latter two are rotated clockwise by 90 and 180 degrees, respectively, and then translated appropriately. The `S' tile in the lower left corner denotes $\sigma_{N,\delta}$. The A, B, C, D and E new tile types suffice to fill in the interior region of the square.}
\end{figure}
Note that the details of a full proof of Theorem~\ref{thm:high-probability-squares} can be inferred from both Figure~\ref{fig:square_construction_high_level} and the proof of Lemma~\ref{lem:high-probability-counter}.
Therefore, we merely give a proof sketch for Theorem~\ref{thm:high-probability-squares}.
\begin{proofsketch}
To prove Theorem~\ref{thm:high-probability-squares}, proceed as follows:
\begin{enumerate}[itemsep=0mm]
	\item Handle the corner cases when $N$ is too small or $\delta$ is set too close to $0$.
	\item Invoke Lemma~\ref{lem:high-probability-counter} three times to get $\mathcal{T}_1$, $\mathcal{T}_2$, and $\mathcal{T}_3$, respectively, making sure that the resulting tile sets $T_1$, $T_2$, and $T_3$ are pairwise disjoint.
	\item Define $\sigma_{N,\delta}$ to be the seed assembly of $\mathcal{T}_1$.
	\item Rotate $\mathcal{T}_2$ and $\mathcal{T}_3$ clockwise by 90 and 180 degrees, respectively, and then translate their rotated seed assemblies appropriately.
	\item Concatenate $\mathcal{T}_2$ to $\mathcal{T}_1$ in the following sense:
	\begin{enumerate}[label=\theenumi\alph*.,itemsep=0mm]
	
		\item Modify the last tile placed by $\mathcal{T}_1$ to expose a positive strength glue pointing east.
		\item Modify the {\tt Seed} gadget of $\mathcal{T}_2$ to expose a positive strength glue pointing west (after rotation) that matches the exposed glue of the last tile placed by $\mathcal{T}_1$.
		\item Concatenate $\mathcal{T}_3$ to $\mathcal{T}_2$ in a similar fashion.
	\end{enumerate}
	\item Let $\mathcal{T}'_1$, $\mathcal{T}'_2$, and $\mathcal{T}'_3$ be the modified versions of $\mathcal{T}_1$, $\mathcal{T}_2$, and $\mathcal{T}_3$, respectively.
	\item Let $T_{N,\delta}$ be the union of the tile sets of $\mathcal{T}'_1$, $\mathcal{T}'_2$, and $\mathcal{T}'_3$.
	\item Add to $T_{N,\delta}$ five additional tile types to fill in the interior region of the $N \times N$ square. 
	\item Define the set of essential POCs $Y=P$ and a winner function $w:Y\rightarrow T_{N,\delta}$ such that the domain of the unique $\mathcal{T}_{N,\delta}$-terminal assembly is $\{0,\ldots, N-1\}^2$.
	\item Prove that $\mathcal{T}_{N,\delta}$ is $w$-sequentially deterministic.
	\item Prove that $\mathcal{T}_{N,\delta}$ strictly self-assembles $\{0,\ldots,N-1\}^2$ with probability at least $1 - \delta$ using the facts that 1) $h$ was set sufficiently large in Lemma~\ref{lem:high-probability-counter} and 2) the proof step 6b above carries over when the value of $r$ is multiplied by 3.
	\item Conclude that $\left| T_{N,\delta} \right| = O\left( \log N + \log \frac{1}{\delta} \right)$ because $\left| T_{N,\delta} \right|$ is equal to five plus the sum of the cardinalities of the respective tile sets of $\mathcal{T}'_1$, $\mathcal{T}'_2$, and $\mathcal{T}'_3$, each one of which is $O\left( \log N + \log \frac{1}{\delta} \right)$.
\end{enumerate}
\end{proofsketch}

\section{Conclusion}
In this paper, we introduced the notion of sequential non-determinism, which is a general framework for analyzing the probabilistic correctness of a TAS in the PTAM that has a unique correct terminal assembly and an arbitrary number of (potentially infinite) incorrect assemblies. 
We then demonstrated that it is a sufficiently general framework to allow the efficient self-assembly of an $N \times N$ square at temperature 1 with high probability, i.e., Theorem~\ref{thm:high-probability-squares}.
Note that our sequential non-determinism framework allows for an infinite set of POCs.
This means our framework is applicable to a counter construction whose correct assembly is finite, but where the counter may erroneously double counts arbitrarily many times and therefore must be supported by infinitely many inessential POCs.
However, observe that, even if our framework did require the set of POCs to be a finite subset of the domain of the unique correct terminal assembly of the counter, then it would  still be possible to implement constructions that utilize a counter.
The idea is to geometrically represent the binary value 0 in the usual manner but immediately followed by the geometric representation of the binary value 1, thhee latter bTeing the error detection bit.
Then, the gadgets that non-deterministically guess the binary value of a geometrically-encoded bit are modified as follows:
\begin{enumerate}
	\item If a geometrically-represented bit has the value 0 and is guessed correctly, then a path of tiles is hard-coded to self-assemble up and over the corresponding error detection bit (whose value is 1).
	\item If a geometrically-represented bit has the value 0 and is guessed incorrectly, then a path of tiles is hard-coded to self-assemble in such a way that it will not try to avoid the corresponding error detection bit (whose value is 1), and therefore be blocked.
\end{enumerate}
Note that the binary value 1 would be represented geometrically in the same fashion, where 1 is represented, followed by 0, but in this case,  a 1 bit cannot be guessed incorrectly. 
It is possible to implement such an error detection scheme by doubling the width of all the reading and writing gadgets for a counter construction to encode the additional error detection bit.
If the binary value 0 was ever guessed incorrectly, then the counter would terminate early before even reading the next bit value and thus would never erroneously double count the same value an arbitrary number of times and not require an infinite set of POCs.
Moreover, implementing such an error detection scheme would not result in an asymptotic increase in the number of tile types over the size of a corresponding standard counter and could be utilized in other non-counter constructions wherever sequences of binary values need to be guessed through self-assembly. 
Figure~\ref{fig:bit-read-error-detection} shows conceptual examples of how such an error detection scheme can be implemented.
\begin{figure}[h!]
    \centering
    \begin{subfigure}[t]{0.45\textwidth}
        \centering
        \includegraphics[width=\textwidth]{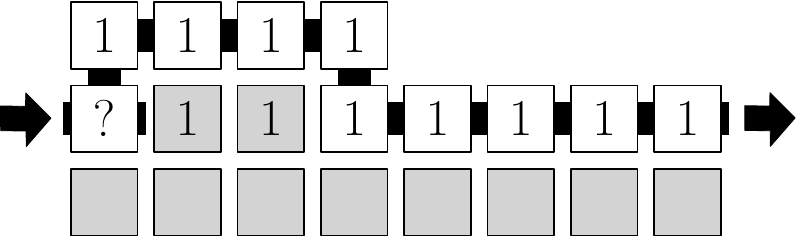}
        \caption{\label{fig:bit-read-error-detection-1-good} If the bit being guessed has the binary value 1, then it cannot be guessed incorrectly. }
    \end{subfigure}
    
    \vspace{10pt}
    \begin{subfigure}[t]{0.45\textwidth}
        \centering
        \includegraphics[width=\textwidth]{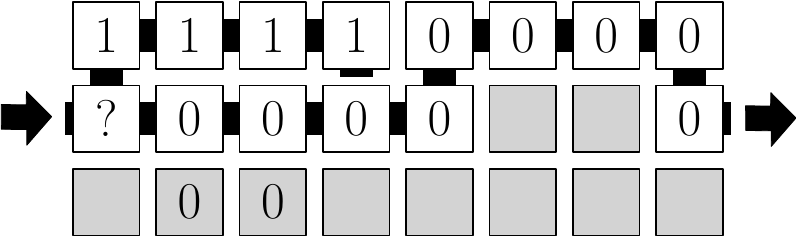}
        \caption{\label{fig:bit-read-error-detection-0-good} In this example, the bit being guessed has the binary value 0, which is guessed correctly. Thus, the `0' path is hard-coded to avoid the error-detection bit, i.e., the protrusion of tiles that represent the opposite of the binary value of the bit that was just guessed. }
    \end{subfigure}
    ~\hspace{20pt}
    \begin{subfigure}[t]{0.41\textwidth}
        \centering
        \includegraphics[width=\textwidth]{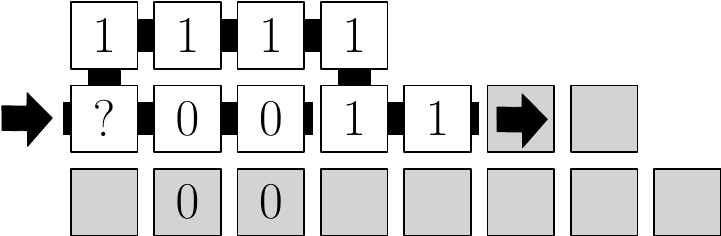}
        \caption{\label{fig:bit-read-error-detection-0-bad} In this example, the bit being guessed has the binary value 0, which is guessed incorrectly. Thus, the `1' path is hard-coded to anticipate the absence of the protrusion of tiles. However, the latter blocks the former, effectively detecting the erroneous guess.  }
    \end{subfigure}
    \caption{\label{fig:bit-read-error-detection} Non-deterministically guessing the binary value of a geometrically-encoded bit and detecting incorrect guesses. In this example, the geometrically-encoded bits are the grey tiles labeled with the binary values 0 or 1. The idea is that each bit is encoded in the usual manner with an upward protrusion representing the binary value 1 and a 0 represented with the absence of such a protrusion. Then, in a similar fashion, the opposite of the encoded bit is geometrically encoded a fixed distance away. In these examples, the grey tiles represent the geometrically-encoded bit and its opposite. The tile labeled with the `?' initiates the process of non-deterministically guessing the bit. }
\end{figure}
While the sequential non-determinism framework introduced here is sufficiently general to allow non-trivial self-assembly, as our Theorem~\ref{thm:high-probability-squares} demonstrates, we speculate that it can be generalized even further. 
For example, it remains future work to determine whether our sequential non-determinism framework can be extended to analyze the probabilistic correctness of the TASes used in \cite{CookFuSch11} for the probabilistic  assembly of squares and the efficient simulation of any time bounded Turing machines. However, we can already describe one direct application of our framework to this prior work.

Consider the self-assembly gadget illustrated in Figure~\ref{fig:conclusion}, which is in essence equivalent to the gadget depicted in Figure~4 of \cite{CookFuSch11}. This gadget exhibits multiple overlapping paths of tiles that self-assemble from the point of the 0 tile to the opposite point across the gap.
	\begin{figure}[h!]
    \centering
    \begin{subfigure}[t]{0.4\textwidth}
        \centering
        \includegraphics[width=\textwidth]{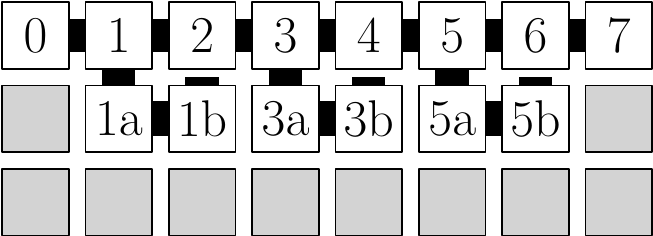}
        \caption{\label{fig:conclusion-1} Incorrectly placing the 7 tile. }
    \end{subfigure}
    ~
    \begin{subfigure}[t]{0.4\textwidth}
        \centering
        \includegraphics[width=\textwidth]{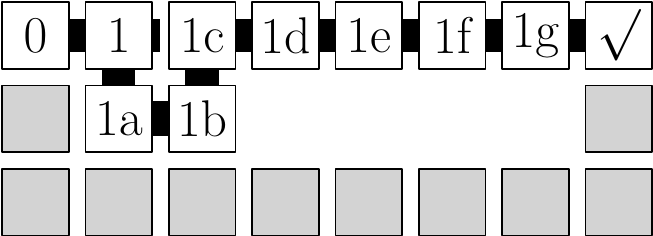}
        \caption{\label{fig:conclusion-3} First way to correctly place the $\surd$ tile. }
    \end{subfigure}

        \begin{subfigure}[t]{0.4\textwidth}
        \centering
        \includegraphics[width=\textwidth]{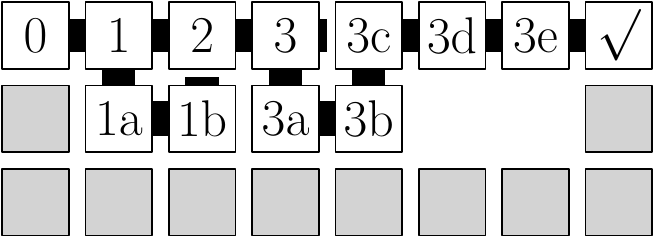}
        \caption{\label{fig:conclusion-2} Second way to correctly place the $\surd$ tile. }
    \end{subfigure}
        ~
        \begin{subfigure}[t]{0.4\textwidth}
        \centering
        \includegraphics[width=\textwidth]{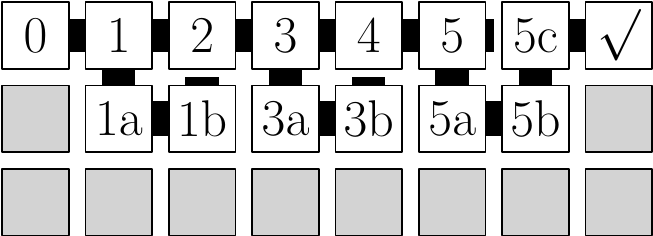}
        \caption{\label{fig:conclusion-4} Third way to correctly place the $\surd$ tile. }
    \end{subfigure}

        \caption{\label{fig:conclusion} This figure depicts the only
          four terminal assemblies for some TAS $\mathcal{T}$ in which
          competing paths all start from the top-left corner point and
          compete to place a tile at the top-right corner point. A
          winning path is any path that self-assembles to the right,
          then down, back up, and eventually places the correct tile
          $\surd$ at the top-right corner point.  (a) The only
          $\mathcal{T}$-terminal assembly that places an incorrect
          tile (namely, the 7 tile) at the top-right corner point. (b)
          through (d) The only three $\mathcal{T}$-terminal assemblies
          that place the correct tile at the top-right corner
          point.  
        }
\end{figure}
	Figure~\ref{fig:conclusion-1} depicts the unique $\mathcal{T}$-terminal assembly containing a path placing the wrong type of tile across the gap, whereas Figures~\ref{fig:conclusion-3} through~\ref{fig:conclusion-4} depict all of the  $\mathcal{T}$-terminal assemblies containing a path placing the correct type of tile across the gap.
	Therefore, this gadget exhibits a kind of generalized competition where the point across the gap can be ``won'' by exactly three out of four competing paths that all self-assemble within a fixed region.
Thus, unlike sequential non-determinism, any type of framework that applies to the gadget in Figure~\ref{fig:conclusion} would not admit a unique correct terminal assembly but rather a set of such terminal assemblies. 

However, careful study of the gadget in Figure~\ref{fig:conclusion} reveals that its TAS $\mathcal{T}$ has a finite set $P=Y=\{\vec{y}_1,\vec{y}_2,\vec{y}_3 \}$ of POCs with a corresponding set of starting points $S_Y=\{\vec{x}_1,\vec{x}_2,\vec{x}_3 \}$, where the elements of $Y$ are the points where the 2, 4, and 6 tiles are placed, respectively, and  the elements of $S_Y$ are the points where the 1, 3, and 5 tiles are placed, respectively. Furthermore, each one of the three pairs of competing paths for $\mathcal{T}$ is essentially identical to the pair of competing paths depicted in Figure~\ref{fig:intro-example1} after a 180-degree rotation. After checking that  $\mathcal{T}$ meets all of the requirements of sequential non-determinism with $r=3$, we can apply our framework to prove that the probability of the correct assembly of the gadget is indeed equal to the stated value of $1 - \left(\frac{7}{8}\right)^3$. Without the sequential non-determinism framework, doing so would involve the extremely tedious consideration of many tile attachment steps that do not affect the result of the computation, for example the placement of the 1a, 1b, 3a, and 3b tiles in the case of Figure~\ref{fig:conclusion-4}. With our framework, we can take advantage of the existence of a unique incorrect assembly, namely the one in  Figure~\ref{fig:conclusion-1}) to:
\begin{itemize}
  \item let $\alpha$ be the assembly shown in Figure~\ref{fig:conclusion-1},
\item define $w$ to be the winner function $\{\left(\vec{y}_1,2\right), \left(\vec{y}_2,4\right), \left(\vec{y}_3,6\right)\}$ and verify that $\mathcal{T}$ is $w$-sequentially non-deterministic,
\item compute the probability of each competition, each one involving only 4 tiles just like in Figure~\ref{fig:intro-example1-noAB}, to yield the same value $\frac{7}{8}$ for all three competitions,
\item use Theorem~\ref{thm:local-non-determinism-theorem} to compute a lower bound on the probability of the incorrect assembly of the gadget, namely $\frac{7}{8}\times \frac{7}{8} \times \frac{7}{8} = \left(\frac{7}{8}\right)^3$,
\item   use the facts that none of the competitions is rigged and only one $\mathcal{T}$-terminal assembly has a shape equal to $\dom{\alpha}$ to infer that the lower bound given by our theorem is tight, i.e., the probability of the incorrect assembly of the gadget is equal to $\left(\frac{7}{8}\right)^3$, and
\item infer the probability of correct assembly of the gadget in  Figure~\ref{fig:conclusion} using the complement rule, to yield $1-\left(\frac{7}{8}\right)^3$.
\end{itemize}
It remains future work to extend our framework to the study of the probabilistic correctness of TASes that assemble shapes by sequencing an arbitrary number of such gadgets as was done in \cite{CookFuSch11}.

Finally, Theorem~\ref{thm:high-probability-squares} improves upon the previous state-of-the-art bound by Cook, Fu and Schweller \cite{CookFuSch11}, but our result is still not optimal from an information-theoretic perspective.
After all, Cook, Fu and Schweller proved that for almost all positive integers $N$, if $\mathcal{T}=(T,\sigma,1)$ is a TAS that strictly self-assembles an $N \times N$ square with probability greater than $\frac{1}{2}$, then $|T| = \Omega\left(\frac{\log N}{\log \log N}\right)$.
It is future work to derermine 1) whether such an asymptotic lower bound can be achieved in the 2DPTAM and 2) the extent to which sequential non-determinism would support or hinder such a result.

\section*{Acknowledgement}
We thank Matthew Patitz for providing helpful insights regarding the implementation of several  aspects of our construction.  

\bibliographystyle{amsplain}
\bibliography{tam}

\end{document}